\documentclass[reqno,hidelinks]{amsart}

\usepackage{a4,amsmath,amssymb,amsthm,slashed}
\usepackage[top=1in, bottom=1in, left=1.25in, right=1.25in]{geometry}
\usepackage{xfrac}
\usepackage{latexsym}
\usepackage{mathrsfs}
\usepackage[numbers]{natbib}
\usepackage{leftidx}
\usepackage{inputenc}

\usepackage{scrextend}

\usepackage{calc}
\usepackage{enumitem}
\usepackage{amsaddr}

\usepackage[T1]{fontenc}

\usepackage{leftidx}
\usepackage{tikz}

\usepackage{hyperref}

\usepackage{url}

\numberwithin{equation}{section}

\usepackage[mathscr]{eucal}

\usepackage{graphicx}
\graphicspath{pic/}

\newcommand{\sfrak}{\mathfrak{s}}

\def\nablaslash{\mbox{$\nabla \mkern -13mu /$ \!}}

\newcommand{\B}{\mathbf}
\newcommand{\half}{\frac{1}{2}}

\newcommand{\veps}{\varepsilon}

\newcommand{\R}{r^2+a^2}
\newcommand{\PR}{r^3-3Mr^2+a^2r+a^2M}

\newcommand{\prb}{\partial_{\rb}}
\newcommand{\di}{\mathrm{d}} 

\newcommand{\Donetwo}{\DOC_{\tb_1,\tb_2}}

\newcommand{\Dtwoinfty}{\DOC_{\tb_2,\infty}}

\newcommand{\NPR}{{\Upsilon}}
\newcommand{\NPRplus}{\NPR_{+1}}
\newcommand{\NPRminus}{\NPR_{-1}}
\newcommand{\NPRpluss}{\NPR_{+\sfrak}}
\newcommand{\NPRminuss}{\NPR_{-\sfrak}}

\newcommand{\psiplus}{\psi_{+1}}
\newcommand{\psiminus}{\psi_{-1}}

\newcommand{\Psiplus}{\Psi_{+1}}
\newcommand{\Psiminus}{\Psi_{-1}}
\providecommand{\PsiminusHigh}[1]{\Psiminus^{(#1)}}

\newcommand{\psipluss}{\psi_{+\sfrak}}
\newcommand{\psiminuss}{\psi_{-\sfrak}}

\newcommand{\Psipluss}{\Psi_{+\sfrak}}
\newcommand{\Psiminuss}{\Psi_{-\sfrak}}

\providecommand{\PsiminussHigh}[1]{\Psiminuss^{(#1)}}

\newcommand{\Phiplus}{\Phi_{+1}^{(0)}}
\providecommand{\Phiminus}[1]{\Phi_{-1}^{(#1)}}

\newcommand{\Phipluss}{\Phi_{+\sfrak}}
\providecommand{\Phiminuss}[1]{\Phi_{-\sfrak}^{(#1)}}

\providecommand{\PhiplussHigh}[1]{\Phi_{+\sfrak}^{(#1)}}

\providecommand{\tildePhiplussHigh}[1]{\widetilde{\Phi}_{+\sfrak}^{(#1)}}

\newcommand{\NPRplustwo}{\NPR_{+2}}
\newcommand{\NPRminustwo}{\NPR_{-2}}

\newcommand{\psiplustwo}{\psi_{+2}}
\newcommand{\psiminustwo}{\psi_{-2}}

\newcommand{\Psiplustwo}{\Psi_{+2}}
\newcommand{\Psiminustwo}{\Psi_{-2}}
\providecommand{\PsiminustwoHigh}[1]{\Psiminustwo^{(#1)}}

\newcommand{\Phiplustwo}{\Phi_{+2}^{(0)}}
\providecommand{\Phiminustwo}[1]{\Phi_{-2}^{(#1)}}

\providecommand{\dotPhiminus}[1]{\dot{\Phi}_{-1}^{(#1)}}
\providecommand{\dotPhiminustwo}[1]{\dot{\Phi}_{-2}^{(#1)}}
\providecommand{\dotPhisHigh}[1]{\dot\Phi_{-\sfrak}^{(#1)}}

\newcommand{\Lxi}{\mathcal{L}_{\xi}}
\newcommand{\Leta}{\mathcal{L}_{\eta}}

\newcommand{\curlV}{\mathcal{V}}

\newcommand{\pu}{\partial_u}
\newcommand{\pv}{\partial_v}

\newcommand{\DOC}{\mathcal{D}}
\newcommand{\tb}{\tau}
\newcommand{\pb}{{\tilde\phi}}
\newcommand{\rb}{\rho}
\newcommand{\Hyper}{\Sigma}

\newcommand{\Sigmazero}{\Hyper_{\tb_0}}
\newcommand{\Sigmatb}{\Hyper_{\tb}}
\newcommand{\Sigmatwo}{\Hyper_{\tb_2}}
\newcommand{\Sigmaone}{\Hyper_{\tb_1}}

\newcommand{\Horizon}{\mathcal{H}^+}
\newcommand{\Scri}{\mathcal{I}^+}
\newcommand{\Horizononetwo}{\Horizon_{\tb_1,\tb_2}}
\newcommand{\Scrionetwo}{\Scri_{\tb_1,\tb_2}}

\newcommand{\CDeri}{\mathbb{D}}

\newcommand{\PDeri}{\mathbb{B}}
\newcommand{\PSDeri}{\widetilde{\mathbb{B}}}

\newcommand{\PriceDeri}{\mathbb{P}}

\newcommand{\VR}{\hat{V}}
\newcommand{\curlVR}{\hat{\mathcal{V}}}
\newcommand{\edthR}{\mathring{\eth}}

\newcommand{\Sphere}{S^2}

\providecommand{\abs}[1]{\lvert#1\rvert}
\providecommand{\norm}[1]{\lVert#1\rVert}

\providecommand{\absHighOrder}[3]{\abs{#1}_{#2,#3}}

\providecommand{\absCDeri}[2]{\absHighOrder{#1}{#2}{\CDeri}}

\makeatletter
  \def\moverlay{\mathpalette\mov@rlay}
  \def\mov@rlay#1#2{\leavevmode\vtop{%
     \baselineskip\z@skip \lineskiplimit-\maxdimen
     \ialign{\hfil$#1##$\hfil\cr#2\crcr}}}
\makeatother

\newcommand{\squareS}{\moverlay{\square\cr {\scriptscriptstyle \mathrm S}}}
\newcommand{\Boxhat}{\widehat{\squareS}}

\newcommand{\reg}{k}

\newcommand{\regl}{k'}

\newcommand{\hhyp}{h_{\text{hyp}}}
\newcommand{\Hhyp}{H_{\text{hyp}}}

\newcommand{\TMEOp}{\mathbb{T}}
\providecommand{\PhisHigh}[1]{\Phi_{s}^{(#1)}}

\providecommand{\ellmode}[2]{(#1)_{#2}}
\providecommand{\mellmode}[3]{(#1)_{#2,#3}}
\providecommand{\Proj}[1]{\mathbf{P}_{#1}}
\newcommand{\Scaling}{w}
\providecommand{\hatPhiminuss}[1]{\hat{\Phi}_{-\sfrak}^{(#1)}}
\providecommand{\hatPhiplussHigh}[1]{\hat{\Phi}_{+\sfrak}^{(#1)}}
\providecommand{\hatPhisHigh}[1]{\hat{\Phi}_{s}^{(#1)}}
\newcommand{\varphipluss}{\varphi_{+\sfrak}}

\providecommand{\tildePhisHighell}[2]{\tilde{\Phi}_{#1, #2}}
\providecommand{\tildePhisHighmell}[3]{\tilde{\Phi}_{#1, #2,#3}}

\newcommand{\tildeV}{\tilde{V}}
\newcommand{\tildeY}{\tilde{Y}}
\newcommand{\ptb}{\partial_{\tb}}

\newcommand{\PJ}{\mathbf{P}}
\newcommand{\MQ}{\mathbb{Q}}
\newcommand{\ql}{\mathfrak{q}}
\newcommand{\G}{G_{+\sfrak,m,\sfrak}}

\providecommand{\IE}[2]{\textbf{I}^{#1,\pm \sfrak}_{\text{total}, #2}}

\providecommand{\Comm}[3]{\mathbf{C}_{#1}^{#2}[#3]}

\newcommand{\PigeonTime}{\timefunc}

\newcommand{\timefunc}{\tb}

\newcommand{\Integers}{\mathbb{Z}}
\newcommand{\Reals}{\mathbb{R}}

\def\pbh{\pb_{\Horizon}}

\def\chitrap{\chi_{\text{trap}}}

\theoremstyle{plain}
\newtheorem{thm}{Theorem}[section]
\newtheorem{cor}[thm]{Corollary}
\newtheorem{lemma}[thm]{Lemma}
\newtheorem{prop}[thm]{Proposition}

\theoremstyle{definition}
\newtheorem{definition}[thm]{Definition}
\newtheorem{assump}[thm]{Assumption}
\newtheorem{remark}[thm]{Remark}

\title{Sharp decay for Teukolsky equation in Kerr spacetimes}

\author{Siyuan Ma$^\dagger$ and Lin Zhang$^\star$}
\email{siyuan.ma@aei.mpg.de, lzhang\_math@cqu.edu.cn}
\address{$^\dagger$Laboratoire Jacques-Louis Lions,
Sorbonne Universit\'{e} Campus Jussieu,
4 place Jussieu 75005 Paris, France\\
and\\
Albert Einstein Institute, Am M\"{u}hlenberg 1, 14476 Potsdam, Germany\\
$^\star$College of Mathematics and Statistics, Chongqing University, Chongqing 401331, China.}

\begin{document}



\allowdisplaybreaks

\begin{abstract}

In this work, we derive the global sharp decay, as both a lower and an upper bounds, for the spin $\pm \mathfrak{s}$ components, which are solutions to the Teukolsky equation, in the black hole exterior and on the event horizon of a slowly rotating Kerr spacetime. These estimates are generalized to any subextreme Kerr background under an integrated local energy decay estimate. Our results apply to the scalar field $(\mathfrak{s}=0)$, the Maxwell field $(\mathfrak{s}=1)$ and the linearized gravity $(\mathfrak{s}=2)$ and confirm the Price's law decay that is conjectured to be sharp. Our analyses rely on a novel global conservation law for the Teukolsky equation, and this new approach can be applied to derive the precise asymptotics for solutions to semilinear wave equations.

\end{abstract}

\maketitle

\tableofcontents


\section{Introduction}

A subextreme Kerr black hole spacetime  $(\mathcal{M}, g_{M.a})$ \cite{kerr63} has metric
of  the form
\begin{align}\label{eq:SchwMetricinHHtetrad}
(g_{M.a})_{\mu\nu}= & -2l_{(\mu}n_{\nu)}+2m_{(\mu}\bar{m}_{\nu)},
\end{align}
where $(l^\nu,n^\mu,m^\nu,\bar{m}^\nu)$ is a Hartle--Hawking (H--H) tetrad\footnote{This tetrad is a Newman--Penrose null tetrad satisfying $g(l, n)=-1$, $g(m,\bar{m})=1$ and the other products being zero, and, more importantly, it is a principle null tetrad in the sense that its elements $l^{\nu}$ and $n^{\nu}$ are aligned with the two principal null directions of the Kerr geometry. Further, the fact that this tetrad is a regular on the future event horizon is manifest by expressing the H--H tetrad in a regular coordinate system, say, the ingoing Eddington--Finkelstein coordinates on the future event horizon.} \cite{HHtetrad72} and reads in the Boyer-Lindquist coordinates $(t,r,\theta,\phi)$ \cite{boyer:lindquist:1967}
\begin{align}\label{eq:HHtetrad}
l^\nu &= \frac{1}{\sqrt{2}\Sigma}(r^2+a^2 , \Delta , 0 , a), \,
n^\nu = \frac{1}{\sqrt{2}\Delta} (r^2+a^2 , - \Delta , 0 , a), \,
m^\nu = \frac{1}{\sqrt{2} (r+ia\cos\theta)}(i a \sin{\theta},0 , 1, \frac{i}{\sin{\theta}}),
\end{align}
and $\bar{m}^\nu$ being the complex conjugate of $m^\nu$.
Here, $\Sigma=r^2+a^2\cos^2\theta$,   $\Delta= r^2 -2Mr+a^2$, $M$ is the mass of the black hole, and $a$ is the angular momentum per unit mass satisfying $\abs{a}<M$.
The larger root $r_+=M+\sqrt{M^2-a^2}$ of function $\Delta$ is the location of the event horizon $\mathcal{H}$, and we define the domain of outer communication (DOC), denoted as $\DOC$, of a subextreme Kerr black hole spacetime to be the closure of $\{(t,r,\theta,\phi)\in \mathbb{R}\times (r_+,\infty)\times {S}^2\}$ in the Kruskal maximal extension (see for instance \cite{hawking1973large}). 
We consider in this work only the future Cauchy problem and denote the future event horizon and the future null infinity as $\Horizon$ and $\Scri$, respectively.

In the end, we define $\tb$ to be a hyperboloidal time function such that the level sets of the time function are spacelike hypersurfaces, cross $\Horizon$ regularly, and are aymptotic $\Scri$ for large $r$. We define the coordinate system $(\tb, \rb=r,\theta,\pb)$ as the hyperboloidal coordinates and denote the level sets of $\tb$ as $\Sigmatb$. Further, denote $v$ the forward time. See Section \ref{sect:foliation}.


\subsection{Main results}

Our results are on sharp asymptotics of the spin $s$ components $\NPR_s$, $s=0, \pm 1, \pm 2$, on subextreme Kerr backgrounds. These spin $s$ components can be defined via the Newman--Penrose (N--P)  formalism \cite{newmanpenrose62,newmanpenrose63errata}: the spin $0$ component $\NPR_0$ is the scalar field solving the scalar wave equation $\Box_{g} \NPR_0=0$; the spin $\pm 1$ components are defined by
\begin{align}\label{eq:MaxwellNPcomponentswithnosuperscript}
\NPRplus ={} &\mathbf{F}_{lm} , &
\quad  \NPRminus ={} & \mathbf{F}_{\bar{m}n} ,
\end{align}
with   $\mathbf{F}_{\alpha\beta}$ a real two-form solving the Maxwell equations; and the spin $\pm 2$ components are defined by
\begin{align}\label{def:regularNPComps}
   \NPRplustwo={}&\B{W}_{{l} m{l}m}, &
     \NPRminustwo={}&\B{W}_{{n}\bar{m} {n}\bar{m}},
\end{align}
where $\B{W}_{\alpha\beta\gamma\delta}$ is the Weyl tensor of the linearized gravity.
The lower index $s$  indicates the spin weight, and \emph{throughout this work, we use $s$ for the spin weight and $\sfrak=\abs{s}$.}

Teukolsky \cite{Teukolsky1973I}  found that the scalars
\begin{align}\label{eq:ssc}
\psipluss\doteq{}&\Sigma^{\sfrak}\NPRpluss, &
\psiminuss\doteq{}&\Sigma^{-\sfrak}(r-ia\cos\theta)^{2\sfrak}\NPRminuss,
\end{align}
called as the spin $s$ components as well for simplicity,  satisfy the so-called \textbf{Teukolsky master equation} (TME), or also called Teukolsky equation. See Section \ref{sect:tme} for the form of TME. Our aim of this paper is to derive the sharp decay, as well as the precise asymptotic profiles, of these spin $s$ components solving TME.

\begin{thm}[Global sharp asymptotics for the spin $\pm \sfrak$ components in Kerr spacetimes]\label{mainthm} Let $M>0$, $\sfrak=0,1,2$, and let $\abs{a}<M$ in the case $\sfrak=0$ and let $\abs{a}/M$ be sufficiently small in the cases $\sfrak=1,2$. Let $j\in\mathbb{N}$ and $\tb_0\geq 1$. Assume  the spin $s=\pm \sfrak$ components  $\psi_s$ satisfying the Teukolsky master equation in the Kerr spacetime $(\mathcal{M}, g_{M,a})$ arise from smooth, compactly supported initial data on $\Sigmazero$. Then there exists an $\veps>0$ such that in the DOC, it holds for any $\tb\geq \tb_0$ that
\begin{enumerate}
\item for $r\geq r_+$,
\begin{align}
\label{mainthm:PL:global}
\hspace{3ex}&\hspace{-3ex}\bigg|\ptb^j\bigg((\R)^{-\sfrak}\psipluss
-\frac{2^{2\sfrak+3}}{(2\sfrak+1)(2\sfrak+2)}\frac{v+(2\sfrak+1)\tb}{v^{2\sfrak+2}\tb^2}
\sum_{|m|\leq\sfrak}\mathfrak{f}_{+\sfrak,m}\MQ_{m,\sfrak}Y_{m,\sfrak}^{+\sfrak}(\cos\theta)e^{im\pb}\bigg)\bigg|\notag\\
&\leq {} C_{+\sfrak,j}v^{-2\sfrak-1}\tb^{-2-j-\veps},\\
\hspace{3ex}&\hspace{-3ex}\bigg|\ptb^j\bigg({\psiminuss}
-\frac{2^{2\sfrak+3}}{(2\sfrak+1)(2\sfrak+2)}
\frac{\tb+(2\sfrak+1)v
}{\tb^{2\sfrak+2}v^2}
\sum_{|m|\leq\sfrak}\MQ_{m,\sfrak}Y_{m,\sfrak}^{-\sfrak}(\cos\theta)e^{im\pb}\bigg)\bigg|\notag\\
&\leq {} C_{-\sfrak,j}v^{-1}\tb^{-2-2\sfrak-j-\veps}.
\end{align}
Here, $\{Y_{m,\sfrak}^{+\sfrak}(\cos\theta)e^{im\pb}\}_{-\sfrak\leq m\leq \sfrak}$ and $\{Y_{m,\sfrak}^{-\sfrak}(\cos\theta)e^{im\pb}\}_{-\sfrak\leq m\leq \sfrak}$ are the spin-weighted spherical harmonic functions, the function $\mathfrak{f}_{+\sfrak, m}$ is a finite function in $ M, a, \sfrak, m, r$ that can be explicitly written down and $\mathfrak{f}_{+\sfrak, m}=\mu^{\sfrak}+amO(r^{-1})$, and the value of $\MQ_{m,\sfrak}$ can be calculated explicitly from the initial data of the spin $\pm \sfrak$ components on $\Sigmazero$.

\item  if $\psipluss$ $(\sfrak=1,2)$ is supported on an azimuthal $m$-mode, then on $\Horizon$,
\begin{align}
\label{mainthm:asymp:psiplus:horizon:pm1:d}
\big|\ptb^j\big(\psipluss|_{\Horizon}
-D_{+\sfrak,\Horizon}\MQ_{m,\sfrak}Y_{m,\sfrak}^{+\sfrak}(\cos\theta)e^{im\pb}\tb^{-2\sfrak-3-j}\big)\big|
\leq{}& C_{+\sfrak,j,\Horizon}\tau^{-2\sfrak-3-j-\veps},
\end{align}
and if moreover $am= 0$, the decay is faster by $\tb^{-1}$:
\begin{align}
\label{mainthm:asymp:psiplus:horizon:pm1:d:v2}
\big|\ptb^j\big({\psipluss}|_{\Horizon}- D'_{+\sfrak,\Horizon}\MQ_{m,\sfrak}Y_{m,\sfrak}^{+\sfrak}(\cos\theta)e^{im\pb}\tb^{-2\sfrak-4-j}\big)
\big|
\leq{}& C'_{+\sfrak,j,\Horizon}\tau^{-2\sfrak-4-j-\veps}.
\end{align}
Here, the constants $D_{+\sfrak, \Horizon}$ and $D'_{+\sfrak,\Horizon}$ are complex-valued constants in $M,a,m,\sfrak$  and can be calculated explicitly, and constant $D_{+\sfrak, \Horizon}$ vanishes if and only if $am=0$.
\end{enumerate}

Furthermore, the above estimates are valid for $\abs{a}/M<1$ in the case $\sfrak=1,2$ under an energy and Morawetz estimate assumption \ref{ass:BEAM:inhomogeneous} for an inhomogeneous Teukolsky master equation.
\end{thm}

\begin{remark}
\begin{itemize}

\item Assumption \ref{ass:BEAM:inhomogeneous} on an energy and Morawetz estimate, also called an integrated local energy decay estimate, is likely to hold true for an inhomogeneous Teukolsky master equation in the cases $\sfrak=1,2$ on a subextreme Kerr. See Section \ref{sect:intro:wed}.

\item{(Extension to non-compactly supported initial data case.)} This theorem  presents a simplified version of Theorems \ref{thm:PL:extregion} and \ref{thm:PL:anymodes:near:pm1}. In Theorems \ref{thm:PL:extregion} and \ref{thm:PL:anymodes:near:pm1},  the requirement for the initial data is specified (thus assumption on the initial data with compact support in the above theorem is not necessary), the value of $\MQ_{m,\sfrak}$ is explicitly calculated  in Lemma \ref{Kerr:prop:VtildePhil-1:nearinf:pm1} by the initial data of the spin $\pm \sfrak$ components on $\Sigmazero$, the expressions of both the function $\mathfrak{f}_{+\sfrak, m}$ and  constant $D_{+\sfrak, \Horizon}$ are explicitly written down, and the constants $C_{+\sfrak,j}$, $C_{-\sfrak,j}$, $C_{+\sfrak,j,\Horizon}$ and $C'_{+\sfrak,j,\Horizon}$ are stated in terms of the initial data. It can also be seen from the expression of $\MQ_{m,\sfrak}$ that the value of $\MQ_{m,\sfrak}$ is nonzero for generic initial data, hence the above asymptotics are generically sharp as both an upper and a lower bounds. 

\item{(Assumptions on the initial data decay.)}  Our assumption in Theorems \ref{thm:PL:extregion} and \ref{thm:PL:anymodes:near:pm1} requires the non-compactly supported initial data to satisfy the so-called peeling property, i.e., $\NPR_{+\sfrak}(\tau_0, \rho,\omega)\sim \rho^{-2\sfrak-1}$ as $\rho\to\infty$ on the initial hypersurface $\Sigma_{\tau_0}$, with $\omega$ being local coordinates on unit $2$ sphere. This peeling property is further shown to also hold in all future time $\tau>\tau_0$, thus the radiation field $\lim\limits_{\rho\to\infty}\rho^{2\sfrak+1}\NPR_{+\sfrak}(\tau, \rho,\omega)$ is a continuous function on future null infinity. In present work, the value of $\MQ_{m,\sfrak}$ is in fact characterized by an integral of the radiation field along the future null infinity.  It should be noted that it is expected that generic physically interesting Cauchy data do not satisfy peeling properties. See, for instance, the recent works of Kehrberger \cite{kehr21I,kehr21II,kehr21III}  in which the author considered the
precise structure of gravitational radiation near infinity for the scalar field on Schwarzschild.

\item{(Relation to the Price's law and the horizon oscillation.)} Our result confirms both the heuristic Price's law \cite{Price1972SchwScalar,Price1972SchwIntegerSpin,Hod99Mode,gleiser2008late}  in the region $r\geq r_+$ of a Kerr spacetime and the claim of Barack--Ori  \cite{bo99} that the spin $+\sfrak$ $(\sfrak=1,2)$ component enjoys faster decay than the Price's law on $\Horizon$ if $am=0$, and generalizes the statements in \cite{MaZhang21PriceSchw} from Schwarzschild to subextreme Kerr backgrounds.\footnote{We thank an anonymous referee in our earlier work \cite{MaZhang20sharp} for bringing  the work of Barack--Ori into our attention.} Note that it is shown in \cite{MaZhang20sharp} that Barack--Ori's claim can not be generalized to $\sfrak=\half$ case which corresponds to the massless Dirac field. Meanwhile, if we introduce a coordinate $\pbh=\pb - \frac{a}{2Mr_+}\tb \text{ mod } 2\pi$ such that it is invariant under the null Killing generator $K=\partial_{\tb}+\frac{a}{2Mr_+}\partial_{\pb}$ along $\Horizon$, then the asymptotics of the spin $\pm \sfrak$ components on $\Horizon$ exhibit the so-called horizon oscillation \cite{bo99} in the sense that the asymptotic profiles for each azimuthal $m$-mode contain an oscillatory factor $e^{\frac{iam}{2Mr_+}\tb}$. This is predicted in \cite{bo99} and first rigorously proven for  $\ell=1$ mode of the scalar field on Kerr in \cite{angelopoulos2021late}.

\item As a corollary, one can utilize the above asymptotics of the spin $\pm \sfrak$ components together with the first-order Maxwell equations to derive the asymptotic decay of the middle component of the Maxwell field to a stationary Coulomb solution. See \cite[Section 4.4]{Ma20almost}.

\end{itemize}
\end{remark}

The spin $s$ components arise from suitable linearizations of the vacuum Einstein equation  and provide high accuracy approximation for its nonlinear dynamics.
In contrast to the flat Minkowski background, the dynamics of the spin $s$ components are known to develop power tails in the future development in the DOC of a Kerr black hole spacetime.
These tails are intimately related to and crucial in addressing some fundamental problems in the theory of General Relativity including for instance the nonlinear stability problem of the black hole exterior and the Strong Cosmic Censorship conjecture concerning the (in)stability of  the Cauchy horizon in the black hole interior.

 In order to put our result into the context, we provide a review of related works in the literature. Physically, the power tails arise because of the backscattering arising from an effective curvature potential that is caused by some non-vanishing Weyl curvature component  on a Kerr background.
These power tails are first predicted by Price \cite{Price1972SchwScalar,Price1972SchwIntegerSpin} and  refined by Price--Burko \cite{price2004late} in a Schwarzschild spacetime saying that the spin $\pm \sfrak$ components have $\tb^{-3-2\sfrak}$ asymptotic decay in a finite radius region and their $\ell$ modes shall have $\tb^{-3-2\ell}$ decay, and then generalized to  Kerr spacetime in \cite{Hod99Mode,gleiser2008late}; they are conjectured to be sharp and called the \emph{Price's law}. Following this, Barack--Ori  \cite{bo99}  found that for $\sfrak\neq 0$, if $am=0$,  the spin $+\sfrak$ component shall actually have faster $\tb^{-1}$ decay, that is, $\tb^{-4-2\sfrak}$ asymptotic decay, on the future event horizon; this is further verified in a recent numerical work of Csuk\'as--R\'acz--T\'oth \cite{Csukas19dynamicsofspinKerr}. As a consequence, in the DOC of a Kerr spacetime, the correct asymptotic decay rates in mind shall be a combination of the Price's law outside the horizon and Barack--Ori's claim on horizon.

There has been much work towards rigorously proving the sharp decay rate for the scalar field in the mathematics literature.
Tataru \cite{tataru2013local} first obtained $t^{-3}$ pointwise decay on a class of stationary spacetimes including the subextreme Kerr spacetimes by assuming an integrated local energy decay estimate, and Donninger--Schlag--Soffer \cite{donninger2011proof} used a different approach to achieve the same decay outside a Schwarzschild black hole; Metcalfe--Tataru--Tohaneanu \cite{metcalfe2012price} further generalized the result of Tataru to a class of nonstationary spacetimes under a similar assumption. Donninger--Schlag--Soffer \cite{donninger2012pointwise} then obtained in a compact region outside a Schwarzschild black hole $t^{-2\ell-2}$ decay (and $t^{-2\ell-3}$ decay for static initial data) for an $\ell$ mode. The globally sharp $v^{-1}\tb^{-2}$ pointwise decay  is first proven by Angelopoulos--Aretakis--Gajic \cite{angelopoulos2018vector,angelopoulos2018late} and the precise late-time asymptotic profile is calculated therein; Hintz \cite{hintz2020sharp} computed the $v^{-1}\tb^{-2}$ leading order term  on both Schwarzschild and subextreme Kerr spacetimes and further obtained $v^{-1}\tb^{-2\ell-2}$ sharp asymptotics for  $\geq \ell$ modes in a compact region on Schwarzschild; Luk--Oh \cite{lukoh17linearinstability} derived sharp decay for the scalar field on a Reissner--Nordstr\"om background and used it to obtain linear instability of the Reissner--Nordstr\"om Cauchy horizon (see also their works \cite{lukoh19SCCI,lukoh19SCCII} on a generalization to a nonlinear setting); Angelopoulos--Aretakis--Gajic based on their own earlier works and re-derived in \cite{angelopoulos2021price} $v^{-1}\tb^{-2\ell-2}$ late time asymptotics for $\geq \ell_0$ modes in a finite radius region on Schwarzschild, and they further computed in \cite{angelopoulos2021late} the asymptotic profiles of the $\ell=0$, $\ell=1$, and $\ell\geq 2$ modes in a subextreme Kerr spacetime; we \cite{MaZhang21PriceSchw} independently computed the global $v^{-1}\tb^{-2\ell-2}$ late time asymptotics for $\geq \ell$ modes in a Schwarzschild spacetime. Additionally, Kehrberger \cite{kehr21I,kehr21II,kehr21III} considered the
precise structure of gravitational radiation near infinity for the scalar field on Schwarzschild. 

  For spin $s$ components, $(s\neq 0)$,  there are no sharp results proven until recently. Donninger--Schlag--Soffer \cite{donninger2012pointwise} obtained in a compact region outside a Schwarzschild black hole $t^{-2\sfrak-2}$ decay  for the spin $\pm \sfrak$ $(\sfrak=1,2)$ components; Metcalfe--Tataru--Tohaneanu \cite{metcalfe2017pointwise} refined the decay for the spin $s$ $(s=\pm 1)$ components of the Maxwell field to a global $v^{-2-s}\tb^{-2+s}$ pointwise decay in a class of nonstationary spacetimes under an integrated local energy decay estimate assumption. The above decay estimates are slower than the sharp Price's law by $\tb^{-1}$ or $\tb^{-\frac{3}{2}}$.
The first author of this current work derived in \cite{Ma20almost}  $v^{-2-s}\tb^{-\frac{3}{2}+s}$ decay in non-static Kerr and  $v^{-2-s}\tb^{-3+s+\epsilon}$ almost sharp decay for all spin $s$ components of the Maxwell field in Schwarzschild towards a stationary/static Coulomb solution, and it also proved the almost sharp $v^{-2-s}\tb^{-2-\ell+s+\veps}$ decay for any $\geq \ell$ modes for the Maxwell field in the region $\rb\gtrsim \tb$ on a Schwarzschild background. If restricted to a Schwarzschild background, we \cite{MaZhang21PriceSchw}  computed  $v^{-1-\sfrak-s}\tb^{-2-\sfrak+s}$ late time asymptotic profiles for the spin $\pm \sfrak$ components globally in the DOC,  and, for $\geq \ell$ modes of the spin $s$ components, computed $v^{-1-\sfrak-s}\tb^{-2-\ell_0+s}$ asymptotics in region $\rb\geq \tb$, $r^{\ell-\sfrak}\tb^{-3-2\ell_0}$ asymptotics  in region $\rb\leq \tb$, and achieved $\tb^{-4-2\ell_0}$ asymptotics for the $\geq \ell_0$ modes for the spin $+\sfrak$ $(\sfrak=1,2)$ components on $\Horizon$; hence, we have confirmed in \cite{MaZhang21PriceSchw} both the Price's law (for $\sfrak=1,2$) and Barack--Ori's claim  (for $s=1,2$) for the spin $s$ component on a Schwarzschild background.
Let us also mention that we \cite{MaZhang20sharp}  generalized the Price's law to the massless Dirac field on Schwarzschild by calculating $v^{-\frac{3}{2}-s}\tb^{-\frac{5}{2}+s}$ asymptotic profiles  for its spin $s=\pm \half$ components.

Apart from the above works working on TME (including scalar wave equation) on Schwarzschild or Kerr spacetimes,  there have been many interesting works  in proving various sharp or almost sharp pointwise decay for wave equations on different backgrounds. We refer to the review paper of Biz\'on \cite{bizon08Huygensprinciple} for relevant physical and numerical results. Interestingly, in \cite{bizon07yangmills,bizon09wavemap}, Biz\'{o}n--Chmaj--Rostworowski (and with Stanisław Zając) found that for Yang--Mills field on Schwarzschild and Einstein--wave map system, the higher $\ell$ modes have $\tb^{-2\ell-2}$ nonlinear tails in a finite radius region, $\tb^{-1}$ slower decay than the linear tails predicted by Price's law. 
In the mathematics literature, 
in an asymptotical flat, stationary spacetime that approaches Minkowski in a rate $|x|^{-k}$, Morgan \cite{morgan2020effect} established $t^{-k-2}$ pointwise decay for scalar field  for $2\leq k\in \mathbb{N}$, and  $t^{-k-2+\veps}$ decay for $k\in (1,+\infty)\setminus\mathbb{N}$ is proved by Morgan-Wunsch \cite{morgan2021generalized}. Looi \cite{looi2021pointwise} obtained pointwise decay estimates for solutions to linear wave equations with variable coefficients. Tohaneanu \cite{tohaneanu2021pointwise} proved the sharp upper bound of pointwise decay for a semilinear wave equation on a slowly rotating Kerr background. 

In the end, we draw attention to the progress on black hole stability problem in recent years. Linear stability of a Schwarzschild or a subextremal  Reissner--Nordstr\"{o}m spacetime has been shown by \cite{dafermos2019linear,hung2020linear,Jinhua17LinGraSchw,johnson2019linear,Hung18odd, Hung19even,Giorgi2019linearRNfullcharge}, and linear stability of a slowly rotating Kerr spacetime is proven in \cite{andersson2019stability,hafner2019linear,andersson2021nonlinear}.
For  nonlinear stability results, we refer to \cite{klainermanszeftel2020global,dafermos2021non} for Schwarzschild, \cite{HintzKds2018}  for  slowly rotating Kerr-de Sitter, and \cite{klainerman2019constructions,giorgi2020general,klainerman2021kerr} for slowly rotating Kerr.


\subsection{Method of the proof}

In this subsection, we provide an outline of the proof. All the estimates are derived via the analyses of the  TME satisfied by  the spin $\pm \sfrak$ $(\sfrak=0,1 ,2)$ components. Our proof can be divided into three steps, each of which is discussed in the following three subsubsections respectively. The first two steps are based on a generalization of the approach developed in our earlier work \cite{MaZhang21PriceSchw} on Schwarzschild to Kerr spacetimes, and the main ingredient of the third step is a novel global conservation law that can be applied to other problems, cf. Section \ref{outlook}.

\subsubsection{Weak energy decay estimates}
\label{sect:intro:wed}

To start with, one has to achieve an energy and Morawetz estimate  for solutions to the TME. These estimates have been proven in a Schwarzschild spacetime for $\sfrak=0$ in \cite{bluesterbenz2006,dafrod09red} and extended to $\sfrak=1,2$ in \cite{pasqualotto2019spin,dafermos2019linear}, and further extensions  are realized in \cite{tataru2011localkerr,larsblue15hidden,dafermos2016decay} for $\sfrak=0$ on any subextreme Kerr and in \cite{Ma2017Maxwell,Ma17spin2Kerr,dafermos2019boundedness} for $\sfrak=1,2$ but on slowly rotating Kerr. See also related works \cite{bluesoffer03mora,blue:soffer:integral,Finster2006,MMTT,tohaneanu2012strichartz,gudapati2017positive,gudapati2019conserved} for $\sfrak=0$ and \cite{blue08decayMaxSchw,larsblue15Maxwellkerr,andersson16decayMaxSchw} for $\sfrak\neq 0$. The basic idea in proving the energy and Morawetz estimates for the TME is  to use certain differential transformations due to Chandrasekhar \cite{chandrasekhar1975linearstabSchw} which are first utilized in \cite{dafermos2019linear} in Schwarzschild, and then treat the  coupled wave systems
\begin{align*}
\textbf{CWS}_{s} = \big\{ \text{the wave system of } \{\curlVR^i (\mu^{\sfrak}\Psiminuss)\}_{i=0,\ldots, \sfrak} \text{ or } \{((\R)Y)^i ((\R)^{-2}\Psipluss)\}_{i=0,\ldots, \sfrak}\big\},
\end{align*}
where  $\mu=\frac{\Delta}{\R}$, $\curlVR=(\R) \VR$, and  $Y=\sqrt{2}n^{\nu}\partial_{\nu}$ and $\VR=\frac{\sqrt{2}\Sigma}{\Delta}l^{\nu}\partial_{\nu}$ are the ingoing and outgoing principal null vectors, and $$
\Psipluss=\sqrt{\R}\psipluss,\qquad \Psiminuss=\sqrt{\R}\psiminuss
$$ are the radiation fields. Of particular importance is that the wave equations of $\curlVR^{\sfrak} (\mu^{\sfrak}\Psiminuss)$ and  $(r^2 Y)^{\sfrak} (r^{-4}\Psipluss)$ on Schwarzschild background are the Regge--Wheeler  equation \cite{ReggeWheeler1957} and decouple from the other equations. By requiring $\abs{a}/M$ sufficiently small, the above coupled wave systems are in fact weakly coupled, and this allows the first author of this paper to complete in \cite{Ma2017Maxwell,Ma17spin2Kerr}  the derivation of a \emph{basic energy and Morawetz (BEAM)}  estimate for TME on slowly rotating Kerr backgrounds. See different proofs in \cite{larsblue15Maxwellkerr, dafermos2019boundedness}  for similar estimates for the Maxwell field and the linearized gravity on slowly rotating Kerr backgrounds.

Generalizing this BEAM estimate for $\sfrak=0$ from slowly rotating Kerr to subextreme Kerr is accomplished in \cite{dafermos2016decay} by combining the approach in treating the slowly rotating Kerr case, a mode stability result \cite{2015AnHP...16..289S} that generalized Whiting's celebrated result \cite{whiting1989mode} and a clever continuity argument, and a BEAM estimate for the scalar wave equation with an inhomogeneous term can be easily derived afterwards. Given that the slowly rotating Kerr case is completed for TME and that mode stability is shown for TME \cite{andersson2017mode,da2019mode} on any subextreme Kerr, it is widely expected that such a BEAM estimate for (an inhomogeneous) TME shall hold true in any subextreme Kerr spacetime. Consequently, we make an assumption that a BEAM estimate holds for solutions to an inhomogeneous TME, and we call it a ``BEAM estimate assumption''. This BEAM estimate assumption is assumed only  for $\sfrak=1,2$ for subextreme Kerr (but not needed for slowly rotating Kerr).

We then generalize the $r^p$ method initiated by Dafermos--Rodnianski \cite{dafermos2009new} to derive a hierarchy of $r$-weighted energy and Morawetz estimates (so-called the $r^p$ estimates) near infinity. Together with the BEAM estimates which encode much of local information of the field, we can deduce certain weak decay for $r$-weighted energies. This approach is developed in \cite{dafermos2009new} for $\sfrak=0$ and in  \cite{andersson2019stability,Ma20almost} for $\sfrak=1,2$, and we describe it in the remainder of this subsubsection.

Due to the gap of the nonpositive spectrum of the spin-weighted spherical Laplacian from zero, one can further commute $\curlVR=(\R)\hat{V}$ up to $\sfrak$ times with the wave equation of $\curlVR^{\sfrak} (\mu^{\sfrak}\Psiminuss)$ and arrive at  larger wave systems
\begin{align}
\label{8}
\textbf{WS}^{(j)}_{-\sfrak} = \{ \text{the system of wave equations of } \{\Phiminuss{i}\}_{i=0,\ldots, j}\},
\end{align}
where $\Phiminuss{i}\triangleq \curlVR^i (\mu^{\sfrak}\Psiminuss)$ and $0\leq j\leq 2\sfrak$. In particular, in the wave equation of $\Phiminuss{2\sfrak}$, we have exhausted out the spectrum gap from zero, and commuting with $\curlVR$ more times would result in a failure of employing the $r^p$ method. The $r^p$ estimates are then derived for each of the wave systems $\{\textbf{WS}^{(j)}_{-\sfrak} \}_{\sfrak\leq j\leq 2\sfrak}$ and yield, for each $j\in \{\sfrak,\sfrak+1,\ldots, 2\sfrak\}$, $\tb^{-2+2\delta}$ decay for $p=\delta$-weighted energy  of the system $\textbf{WS}^{(j)}_{-\sfrak}$ in terms of $p=2-\delta$-weighted energy of this system. 
 Combined with the fact that $p=2-\delta$-weighted energy of the system $\textbf{WS}^{(j)}_{-\sfrak}$ is bounded by $p=\delta$-weighted energy  of the system $\textbf{WS}^{(j+1)}_{-\sfrak}$,
one eventually obtains $\tb^{-(2-2\delta)(\sfrak+1)}$ decay for the $p=\delta$-weighted energy of system $\textbf{WS}^{(\sfrak)}_{-\sfrak}$ in terms of the $p=2-\delta$-weighted energy of system $\textbf{WS}^{(2\sfrak)}_{-\sfrak}$. Further, one achieves extra $\tb^{-(2-2\delta)j}$ energy decay for $\ptb^j$ derivatives. By a standard Sobolev imbedding estimate, this proves $r v^{-1}\tb^{-(1-\delta) (\sfrak+j)-\half+\delta}$ pointwise decay  for $\{\ptb^j\curlV^i\Psiminuss\}_{0\leq i\leq\sfrak}$, with $\curlV=\mu \curlVR$.

For the spin $+\sfrak$ component, we simply consider the wave equation of $\PhiplussHigh{0}=\mu^{\sfrak}\Psipluss$:
\begin{align*}
\textbf{WS}^{(0)}_{+\sfrak} = \{ \text{the wave equation  } \PhiplussHigh{0}\}
\end{align*}
and easily achieve the $r^p$ estimates, thus concluding $\tb^{-2(1-\delta)}$ decay for $p=\delta$-weighted energy of $\textbf{WS}^{(0)}_{+\sfrak}$ and $r v^{-1}\tb^{-\half +\delta -(1-\delta)j}$ pointwise decay for $\ptb^j\Psipluss$ in terms of $p=2-\delta$-weighted energy of $\textbf{WS}^{(0)}_{+\sfrak}$.

\subsubsection{Almost sharp energy and pointwise decay estimates for the modes}
\label{intro:almostsharp}

 To deduce further energy decay, it is convenient to decompose the field into spin-weighted spherical harmonic modes and employ different techniques to obtain almost sharp decay for the modes. See \cite{angelopoulos2018vector,angelopoulos2018late,angelopoulos2021price}
for $\sfrak=0$ and \cite{MaZhang21PriceSchw} for general $\sfrak$ in Schwarzschild spacetimes.

In a non-static Kerr spacetime, however, these modes are coupled in the evolution due to the presence of $\theta$-dependent operator $a^2\sin^2\theta \ptb^2 -2ias\cos\theta \ptb$  in the TME. Notwithstanding, since the terms arising from mode coupling are with $\ptb$-derivatives and have faster $\tb$-decay by the claim in the previous subsubsection, Angelopoulos--Aretakis--Gajic \cite{angelopoulos2021late} were able to treat these mode coupling terms as inhomogeneous terms and derived almost sharp decay for $\sfrak=0$.

We follow this idea and further generalize it by decomposing the spin $\pm \sfrak$ components into spin-weighted spherical harmonic modes $\ell=\sfrak$, $\ell=\sfrak+1$ and $\ell\geq \sfrak+2$. It turns out that it suffices to consider the spin $+\sfrak$ component since there is a special combination $\dotPhisHigh{2\sfrak}=\Phiminuss{2\sfrak}+\sum_{i=0}^{2\sfrak-1}C_i \Phiminuss{i}$ such that this scalar satisfies essentially the same wave equation as $\PhiplussHigh{0}$,  thus a similar approach as the one for the spin $+\sfrak$ component works for the spin $-\sfrak$ component.

Following our earlier work \cite{MaZhang21PriceSchw}  on TME in Schwarzschild, we first derive  equations of $\PhiplussHigh{i}=\curlVR^i\PhiplussHigh{0}$:
\begin{align*}
\Boxhat_{+\sfrak}\PhiplussHigh{i}
={}&\frac{2(\sfrak+i)(r^3-3Mr^2 +a^2 r+a^2 M)}{(\R)^2} \curlVR\PhiplussHigh{i}
-(2\sfrak+i)(i+1)\PhiplussHigh{i}
\notag\\
&
+\sum_{{0\leq j \leq i-1, \frac{i-j-1}{2}\in \mathbb{N}}}X_{+\sfrak,i,j}\Leta \PhiplussHigh{j}
-\sum_{j=0}^{i-1}{Z_{+\sfrak,i,j}}\PhiplussHigh{j}
+\sum_{n=0,1}\sum_{j=0}^i O(r^{-1})\Leta^n \PhiplussHigh{j},
\end{align*}
where $\Boxhat_{+\sfrak}$ is a spin-weighted wave operator, $\Leta$ is the Killing vector $\partial_{\pb}$, and $X_{+\sfrak,i,j}$ and $Z_{+\sfrak,i,j}$ are constants depending only on $\sfrak, i, j$. The terms with coefficients $X_{+\sfrak,i,j}$ and $Z_{+\sfrak,i,j}$ are one of the main obstructions in extending the $r^p$ method to an almost maximal range of $p$ after decomposing into modes. Fortunately, there exists a unique linear combination of the form
\begin{align}
\hatPhiplussHigh{i}\doteq{}&\PhiplussHigh{i}+\sum_{j=0}^{i-1} \sum_{n=0}^{i-j} x_{\sfrak,i,j,n}\Leta^n\hatPhiplussHigh{j}
\end{align}
with $\{x_{\sfrak,i,j,n}\}_{0\leq j\leq i-1, 0\leq n\leq i-j}$ being constants such that the scalars $\hatPhiplussHigh{i}$
solve the following wave equations that successfully remove the above troublesome constant coefficient terms:
\begin{align}
\Boxhat_{+\sfrak}\hatPhiplussHigh{i}
={}&\frac{2(\sfrak+i)(r^3-3Mr^2 +a^2 r+a^2 M)}{(\R)^2}\curlVR\hatPhiplussHigh{i}
-(2\sfrak+i)(i+1)\hatPhiplussHigh{i}
+\hat{H}_{+\sfrak,i},
\end{align}
with $d_i$ a constant depending on $i$ and $\hat{H}_{+\sfrak,i}=\sum_{0\leq j\leq i}\sum_{n\leq d_i}O(r^{-1})\Leta^n\PhiplussHigh{j}$.
By projecting onto $\ell$ modes, we obtain
\begin{align}
\label{13}
\Boxhat_{+\sfrak}\ellmode{\hatPhiplussHigh{i}}{\ell}
+(2\sfrak+i)(i+1)\ellmode{\hatPhiplussHigh{i}}{\ell}
-\frac{2(\sfrak+i)(r^3-3Mr^2 +a^2 r+a^2 M)}{(\R)^2}\curlVR\ellmode{\hatPhiplussHigh{i}}{\ell}
={}&\hat{\mathbf{N}}[\ellmode{\hatPhiplussHigh{i}}{\ell}],
\end{align}
with $\ellmode{\varphi}{\ell}$ being the $\ell$ mode of $\varphi$, $\hat{\mathbf{N}}[\ellmode{\hatPhiplussHigh{i}}{\ell}]=\ellmode{\hat{H}_{+s,i}}{\ell}+\mathbf{MC}[\ellmode{\hatPhiplussHigh{i}}{\ell}]$ and $\mathbf{MC}[\ellmode{\hatPhiplussHigh{i}}{\ell}]$ arising from the mode coupling.
This equation can be put into a form of an inhomogeneous spin-weighted wave equation to which $r^p$ estimates with $p\in (\delta,2-\delta)$ can be applied iff $i\leq \ell-\sfrak$.

To go beyond $p=2$, one shall consider $i=\ell-\sfrak$ in the above equation  for the reason that $(2\sfrak+i)(i+1)\ellmode{\hatPhiplussHigh{i}}{\ell}$ offsets the spin-weighted angular operator acting on $\ellmode{\hatPhiplussHigh{i}}{\ell}$ in the term $\Boxhat_{+\sfrak}\ellmode{\hatPhiplussHigh{i}}{\ell}$. The other obstruction to extending the $r^p$ hierarchy for $i=\ell-\sfrak$ is exactly the mode coupling terms $\mathbf{MC}[\ellmode{\hatPhiplussHigh{\ell-\sfrak}}{\ell}]$ together with $2a\ptb\Leta \ellmode{\hatPhiplussHigh{\ell-\sfrak}}{\ell}$ in $\Boxhat_{+\sfrak}\ellmode{\hatPhiplussHigh{\ell-\sfrak}}{\ell}$ since they are with constant coefficients. By introducing a scalar
\begin{align}
\label{14}
\tildePhisHighell{+\sfrak}{\ell}\doteq\Proj{\ell}\Big(\curlVR\hatPhiplussHigh{\ell-\sfrak}
-\half\big(2a\Leta\hatPhiplussHigh{\ell-\sfrak}
+a^2\sin^2 \theta\Lxi \hatPhiplussHigh{\ell-\sfrak}
-2ia\sfrak\cos\theta \hatPhiplussHigh{\ell-\sfrak}\big)\Big),
\end{align}
with $\Proj{\ell}$ being the projection onto $\ell$ mode,
it  satisfy a simple inhomogeneous transport equation
\begin{align}
\label{15}
-\mu Y\tildePhisHighell{+\sfrak}{\ell}
-\frac{2(\ell+1)(r^3-3Mr^2 +a^2 r+a^2 M)}{(\R)^2}\tildePhisHighell{+\sfrak}{\ell}
={}&\tilde{\mathbf{N}}[\tildePhisHighell{+\sfrak}{\ell}],
\end{align}
where $\tilde{\mathbf{N}}[\tildePhisHighell{+\sfrak}{\ell}]=O(r^{-1})(\cdot)$ with $(\cdot)$ being a complicated form of derivatives of $\{\hatPhiplussHigh{j}\}_{0\leq j\leq \ell-\sfrak}$, and the common $O(r^{-1})$ coefficients in $\tilde{\mathbf{N}}[\tildePhisHighell{+\sfrak}{\ell}]$ allows us to easily derive extended $r^p$ hierarchy for this transport equation and regain refined energy decay estimates. 
	
We list in the following table how we achieve $r^p$ estimates for the $\sfrak$, $\sfrak+1$, $\geq\sfrak+2$ modes in different ranges of $p$ in the $r^p$ hierarchy, respectively. One should note that the $r^p$ estimates for these modes shall be coupled together in order to get the error terms arising from the right-hand sides of equations \eqref{13} and \eqref{15} under control.  Specifically, we pose the following condition on a weighted initial energy of the spin $+\sfrak$ component: 
\begin{align}
\label{eq:initialweightedenergyspin+sfrakcomponent}
&\norm{\Psi_{+\sfrak}}_{W_{-2}^{k}(\Sigma_{\tau_0}^{\leq 4M})}\notag\\
	&+\norm{\tildePhisHighell{+\sfrak}{\sfrak}}_{W_{1-\delta}^{k-1}(\Sigma_{\tau_0}^{\geq 4M})}
	+\norm{rV\ellmode{\hatPhiplussHigh{0}}{\sfrak}}_{W_{-\delta}^{k-1}(\Sigma_{\tau_0}^{\geq 4M})}
	+\norm{\ellmode{\hatPhiplussHigh{0}}{\sfrak}}_{W_{-2}^{k}(\Sigma_{\tau_0}^{\geq 4M})}\notag\\
	&+\norm{\tildePhisHighell{+\sfrak}{\sfrak+1}}_{W_{-\delta}^{k-2}(\Sigma_{\tau_0}^{\geq 4M})}
	+\norm{rV\ellmode{\hatPhiplussHigh{1}}{\sfrak+1}}_{W_{-\delta}^{k-2}(\Sigma_{\tau_0}^{\geq 4M})}
	+\norm{\ellmode{\hatPhiplussHigh{\leq1}}{\sfrak+1}}_{W_{-2}^{k-1}(\Sigma_{\tau_0}^{\geq 4M})}\notag\\
	&+\norm{rV\ellmode{\hatPhiplussHigh{2}}{\geq\sfrak+2}}_{W_{2-\delta}^{k-3}(\Sigma_{\tau_0}^{\geq 4M})}
	+\norm{rV\ellmode{\hatPhiplussHigh{1}}{\geq\sfrak+2}}_{W_{-\delta}^{k-2}(\Sigma_{\tau_0}^{\geq 4M})}
	+\norm{\ellmode{\hatPhiplussHigh{\leq2}}{\geq\sfrak+2}}_{W_{-2}^{k-2}(\Sigma_{\tau_0}^{\geq 4M})}\notag\\
	<&+\infty
\end{align}
for a suitably large $k>0$,
 where the weighted energy on such a spacelike hypersurface $\Sigma$ (we may take $\Sigma=\Sigma_{\tau_0}^{\geq 4M}$ or $\Sigma=\Sigma_{\tau_0}^{\leq 4M}$) defined by
\begin{align}
	\norm{h}_{W_{\gamma}^{k}(\Sigma)}^2=\int_{\Sigma}r^\gamma \sum_{\abs{\alpha}\leq k}|\CDeri^{\alpha}h|^2 \di\rho\di\omega, 
\end{align}
where $\di\omega$ is the volume form on unit $2$-sphere and $\CDeri=\{Y, rV, \partial_\tau, \edthR, \edthR'\}$
with $\edthR$ and $\edthR'$ being first-order spin-weighted angular operators on unit $2$-sphere. This weighted initial energy arises naturally from the $r^p$ hierarchies for the scalars that are presented in the following Table \ref{tabel:rp}.  We shall refer to Definition \ref{defineenergy-initial:111111} for the explicit definitions of the relevant weighted initial energies for both of the spin $\pm\sfrak$ components.

\begin{table}[htbp]
\begin{center}
\begin{tabular}{|l|l|l|}
\hline
 \qquad \,\,\, scalar & \qquad {equation to use} & $p$ range in $r^p$ hierarchy\\
\hline
  \qquad\,\,\,$\ellmode{\hatPhiplussHigh{0}}{\sfrak}$& \quad wave equation \eqref{13}  & \qquad\quad$(\delta, 2-\delta)$  \\
  \hline
   \qquad\,\,\,\,\,$\tildePhisHighell{+\sfrak}{\sfrak}$\,\,\, & transport equation  \eqref{15} &   \qquad\quad$[2, 5-\delta)$\\
  \hline
 $ \qquad\, \ellmode{\hatPhiplussHigh{0}}{\sfrak+1}$& \quad  wave equation \eqref{13}  & \qquad\quad $(\delta, 2-\delta)$  \\
  \hline
   $ \qquad\, \ellmode{\hatPhiplussHigh{1}}{\sfrak+1}$& \quad  wave equation \eqref{13}  & \qquad\quad $(\delta, 2-\delta)$  \\
  \hline
   \qquad\,\,\,$\tildePhisHighell{+\sfrak}{\sfrak+1}$\,\,\, & transport equation \eqref{15} & \qquad\quad  $[2, 4-\delta)$\\
  \hline
  $\{\ellmode{\hatPhiplussHigh{i}}{\geq\sfrak+2}\}_{0\leq i \leq 2}$ &\quad wave equation \eqref{13} & \qquad\quad$(\delta,2-\delta)$ \\
  \hline
\end{tabular}
\end{center}
\caption{Coupled $r^p$ hierarchies for the modes}
\label{tabel:rp}
\end{table}

The second and third lines together in the above table are used to derive energy decay for the $\sfrak$ mode, the last line is to derive energy decay for $\geq \sfrak+2$ modes, and the lines in between are to derive energy decay for the $\sfrak+1$ mode.
The above coupled $r^p$ hierarchy for different modes eventually implies $\tb^{-5-2j+C_j\delta}$ and $\tb^{-6-2j+C_j\delta}$ energy decay for $p=\delta$-weighted energy of $\ptb^j\ellmode{\Psipluss}{\sfrak}$ and $\ptb^j\ellmode{\Psipluss}{\geq \sfrak+1}$ and $r v^{-1}\tb^{-2-j+C_j\delta}$ and $r v^{-1}\tb^{-\frac{5}{2}-j+C_j\delta}$ global pointwise decay for $\ellmode{\Psipluss}{\sfrak}$ and $\ellmode{\Psipluss}{\geq \sfrak+1}$, respectively, in terms of some suitable initial energy of the spin $+\sfrak$ component.
Analogously, one achieves  $\tb^{-5-2\sfrak-2j+C_j\delta}$ and $\tb^{-6-2\sfrak-2j+C_j\delta}$ energy decay for $p=\delta$-weighted energy of $\{\ptb^j\ellmode{\curlV^i\Psiminuss}{\sfrak}\}_{0\leq i\leq\sfrak}$ and $\{\ptb^j\ellmode{\curlV^i\Psiminuss}{\geq \sfrak+1}\}_{0\leq i\leq\sfrak}$ and $r v^{-1}\tb^{-2-\sfrak-j+C_j\delta}$ and $r v^{-1}\tb^{-\frac{5}{2}-\sfrak-j+C_j\delta}$ global pointwise decay for $\{\ptb^j\ellmode{\curlV^i\Psiminuss}{\sfrak}\}_{0\leq i\leq\sfrak}$ and $\{\ptb^j\ellmode{\curlV^i\Psiminuss}{\geq \sfrak+1}\}_{0\leq i\leq\sfrak}$, respectively, in terms of some suitable initial energy of the spin $-\sfrak$ component.

The final step is to further improve these decay estimates of the spin $\pm \sfrak$ components to almost sharp decay estimates, that is, $v^{-1-2\sfrak}\tb^{-2-j+C_j\delta}$ for $\ptb^j (r^{-2\sfrak}\ellmode{\psipluss}{\sfrak})$, $v^{-1}\tb^{-2-2\sfrak-j+C_j\delta}$ for $\ptb^j (\ellmode{\psiminuss}{\sfrak})$, and extra $\tb^{-\half}$ decay for $\geq\sfrak +1$ modes.
This is realized in two separate regions: the exterior region $\{r\geq \tb\}$ and the interior region $\{r\leq \tb\}$. Again, the idea follows from our earlier work \cite{MaZhang21PriceSchw} on Schwarzschild, and we generalize the method therein to subextreme Kerr.

In the exterior region, because of $r\gtrsim v$, we immediately obtain $v^{-1-2\sfrak}\tb^{-2-j+C_j\delta}$ for $\ptb^j(r^{-2\sfrak}\ellmode{\psipluss}{\sfrak})$  and $v^{-1-2\sfrak}\tb^{-\frac{5}{2}-j+C_j\delta}$ for $\ptb^j (r^{-2\sfrak}\ellmode{\psipluss}{\geq \sfrak+1})$. To achieve the almost sharp decay for the spin $-\sfrak$ component, an efficient way is to make use of the \textbf{Teukolsky--Starobinsky identities} (TSI) \cite{TeuPress1974III,starobinsky1973amplification} that are two $2\sfrak$-order differential identities between the spin $\pm \sfrak$ components. See Section \ref{sect:TSI} for the TSI. The rough form of TSI is
\begin{subequations}
\begin{align}
\label{17a}
(\edthR'-ia\sin\theta\ptb)^{2\sfrak}\psipluss
\approx{}&\Delta^{\sfrak}\VR^{2\sfrak}(\Delta^{\sfrak}\psiminuss),\\
\label{17b}
(\edthR+ia\sin\theta\ptb)^{2\sfrak} \psiminuss\approx{}&
Y^{2\sfrak} \psipluss,
\end{align}
\end{subequations}
where $\edthR$ and $\edthR'$ are first-order spin-weighted angular operators on spheres. The TSI are ubiquitous tools in the analyses of linear or nonlinear TME for the reason that one can retrieve the estimates for one spin component from the estimates of the other spin component, and many works on  Schwarzschild or Kerr stability, for instance, \cite{klainermanszeftel2017Schw,andersson2019stability}, have witnessed their indispensable importance. The left-hand sides of TSI \eqref{17a} and \eqref{17b} are elliptic operators over sphere, modulus terms with $\ptb$-derivatives that have faster decay. 
An application of the TSI \eqref{17b} and the almost sharp decay for the spin $+\sfrak$ component together with an elliptic estimate over sphere then prove the almost sharp decay for the modes of the spin $-\sfrak$ component via a simple elliptic estimate.

In the interior region, we shall instead first analyze the spin $-\sfrak$ component and then derive the almost sharp decay for the spin $+\sfrak$ component via the other TSI \eqref{17a}. We rely on  two types of elliptic estimates: one on $2$-dimensional spheres to gain $r^{-\sfrak}$ further decay for $\psiminuss$, and the other being a hierarchy of $r$-weighted elliptic estimates on a $3$-dimensional space to trade this extra $r^{-\sfrak}$ decay for extra $\tb^{-\sfrak}$ decay, thus proving the almost sharp decay for the spin $-\sfrak$ component. For the first one, we take $\sfrak=1$ without loss of generality. By isolating out the spin-weighted spherical part of equation $\textbf{WS}^{(0)}_{-1}$  as defined in \eqref{8} to the left-hand side and putting the extra terms to the right-hand side, and writing the main extra term $Y\curlVR\Phiminus{0}=Y\Phiminus{1}$, all the terms on the right-hand side have  faster $r^{-1}$ decay, hence a standard elliptic estimate over sphere yields the desired result. For the other one, we can simply write the TME of $\psiminuss$ as a second-order spatial operator on $\psiminuss$ equal $\ptb$ acting on the rest. The right-hand side with $\ptb$-derivative has faster $\tb^{-1}$ pointwise and $\tb^{-2}$ energy decay, and we are able to derive a sequence of elliptic estimates that eventually improve the extra $r^{-\sfrak}$ decay to $\tb^{-\sfrak}$ decay.  It is worth to remark that we can also derive  $v^{-1}\tb^{-3-2\sfrak-j+C_j\delta}$ for $\Lxi^j\prb \ellmode{\psiminuss}{\sfrak}$ in the interior region $\{r\leq \tb\}$, which in particular suggests faster $\tb^{-1+C\delta}$ decay for $\prb \ellmode{\psiminuss}{\sfrak}$ in a finite $r$ region than $\ellmode{\psiminuss}{\sfrak}$.

\subsubsection{A global conservation law and proof of the sharp decay}
\label{intro:conser}

The foremost gist is a global conservation law for the spin $+\sfrak$ component. By projecting the TME of $\psipluss$ onto an $(m,\sfrak)$ mode, we obtain
\begin{align}
\label{18}
\prb(\Delta^{\sfrak+1}\prb(\Delta^{-\sfrak}\mellmode{\psipluss}{m}{\sfrak})+2iam\mellmode{\psipluss}{m}{\sfrak})
=\ptb \mathbf{H}_{m,\sfrak}[\psipluss],
\end{align}
and an integration of this equation over the future Cauchy development of the initial hypersurface $\Sigmazero$ leads to a global conservation law. With a bit more details, this global conservation law indicates\footnote{We remark that the LHS of this conservation law is in fact equal to the second term in the formula of  $\mathcal{L}$ in \cite[Equation (1.13)]{lukoh17linearinstability}  if restricting to the scalar field $(\sfrak=0)$ on a Reissner--Nordstr\"om background.}
\begin{align}
(2\sfrak+1)\int_{\Scri\cap[\tb_0,\infty)}\mellmode{\Psipluss}{m}{\sfrak}={}\int_{\Sigmazero}\mathbf{H}_{m,\sfrak}[\psipluss]-(2iam-2\sfrak(r_+-M))\int_{\Horizon\cap[\tb_0,\infty)}\mellmode{\psipluss}{m}{\sfrak}.
\end{align}
Using again the mode projection form of the TSI \eqref{17a}, we can express $\int_{\Horizon\cap[\tb_0,\infty)}\mellmode{\psipluss}{m}{\sfrak}$ in terms of the initial data of the spin $\pm \sfrak$ components and $\int_{\Horizon\cap[\tb_0,\infty)}\mellmode{\psiminuss}{m}{\sfrak}$.

Our next aim is to calculate $\int_{\Scri\cap[\tb_0,\infty)}\mellmode{\Psipluss}{m}{\sfrak}$ in terms of the initial data, hence it suffices to compute $\int_{\Horizon\cap[\tb_0,\infty)}\mellmode{\psiminuss}{m}{\sfrak}$ in terms of the initial data. This is in turn achieved by first integrating an analogous equation for the $(m,\sfrak)$ mode of the spin $-\sfrak$ component as \eqref{18} such that $\mellmode{\psiminuss}{m}{\sfrak}(\rb,\tb)$ can be expressed as a weighted double integral of $\ptb \mathbf{H}_{m,\sfrak}[\psipluss]$ in $\rb$ and then integrating over horizon.
Further, we can also compute the integrals $\int_{\Scri\cap[\tb_0,\infty)}\mellmode{\PhiplussHigh{j}}{m}{\ell}$ for any $\ell >\sfrak$ and $0\leq j<\ell-\sfrak$ in terms of the initial data information.

Given the above integrals of the radiation fields along null infinity, we are now able to demonstrate how they can be used to derive the asymptotic profiles.  By projecting equation \eqref{15} onto an $m$ mode, denoting $\tildePhisHighell{+\sfrak}{m,\sfrak}$ as the $m$ mode of $\tildePhisHighell{+\sfrak}{\sfrak}$, and applying a simple scaling, we get
\begin{align}
\label{19}
-\mu Y ( \mu^{\sfrak+1} (\R)^{-\sfrak-1}\tildePhisHighell{+\sfrak}{m,\sfrak})
= {}&\mu^{\sfrak+1} (\R)^{-\sfrak-1}\tilde{\mathbf{N}}[\tildePhisHighell{+\sfrak,m}{\sfrak}].
\end{align}
One finds $\tilde{\mathbf{N}}[\tildePhisHighell{+\sfrak,m}{\sfrak}]= C_1 r^{-1}\mellmode{\Psipluss}{m}{\sfrak} +r^{-1}\sum_{i\leq 1}\sum_{\ell=\sfrak+1,\sfrak+2}C_{2,i,\ell}\Lxi^{i}\mellmode{\Psipluss}{m}{\ell} +O(r^{-2})v^{-1+\veps}$ and $\hat{V}^j (r\tilde{\mathbf{N}}[\tildePhisHighell{+\sfrak,m}{\sfrak}]) = O(r^{-1-j}) v^{-1+\veps}$ for any $j>1$, these properties enable us to integrate \eqref{19} along the integral curve of $-\mu Y$ from initial hypersurface to any point $(\tb,\rb)\in \{r\geq v^{\alpha}\}$ for some $\alpha\in (0,1)$ close to $1$. The value of $v^{2\sfrak+3}(\R)^{-\sfrak-1}\tildePhisHighell{+\sfrak}{m,\sfrak} (\tb,\rb)$ can then be computed, up to some terms with faster decay, by the initial data asymptotics and the integral of $v^{2\sfrak+3}(\R)^{-\sfrak-1}\tilde{\mathbf{N}}[\tildePhisHighell{+\sfrak,m}{\sfrak}]$ whose leading order behaviour is determined by the integrals $\int_{\Scri\cap[\tb_0,\infty)}\mellmode{\Psipluss}{m}{\sfrak}$ and  $\{\int_{\Scri\cap[\tb_0,\infty)}\mellmode{\PhiplussHigh{0}}{m}{\ell}\}_{\ell=\sfrak+1,\sfrak+2}$ that are already known in the above discussions. Given now the asymptotic profile of $(\R)^{-\sfrak-1}\tildePhisHighell{+\sfrak}{m,\sfrak} (\tb,\rb)$, one can simply integrate the $m$-mode projection form of  \eqref{14}  to deduce the asymptotic profile of $r^{-2\sfrak-1}\mellmode{\Phipluss}{m}{\sfrak}$ at any point $(\tb,\rb)\in \{r\geq v^{\alpha'}\}$ for some suitable $\alpha'\in (\alpha, 1)$.  In this region $\{r\geq v^{\alpha'}\}$, the asymptotic profiles of derivatives of $r^{-2\sfrak-1}\mellmode{\Phipluss}{m}{\sfrak}$ can also be computed, and the asymptotic profiles of derivatives of the spin $-\sfrak$ component are obtained utilizing the TSI \eqref{17b}.

The asymptotic profiles in the complement of region $\{r\geq v^{\alpha'}\}$ are easier to attain. Because of the proven faster decay of $\prb \ellmode{\psiminuss}{\sfrak}$ in region $r\leq \tb$, by choosing $\delta$ sufficiently small, the asymptotic profile of the spin $-\sfrak$ component simply propagates from $\{r=v^{\alpha'}\}$ to the region $\{r<v^{\alpha'}\}$. This asymptotic profile is finally utilized together with the TSI \eqref{17a} to compute the asymptotics of the spin $+\sfrak$ component in region $\{r<v^{\alpha'}\}$ as well as on $\Horizon$.

It is worthy noticing that the application of TSI is  imperative not only in deriving the almost sharp decay estimates in Section \ref{intro:almostsharp}, but also in computing the global asymptotic profiles of the spin $\pm \sfrak$ components.

\subsubsection{Comparison and relation to previous works}
\label{subsubsect:intro:comparison}

The main results of the current work can be viewed as an extension of the ones of our previous works  \cite{Ma20almost,MaZhang21PriceSchw} from Schwarzschild to Kerr, or of the works \cite{hintz2020sharp,angelopoulos2021late} from scalar field to spin $s$ fields. We compare and relate the techniques, the ideas and the results in this work to these relevant works in the following context. 

As can be seen in the above sections, many techniques and ideas in this work are direct, but still complicated, generalizations of the ones developed in \cite{Ma20almost,MaZhang21PriceSchw}. In Section \ref{sect:intro:wed}, we developed more complicated wave systems compared to the Schwarzschild case that is treated in \cite{MaZhang21PriceSchw} and none of the equations in the system is decoupled from the rest (this is in contrast to the Schwarzschild case where one does obtain a decoupled Regge--Wheeler equation in the wave system). This part in particular follows closely the analysis in \cite{andersson2019stability,Ma20almost}. 

The second main difference lies in obtaining the almost sharp decay estimates for the modes.  One central idea is to exploit the fact that the spectrum of the spin-weighted spherical Laplacian is away from $0$ for higher spin and higher modes, and this enables us to derive the $r^p$ estimates for an extended range of $p$ value, thus establish faster energy decay estimates. Such an extension of the $p$ range beyond $2$ is first due to
Angelopoulos--Aretakis--Gajic \cite{angelopoulos2018vector} where they derived the $r^p$ estimates for an extended range of $p$ value  for
the spherically symmetric part (that is, the $\ell= 0$ mode) of the scalar field on a Reissner-Nordstr\"om background. We  \cite{Ma20almost,MaZhang21PriceSchw} exploit further this property in the case of non-zero spin fields on Schwarzschild.
As described in Section \ref{intro:almostsharp}, the modes are actually coupled to each other in their governing equations, and this fact significantly increases the technical difficulty in deriving the $r^p$ estimates and the almost sharp pointwise decay estimates for the modes. 

The last main difference is a new, different proof for the sharp decay. Our previous work \cite{MaZhang21PriceSchw} in Schwarzschild follows the approach of Angelopoulos--Aretakis--Gajic \cite{angelopoulos2018late} for the scalar field on Reissner--Nordstr\"om by defining the so-called ``time-inverted Newman--Penrose constants'' from the Newman--Penrose constants that are fixed constants along null infinity. This is no longer straightforward in the Kerr case since one needs to invert an operator that is however non-elliptic inside the ergoregion of the Kerr spacetime, hence introducing one of the two main difficulties in generalizing to Kerr for the scalar field as shown by Angelopoulos--Aretakis--Gajic \cite{angelopoulos2021late}. 

Our new approach determines the coefficient of the leading order term for the spin $\pm \sfrak$ components in the asymptotics via an integral of the radiation field along null infinity.  
Such an integral of the radiation field along null infinity is first exploited by Luk--Oh in  \cite{lukoh17linearinstability} to prove the generic instability of the Cauchy horizon in subextreme Reissner--Nordstr\"om spacetimes against  linear scalar perturbations.  They identify a quantity $\mathcal{L}$, which is related to the integral of radiation field and defined by
\begin{align}
	\mathcal{L}\doteq\lim\limits_{v\to\infty}r^3\partial_v(r\phi)(u_0, v)-2M \int_{u_0}^\infty\Phi(u)\di u,
\end{align} 
where $\Phi(u)=\lim\limits_{v\to\infty}(r\phi)(u,v)$ is the radiation field of the scalar field $\phi$ on null infinity, and the generically nonvanishing property of this quantity $\mathcal{L}$ is fundamental in their proof. This quantity is further related to the time-inverted Newman--Penrose constant in \cite[Section 1.6]{angelopoulos2021late}.

In our present work, we compute this integral along null infinity purely from the initial data by employing a novel global conservation law of the field as described in Section \ref{intro:conser}. This enables us to treat the different spin $s$ components in a unified manner and deduce their precise late-time asymptotics globally.

\subsection{Outlook and future applications}
\label{outlook}

To end this introduction, we propose some potential applications of our result and method as well as some further problems.

\begin{enumerate}

\item Given the asymptotics on $\Horizon$  of the spin $\pm 2$ components of the lienarized gravity in subextreme Kerr spacetimes, it is interesting to consider the Strong Cosmic Censorship conjecture in the setting of the linearized gravity in the interior of  subextreme Kerr black holes.

\item It is natural to investigate the sharp asymptotics of higher modes of the spin $\pm \sfrak$ components in non-static subextreme Kerr spacetimes. The asymptotic decay rates for any $\ell$ mode in the region $\{r\geq \tb\}$ will be the same as the Schwarzschild case (that is, $v^{-1-\sfrak-s}\tb^{-2-\ell+s}$ asymptotic decay) but different in the region $\{r\leq \tb\}$. This has been verified in \cite{angelopoulos2021late} for scalar field in non-static Kerr spacetimes, and since the asymptotic decay rate of the $\ell$ mode  of $a^2\sin^2\theta \ptb^2\psi$ are determined by the rate of the $\tilde{\ell}$ mode of $\ptb^2 \psi$ with $\tilde{\ell}=\max\{0,\ell-2\}$, $\ellmode{\psi}{\ell}$ has $\tb^{-3-\ell}$ asymptotic decay for even $\ell$ and $r\tb^{-4-\ell}$ for odd $\ell$ in region $\{r\leq \tb\}$.
For $\sfrak\neq 0$,  in contrast to $\sfrak=0$ case,  the mode coupling arising from $ias\cos\theta  \ptb$ part will dominate the asymptotic decay rate, therefore, the scenario $\ellmode{\psi_{\pm\sfrak}}{\ell}\sim \ptb\ellmode{\psi_{\pm \sfrak}}{\ell-1}$ for any $\ell\geq  \sfrak+1$ is likely to be true, thereby, the $(m,\ell)$ mode $\ellmode{r^{-\sfrak-s}\psi_s}{m,\ell}$ is conjectured to have $v^{-1-\sfrak-s}\tb^{-2-\ell+s}$ global asymptotic decay for $s=\pm 1,\pm 2$ and have extra $\tb^{-1}$ decay on $\Horizon$ in the case that $s=1,2$ and $m=0$. (Note that this naive scenario may be invalid in some special cases, see \cite{Csukas19dynamicsofspinKerr} for more numerical discussions.)

\item It is of much importance to consider the asymptotics of the solutions to the following semilinear wave equations
\begin{align}
\label{semi1}
\Box_g \psi={} &N_1[\psi]\sim \pm \psi^p,\\
\label{semi2}
\Box_g \psi= {} & N_2[\psi] \sim  Y\psi V\psi +\nablaslash\psi \nablaslash\psi
\end{align}
arising from small initial data that are of size $\veps$ and decay rapidly as $\rb\to +\infty$.
Here, $p\geq 4$ is an integer, $Y$ and $V$ are the regular ingoing and outgoing derivatives, and $\nablaslash$ is the covariant angular derivative over $S^2(r)$. 

The first model problem \eqref{semi1}  has been intensively studied in the literature for small initial data in both aspects of global existence (related to the Strauss conjecture) in \cite{georgiev1997weighted,tataru2001strichartz,lindblad2014strauss} and references therein and  sharp decay rates \cite{szpak2008linear,bieli2010global}. For large initial data, see \cite{grillakis1990regularity}. Quite recently, Tohaneanu \cite{tohaneanu2021pointwise} proved the optimal pointwise upper bounds $\langle t\rangle^{-1}\langle t-r\rangle^{-\kappa}$ with $p\geq 3$ and $\kappa=\min\{2,p-2\}$ for solutions arsing from small initial data in Kerr spacetimes. The second model problem \eqref{semi2} is a prototye of wave equations respecting the null condition \cite{klainerman1986null,christodoulou1986global}.

We are interested in providing the precise asymptotic profiles for both models \eqref{semi1} and \eqref{semi2} on Kerr backgrounds. To briefly illustrate how our novel idea of global conservation law can be employed to derive the asymptotic profiles, we take the model problem \eqref{semi1} with $g$ being the Schwarzschild metric as an example.  The approach developed in this work is expected to be adapted to show suitable decay for $\psi$, and, in particular, one can still derive  an \textbf{almost}, global conservation law that provides the approximate value of the integral of the radiation field along future null infinity, in view that the integral from the source term $N_1[\psi]$ is bounded by $O(\veps^p)$, negligible compared to the contribution from the initial data of size $\veps$. The remaining discussions in Section \ref{intro:conser} apply and yield that the asymptotic profiles for $\psi$ in Theorem \ref{mainthm} are valid up to a correction term which is $O(\veps^p)$ times the same asymptotic decay rate. We will address rigorously the asymptotic profiles of solutions to the semilinear models \eqref{semi1} and \eqref{semi2} in a future work.
\end{enumerate}

\subsection*{Overview of the paper}

In Section \ref{sect:prel}, we define the hyperboloidal coordinates, a few sets of operators and norms, discuss the mode projection and present some elementary estimates.  We then introduce the TME and TSI and derive various systems of equations from the TME in Section \ref{sect:sys:of:eqs}.
In Section  \ref{sect:APL}, the BEAM estimate assumption is introduced,  and based on this assumption, we show almost sharp decay for the spin $s$ components.  Section  \ref{sect:PL} is devoted to  proving a global conservation law and deriving the globally precise late time asymptotics. In the end, we provide in Appendix \ref{app:scalars} a table of notations for the scalars constructed from the spin $\pm \sfrak$ components.

\section{Geometry and preliminaries}
\label{sect:prel}

\subsection{A hyperboloidal foliation of the spacetime}
\label{sect:foliation}

Let 
\begin{align}\mu=\frac{\Delta}{r^2+a^2},
\end{align}
 and define a tortoise coordinate $r^*$ by
\begin{align}
\di r^*=\mu^{-1}\di r,\qquad r^*(3M)=0.
\end{align}
The Boyer--Lindquist coordinate system is not regular at the event horizon, so we shall use a different coordinate system--the ingoing Eddington--Finkelstein coordinate system $(v,r,\theta,\tilde{\phi})$--which is regular at the future event horizon $\Horizon$ and defined by
\begin{equation}\label{def:IngoingEddiFinkerCoord}
\left\{
  \begin{array}{ll}
   v=t +r^*, \\
    \di \pb=\di \phi +a(r^2 +a^2)^{-1}\di r^*,\\
   r=r,\\
    \theta=\theta.\\
  \end{array}
\right.
\end{equation}
The coordinate $v$ is known as the forward time, and there is an analogous retarded time $u$ which is defined by $u= t -r^*$.

 Define  a \textit{hyperboloidal coordinate system} $(\tb,\rb=r,\theta,\pb)$ as in \cite{andersson2019stability}, with $\rb=r$, $\tb=v-\hhyp$ and $\hhyp=\hhyp(r)$, such that the level sets of the time function $\tb$ are strictly spacelike with
\begin{align}
c(M,a)r^{-2}\leq -g(\nabla \tb,\nabla\tb)\leq C(M,a) r^{-2}
\end{align}
for two positive universal constants $c(M)$ and $C(M)$
and they cross the future event horizon regularly and are asymptotic to future null infinity $\Scri$, and  for large $r$, $1\lesssim \lim_{\rb\to \infty}r^2 (\partial_r\hhyp - 2\mu^{-1})\vert_{\Sigmatb}<\infty$.

Define a function related to the hyperboloidal foliation
\begin{align}
\Hhyp\doteq 2\mu^{-1}-\partial_r \hhyp.
\end{align}
By the choice of the hyperboloidal coordinates, \begin{align}
\label{eq:propertyofHfunction}
r^2 \Hhyp\lesssim 1 \quad
\text{for  } r  \text{ large}, \quad \text{and} \quad \abs{ \Hhyp-2\mu^{-1}}\lesssim 1  \quad
\text{as } r\to r_+.
\end{align}

Let $\Sigmatb$ be the constant $\tb$ hypersurface in the domain of outer communication $\DOC$.  Let $\tb_0\geq 1$, and let $\Sigmazero$ be our initial hypersurface on which the initial data are imposed. For any $\tb_0\leq \tb_1<\tb_2$, let $\Donetwo$, $\Scrionetwo$ and $\Horizononetwo$ be the truncated parts of $\DOC$, $\Scri$ and $\Horizon$ on $\tb\in [\tb_1,\tb_2]$, respectively. See Figure \ref{fig:2}.

\begin{figure}[htbp]
\begin{center}
\begin{tikzpicture}[scale=0.8]
\draw[thin]
(0,0)--(2.45,2.45);
\draw[very thin]
(2.5,2.5) circle (0.05);
\coordinate [label=90:$i_+$] (a) at (2.5,2.5);
\draw[thin]
(0,0)--(2.45,-2.45);
\draw[dashed]
(2.55,2.45)--(4.95,0.05);
\draw[very thin]
(5,0) circle (0.05);
\coordinate [label=360:$i_0$] (a) at (5,0);
\draw[dashed]
(4.95,-0.05)--(2.55,-2.45);
\draw[very thin]
(2.5,-2.5) circle (0.05);
\coordinate [label=270:$i_-$] (a) at (2.5,-2.5);
\draw[thin]
(0.3,0.3) arc (210:327:2.65 and 2.2);
\node at (2.5,-1.1) {\small $\Sigma_{\tau_0}$};
\draw[thin]
(0.9,0.9) arc (215:324:2.0 and 1.5);
\node at (2.5,-0.15) {\small $\Sigma_{\tau_1}$};
\draw[thin]
(1.5,1.5) arc (212:323:1.25 and 0.9);
\node at (2.5,1.35) {\small $\Sigma_{\tau_2}$};
\draw[very thick]
(0.9,0.9)--(1.5,1.5);
\draw[dashed, very thick]
(3.57,1.43)--(4.28,0.72);
\node at (0.95,1.45) [rotate=45] {\small $\mathcal{H}_{\tau_1,\tau_2}^+$};
\node at (4.15,1.35) [rotate=-45] {\small $\mathcal{I}_{\tau_1,\tau_2}^+$};
\node at (2.5,0.7) {\small $\Donetwo$};
\end{tikzpicture}
\end{center}
\caption{Hyperboloidal foliation and related definitions.}
\label{fig:2}
\end{figure}
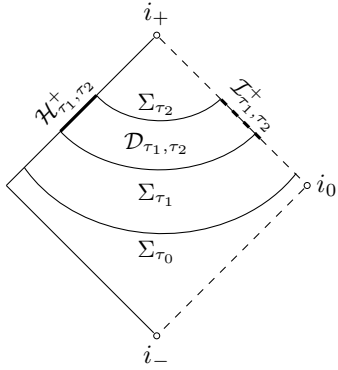

Furthermore, we define a few $3$- and $4$-dimensional subregions of $\Sigmatb$ and $\DOC$.

\begin{definition}
Let $\tb_2>\tb_1\geq \tb_0$  and let $r_2>r_1\geq r_+$. Define
\begin{subequations}
\label{def:domainnotations:subdomains}
\begin{align}
\Sigmaone^{\geq r_1}={}&\Sigmaone\cap \{r\geq r_1\}, & \Donetwo^{\geq r_1}={}&\Donetwo\cap\{r\geq r_1\},\\
\Sigmaone^{r_1,r_2}={}&\Sigmaone\cap \{r_1\leq r\leq r_2\}, &\Donetwo^{r_1,r_2}={}&\Donetwo\cap\{r_1\leq r\leq r_2\},\\
\Sigmaone^{\leq r_1}={}&\Sigmaone\cap \{r_+\leq r\leq r_1\}, & \Donetwo^{\leq r_1}={}&\Donetwo\cap\{r_+\leq r\leq r_1\}.
\end{align}
\end{subequations}
\end{definition}



\subsection{General conventions}

$\mathbb{N}$ is denoted as the natural number set $\{0,1,\ldots\}$, $\mathbb{N}^+$ the positive natural number set, $\mathbb{Z}^+$ the positive integer set, $\mathbb{R}$ the real number set, and $\mathbb{R}^+$ the positive real number set.  Denote $\Re(\cdot)$ as the real part.

LHS and RHS are short for left-hand side and right-hand side, respectively.

Constants in this work may depend on the hyperboloidal foliation via the function $\hhyp$. For simplicity, we shall always suppress this dependence throughout this work as one can fix this function once for all. For the same reason, the dependence on the mass parameter $M$ and angular momentum per mass $a$ is always suppressed as well.

We denote a universal constant by $C$ if it depends only on the hyperboloidal foliation (via the function $\hhyp$), mass $M$ and angular momentum $a$. If a universal constant depends on a set of other parameters $\mathbf{P}$, we denote it by $C(\mathbf{P})$. Regularity parameters are generally denoted by $\reg$, and $\regl$  is a universal constant. Also, $\regl(\mathbf{P})$ means a regularity constant depending on the parameters in the set $\mathbf{P}$.

We say $F_1\lesssim F_2$ if there exists a universal constant $C$ such that $F_1\leq CF_2$. Similarly for $F_1\gtrsim F_2$. If both $F_1\lesssim F_2$ and $F_1\gtrsim F_2$ hold, we say $F_1\sim F_2$.

Let $\mathbf{P}$ be a set of parameters. We say $F_1\lesssim_{\mathbf{P}} F_2$ if there exists a universal constant $C(\mathbf{P})$ such that $F_1\leq C(\mathbf{P})F_2$. Similarly for $F_1\gtrsim_{\mathbf{P}} F_2$. If both $F_1\lesssim_{\mathbf{P}} F_2$ and $F_1\gtrsim_{\mathbf{P}} F_2$ hold, then we say $F_1\sim_{\mathbf{P}} F_2$.

For any $\alpha \in \mathbb{R}^+\cup\{0\}$, we say a function $f(r,\theta,\pb)$ is $O(r^{-\alpha})$ if  for any $j\in \mathbb{N}$,
$\abs{(\partial_r)^j f_2}\leq C_j r^{-\alpha-j}$ as $r\to \infty$.

For any $x\in \mathbb{R}$, let the Japanese bracket be defined by $\langle x\rangle=\sqrt{x^2+1}$.


\subsection{Operators and norms}

In this subsection, we define various operators and introduce relevant norms.

To start with, we need the following definitions of spin-weighted scalars and spin-weighted operators.

\begin{definition}
\begin{itemize}
\item A scalar which has proper spin weight and zero boost weight in the sense of Geroch, Held and Penrose \cite{geroch1973space} is called a \textit{spin-weighted scalar}.\footnote{In particular, the spin-weighted scalars are sections of complex line bundles.}
    Unless otherwise stated, we shall always denote $s$ the spin weight, and we call a spin-weighted scalar with spin weight $s$ as a \textit{spin $s$ scalar}.
\item A differential operator is a \textit{spin-weighted operator} if it takes a spin-weighted scalar to a spin-weighted scalar.
\end{itemize}
\end{definition}

Our abstract definition of the pointwise norms of a spin-weighted scalar is as follows.

\begin{definition}
Let $\mathbb{X}=\{X_1, X_2, \ldots, X_n\}$, $n\in \mathbb{N}^+$, be a set of spin-weighted operators, and let a multi-index $\mathbf{a}$ be an ordered set $\mathbf{a}=(a_1,a_2,\ldots,a_m)$ with all $a_i\in \{1,\ldots, n\}$.  Let $m=|\mathbf{a}|$, and define $\mathbb{X}^{\mathbf{a}}=X_{a_1}X_{a_2}\cdots X_{a_m}$ if $m\in \mathbb{N}^+$ and $\mathbb{X}^{\mathbf{a}}$ as the identity operator if $m=0$. Let $\varphi$ be a spin-weighted scalar, and define its pointwise norm of order $m$, $m\in \mathbb{N}$, as
\begin{align}
\absHighOrder{\varphi}{m}{\mathbb{X}}\doteq{}\sqrt{\sum_{\abs{\mathbf{a}}\leq m}\abs{\mathbb{X}^{\mathbf{a}}\varphi}^2}.
\end{align}
\end{definition}

In order to properly define the above norms, we shall introduce  (spin-weighted) operators.

\begin{definition}
\label{def:basic:vectorfields}
\begin{itemize}
\item
For a spin $s$ scalar $\varphi_s$, define
the \textit{spherical edth operators} $\edthR$ and $\edthR'$  by
\begin{align}\label{def:setsofopers}
\edthR\varphi_s\doteq{}&\partial_{\theta}\varphi_s
+{i}csc\theta\partial_\pb\varphi_s
-{s}cot\theta\varphi_s,&
\edthR'\varphi_s\doteq{}&\partial_{\theta}\varphi_s
-{i}\csc\theta\partial_\pb\varphi_s
+{s}cot\theta\varphi_s.
\end{align}
\item
Define two Killing vector fields
\begin{align}
\label{def:Killingvectors}
\Lxi\doteq{}\partial_{\tb}, \quad \Leta\doteq\partial_{\pb}.
\end{align}

\item
Define the regular, future-directed ingoing and outgoing principal null vector fields
\begin{align}
Y\doteq \sqrt{2}n^\mu\partial_\mu=\frac{(r^2+a^2)\partial_t + a\partial_{\phi}}{\Delta}-\partial_r , \,\,V\doteq\frac{\sqrt{2}\Sigma}{r^2+a^2}l^\mu \partial_\mu=\frac{(r^2+a^2)\partial_t +a\partial_\phi}{r^2+a^2}+\frac{\Delta}{r^2+a^2}\partial_{r}.
\end{align}
Further, define
\begin{align}
\VR\doteq \mu^{-1}V=\frac{(r^2+a^2)\partial_t + a\partial_{\phi}}{\Delta}+\partial_r.
\end{align}
Last, for latter use of application, define vector fields
\begin{align}\label{def:curlVR}
\curlVR\doteq{}(\R)\VR, \quad \curlV\doteq{} (\R)V
\end{align}
that are conformally regular near null infinity.
\item Define two vector fields
\begin{align}
\tildeV\doteq V-\frac{2a}{\R}\Leta, \quad \tildeY \doteq Y -\frac{2a}{\Delta}\Leta.
\end{align}
They satisfy
\begin{align}
\tildeV+\mu Y=\mu \tilde Y +V = 2\Lxi.
\end{align}
\end{itemize}
\end{definition}

\begin{remark}
\begin{itemize}
\item
Note that if $\varphi_s$ is a spin $s$ scalar, then $\edthR\varphi_s$ and $\edthR'\varphi_s$ are spin $s+1$ and $s-1$ scalars, respectively. That is, $\edthR$ increases the spin weight by $1$, while $\edthR'$ decreases it by $1$.
\item The second-order angular operators $\edthR\edthR'$ and $\edthR'\edthR$, which are both Killing $(2,0)$ tensors, are 
\begin{subequations}
\begin{align}
 \edthR\edthR'\varphi_s
={}&\frac{1}{\sin{\theta}} \partial_{\theta}( \sin{\theta} \partial_{\theta}\varphi_s)
+\frac{1}{\sin^2{\theta}}\partial_{\pb\pb}^2\varphi _s+ \frac{2i s\cos{\theta}}{\sin^2{\theta} }  \partial_{\pb} \varphi_s- (s^2 \cot^2{\theta} +s)\varphi_s,\\
 \edthR'\edthR\varphi_s
={}&\frac{1}{\sin{\theta}} \partial_{\theta}( \sin{\theta} \partial_{\theta}\varphi_s)
+\frac{1}{\sin^2{\theta}}\partial_{\pb\pb}^2\varphi _s+ \frac{2i s\cos{\theta}}{\sin^2{\theta} }  \partial_{\pb} \varphi_s- (s^2 \cot^2{\theta} -s)\varphi_s,
\end{align}
\end{subequations}
when acting on a spin $s$ scalar $\varphi_s$.
\item
One can express  $Y$, $V$ and $\VR$ in the hyperboloidal coordinates as
\begin{align}
\label{def:vectorVRintermsofprb}
Y={}-\prb+(2\mu^{-1}-\Hhyp)\Lxi,\quad V=\mu \prb +\mu \Hhyp \Lxi+\frac{2a}{\R}\Leta, \quad
\VR={}\partial_{\rb}+\Hhyp \Lxi+\frac{2a}{\Delta}\Leta.
\end{align}
\end{itemize}
\end{remark}

We derive the commutators between different operators.

\begin{prop}
\label{prop:comms}
It holds that
\begin{subequations}
\label{eq:comms}
\begin{align}
[Y,\edthR]={}&[Y,\edthR']=[V,\edthR]=[V,\edthR']=[Y,\Lxi]=[V,\Lxi]=[Y,\Leta]=[V,\Leta]=0,\\
\label{comm:muYandV}
[\mu Y, V]={}&\frac{4ar \mu }{(\R)^2}\Leta,\\
\label{comm:muYandtildeVandvv}
[\mu Y , \tildeV]={}&[\mu \tildeY , V]= 0.
\end{align}
\end{subequations}
\end{prop}

\begin{proof}
The first formula is manifest.
Formula \eqref{comm:muYandV} follows from
\begin{align*}
[\mu Y, V]={}&[\partial_t +\frac{a}{\R}\partial_{\phi}-\mu \partial_r , \partial_t +\frac{a}{\R}\partial_{\phi}+\mu \partial_r ]=-2\mu \partial_r \Big(\frac{a}{ \R}\Big) \partial_{\phi}.
\end{align*}
The last formula \eqref{comm:muYandtildeVandvv} can be seen by substituting in $\tildeV=2\Lxi - \mu Y$ and $\mu Y=2\Lxi- V$.
\end{proof}

Define a few operator sets as follows:

\begin{definition}
\label{def:setsofopers:commutators}
Define a set of operators
\begin{subequations}
\begin{align}
\PDeri\doteq{}\{Y,V, r^{-1}\edthR,r^{-1}\edthR'\}
\end{align}
adapted to the Hartle--Hawking tetrad, and its rescaled one
\begin{align}
\PSDeri\doteq{}\{rY,rV, \edthR,\edthR'\}.
\end{align}
Define a set of operators
\begin{align}
\CDeri\doteq{}\{Y,rV, \Leta, \edthR,\edthR'\}
\end{align}
\end{subequations}
which is adapted to the hyperboloidal foliation and will be the set of commutators.
\end{definition}

In the end,  we define a few energy norms and (spacetime) Morawetz norms for spin-weighted scalars.

\begin{definition}
Define the following reference volume elements
\begin{align}
\di^2\mu=\sin\theta \di \theta \wedge \di \pb, \quad \di^3 \mu =\di \rb\wedge \di^2\mu, \quad \di^4\mu =\di \tb\wedge\di^3 \mu.
\end{align}
\end{definition}

\begin{definition}
\label{def:basicweightednorm}
Let $\varphi$ be a spin-weighted scalar. Let $k\in \mathbb{N}$ and $\gamma\in \mathbb{R}$. Let $\Omega$ be a $4$-dimensional subspace of the DOC and let $\Sigma$ be a $3$-dimensional space that can be parameterized by $(\rb,\theta,\pb)$.  Define energy norms and Morawetz norms by
\begin{subequations}
\begin{align}
\norm{\varphi}_{W_{\gamma}^{\reg}(\Sigma)}^2
={}&\int_{\Hyper} r^{\gamma}\absCDeri{\varphi}{\reg}^2\di^3\mu,\\
\norm{\varphi}_{W_{\gamma}^{\reg}(\Omega)}^2
={}&\int_{\Omega} r^{\gamma}\absCDeri{\varphi}{\reg}^2\di^4\mu.
\end{align}
\end{subequations}
\end{definition}

\subsection{Spin-weighted spherical harmonic mode projection}
\label{sect:decompIntoModes}

In this subsection, we define the projection of a spin $s$ scalar onto spin-weighted spherical harmonic modes and discuss a few properties of the mode projection.

Recall that $\{Y_{m,\ell}^{s}(\cos\theta)e^{im\pb}\}_{m,\ell}$ are the eigenfunctions, called as  the \textquotedblleft{spin-weighted spherical harmonics,\textquotedblright}
 of a self-adjoint operator
$\edthR'\edthR$:
\begin{equation}
\label{eq:eigenvalueSWSHO}
\edthR'\edthR(Y_{m,\ell}^{s}(\cos\theta)e^{im\pb})=
-(\ell-s)(\ell+s+1)
Y_{m,\ell}^{s}(\cos\theta)e^{im\pb}.
\end{equation}
They form a complete orthonormal basis on $L^2(\di^2\mu)$. Further,
\begin{subequations}
\label{eq:ellipticop:eigenvalue:fixedmode}
\begin{align}
\edthR (Y_{m,\ell}^{s}(\cos\theta)e^{im\pb})={}&-\sqrt{(\ell-s)(\ell+s+1)}Y_{m,\ell}^{s+1}(\cos\theta)e^{im\pb}, \\
\edthR' (Y_{m,\ell}^{s}(\cos\theta)e^{im\pb})={}& \sqrt{(\ell+s)(\ell-s+1)}Y_{m,\ell}^{s-1}(\cos\theta)e^{im\pb}.
\end{align}
\end{subequations}

The mode projection is defined as follows.

\begin{definition}
\label{def:projectiondefinitions}
For any $(m,\ell)$ with $-\ell\leq m\leq\ell$ and $\ell\geq |s|$, we define the projection of
a spin $s$ scalar $\varphi_s$ onto a fixed spin-weighted spherical harmonic mode as
\begin{align}
\PJ_{m,\ell}^{s}(\varphi_s)\doteq \int_{S^2} \varphi_s\cdot \overline{Y_{m,\ell}^s(\cos\theta) e^{i m \pb}}\di^2\mu.
\end{align}
Meanwhile,  define the projection of $\varphi_s$ onto an $\ell$ mode as
\begin{align}
\PJ_{\ell}^s(\varphi_s)\doteq\sum_{m=-\ell}^{\ell}\PJ_{m,\ell}^{s}(\varphi_s) Y_{m,\ell}^s(\cos\theta) e^{im\pb}.
 \end{align}
 Further, we can define the projection onto $\geq \ell$ modes by
 \begin{align}
\PJ_{\geq\ell}^s(\varphi_s)\doteq\sum_{\ell'\geq \ell}\PJ_{\ell'}^s(\varphi_s).
 \end{align}
 When there is no confusion, we may drop the superscript $s$ that indicates the spin weight, and write $\PJ_{m,\ell}^{s}(\varphi_s)$, $\PJ_{\ell}^s(\varphi_s)$ and $\PJ_{\geq\ell}^s(\varphi_s)$ as $\PJ_{m,\ell}(\varphi_s)$, $\PJ_{\ell}(\varphi_s)$ and $\PJ_{\geq\ell}(\varphi_s)$ respectively. For simplicity, we may denote them by $\mellmode{\varphi_s}{m}{\ell}$, $\ellmode{\varphi_s}{\ell}$ and $\ellmode{\varphi_s}{\geq\ell}$ respectively.
\end{definition}

\begin{remark}
We shall make the following conventions. For an $(m,\ell)$ mode $\mellmode{\varphi_s}{m}{\ell}$ of a spin $s$ scalar $\varphi_s$, we shall use the convention:
\begin{align}
\Leta \mellmode{\varphi_s}{m}{\ell}\doteq \mellmode{\Leta\varphi_s}{m}{\ell}=im\mellmode{\varphi_s}{m}{\ell}.
\end{align}
Similarly, we adopt the convention $V\mellmode{\varphi_s}{m}{\ell}= \mu \prb \mellmode{\varphi_s}{m}{\ell}+\mu \Hhyp \Lxi \mellmode{\varphi_s}{m}{\ell}+\frac{2iam}{\R}\mellmode{\varphi_s}{m}{\ell}$.
 Further, its norm shall be understood as follows
 \begin{align}
 \absCDeri{\mellmode{\varphi_s}{m}{\ell}}{\reg}^2\doteq  \absCDeri{\mellmode{\varphi_s}{m}{\ell}Y^{s}_{m,\ell}e^{im\pb}}{\reg}^2.
 \end{align}
\end{remark}

In particular, by definition, it holds in $L^2 (\Sphere)$ that
\begin{align}
\label{eq:l=l0mode:eigenvalue}
\edthR\edthR'\ellmode{\varphi_s}{\ell}
={}-(\ell+s)(\ell-s+1)\ellmode{\varphi_s}{\ell}, \quad
\edthR'\edthR\ellmode{\varphi_s}{\ell}
={}-(\ell-s)(\ell+s+1)\ellmode{\varphi_s}{\ell}.
\end{align}

\begin{lemma}
\label{lem:eigenvalue:edthandprime}
Let $\varphi_s$ be a spin $s$ scalar, then
\begin{align}
\label{eq:ellipestis}
\int_{S^2}\big(\abs{\edthR'\varphi_s}^2
-({s+\abs{s}})\abs{\varphi_s}^2\big) \di^2\mu =\int_{S^2}\big(\abs{\edthR\varphi_s}^2
-({\abs{s}-s})\abs{\varphi_s}^2\big) \di^2\mu \geq 0.
\end{align}
If $\varphi_s$ is a spin $s$ scalar and supported on $\geq \ell$ modes, then
\begin{align}
\label{eq:ellip:highermodes}
\int_{S^2}\big(\abs{\edthR'\varphi_s}^2
-{(\ell+s)(\ell-s+1)}\abs{\varphi_s}^2\big) \di^2\mu
={}&\int_{S^2}\big(\abs{\edthR\varphi_s}^2
-{(\ell-s)(\ell+s+1)}\abs{\varphi_s}^2\big) \di^2\mu \geq 0.
\end{align}
\end{lemma}

The following mode projection statements are necessary when projecting the TME \eqref{eq:TME} or \eqref{eq:TME:radfield} onto modes.
\begin{prop}
\label{prop:modeprojection:1}
 Let $s=0, \pm1,\pm2$, and let $\ell\geq |s|$. Let  $\varphi_s$ be a spin $s$ scalar. Then there exist constants $\{c_{m,\ell}^s\}$ and $\{b_{m,\ell}^s\}$, with $\abs{m}\leq \ell$, such that
\begin{align}
\PJ_{m,\ell}^s(\sin^2\theta \varphi_s)&=\sum_{\ell'=\ell-2}^{\ell+2}c_{m,\ell'}^s \PJ_{m,\ell'}^s(\varphi_s),\\
\PJ_{m,\ell}^s(\cos\theta \varphi_s)&=\sum_{\ell'=\ell-1}^{\ell+1}b_{m,\ell'}^s \PJ_{m,\ell'}^s(\varphi_s).
\end{align}
In the above relations, we have set all $c_{m,\ell}^s$ and $b_{m,\ell}^s$, for $\ell<\sfrak$, to zero. Moreover, the constants $c_{0,\ell\pm 1}^s$ and $b_{0,\ell}^s$ in the above formulas vanish.
\end{prop}

\begin{proof}
By definition, we have
\begin{align}\begin{split}
\PJ_{m,\ell}^{s}(\sin^2\theta \varphi_s)=&\sum\limits_{\ell'\geq\max\{|s|,|m|\}}\PJ_{m,\ell'}^s(\varphi_s)\int_{S^2}\sin^2\theta Y_{m,\ell'}^{s}(\cos\theta)\overline{Y}_{m,\ell}^s(\cos\theta) \sin\theta \di\theta \di \pb,\\
\PJ_{m,\ell}^{s}(\cos\theta \varphi_s)=&\sum\limits_{\ell'\geq\max\{|s|,|m|\}}\PJ_{m,\ell'}^s(\varphi_s)\int_{S^2}\cos\theta Y_{m,\ell'}^{s}(\cos\theta)\overline{Y}_{m,\ell}^s(\cos\theta) \sin\theta \di\theta \di \pb.
\end{split}\end{align}
Then the desired result follows from the properties of Wigner $3j$-functions and the Clebsch--Gordan coefficients. See \cite{Hod99Mode} for more details.
\end{proof}


\subsection{Elementary analytic estimates}

Since we are treating complex spin-weighted scalars, the following integration by parts in terms of the edth operators $\edthR$ and $\edthR'$ over sphere is necessary. It is a standard fact.

\begin{lemma}
\label{lem:IBPonSphere}
Let $s\in \half \mathbb{Z}$. For two spin-weighted scalars $f$ and $h$ with spin weight $s+1$ and $s$ respectively, we have
\begin{align}
\label{eq:IBPonsphere:1}
\int_{\Sphere}\Re(\bar{f}\edthR h)\di^2\mu={}&
-\int_{\Sphere}\Re(\overline{\edthR'f}h)\di^2 \mu,
\end{align}
\end{lemma}

\begin{proof}
By using the expression \eqref{def:setsofopers} of $\edthR$ to expand the LHS of \eqref{eq:IBPonsphere:1}:
\begin{align*}
\int_{\Sphere}\Re(\bar{f}\edthR h)\di^2\mu ={}&
\int_{\Sphere}\Re(\bar{f} (\partial_{\theta}h +i\csc\theta\partial_{\pb}h-s\cot\theta h)\sin\theta\di\theta\di\pb\notag\\
={}&\int_{\Sphere}\Big(\partial_{\theta}\big(\Re(\bar{f} h \sin\theta)\big)
+\partial_{\pb}\big(\Re(i \bar{f}h)\big)\Big)\di\theta\di\pb\notag\\
&
+\int_{\Sphere}\Re\Big(
-\overline{\partial_{\theta}f}h+\overline{i\csc\theta\partial_{\pb}f}h-(s+1)\cot\theta\bar{f}h\Big)
\Big)\sin\theta\di \theta\di \pb,
\end{align*}
one finds that the second last line vanishes and the last line equals the RHS of \eqref{eq:IBPonsphere:1} in view of the expression \eqref{def:setsofopers} of the operator $\edthR'$.
\end{proof}

The following simple Hardy's inequality will be useful.
\begin{lemma}
Let $\varphi_s$ be a spin $s$ scalar. Then for any $r'>r_+$,
\begin{align}
\label{eq:Hardy:trivial}
\int_{r_+}^{r'}\abs{\varphi_s}^2\di r
\lesssim{}&\int_{r_+}^{r'}\mu^2r^2\abs{\partial_r\varphi_s}^2\di r
+(r'-r_+)\abs{\varphi_s(r')}^2.
\end{align}
If, moreover, $\lim\limits_{r\to \infty} r\abs{\varphi_s}^2 =0$, then
\begin{align}
\label{eq:Hardy:trivial:1}
\int_{r_+}^{\infty}\abs{\varphi_s}^2\di r
\lesssim{}&\int_{r_+}^{\infty}\mu^2r^2\abs{\partial_r\varphi_s}^2\di r.
\end{align}
\end{lemma}

\begin{proof}
It follows easily by integrating the following equation
\begin{align}
\partial_r((r-r_+)\abs{\varphi}^2)=\abs{\varphi}^2
+2(r-r_+)\Re(\bar{\varphi}\partial_r\varphi)
\end{align}
from $r_+$ to $r'$ and applying the Cauchy-Schwarz inequality to the last product term.
\end{proof}

We will also use the following standard Hardy's inequality cited from \cite[Lemma 4.30]{andersson2019stability}. Its proof is standard and can be found therein.

\begin{lemma}[One-dimensional Hardy estimates]
\label{lem:HardyIneq}
Let $\alpha \in \mathbb{R}\setminus \{0\}$  and $h: [r_0,r_1] \rightarrow \mathbb{R}$ be a $C^1$ function.
\begin{enumerate}
\item \label{point:lem:HardyIneqLHS} If $r_0^{\alpha}\vert h(r_0)\vert^2 \leq D_0$ and $\alpha<0$, then
\begin{subequations}
\label{eq:HardyIneqLHSRHS}
\begin{align}\label{eq:HardyIneqLHS}
-2\alpha^{-1}r_1^{\alpha}\vert h(r_1)\vert^2+\int_{r_0}^{r_1}r^{\alpha -1} \vert h(r)\vert ^2 \di r \leq \frac{4}{\alpha^2}\int_{r_0}^{r_1}r^{\alpha +1} \vert \partial_r h(r)\vert ^2 \di r-2\alpha^{-1}D_0;
\end{align}
\item \label{point:lem:HardyIneqRHS} If $r_1^{\alpha}\vert h(r_1)\vert^2 \leq D_0$ and $\alpha>0$, then
\begin{align}\label{eq:HardyIneqRHS}
2\alpha^{-1}r_0^{\alpha}\vert h(r_0)\vert^2+\int_{r_0}^{r_1}r^{\alpha -1} \vert h(r)\vert ^2 \di r \leq \frac{4}{\alpha^2}\int_{r_0}^{r_1}r^{\alpha +1} \vert \partial_r h(r)\vert ^2 \di r +2\alpha^{-1}D_0.
\end{align}
\end{subequations}
\end{enumerate}
\end{lemma}

Further, recall the following Sobolev-type estimates from \cite[Lemmas 4.32 and 4.33]{andersson2019stability} where the proof is also provided.

\begin{lemma}
\label{lem:Sobolev}
Let $\varphi_s$ be a spin $s$ scalar. Then
\begin{align}
\label{eq:Sobolev:1}
\sup_{\Sigmatb}\abs{\varphi_s}^2\lesssim_{s}{} \norm{\varphi_s}_{W_{-1}^3(\Sigmatb)}^2.
\end{align}
If $\alpha\in (0,1]$, then
\begin{align}
\label{eq:Sobolev:2}
\sup_{\Sigmatb}\abs{\varphi_s}^2\lesssim_{s,\alpha} {} (\norm{\varphi_s}_{W_{-2}^3(\Sigmatb)}^2
+\norm{rV\varphi_s}_{W_{-1-\alpha}^2(\Sigmatb)}^2)^{\half}
(\norm{\varphi_s}_{W_{-2}^3(\Sigmatb)}^2
+\norm{rV\varphi_s}_{W_{-1+\alpha}^2(\Sigmatb)}^2)^{\half}.
\end{align}
If $\lim\limits_{{\tb\to\infty}}\abs{r^{-1}\varphi_s}=0$ pointwise in $(\rb,\theta,\pb)$, then
\begin{align}
\label{eq:Sobolev:3}
\abs{r^{-1}\varphi_s}^2\lesssim_{s} {}\norm{\varphi_s}_{W_{-3}^3(\DOC_{\tb,\infty})}
\norm{\Lxi\varphi_s}_{W_{-3}^3(\DOC_{\tb,\infty})}.
\end{align}
\end{lemma}

Finally, we provide a lemma showing that a hierarchy of energy and Morawetz estimates implies a rate of decay for the energy in the hierarchy. The way this lemma is stated is the same as \cite[Lemma 5.2]{andersson2019stability} and we have taken the simpler case $\gamma=0$. In applications, $\reg$ represents a level of regularity, $p$ represents a weight, and $\PigeonTime$ represents a time coordinate. Further, $\regl$ characterizes the potential loss of regularity in the hierarchy of energy and Morawetz estimates.

\begin{lemma}[A hierarchy of energy and Morawetz estimates implies energy decay]
\label{lem:hierarchyImpliesDecay:73}
Let $p_1,p_2\in\Reals$ be such that $p_1\leq p_2-1$, let $\regl\geq 0$, and let $\reg_0\in\Integers^+$ be suitably large. Let $F:\{0,\ldots,\reg_0\}\times[p_1-1,p_2]\times[\tb_0,\infty)\rightarrow[0,\infty)$ be such that $F(\reg,p,\tb)$ is Lebesgue measurable in $\tb$ for each $p$ and $\reg$. Let $D: \{0,\ldots,\reg_0\}\times [p_1,p_2]\times [\tb_0,\infty)\rightarrow[0,\infty)$ be such that $D(\reg, p,\tb)$ is Lebesgue measurable in $\tb$ for each $p$ and $\reg$.

If
\begin{subequations}
\begin{enumerate}
\item{} [monotonicity] \label{assump:HierarchyToDecay(1)}for all $\reg,\reg_1,\reg_2\in\{0,\ldots,\reg_0\}$ with $\reg_1\leq \reg_2$, all $p, \beta_1,\beta_2\in[p_1,p_2]$ with $\beta_1\leq\beta_2$, and all $\tb\geq \tb_0$,
\begin{align}
F(\reg_1,p,\tb)\lesssim{}& F(\reg_2,p,\tb) ,
\label{eq:Rev:HierarchyToDecayReal:MonotonicityHypothesis:j}\\
F(\reg,\beta_1,\tb)\lesssim{}& F(\reg,\beta_2,\tb) ,
\label{eq:Rev:HierarchyToDecayReal:MonotonicityHypothesis:beta}
\end{align}
and the same for $D(\reg,p,\tb)$,
\item{} [interpolation] \label{assump:HierarchyToDecay(2)}for all $\reg\in\{0,\ldots,\reg_0\}$, all $p,\beta_1,\beta_2\in[p_1,p_2]$ such that $\beta_1\leq p \leq\beta_2$, and all $\tb\geq \tb_0$,
\begin{align}
F(\reg,p,\tb)
\lesssim{}&
F(\reg,\beta_1,\tb)^{\frac{\beta_2-p}{\beta_2-\beta_1}}
F(\reg,\beta_2,\tb)^{\frac{p-\beta_1}{\beta_2-\beta_1}} ,
\label{eq:Rev:HierarchyToDecayReal:InterpolationHypothesis}
\end{align}
\item{} [energy and Morawetz estimate] for all $\reg\in\{0,\ldots,\reg_0-\regl\}$, $p\in[p_1,p_2]$, and $\tb_2\geq \tb_1>\tb_1'\geq \tb_0$,
\begin{align}
F(\reg,p,\tb_2)
+\int_{\tb_1}^{\tb_2} F(\reg-\regl,p-1,\tb) \di \PigeonTime
\lesssim F(\reg+\regl,p,\tb_1) +\langle \tb_1 -\tb_1'\rangle^{p-p_2}D(\reg+\regl, p,\tb_1') ,
\label{eq:Rev:HierarchyToDecayReal:EvolutionHypothesis}
\end{align}
\end{enumerate}
\end{subequations}
then there exists a constant $C>0$ such that for all $\reg\in\{0,\ldots,\reg_0-C\regl\}$, all $p\in[p_1,p_2]$, and all $\tb_2>\tb_1\geq \tb_0$,
\begin{align}
F(\reg,p,\tb_2) \lesssim_{p_2, p_1}{}& \langle \tb_2-\PigeonTime_1\rangle^{p-p_2}  (F(\reg+C\regl,p_2,\tb_1) +D(\reg+C\regl, p_2,\tb_1)).
\label{eq:Rev:HierarchyToDecayReal}
\end{align}
\end{lemma}


\section{System of equations}
\label{sect:sys:of:eqs}

In this section, we derive various systems of equations from the Teukolsky master equation (TME) satisfied by the spin $\pm \sfrak$ components. The TME is introduced in Section \ref{sect:tme}. Then we derive in Section \ref{subsect:wavesys:hatPhisHighi} coupled wave systems for each of the spin $\pm \sfrak$ components, followed by a derivation of the wave equations for the modes in Section \ref{subsect:waveeq:modes:gen}.  In the end, we discuss the Teukolsky--Starobinsky identities (TSI) in Section \ref{sect:TSI}.

\subsection{Teukolsky master equation}
\label{sect:tme}

We introduce a few scalars defined from the spin $\pm\sfrak$ components.

 \begin{definition}
Define two rescaled spin $\pm \sfrak$ components
\begin{align}\label{eq:spinsfields}
\psipluss\doteq{}&\Sigma^{\sfrak}\NPRpluss, &
\psiminuss\doteq{}&\Sigma^{-\sfrak}(r-ia\cos\theta)^{2\sfrak}\NPRminuss.
\end{align}
Define their radiation field
\begin{align}\label{definition:radiationfield:kerr:spin}
\Psipluss\doteq{}& \sqrt{r^2+a^2}\psipluss, &
\Psiminuss\doteq{}& \sqrt{r^2+a^2}\psiminuss.
\end{align}
\end{definition}

It is a remarkable discovery by Teukolsky \cite{Teukolsky1973I} that the scalars $\psi_s$ in a Kerr spacetime satisfy the celebrated
\textit{Teukolsky Master Equation} (TME), a separable, decoupled wave equation.
\begin{prop}[TME of the spin $s$ components]
\label{prop:TME:originalform}
In a Kerr spacetime, the scalars $\psi_s$ solve the following TME in the Boyer--Lindquist coordinates:
\begin{align}\label{eq:TME}
0=\TMEOp_s\psi_s=&-\frac{(r^2+a^2)^2-a^2 \sin^2{\theta} \Delta}{\Delta}
\partial_{tt}^2\psi_s
+\partial_r( \Delta\partial_r\psi_{s})
- \frac{4aMr}{\Delta}\partial_{t\phi}^2 \psi_{s}
-\frac{a^2}{\Delta}\partial_{\phi\phi}^2\psi_s
  \notag\\
&
+ \frac{1}{\sin{\theta}} \partial_{\theta}( \sin{\theta} \partial_{\theta}\psi_{s})
+\frac{1}{\sin^2{\theta}}\partial_{\phi\phi}^2\psi_s + \frac{2i s\cos{\theta}}{\sin^2{\theta} }  \partial_{\phi} \psi_{s}- (s^2 \cot^2{\theta} +s) \psi_{s} \notag\\
&  -2ias \cos{\theta} \partial_t \psi_{s}
+2s[(r-M)Y-2r\partial_t]\psi_s .
\end{align}
\end{prop}

We remark that these N--P scalars satisfying TME differ from the ones used in \cite{Teukolsky1973I} by a rescaling factor of $2^{-s/2}\Delta^s$, and the reason that we use these scalars lies in the fact that both of they are regular at $\Horizon$ from formula \eqref{def:regularNPComps}. Note that the second line of \eqref{eq:TME} equals
$\edthR \edthR'\psi_s$, and this makes the TME a spin-weighted wave equation in the sense that the TME operator $\TMEOp_s$ is a second-order spin-weighted operator.   It serves as a starting model for quite many results   in obtaining quantitative estimates for these fields, including  the scalar field, the Maxwell field and the linearized gravity.

We define a (spin-weighted) wave operator that is different from the TME operator $\TMEOp_s$ and useful in deriving the wave equations for the radiation fields.

\begin{definition}\label{def:squareShat}
Define a spin-weighted wave operator
\begin{align}
\label{eq:squareShat}
\Boxhat_{s} \doteq {}&-(\R)YV
+\edthR\edthR'
+2a\Lxi\Leta+a^2 \sin^2 \theta\Lxi^2
-2ias\cos\theta \Lxi.
\end{align}
\end{definition}

The two wave operators $\Boxhat_s$ and $\TMEOp_s$ can be related via the following statement.

\begin{lemma}
For any spin $s$ scalar $\varphi$,
\begin{align}
\Boxhat_{s}(\sqrt{\R}\varphi)={}&
\sqrt{\R}\bigg(\TMEOp_s-2s[(r-M)Y-2r\partial_t]-\frac{2ar}{\R}\Leta+\frac{2Mr^3+a^2r^2-4a^2Mr+a^4}{(\R)^2}\bigg)\varphi.
\end{align}
\end{lemma}

\begin{proof}
We calculate in the Boyer--Lindquist coordinates that
\begin{align*}
\hspace{2ex}&\hspace{-2ex}(\R)Y V (\sqrt{\R}\varphi)\notag\\
={}&\frac{(\R)^2}{\Delta}(\partial_t +\frac{a}{\R}\partial_{\phi}-\mu \partial_r)(\partial_t +\frac{a}{\R}\partial_{\phi}+\mu \partial_r)(\sqrt{\R}\varphi)
\notag\\
={}&\frac{(\R)^2}{\Delta}(\partial_t +\frac{a}{\R}\partial_{\phi}-\mu \partial_r)\big(\sqrt{\R}\partial_t +\frac{a}{\sqrt{\R}}\partial_{\phi}+\frac{\Delta}{\sqrt{\R}}\partial_r
+\frac{r\Delta}{(\R)^{\frac{3}{2}}}\big)\varphi.
\end{align*}
By expanding this formula, one finds
\begin{align}
-(\R)Y V (\sqrt{\R}\varphi)
={}&\sqrt{\R}\bigg(
-\frac{(\R)^2}{\Delta}\partial_{tt}^2\varphi+\partial_r(\Delta\partial_r\varphi)
-\frac{2a(\R)}{\Delta}\partial_{t\phi}^2\varphi\notag\\
&
-\frac{a^2}{\Delta}\partial_{\phi\phi}^2\varphi-\frac{2ar}{\R}\partial_{\phi}\varphi
+\frac{2Mr^3+a^2r^2-4a^2Mr+a^4}{(\R)^2}\varphi
\bigg).
\end{align}
In view of the definitions of the TME operator $\TMEOp_s$ in \eqref{eq:TME} and
the wave operator $\Boxhat_{s}$ in \eqref{eq:squareShat}, the claim then follows.
\end{proof}

\begin{cor}[TME for radiation fields of the spin $s$ components]
The radiation field scalars $\Psi_s$ then satisfy the following wave equation that we call as TME as well:
\begin{align}
\label{eq:TME:radfield}
\Boxhat_s\Psi_s={}&-2s((r-M)Y-2r\Lxi)\Psi_s
-\frac{2ar}{\R}\Leta\Psi_s\notag\\
&
-\bigg(\frac{2s r(r-M)}{\R}-\frac{2Mr^3+a^2r^2-4a^2Mr+a^4}{(\R)^2}\bigg)\Psi_s.
\end{align}
\end{cor}

\subsubsection{Alternative form of TME in hyperboloidal coordinates}
\label{subsect:alternative:TME}

We recast the TME under the hyperboloidal coordinates.

\begin{prop}
\label{prop:goodsecondorderODE:varphis}
The scalars $\psi_s$ satisfy the following wave equation
\begin{align}
\label{eq:TME:psis:hypercoords}
\prb(\Delta^{-s+1}\prb\psi_s)+2a\Delta^{-s}\Leta\prb\psi_s
+\Delta^{-s}\edthR\edthR'\psi_s
=\Delta^{-s}\Lxi H[\psi_s]
\end{align}
with
\begin{align}
\label{def:Hpsis}
H[\psi_s]={}&\frac{-1}{\sqrt{\R}}\Big(2(\R)(\mu \Hhyp-1)\prb (\sqrt{\R}\psi_s)
\notag\\
&
+\sqrt{\R}\big[a^2 \sin^2 \theta+(\R)\Hhyp
(\mu\Hhyp -2)\big]\Lxi\psi_s
+2a\sqrt{\R}\big[1+\mu^{-1}(\mu \Hhyp-2)\big]\Leta\psi_s
\notag\\
&
+\sqrt{\R}\big[(\R)\partial_r (\mu \Hhyp)+2s((r-M)(2\mu^{-1}-\Hhyp)-2r)-2ias\cos\theta
\big] \psi_s\Big).
\end{align}
\end{prop}

\begin{proof}
We substitute  in the formula \eqref{def:vectorVRintermsofprb}  to deduce
\begin{align*}
-(\R)YV\varphi
={}&-(\R)(-\prb+(2\mu^{-1}-\Hhyp)\Lxi)\bigg(\mu \prb +\mu \Hhyp \Lxi+\frac{2a}{\R}\Leta\bigg)\varphi\notag\\
={}&(\R)\prb(\mu\prb)\varphi+2(\R)(\mu \Hhyp-1)\Lxi\prb \varphi +(\R)\Hhyp
(\mu\Hhyp -2) \Lxi^2\varphi\notag\\
&
+2a\mu^{-1}(\mu \Hhyp-2) \Lxi\Leta\varphi
+(\R)\partial_r (\mu \Hhyp) \Lxi\varphi
-\frac{4ar}{\R}\Leta\varphi +2a\Leta\prb\varphi
\end{align*}
and
\begin{align*}
2s((r-M)Y-2r\Lxi)\Psi_s
={}&2s((r-M)(-\prb+(2\mu^{-1}-\Hhyp)\Lxi)-2r\Lxi)\Psi_s
\notag\\
={}&-2s(r-M)\prb\Psi_s+2s((r-M)(2\mu^{-1}-\Hhyp)-2r)\Lxi\Psi_s.
\end{align*}
Then by the TME \eqref{eq:TME:radfield} of $\Psi_s$ and the definition of the wave operator $\Boxhat_s$ in \eqref{eq:squareShat}, we obtain the following wave equation in the hyperboloidal coordinates for $\Psi_s$:
\begin{align}
\label{TME:Psis:hypercoords}
0={}&\edthR\edthR'\Psi_s
+2a\Lxi\Leta\Psi_s+a^2 \sin^2 \theta\Lxi^2\Psi_s
-2ias\cos\theta \Lxi\Psi_s\notag\\
&+(\R)\prb(\mu\prb)\Psi_s+2(\R)(\mu \Hhyp-1)\Lxi\prb \Psi_s
 +(\R)\Hhyp
(\mu\Hhyp -2) \Lxi^2\Psi_s\notag\\
&
+2a\mu^{-1}(\mu \Hhyp-2) \Lxi\Leta\Psi_s
+(\R)\partial_r (\mu \Hhyp) \Lxi\Psi_s
-\frac{4ar}{\R}\Leta\Psi_s+2a\Leta\prb\Psi_s\notag\\
&-2s(r-M)\prb\Psi_s+2s((r-M)(2\mu^{-1}-\Hhyp)-2r)\Lxi\Psi_s
+\frac{2ar}{\R}\Leta\Psi_s\notag\\
&
+\bigg(\frac{2s r(r-M)}{\R}-\frac{2Mr^3+a^2r^2-4a^2Mr+a^4}{(\R)^2}\bigg)\Psi_s\notag\\
={}&(\R)\prb(\mu\prb)\Psi_s
-2s(r-M)\prb\Psi_s
+2a\Leta\prb\Psi_s
-\frac{2ar}{\R}\Leta\Psi_s
\notag\\
&+\edthR\edthR'\Psi_s
+\bigg(\frac{2s r(r-M)}{\R}-\frac{2Mr^3+a^2r^2-4a^2Mr+a^4}{(\R)^2}\bigg)\Psi_s
+\Lxi H[\Psi_s]
\end{align}
with
\begin{align*}
\hat{H}[\Psi_s]={}&2(\R)(\mu \Hhyp-1)\prb \Psi_s
+(a^2 \sin^2 \theta+(\R)\Hhyp
(\mu\Hhyp -2))\Lxi\Psi_s
+2a(1+\mu^{-1}(\mu \Hhyp-2))\Leta\Psi_s
\notag\\
&
+\big((\R)\partial_r (\mu \Hhyp)
+2s((r-M)(2\mu^{-1}-\Hhyp)-2r)-2ias\cos\theta
\big) \Psi_s.
\end{align*}
Hence, with the definition $\psi_s=\sqrt{\R}\Psi_s$, one finds
\begin{align*}
&(\R)\prb(\mu\prb)\Psi_s
-2s(r-M)\prb\Psi_s
+2a\Leta\prb\Psi_s
-\frac{2ar}{\R}\Leta\Psi_s
\notag\\
={}&(\R)\prb(\mu\prb)(\sqrt{\R}\psi_s)
-2s(r-M)\prb(\sqrt{\R}\psi_s)
+2a\sqrt{\R}\Leta\prb\psi_s
\notag\\
={}&\sqrt{\R}\Big(\sqrt{\R}\prb(\sqrt{\R}\mu \prb \psi_s)
+\mu \sqrt{\R}\partial_r (\sqrt{\R})\prb\psi_s
-2s(r-M)\prb\psi_s
\Big)\notag\\
&+\big({\R}\partial_r (\mu \partial_r (\sqrt{\R}))-2s(r-M)\partial_r (\sqrt{\R})\big)\psi_s
+2a\sqrt{\R}\Leta\prb\psi_s
\notag\\
={}&\sqrt{\R}\big(\Delta^s\prb(\Delta^{-s+1}\prb\psi_s)
\big)-\bigg(\frac{2s r(r-M)}{\R}-\frac{2Mr^3+a^2r^2-4a^2Mr+a^4}{(\R)^2}\bigg)\Psi_s+2a\sqrt{\R}\Leta\prb\psi_s.
\end{align*}
Plugging this back into equation \eqref{TME:Psis:hypercoords} yields equation \eqref{eq:TME:psis:hypercoords} for $\psi_s$.
\end{proof}

In addition, for the spin $+\sfrak$ component, we have
\begin{cor}
\label{cor:varphipluss:eq}
Let
\begin{align}
\label{def:varphipluss}
\varphipluss\doteq \Delta^{-\sfrak}\psipluss.
\end{align}
It then satisfies
\begin{align}
\label{eq:TME:varphipluss:hypercoords}
\prb(\Delta^{\sfrak+1}\prb\varphipluss)+2a\Leta\prb(\Delta^{\sfrak}\varphipluss)
+\Delta^{\sfrak}(\edthR\edthR'+2\sfrak)\varphipluss
=\Lxi H[\psipluss]
\end{align}
with $H[\psipluss]$ defined as in equation \eqref{def:Hpsis}.
\end{cor}

\begin{proof}
With the definition \eqref{def:Hpsis}, we substitute $\psipluss=\Delta^{\sfrak}\varphipluss$ into \eqref{eq:TME:psis:hypercoords} with $s=+\sfrak$ and find  that the LHS equals
\begin{align*}
\hspace{4ex}&\hspace{-4ex}
\prb(\Delta^{-\sfrak+1}\prb(\Delta^{\sfrak}\varphipluss))+2a\Delta^{-\sfrak}\Leta\prb(\Delta^{\sfrak}\varphipluss)
+\Delta^{-\sfrak}\edthR\edthR'(\Delta^{\sfrak}\varphipluss)\notag\\
={}&\prb(\Delta\prb\varphipluss+2\sfrak(r-M)\varphipluss)+2a\Delta^{-\sfrak}\Leta\prb(\Delta^{\sfrak}\varphipluss)
+\edthR\edthR'\varphipluss\notag\\
={}&\prb(\Delta\prb\varphipluss)+2\sfrak(r-M)\prb\varphipluss+2a\Delta^{-\sfrak}\Leta\prb(\Delta^{\sfrak}\varphipluss)
+(\edthR\edthR'+2\sfrak)\varphipluss\notag\\
={}&\Delta^{-\sfrak}\big(\prb(\Delta^{\sfrak+1}\prb\varphipluss)+2a\Leta\prb(\Delta^{\sfrak}\varphipluss)
+\Delta^{\sfrak}(\edthR\edthR'+2\sfrak)\varphipluss\big).
\end{align*}
 This thus yields equation \eqref{eq:TME:varphipluss:hypercoords}.
\end{proof}

\begin{remark}
The main reason that we derive equation \eqref{eq:TME:psis:hypercoords} (actually mainly for the spin $-\sfrak$ component) and equation \eqref{eq:TME:varphipluss:hypercoords} for the spin $+\sfrak$ component is that when projecting both equations on the $\sfrak$ mode, the terms $\Delta^{\sfrak}\edthR\edthR'\ellmode{\psiminuss}{\sfrak}$  and $\Delta^{\sfrak}(\edthR\edthR'+2\sfrak)\ellmode{\varphipluss}{\sfrak}$ vanish due to \eqref{eq:l=l0mode:eigenvalue}. This property is essential in the analysis in Sections \ref{subsect:APL:general} and \ref{sect:PL}.
\end{remark}

\subsection{Wave systems for the spin $\pm \sfrak$ components}
\label{subsect:wavesys:hatPhisHighi}

In this subsection, we define a few scalars constructed from the spin $\pm\sfrak$ components and derive their governing equations. These equations are crucial in deriving the energy decay estimates for the spin $\pm \sfrak$ components.

We begin with a definition of these scalars.

\begin{definition}
\label{def:Phiminusi}
Let $i\in \mathbb{N}$ and define for the spin $s$ components the following spin $s$ scalars
\begin{align}
\label{def:PhisHigh}
\PhisHigh{0}\doteq{}&\mu^{-s}\Psi_s, \qquad \PhisHigh{i}\doteq{}\curlVR^i\PhisHigh{0}.
\end{align}
Define additionally the following spin $+\sfrak$ scalars
\begin{align}\label{eq:DefOfphi012PosiSpinS2}
    \Xi^{(0)}_{+\sfrak}\doteq{}&(\R)^{-\sfrak}\Psipluss, \qquad
\Xi^{(i)}_{+\sfrak}\doteq {}(-(\R) Y)^i\Xi^{(0)}_{+\sfrak}.
\end{align}
\end{definition}

To derive the governing equations of the above defined scalars, we calculate the commutators between the wave operator $\Boxhat$ and some other operators.

\begin{prop}
Let $\varphi$ be a spin $s$ scalar.
\begin{itemize}
\item
For any function $f=f(r)$,
\begin{align}
\label{eq:comm:Boxhatandf}
\Boxhat_s (f\varphi)=f\Boxhat_s\varphi +2\Delta \partial_r f \partial_r \varphi+(\R)\partial_r(\mu \partial_r f)\varphi.
\end{align}
\item The commutator between $\Boxhat_s$ and $\curlVR$ is
\begin{align}
\label{comm:BoxhatandcurlVR}
[\Boxhat_s, \curlVR]\varphi={}&\frac{2(\PR)}{(\R)^2} \curlVR^2\varphi
-\frac{4ar}{\R}\Leta\curlVR\varphi
-\frac{2(r^4 -6Mr^3 +10a^2 Mr-a^4)}{(\R)^2}\curlVR\varphi.
\end{align}
\end{itemize}
\end{prop}

\begin{proof}
Formula \eqref{eq:comm:Boxhatandf} can be directly verified.

By formula \eqref{eq:squareShat} and the commutator relations in Proposition \ref{prop:comms},
\begin{align}
[\Boxhat_s, \curlVR]\varphi={}&[-(\R) YV, \curlVR]\varphi\notag\\
={}&\curlVR((\R)YV\varphi)-(\R) YV\curlVR\varphi\notag\\
={}&\curlVR(\mu Y\curlVR\varphi) - \curlVR\bigg(\partial_r \bigg(\frac{\Delta}{(\R)^2}\bigg)(\R)\curlVR\varphi\bigg)
-\curlVR(\mu Y\curlVR\varphi) -\mu^{-1}(\R) [\mu Y, V] \curlVR\varphi\notag\\
={}&\frac{2(\PR)}{(\R)^2} \curlVR\curlVR\varphi
+(\R)\partial_r \bigg(\frac{2(\PR)}{(\R)^2}\bigg)\curlVR\varphi \notag\\
&-\mu^{-1}(\R) [\mu Y, V] \curlVR\varphi\notag\\
={}&\frac{2(\PR)}{(\R)^2} \curlVR^2\varphi
-\frac{4ar}{\R}\Leta\curlVR\varphi\notag\\
&+(\R)\partial_r \bigg(\frac{2(\PR)}{(\R)^2}\bigg)\curlVR\varphi .
\end{align}
Calculating the coefficient of the last term then yields \eqref{comm:BoxhatandcurlVR}.
\end{proof}

The following two propositions then provide the governing equations of the scalars $ \PhisHigh{i}$.

\begin{prop}
\label{prop:wavesys:PhisHigh0}
The scalar $\PhisHigh{0}$ defined above satisfies a wave equation
\begin{align}
\label{eq:wave:PhisHigh0}
\Boxhat_{s}\PhisHigh{0}
={}&\frac{2s(r^3-3Mr^2 +a^2 r+a^2 M)}{(\R)^2} \curlVR \PhisHigh{0} -\frac{2(2s+1)ar}{\R}\Leta \PhisHigh{0}\notag\\
&-\bigg(2s-\frac{(2s+1)(2(s+1)Mr^3+a^2 r^2 -2(s+2)a^2 Mr+a^4)}{(\R)^{2}}\bigg)\PhisHigh{0}.
\end{align}
\end{prop}

\begin{proof}
Since $\Psi_s$ satisfies the TME \eqref{eq:TME:radfield}, we obtain by taking $f=\mu^{-s}$ in \eqref{eq:comm:Boxhatandf} that
\begin{align}
\label{eq:wave:PhisHigh0:74938}
\Boxhat_{s}\PhisHigh{0}={}&\mu^{-s}\Boxhat_{s}\Psi_s +2\Delta \partial_r (\mu^{-s}) \partial_r \Psi_s + (\R) \partial_r (\mu \partial_r (\mu^{-s}))\Psi_s\notag\\
={}&\big(-2s\mu^{-s}((r-M)Y-2r\Lxi)
-2ar (\R)^{-1}\mu^{-\sfrak}\Leta+2\Delta \partial_r (\mu^{-s}) \partial_r\big) \Psi_s \notag\\
&+\bigg((\R) \partial_r (\mu \partial_r (\mu^{-s}))
-\mu^{-s}\bigg(\frac{2s r(r-M)}{\R}
-\frac{2Mr^3+a^2r^2-4a^2Mr+a^4}{(\R)^2}\bigg)\bigg)\Psi_s.
\end{align}
The second line equals
\begin{align*}
&
2s\mu^{-s}\bigg(\frac{r^3-3Mr^2 +a^2 r+a^2 M}{\R} \VR \Psi_s -\frac{2ar}{\R}\Leta \Psi_s\bigg)-\frac{2ar}{\R}\mu^{-\sfrak}\Leta \Psi_s\notag\\
={}&\frac{2s(r^3-3Mr^2 +a^2 r+a^2 M)}{\R} \VR \PhisHigh{0} -\frac{(4s+2)ar}{\R}\Leta \PhisHigh{0}
+\frac{2s(r^3-3Mr^2 +a^2 r+a^2 M)\partial_r (\mu^s)}{\mu^{s}(\R)}  \PhisHigh{0}.
\end{align*}
Putting this into \eqref{eq:wave:PhisHigh0} and substituting in $\Psi_s=\mu^s\PhisHigh{0}$, we find that the coefficient of the $\PhisHigh{0}$ term on the RHS of \eqref{eq:wave:PhisHigh0}  is equal to
 \begin{align*}
 -\bigg(\frac{2s r(r-M)}{\R}
-\frac{2Mr^3+a^2r^2-4a^2Mr+a^4}{(\R)^2}+s(\R)^{1+s} \partial_r \bigg(\frac{\partial_r \mu}{ (\R)^{s}}\bigg)\bigg)
\end{align*}
which further equals the coefficient of the $\PhisHigh{0}$ term in equation \eqref{eq:wave:PhisHigh0}.
 Thus, we achieve \eqref{eq:wave:PhisHigh0}.
\end{proof}

\begin{prop}
\label{prop:wavesys:PhisHighi}
The scalars $\PhisHigh{i}$ defined in Definition \ref{def:Phiminusi} satisfy the following wave equations
\begin{align}
\label{eq:wave:PhisHighi:general}
\Boxhat_{s}\PhisHigh{i}
={}&\frac{2(s+i)(r^3-3Mr^2 +a^2 r+a^2 M)}{(\R)^2} \curlVR\PhisHigh{i}
+\sum_{{0\leq j \leq i-1, \frac{i-j-1}{2}\in \mathbb{N}}}X_{s,i,j}\Leta \PhisHigh{j}
\notag\\
&-(2s+i)(i+1)\PhisHigh{i}
-\sum_{j=0}^{i-1}{Z_{s,i,j}}\PhisHigh{j}
+\sum_{n=0,1}\sum_{j=0}^i w_{s,i,j,n}\Leta^n \PhisHigh{j},
\end{align}
with functions $w_{s,i,j,n}=O(r^{-1})$.
Here, $Z_{s,i,j}$ are constants which can be calculated as in the proof and the constants $X_{s,i,j}$ are
\begin{equation}
\label{eq:Xsij:coeffs}
\begin{split}
X_{s,i,j}={}&(-1)^{\frac{i-j-1}{2}}(2a)^{i-j}((2s+1)C_i^j +2C_{i}^{j-1}), \quad \forall i\in \mathbb{N}, 1\leq j \leq i,\\
X_{s,i,0}={}&(-1)^{\frac{i-1}{2}}(2a)^{i}(2s+1), \qquad\qquad\qquad\,\,\,\,\qquad\forall i\in \mathbb{N}
\end{split}
\end{equation}
with $C_i^j=\frac{i!}{j! (i-j)!}$. 
\end{prop}

\begin{proof}
Applying once $\curlVR$ on both sides of the wave equation \eqref{eq:wave:PhisHigh0} and using the commutator formula \eqref{comm:BoxhatandcurlVR}, the LHS equals
\begin{align*}
\Boxhat_s\PhisHigh{1}-\frac{2(\PR)}{(\R)^2} \curlVR\PhisHigh{1}
+\frac{4ar}{\R}\Leta\PhisHigh{1}
+\frac{2(r^4 -6Mr^3 +10a^2 Mr-a^4)}{(\R)^2}\PhisHigh{1},
\end{align*}
and the RHS equals
\begin{align*}
&\frac{2s(r^3-3Mr^2 +a^2 r+a^2 M)}{(\R)^2} \curlVR \PhisHigh{1}
+(\R)\partial_r\bigg( \frac{2s(r^3-3Mr^2 +a^2 r+a^2 M)}{(\R)^2} \bigg) \PhisHigh{1}\notag\\
&
-\frac{2(2s+1)ar}{\R}\Leta \PhisHigh{1}
-(\R)\partial_r\bigg(\frac{2(2s+1)ar}{\R}\bigg)\Leta \PhisHigh{0}\notag\\
&-\bigg(2s-\frac{(2s+1)(2(s+1)Mr^3+a^2 r^2 -2(s+2)a^2 Mr+a^4)}{(\R)^{2}}\bigg)\PhisHigh{1}\notag\\
&+(\R)\partial_r \bigg(\frac{(2s+1)(2(s+1)Mr^3+a^2 r^2 -2(s+2)a^2 Mr+a^4)}{(\R)^{2}}\bigg)\PhisHigh{0}\notag\\
={}&\frac{2s(r^3-3Mr^2 +a^2 r+a^2 M)}{(\R)^2} \curlVR \PhisHigh{1}
-\frac{2(2s+1)ar}{\R}\Leta \PhisHigh{1}
+\frac{2(2s+1)a(r^2-a^2)}{\R}\Leta \PhisHigh{0}
\notag\\
&
-\bigg(2s-\frac{(2s+1)(2(s+1)Mr^3+a^2 r^2 -2(s+2)a^2 Mr+a^4)
-2s(r^4-6Mr^3+10a^2 Mr -a^4)}{(\R)^{2}}\bigg)\PhisHigh{1}\notag\\
&-\frac{2(2s+1)(
(s+1)Mr^4+a^2 r^3- (6s+9)a^2 Mr^2 +a^4 r+ (s+2)a^4 M)}{(\R)^{2}}\PhisHigh{0}.
\end{align*}
Therefore,
\begin{align}
\label{eq:wave:PhisHigh1:gene}
\Boxhat_{s}\PhisHigh{1}={}&\frac{2(s+1)(r^3-3Mr^2 +a^2 r+a^2 M)}{(\R)^2} \curlVR \PhisHigh{1} -\frac{2(2s+3)ar}{\R}\Leta \PhisHigh{1}+\frac{2(2s+1)a(r^2-a^2)}{\R}\Leta \PhisHigh{0}\notag\\
&-\bigg(2(2s+1)-\frac{2(s+1)(2s+7)Mr^3 + (6s+5)a^2 r^2 -2(2s^2+15s+12)a^2 Mr +(6s+5)a^4}{(\R)^{2}}\bigg)\PhisHigh{1}\notag\\
&
-\frac{2(2s+1)(
(s+1)Mr^4+a^2 r^3- (6s+9)a^2 Mr^2 +a^4 r+ (s+2)a^4 M)}{(\R)^{2}}\PhisHigh{0}.
\end{align}
Applying further the operator $\curlVR$ on both sides of \eqref{eq:wave:PhisHigh1:gene} and repeated application of the commutator formula \eqref{comm:BoxhatandcurlVR} yields that the scalars $\PhisHigh{i}$ $(i\geq 0)$ satisfy the following equation
\begin{align}
\label{eq:wave:PhisHigh:gene}
\Boxhat_{s}\PhisHigh{i}
={}&\frac{2(s+i)(r^3-3Mr^2 +a^2 r+a^2 M)}{(\R)^2} \curlVR\PhisHigh{i}\notag\\
& -\sum_{0\leq j \leq i, \frac{i-j}{2}\in \mathbb{N}}(-1)^{\frac{i-j}{2}}(2a)^{i-j}\tilde{X}_{s,i,j}\frac{2ar}{\R}\Leta \PhisHigh{j}
+\sum_{0\leq j \leq i, \frac{i-j-1}{2}\in \mathbb{N}}(-1)^{\frac{i-j-1}{2}}(2a)^{i-j}\tilde{X}_{s,i,j}\frac{(r^2-a^2)}{\R}\Leta \PhisHigh{j}
\notag\\
&-\bigg((2s+i)(i+1)+\frac{W_{s,i,i}}{(\R)^{2}}\bigg)\PhisHigh{i}
-\sum_{j=0}^{i-1}\frac{W_{s,i,j}}{(\R)^2}\PhisHigh{j}
\end{align}
with the following iterative relations for the appeared constants and functions:
the constants $\tilde{X}_{s,i,j}$ obey
\begin{align*}
\tilde{X}_{s,i,i}={}&2s+2i+1, \qquad\qquad\qquad\forall i\in \mathbb{N},\\
\tilde{X}_{s,i,j}={}&\tilde{X}_{s, i-1, j-1}+\tilde{X}_{s,i-1,j}, \quad\,\,\,\, \forall 1\leq j\leq i-1,\\
\tilde{X}_{s,i,0}={}&2s+1, \qquad\qquad\qquad\qquad\forall i\in \mathbb{N}
\end{align*}
and the functions $W_{s,i,j}$ obey
\begin{align*}
W_{s,i,i}={}&W_{s,i-1,i-1}-4(s+i)(3Mr^3 +a^2 r^2 -5a^2 Mr +a^4),\\
W_{s,i, j}={}&(\R)^3\partial_r \bigg(\frac{W_{s,i-1,j}}{(\R)^2}\bigg)
+W_{s,i-1,j-1}=(\R)\partial_r W_{s,i-1,j} - 4r W_{s,i-1,j} +W_{s,i-1,j-1}
\end{align*}
with the initial one
$
W_{s,0,0}=-(2s+1)(2(s+1)Mr^3+a^2 r^2 -2(s+2)a^2 Mr+a^4)$
that can be read off from equation \eqref{eq:wave:PhisHigh0} and $W_{s,i,-1}=0$ for all $i\in \mathbb{N}$.
The above iterative relations for constants $\tilde{X}_{s,i,j}$ yield that
\begin{align*}
\tilde{X}_{s,i,j}={}&(2s+1)C_i^j +2C_{i}^{j-1}, \quad \forall i\in \mathbb{N}, 1\leq j \leq i,\\
\tilde{X}_{s,i,0}={}&2s+1, \qquad\qquad\qquad\quad\forall i\in \mathbb{N}.
\end{align*}
Meanwhile, one can compute the functions $W_{s,i,j}$ from the above iterative relations.
 By defining the coefficient of $r^4$ term in each $W_{s,i,j}$ as the value of $Z_{s,i,j}$ and isolating the constant part of the coefficients in the second line of equation \eqref{eq:wave:PhisHigh:gene}, the claim then follows.
\end{proof}

The above also yields equations for $ \{\Xi^{(i)}_{+\sfrak}\}_{0\leq i\leq \sfrak}$. The wave systems for the scalars $\{\Phiminuss{i}\}_{0\leq i\leq \sfrak}$ and $ \{\Xi^{(i)}_{+\sfrak}\}_{0\leq i\leq \sfrak}$ are derived below, and the importance of these systems are crucial in obtaining the basic energy and Morawetz estimates for the spin $\pm \sfrak$ components in Kerr spacetime \cite{Ma2017Maxwell, Ma17spin2Kerr}. The following equations for the radiation fields in $\sfrak=1$ and $\sfrak=2$ cases are also derived in \cite{Ma20almost, andersson2019stability} respectively.

\begin{cor}
\label{cor:wavesys:basicRW}
We have the following basic wave systems for the scalars $\{\Phiminuss{i}\}_{0\leq i\leq \sfrak}$ and $ \{\Xi^{(i)}_{+\sfrak}\}_{0\leq i\leq \sfrak}$ defined in Definition \ref{def:Phiminusi}:
\begin{itemize}
\item for $s=0$,
\begin{align}
\Boxhat_{0}\Phi^{(0)}_0
={}&-\frac{2ar}{\R}\Leta\Phi^{(0)}_0+\frac{2Mr^3+a^2 r^2 -4a^2 Mr+a^4}{(\R)^{2}}\Phi^{(0)}_0;
\end{align}

\item for $s=-1$,
\begin{subequations}
\label{eq:basicwavesys:-1}
\begin{align}
\Boxhat_{-1}\Phiminus{0}
={}&-\frac{2(r^3-3Mr^2 +a^2 r+a^2 M)}{(\R)^2}\Phiminus{1}
+\frac{2ar}{\R}\Leta \Phiminus{0}
+\bigg(2-\frac{a^2\Delta}{(\R)^2}\bigg)\Phiminus{0},\\
\Boxhat_{-1}\Phiminus{1}
={}&
\bigg(2-\frac{a^2\Delta}{(\R)^2}\bigg)\Phiminus{1}
-\frac{2ar}{\R}\Leta \Phiminus{1}
-\frac{2a(r^2-a^2)}{\R}\Leta \Phiminus{0}
+\frac{2a^2 (\PR)}{(\R)^2}\Phiminus{0}
\end{align}
\end{subequations}
and
\begin{align}
\Boxhat_{-1}\Phiminus{2}
={}&
\frac{2(r^3-3Mr^2 +a^2 r+a^2 M)}{(\R)^2}\Phiminus{2}
-\frac{6ar}{\R}\Leta \Phiminus{2}
+\frac{12Mr^3+3a^2r^2-18a^2 Mr +3a^4}{(\R)^2}\Phiminus{2}\notag\\
&+\frac{4a^2(\PR)}{(\R)^2}\Phiminus{1}
-\frac{8a^3 r}{\R}\Leta\Phiminus{0}
-\frac{2a^2(r^4-6Mr^3+10a^2 Mr- a^4)}{(\R)^2}\Phiminus{0};
\end{align}

\item for $s=+1$,
\begin{subequations}
\begin{align}
\Boxhat_{+1}\Xi_{+1}^{(0)}
={}&-\frac{2(r^3-3Mr^2 +a^2 r+a^2 M)}{(\R)^2}\Xi_{+1}^{(1)}
-\frac{6ar}{\R}\Leta \Xi_{+1}^{(0)}
+\bigg(2-\frac{a^2\Delta}{(\R)^2}\bigg)\Xi_{+1}^{(0)},\\
\Boxhat_{+1}\Xi_{+1}^{(1)}
={}&
\bigg(2-\frac{a^2\Delta}{(\R)^2}\bigg)\Xi_{+1}^{(1)}
-\frac{2ar}{\R}\Leta \Xi_{+1}^{(1)}
+\frac{2a(r^2-a^2)}{\R}\Leta\Xi_{+1}^{(0)}+\frac{2a^2 (\PR)}{(\R)^2}\Xi_{+1}^{(0)};
\end{align}
\end{subequations}

\item
for spin $-2$,
\begin{subequations}
\label{eq:basicwavesys:-2}
\begin{align}
\Boxhat_{-2}\Phiminustwo{0}
={}&-\frac{4(r^3-3Mr^2 +a^2 r+a^2 M)}{(\R)^2} \Phiminustwo{1}
+\frac{6ar}{\R}\Leta \Phiminustwo{0}
+\bigg(4
+\frac{6Mr^3-3a^2 r^2 -3a^4 }{(\R)^2}\bigg)\Phiminustwo{0},\\
\Boxhat_{-2}\Phiminustwo{1}
={}&-\frac{2(r^3-3Mr^2 +a^2 r+a^2 M)}{(\R)^2} \Phiminustwo{2}
+\bigg(6
-\frac{6Mr^3+7a^2 r^2 -20a^2 Mr+7a^4 }{(\R)^2}\bigg)\Phiminustwo{1}\notag\\
&
+\frac{2ar}{\R}\Leta \Phiminustwo{1}
-\frac{6a(r^2 -a^2)}{\R}\Leta\Phiminustwo{0}
-\frac{6Mr^4 - 6a^2 r^3 - 18 a^2 Mr^2 -6a^4 r}{(\R)^2}\Phiminustwo{0},\\
\Boxhat_{-2}\Phiminustwo{2}
={}&
\bigg(6
-\frac{6Mr^3+7a^2 r^2 -20a^2 Mr+7a^4 }{(\R)^2}\bigg)\Phiminustwo{2}
-\frac{2ar}{\R}\Leta\Phiminustwo{2}
\notag\\
&
+\frac{20a^2(\PR)}{(\R)^2}\Phiminustwo{1}-\frac{8a(r^2 -a^2)}{\R}\Leta\Phiminustwo{1}
\notag\\
&-\frac{24a^3r}{\R}\Leta \Phiminustwo{0}-\frac{6a^2 (r^4+10Mr^3-6a^2 Mr-a^4)}{(\R)^2}\Phiminustwo{0}
\end{align}
\end{subequations}
and
\begin{subequations}
\label{eq:waveeq:Phiminustwo3and4}
\begin{align}
\label{eq:waveeq:Phiminustwo3}
\Boxhat_{-2}\Phiminustwo{3}={}&
\frac{2(r^3-3Mr^2 +a^2 r+a^2 M)}{(\R)^2} \curlVR\Phiminustwo{3}
-\frac{6ar}{\R}\Leta\Phiminustwo{3}
+\bigg(4+\frac{6Mr^3-3a^2 r^2-3a^4}{(\R)^2}\bigg)\Phiminustwo{3}\notag\\
&-\frac{6a(r^2-a^2)}{\R}\Leta\Phiminustwo{2}+\frac{6Mr^4 +34a^2 r^3 -138a^2 Mr^2 +34a^4 r +40a^4 M}{(\R)^2}\Phiminustwo{2}\notag\\
&-\frac{56a^3 r}{\R}\Leta\Phiminustwo{1}
-\frac{26a^2 r^4 -60a^2 Mr^3 +164a^4 Mr-26a^6}{(\R)^2}\Phiminustwo{1}\notag\\
&+\frac{24a^3 (r^2-a^2)}{\R}\Leta\Phiminustwo{0}
+\frac{60a^2 Mr^4 -24 a^4 r^3 -288a^4 M r^2 -24a^6 r +36a^6 M}{(\R)^2}\Phiminustwo{0},\\
\label{eq:waveeq:Phiminustwo4}
\Boxhat_{-2}\Phiminustwo{4}={}&
\frac{4(\PR)}{(\R)^2} \curlVR\Phiminustwo{4}
+\frac{10ar}{\R}\Leta\Phiminustwo{4}
+\frac{5(6 M r^3+a^2 r^2 - 8 a^2 M r+ a^4 )  }{(\R)^2}\Phiminustwo{4}\notag\\
&+\frac{40 a^2 (\PR) }{(\R)^2}\Phiminustwo{3}\notag\\
&+\frac{80a^3 r}{\R}\Leta\Phiminustwo{2}
-\frac{60 a^2(r^4 -6Mr^3 +10 a^2 Mr+a^4)}{(\R)^2}
\Phiminustwo{2}\notag\\
&+\frac{80a^3(r^2-a^2)}{\R}\Leta\Phiminustwo{1}
-\frac{128a^4(\PR)}{(\R)^2}
\Phiminustwo{1}\notag\\
&-\frac{96a^5r}{\R}\Leta\Phiminustwo{0}
+\frac{24a^4(r^4+34Mr^3 -30a^2 Mr -a^4)}{(\R)^2}
\Phiminustwo{0};
\end{align}
\end{subequations}

\item for $s=+2$,
\begin{subequations}
\begin{align}
\Boxhat_{+2}\Xi_{+2}^{(0)}
={}&-\frac{4(r^3-3Mr^2 +a^2 r+a^2 M)}{(\R)^2} \Xi_{+2}^{(1)}
-\frac{10ar}{\R}\Leta\Xi_{+2}^{(0)}
+\bigg(4
+\frac{6Mr^3-3a^2 r^2 -3a^4 }{(\R)^2}\bigg)\Xi_{+2}^{(0)},\\
\Boxhat_{+2}\Xi_{+2}^{(1)}
={}&-\frac{2(r^3-3Mr^2 +a^2 r+a^2 M)}{(\R)^2}\Xi_{+2}^{(2)}
+\bigg(6
-\frac{6Mr^3+7a^2 r^2 -20a^2 Mr+7a^4 }{(\R)^2}\bigg)\Xi_{+2}^{(1)}\notag\\
&
-\frac{6ar}{\R}\Leta \Xi_{+2}^{(1)}+\frac{6a(r^2 -a^2)}{\R}\Leta\Xi_{+2}^{(0)}-\frac{6Mr^4 - 6a^2 r^3 - 18 a^2 Mr^2 -6a^4 r}{(\R)^2}\Xi_{+2}^{(0)},\\
\Boxhat_{+2}\Xi_{+2}^{(2)}
={}&
\bigg(6
-\frac{6Mr^3+7a^2 r^2 -20a^2 Mr+7a^4 }{(\R)^2}\bigg)\Xi_{+2}^{(2)}
-\frac{2ar}{\R}\Leta\Xi_{+2}^{(2)}
\notag\\
&
+\frac{20a^2(\PR)}{(\R)^2}\Xi_{+2}^{(1)}-\frac{8a(r^2 -a^2)}{\R}\Leta\Xi_{+2}^{(1)}\notag\\
&+\frac{24a^3r}{\R}\Leta \Xi_{+2}^{(0)}
-\frac{6a^2 (r^4+10Mr^3-6a^2 Mr-a^4)}{(\R)^2}\Xi_{+2}^{(0)}.
\end{align}
\end{subequations}
\end{itemize}
\end{cor}

For the spin $-\sfrak$ component, it is surprising that a linear combination of $\{\Phiminuss{i_0}\}_{i_0\leq i}$ satisfies the basically the same equation as the one of $\PhiplussHigh{i-2\sfrak}$, for any $i\geq 2\sfrak$. This allows us to focus on one single spin component when deriving the energy decay estimates as the argument for the other spin component is similar. Cf. Section \ref{sect:APL}. Such a linear combination is as follows.

\begin{definition}
\label{def:dotPhisHighi}
\begin{itemize}
\item
For $\sfrak=0$, define $\dot\Phi_{0}^{(i)}\doteq\Phi_{0}^{(i)}$ for any $i\in \mathbb{N}$;
\item
For $\sfrak=1$, define
\begin{subequations}
\begin{align}
\dotPhiminus{2}\doteq{}&\Phiminus{2}+a^2\Phiminus{0},\\
\dotPhiminus{i}\doteq{}&\curlVR^{i-2}\dotPhiminus{2}, \quad \forall i>2;
\end{align}
\end{subequations}
\item For $\sfrak=2$, define 
\begin{subequations}
\begin{align}
\dotPhiminustwo{4}\doteq{}&\Phiminustwo{4}+10a^2\Phiminustwo{2}+9a^4 \Phiminustwo{0},\\
\dotPhiminustwo{i}\doteq{}&\curlVR^{i-4}\dotPhiminustwo{4}, \quad \forall i>4.
\end{align}
\end{subequations}
\end{itemize}
\end{definition}

We can derive the governing equations for the above-defined scalars $\dotPhisHigh{i}$ for $i\geq 2\sfrak$.

\begin{prop}
\label{prop:wavesys:dotPhisHighi}
Let $\sfrak\in \{0,1,2\}$ and let $i\in\mathbb{N}$.
The scalars $\dotPhisHigh{2\sfrak+i}$ satisfy the following wave equations
\begin{align}
\label{eq:wave:dotPhisHighi:general}
\Boxhat_{-\sfrak}\dotPhisHigh{2\sfrak+i}
={}&\frac{2(\sfrak+i)(r^3-3Mr^2 +a^2 r+a^2 M)}{(\R)^2} \curlVR\dotPhisHigh{2\sfrak+i}
+\sum_{\substack{0\leq j \leq i-1, \\ \frac{i-j-1}{2}\in \mathbb{N}}}X_{\sfrak,i,j}\Leta\dotPhisHigh{2\sfrak+j}
\notag\\
&-i(2\sfrak+i+1)\dotPhisHigh{2\sfrak+i}
-\sum_{j=0}^{i-1}{Z_{\sfrak,i,j}}\dotPhisHigh{2\sfrak+j}
+\sum_{n=0,1}\sum_{j=0}^i w_{\sfrak,i,j,n}\Leta^n \dotPhisHigh{2\sfrak+j},
\end{align}
with functions $w_{\sfrak,i,j,n}$ and constants
$Z_{\sfrak,i,j}$ and $X_{\sfrak,i,j}$ being the same as in Proposition \ref{prop:wavesys:PhisHighi}.
\end{prop}

\begin{proof}
First,  equations of $\dotPhiminus{2}$ and $\dotPhiminustwo{4}$ can be verified directly from Definition \ref{def:dotPhisHighi} and using the equations in Corollary  \ref{cor:wavesys:basicRW}. This proves $i=0$ case.

Then, one notices that the RHS of the governing equation of $\dotPhisHigh{2\sfrak}$ is in the same form as the one of equation \eqref{eq:wave:PhisHighi:general} for $i=0$ and $s=+\sfrak$. (Note that however the constant coefficient of $\dotPhisHigh{2\sfrak}$ term on the RHS differs from the one of $\Phiplustwo$ term on the RHS of equation \eqref{eq:wave:PhisHighi:general} for $i=0$ and $s=+\sfrak$.) Equation \eqref{eq:wave:dotPhisHighi:general} for general $i>0$ can then be proven in an exactly same manner as proving equation \eqref{prop:wavesys:PhisHighi} in the proof of Proposition \ref{prop:wavesys:PhisHighi}.
\end{proof}

We then define new scalars $\hatPhiplussHigh{i}$ (resp. $\hatPhiminuss{2\sfrak+i}$) constructed from a linear combination  (with constant coefficients)  of  $\{\PhiplussHigh{i'}\}_{i'\leq i}$ (resp. $\{\dotPhisHigh{i'+2\sfrak}\}_{i'\leq i}$) such that we can eliminate the term $-\sum_{j=0}^{i-1}{Z_{s,i,j}}\PhisHigh{j}$ in equation \eqref{prop:wavesys:PhisHighi} (resp. the term $-\sum_{j=0}^{i-1}{Z_{\sfrak,i,j}}\dotPhisHigh{2\sfrak+j}$ in equation \eqref{eq:wave:dotPhisHighi:general}). These eliminated terms are obstructions to deriving $r^p$ estimates for an extended range of $p$, thus to deriving further energy decay estimates for the spin $\pm \sfrak$ components. It is these linear combinations that successfully remove these terms and these combinations are unique\footnote{The uniqueness can be seen from the proof.} up to an overall nonzero multiplicative constant.
\begin{prop}
\label{prop:wavesys:hatPhisHighi}
Let $i\in \mathbb{N}$. There exist constants $\{x_{\sfrak,i,j,n}\}_{0\leq j\leq i-1, 0\leq n\leq i-j}$ such that the scalars $\hatPhisHigh{i}$ defined by
\begin{subequations}\label{ansatz:hatPhisHigh}
\begin{align}
\hatPhiplussHigh{i}\doteq{}&\PhiplussHigh{i}+\sum_{j=0}^{i-1} \sum_{n=0}^{i-j} x_{\sfrak,i,j,n}\Leta^n\hatPhiplussHigh{j}\\
\hatPhiminuss{2\sfrak+i}\doteq{}&\dotPhisHigh{2\sfrak+i}+\sum_{j=0}^{i-1} \sum_{n=0}^{i-j} x_{\sfrak,i,j,n}\Leta^n\hatPhisHigh{2\sfrak+j}
\end{align}
\end{subequations}
satisfy the following wave equations
\begin{align}
\label{eq:wave:hatPhisHighi:an}
\Boxhat_{s}\hatPhisHigh{i}
={}&\frac{2(s+i)(r^3-3Mr^2 +a^2 r+a^2 M)}{(\R)^2}\curlVR\hatPhisHigh{i}
-(i+2s)(i+1)\hatPhisHigh{i}
+\hat{H}_{s,i}
\end{align}
with
\begin{subequations}
\label{def:hatHsi:error:47}
\begin{align}
\label{def:hatHsi:error:various:weak}
\hat{H}_{+\sfrak,i}={}&\sum_{0\leq j\leq i}\sum_{n\leq d_i}O(r^{-1})\Leta^n\PhiplussHigh{j},\\
\hat{H}_{-\sfrak,i-2\sfrak}={}&\sum_{0\leq j\leq i}\sum_{n\leq d_i}O(r^{-1})\Leta^n\dotPhisHigh{j}
\end{align}
\end{subequations}
where the coefficient of the term $\Leta^n\PhiplussHigh{j}$ is the same as the coefficient of the term $\Leta^n\dotPhisHigh{j}$ in the above formulas \eqref{def:hatHsi:error:47}
and $d_i=i+1$.\footnote{The proof actually shows that one can take $d_i=i+1-j$.}
\end{prop}

\begin{proof}
It suffices to consider $s=+\sfrak$ case, since the proof for $s=-\sfrak$ case is exactly the same in view of the fact that equation \eqref{eq:wave:dotPhisHighi:general} of $\dotPhisHigh{2\sfrak+i}$ is in a same form as equation \eqref{eq:wave:PhisHighi:general} of $\PhiplussHigh{i}$ for any $i\in \mathbb{N}$.

To illustrate better the idea of this proof, we define the constants $V_{\sfrak,i}=2(\sfrak+i)$ and $Y_{\sfrak,i}=(2\sfrak+i)(i+1)$ and denote the last two terms in \eqref{eq:wave:PhisHighi:general} as $H_{+\sfrak,i}$, that is, $H_{+\sfrak,i}=\sum_{j=0}^i O(r^{-1}) \PhiplussHigh{j}+\sum_{j=0}^i O(r^{-1}) \Leta\PhiplussHigh{j}$. Equation \eqref{eq:wave:PhisHighi:general} can then be written as
\begin{align*}
\Boxhat_{+\sfrak}\PhiplussHigh{i}
={}&\frac{(r^3-3Mr^2 +a^2 r+a^2 M)}{(\R)^2}V_{\sfrak,i} \curlVR\PhiplussHigh{i}-Y_{\sfrak,i}\PhiplussHigh{i}
+\sum_{\substack{0\leq j \leq i-1, \\
 \frac{i-j-1}{2}\in \mathbb{N}}}{X}_{\sfrak,i,j}\Leta \PhiplussHigh{j}
-\sum_{j=0}^{i-1}{Z_{\sfrak,i,j}}\PhiplussHigh{j}
+H_{+\sfrak,i}.
\end{align*}

We shall prove the statement by induction.
In view of equation \eqref{eq:wave:PhisHigh0}, $\hatPhiplussHigh{0}=\PhiplussHigh{0}$ clearly satisfies \eqref{eq:wave:hatPhisHighi:an} with $i=0$. We then proceed by
assuming that we have chosen the constants $\{x_{\sfrak,i,j,n}\}_{0\leq j\leq i-1, 0\leq n \leq i-j}$ such that $\{\hatPhiplussHigh{j}\}_{0\leq j\leq i}$ satisfy \eqref{eq:wave:hatPhisHighi:an}, that is,
\begin{align*}
\Boxhat_{+\sfrak}\hatPhiplussHigh{j}
={}&\frac{(r^3-3Mr^2 +a^2 r+a^2 M)}{(\R)^2}V_{\sfrak,j} \curlVR\hatPhiplussHigh{j}-Y_{\sfrak,j}\hatPhiplussHigh{j}
+\hat{H}_{+\sfrak,j}.
\end{align*}
Using the general ansatz \eqref{ansatz:hatPhisHigh},
the above two equations then yield that $\hatPhiplussHigh{i+1}$ satisfies
\begin{align}
\label{eq:wave:hatPhisHighi+1:an}
\Boxhat_{+\sfrak}\hatPhiplussHigh{i+1}
={}&\frac{(r^3-3Mr^2 +a^2 r+a^2 M)}{(\R)^2}V_{\sfrak,i+1} \curlVR\hatPhiplussHigh{i+1}
\notag\\
&-Y_{\sfrak,i+1}\PhiplussHigh{i+1}
-\sum_{j=0}^{i}{Z_{\sfrak,i+1,j}}\PhiplussHigh{j}
+\sum_{\substack{0\leq j \leq i, \\ \frac{i-j}{2}\in \mathbb{N}}}{X}_{\sfrak,i+1,j}\Leta \PhiplussHigh{j}
-\sum_{j=0}^i\sum_{n=0}^{i+1-j} Y_{\sfrak,j}x_{\sfrak,i+1,j,n}\Leta^n\hatPhiplussHigh{j}\notag\\
&
+\frac{(r^3-3Mr^2 +a^2 r+a^2 M)}{(\R)^2}
\sum_{j=0}^i\sum_{n=0}^{i+1-j}(V_{\sfrak,j}-V_{\sfrak,i+1} )x_{\sfrak,i+1,j,n}\Leta^n\curlVR\hatPhiplussHigh{j}\notag\\
&
+{H}_{+\sfrak,i+1}
+\sum_{j=0}^i\sum_{n=0}^{i+1-j}x_{\sfrak,i+1,j,n}\Leta^n\hat{H}_{+\sfrak,j}.
\end{align}
The remaining step is to choose the constants $x_{\sfrak,i+1,j,n}$ such that the second line of the above equation equal $-Y_{\sfrak,i+1}\hatPhiplussHigh{i+1}$.  This is equivalent to requiring
\begin{align*}
&Y_{\sfrak,i+1}\sum_{j=0}^i \sum_{n=0}^{i+1-j} x_{\sfrak,i+1,j,n}\Leta^n\hatPhiplussHigh{j}
-\sum_{j=0}^{i}{Z_{\sfrak,i+1,j}}\PhiplussHigh{j}+\sum_{\substack{0\leq j \leq i, \\ \frac{i-j}{2}\in \mathbb{N}}}{X}_{\sfrak,i+1,j}\Leta \PhiplussHigh{j}
-\sum_{j=0}^i\sum_{n=0}^{i+1-j} Y_{\sfrak,j}x_{\sfrak,i+1,j,n}\Leta^n\hatPhiplussHigh{j}=0.
\end{align*}
By substituting in $\PhiplussHigh{j}= \hatPhiplussHigh{j}-\sum\limits_{j'=0}^{j-1}\sum\limits_{n=0}^{j-j'}x_{\sfrak,j,j',n}\Leta^n\hatPhiplussHigh{j'}$ that comes from \eqref{ansatz:hatPhisHigh}, the above equation becomes
\begin{align}
\label{eq:determine:xsijn}
\sum_{j=0}^i \sum_{n=0}^{i+1-j} (Y_{\sfrak,i+1}-Y_{\sfrak,j})x_{\sfrak,i+1,j,n}\Leta^n\hatPhiplussHigh{j}
={}&\sum_{j=0}^{i}{Z_{\sfrak,i+1,j}}\bigg( \hatPhiplussHigh{j}-\sum\limits_{j'=0}^{j-1}\sum\limits_{n=0}^{j-j'}x_{\sfrak,j,j',n}\Leta^n\hatPhiplussHigh{j'}\bigg)
\notag\\
&-\sum_{\substack{0\leq j \leq i, \\ \frac{i-j}{2}\in \mathbb{N}}}{X}_{\sfrak,i+1,j} \bigg( \Leta\hatPhiplussHigh{j}-\sum\limits_{j'=0}^{j-1}\sum\limits_{n=0}^{j-j'}x_{\sfrak,j,j',n}\Leta^{n+1}\hatPhiplussHigh{j'}\bigg).
\end{align}
Since the values of the constants $\{X_{\sfrak,i+1,j}\}_{0\leq j\leq i}$ and $\{Z_{\sfrak,i+1,j}\}_{0\leq j\leq i}$ are given in Proposition \ref{prop:wavesys:PhisHighi} and the difference $Y_{\sfrak,i+1}-Y_{\sfrak,j}=(i-j+1)(i+j+2\sfrak+2)$ is non-zero for any $i\in \mathbb{N}$, and since the values of constants $\{x_{\sfrak,j,j',n}\}_{0\leq j\leq i, 0\leq j'\leq j-1, 0\leq n\leq j-j'}$ are given, there is a unique solution for $\{x_{\sfrak,i+1,j,n}\}_{0\leq j\leq i, 0\leq n\leq i+1-j}$ to equation \eqref{eq:determine:xsijn}.

Finally, we denote the last two lines of \eqref{eq:wave:hatPhisHighi+1:an} as $\hat{H}_{+\sfrak,i+1}$. We shall show that
\begin{align}
\label{def:hatHsi:error:various:3}
\hat{H}_{+\sfrak,i}={}&\sum_{0\leq j\leq i}\sum_{n\leq i+1-j}O(r^{-1})\Leta^n\PhiplussHigh{j}.
\end{align}
which clearly yields \eqref{def:hatHsi:error:various:weak}.
Note from \eqref{ansatz:hatPhisHigh} that 
$\hatPhiplussHigh{i}=\sum_{j=0}^i\sum_{n=0}^{i-j} O(1)\Leta^n\PhiplussHigh{j}$, and by substituting this into the second last line of \eqref{eq:wave:hatPhisHighi+1:an}, one finds this second last line equals 
\begin{align*}
\sum_{j=0}^i\sum_{n=0}^{i+1-j}O(r^{-1})\Leta^n\curlVR\hatPhiplussHigh{j}
={}&\sum_{j=0}^i\sum_{n=0}^{i+1-j}O(r^{-1})\Leta^n\curlVR
\bigg( \sum_{x=0}^j \sum_{y=0}^{j-x}O(1)\Leta^y \PhiplussHigh{x}\bigg)\notag\\
={}&\sum_{j=0}^{i}\sum_{n=0}^{i+1-j}O(r^{-1}) \Leta^n\PhiplussHigh{j+1}.
\end{align*}
By the induction assumption \eqref{def:hatHsi:error:various:3} for $\hat{H}_{+\sfrak, i}$,  the last term in \eqref{eq:wave:hatPhisHighi+1:an} equals
\begin{align*}
\sum_{j=0}^i\sum_{n=0}^{i+1-j}\Leta^n\bigg(\sum_{x\leq j+1-y} \sum_{y=0}^j O(r^{-1}) \Leta^x \PhiplussHigh{y}\bigg)=\sum_{j=0}^i \sum_{n=0}^{i+2-j} O(r^{-1}) \Leta^{n} \PhiplussHigh{j}.
\end{align*}
In view of these discussions, we therefore conclude 
\eqref{def:hatHsi:error:various:3} for $\hat{H}_{+\sfrak, i+1}$ and prove \eqref{def:hatHsi:error:various:3} for general $i\in \mathbb{N}$. 
\end{proof}

\subsection{Wave equations for the modes of spin $\pm \sfrak$ components}
\label{subsect:waveeq:modes:gen}

The following definition is useful to calculate the commutator between the wave operator $\Boxhat_s$ and mode projection operators.

\begin{definition}
\label{def:Commells}
Let $\varphi_s$ be a spin $s$ scalar. Define
\begin{subequations}
\label{def:eq:Commells}
\begin{align}
\Comm{\ell}{s}{\varphi_s}\doteq{}& -a^2[\PJ_{\ell}^s, \sin^2\theta](\Lxi \varphi_s) +2ias[\PJ_{\ell}^s, \cos\theta](\varphi_s),\\
\Comm{m,\ell}{s}{\varphi_s}\doteq{}& -a^2[\PJ_{m,\ell}^s, \sin^2\theta](\Lxi \varphi_s) +2ias[\PJ_{m,\ell}^s, \cos\theta](\varphi_s),\\
\Comm{\geq\ell}{s}{\varphi_s}\doteq{}&\sum_{\ell'\geq \ell}\Comm{\ell'}{s}{\varphi_s}.
\end{align}
\end{subequations}
It holds
\begin{align}
\label{eq:Commells:sum}
\Comm{\geq \ell}{s}{\varphi_s}+\sum_{\sfrak\leq \ell'\leq \ell-1}\Comm{\ell'}{s}{\varphi_s}=0
\end{align}
and
\begin{subequations}
\label{comm:BoxhatsandPJ}
\begin{align}
[\Boxhat_s, \PJ_{\ell}^s]\varphi_s={}&\Lxi\Comm{\ell}{s}{\varphi_s},\\
[\Boxhat_s, \PJ_{m,\ell}^s]\varphi_s={}&\Lxi\Comm{m,\ell}{s}{\varphi_s},\\
[\Boxhat_s, \PJ_{\geq\ell}^s]\varphi_s={}&\Lxi\Comm{\geq\ell}{s}{\varphi_s}=-\sum_{\sfrak\leq\ell'\leq \ell-1}\Lxi\Comm{\ell'}{s}{\varphi_s}.
\end{align}
\end{subequations}
\end{definition}

By projecting \eqref{eq:wave:hatPhisHighi:an} onto an $\ell$ mode and using the above definition, we achieve

\begin{prop}
\label{prop:wavesys:hatPhisHighi:ellmode}
Let $\ell\geq \sfrak$, and let $\sfrak-s\leq i\leq \ell-s$.
The scalars $\ellmode{\hatPhisHigh{i}}{\ell}$, the $\ell$ mode of $\hatPhisHigh{i}$ that is defined in \eqref{ansatz:hatPhisHigh},
satisfy the following wave equations
\begin{align}
\label{eq:wave:hatPhisHighi:an:ellmode}
\Boxhat_{s}\ellmode{\hatPhisHigh{i}}{\ell}
={}&\frac{2(s+i)(r^3-3Mr^2 +a^2 r+a^2 M)}{(\R)^2}\curlVR\ellmode{\hatPhisHigh{i}}{\ell}
-(2s+i)(i+1)\ellmode{\hatPhisHigh{i}}{\ell}
+\ellmode{\hat{H}_{s,i}}{\ell}+\Lxi\Comm{\ell}{s} {\hatPhisHigh{i}},
\end{align}
with $\ellmode{\hat{H}_{s,i}}{\ell}$ being the $\ell$ mode of $\hat{H}_{s,i}$ defined in \eqref{def:hatHsi:error:47}.
\end{prop}

Further, we base on the above result and define a new scalar supported on a fixed mode such that it satisfies a transport equation with the source enjoying faster decay in $r$, a property that is essential in further extending the $r^p$ hierarchy in order to achieve almost sharp decay in Section \ref{sect:ED:modes}.

\begin{prop}
\label{prop:wavesys:tildePhisHighi:ellmode}
Let $\ell\geq \sfrak$, and let $i\in \mathbb{N}$. The scalars $\hatPhisHigh{i}$ defined by
\begin{align}\label{ansatz:tildePhisHigh:ellmode}
\tildePhisHighell{s}{\ell}\doteq\Proj{\ell}\Big(\curlVR\hatPhisHigh{\ell-s}
-\half\big(2a\Leta\hatPhisHigh{\ell-s}
+a^2\sin^2 \theta\Lxi \hatPhisHigh{\ell-s}
-2ias\cos\theta \hatPhisHigh{\ell-s}\big)\Big)
\end{align}
satisfy the following wave equations
\begin{align}
\label{eq:wave:tildePhisHighi:an:ellmode}
-\mu Y\tildePhisHighell{s}{\ell}
-\frac{2(\ell+1)(r^3-3Mr^2 +a^2 r+a^2 M)}{(\R)^2}\tildePhisHighell{s}{\ell}
={}&\tilde{H}_{s,\ell},
\end{align}
with
\begin{align}
\label{eq:tildeHsell-s}
\tilde{H}_{s,\ell}
={}&\sum_{n\leq d_{\ell-s}}\sum_{\sfrak-s\leq j\leq \ell-s}O(r^{-1})\Leta^n\ellmode{\hatPhisHigh{j}}{\ell}\notag\\
&
+\sum_{j=0,1}O(r^{-1}) (rV)^j \Proj{\ell}\big(2a\Leta\hatPhisHigh{\ell-s}
+a^2\sin^2 \theta\Lxi \hatPhisHigh{\ell-s}
-2ias\cos\theta \hatPhisHigh{\ell-s}\big)\notag\\
&+O(r^{-2}) \Leta \Proj{\ell}\big(2a\Leta\hatPhisHigh{\ell-s}
+a^2\sin^2 \theta\Lxi \hatPhisHigh{\ell-s}
-2ias\cos\theta \hatPhisHigh{\ell-s}\big)
\end{align}
and $d_{\ell-s}$ a constant depending only on $\ell-s$.

Further, by defining $\tildePhisHighmell{s}{m}{\ell}$ and $\tilde{H}_{s,m,\ell}$ as the $m$ azimuthal  modes of $\tildePhisHighell{s}{\ell}$ and $\tilde{H}_{s,\ell}$ respectively, it satisfies
\begin{align}
\label{ansatz:tildePhisHigh:mellmode:2}
-\mu Y\tildePhisHighmell{s}{m}{\ell}
-\frac{2(\ell+1)(r^3-3Mr^2 +a^2 r+a^2 M)}{(\R)^2}\tildePhisHighmell{s}{m}{\ell}
={}&\tilde{H}_{s,m,\ell}.
\end{align}
\end{prop}

\begin{remark}\label{definition:tildemodephi:general}
The scalar $\tildePhisHighmell{s}{m}{\ell}$ actually equals the Newman--Penrose constant of the $(m,\ell)$ mode of the spin $s$ component in the nonvanishing N--P constant case in \cite{angelopoulos2018vector,Ma20almost,angelopoulos2021late,MaZhang21PriceSchw}.
\end{remark}

\begin{proof}
We have shown in the above proposition that projecting  \eqref{eq:wave:hatPhisHighi:an} onto an $\ell$ mode leads to equation \eqref{eq:wave:hatPhisHighi:an:ellmode}, which can be expanded into
\begin{align}
\label{eq:wave:hatPhisHighi:ellmode:i=ell-s}
\hspace{4ex}&\hspace{-4ex}
-(\R)YV\ellmode{\hatPhisHigh{i}}{\ell} +2a\Lxi\Leta\ellmode{\hatPhisHigh{i}}{\ell}
+a^2\Lxi(\Proj{\ell}(\sin^2 \theta\Lxi \hatPhisHigh{i}))
-2ias\Lxi(\Proj{\ell}(\cos\theta \hatPhisHigh{i}))\notag\\
={}&-(\edthR\edthR'+
(2s+i)(i+1))\ellmode{\hatPhisHigh{i}}{\ell}
+\frac{2(s+i)(r^3-3Mr^2 +a^2 r+a^2 M)}{(\R)^2}\curlVR\ellmode{\hatPhisHigh{i}}{\ell}
+\ellmode{\hat{H}_{s,i}}{\ell} .
\end{align}
Substituting in $\Lxi=\half (\mu Y +V)-\frac{a}{\R}\Leta$,  the LHS of equation \eqref{eq:wave:hatPhisHighi:ellmode:i=ell-s} equals
\begin{align*}
\hspace{2ex}&\hspace{-2ex}
-\mu Y\curlVR\ellmode{\hatPhisHigh{i}}{\ell}
+\Lxi\big(2a\Leta\ellmode{\hatPhisHigh{i}}{\ell}
+a^2\Proj{\ell}(\sin^2 \theta\Lxi \hatPhisHigh{i})
-2ias\Proj{\ell}(\cos\theta \hatPhisHigh{i})\big)
-\frac{2(\PR)}{(\R)^2}\curlVR\ellmode{\hatPhisHigh{i}}{\ell}
\notag\\
={}&-\mu Y\bigg(\curlVR\ellmode{\hatPhisHigh{i}}{\ell}
-\half \big(2a\Leta\ellmode{\hatPhisHigh{i}}{\ell}
+a^2\Proj{\ell}(\sin^2 \theta\Lxi \hatPhisHigh{i})
-2ias\Proj{\ell}(\cos\theta \hatPhisHigh{i})\big)
\bigg)
-\frac{2(\PR)}{(\R)^2}\curlVR\ellmode{\hatPhisHigh{i}}{\ell}
\notag\\
&+\bigg(\half V-\frac{a}{\R}\Leta\bigg)\big(2a\Leta\ellmode{\hatPhisHigh{i}}{\ell}
+a^2\Proj{\ell}(\sin^2 \theta\Lxi \hatPhisHigh{i})
-2ias\Proj{\ell}(\cos\theta \hatPhisHigh{i})\big).
\end{align*}
From now on, take $i=\ell-s$. Then by \eqref{eq:l=l0mode:eigenvalue},
\begin{align*}
(\edthR\edthR' +(2s+i)(i+1))\ellmode{\hatPhisHigh{i}}{\ell}=(-(\ell+s)(\ell-s+1)+(\ell+s)(\ell-s+1))\ellmode{\hatPhisHigh{\ell-s}}{\ell}=0.
\end{align*}
The above discussions together thus yield that the scalar $\tildePhisHighell{s}{\ell}$ defined in \eqref{ansatz:tildePhisHigh:ellmode} satisfies
\begin{align}
\label{eq:wave:hatPhisHighi:ellmode:i=ell-s:v4}
-\mu Y\tildePhisHighell{s}{\ell}
={}&\frac{2(\ell+1)(r^3-3Mr^2 +a^2 r+a^2 M)}{(\R)^2}\curlVR\ellmode{\hatPhisHigh{\ell-s}}{\ell}\notag\\
&
+\ellmode{\hat{H}_{s,\ell-s}}{\ell}
+\bigg(\half V-\frac{a}{\R}\Leta\bigg)\big(2a\Leta\ellmode{\hatPhisHigh{\ell-s}}{\ell}
+a^2\Proj{\ell}(\sin^2 \theta\Lxi \hatPhisHigh{\ell-s})
-2ias\Proj{\ell}(\cos\theta \hatPhisHigh{\ell-s})\big).
\end{align}
We use \eqref{ansatz:tildePhisHigh:ellmode} to rewrite $\curlVR\ellmode{\hatPhisHigh{\ell-s}}{\ell}$ as
\begin{align*}
\curlVR\ellmode{\hatPhisHigh{\ell-s}}{\ell}=\tildePhisHighell{s}{\ell}+\half\Proj{\ell}\big(2a\Leta\hatPhisHigh{\ell-s}
+a^2\sin^2 \theta\Lxi \hatPhisHigh{\ell-s}
-2ias\cos\theta \hatPhisHigh{\ell-s}\big)
\end{align*}
and substitute this into equation \eqref{eq:wave:hatPhisHighi:ellmode:i=ell-s:v4}, then
the desired equation \eqref{eq:wave:tildePhisHighi:an:ellmode} holds with
\begin{align}\label{definition:tildeH:2345}
\tilde{H}_{s,\ell}
={}&\ellmode{\hat{H}_{s,\ell-s}}{\ell}
+\bigg(\half V-\frac{a}{\R}\Leta\bigg)\big(2a\Leta\ellmode{\hatPhisHigh{\ell-s}}{\ell}
+a^2\Proj{\ell}(\sin^2 \theta\Lxi \hatPhisHigh{\ell-s})
-2ias\Proj{\ell}(\cos\theta \hatPhisHigh{\ell-s})\big)\notag\\
&+\frac{(\ell+1)(r^3-3Mr^2 +a^2 r+a^2 M)}{(\R)^2}\Proj{\ell}\big(2a\Leta\hatPhisHigh{\ell-s}
+a^2\sin^2 \theta\Lxi \hatPhisHigh{\ell-s}
-2ias\cos\theta \hatPhisHigh{\ell-s}\big).
\end{align}
This expression can manifestly be put into the form of \eqref{eq:tildeHsell-s}.
\end{proof}

\subsection{Teukolsky--Starobinsky identities}
\label{sect:TSI}

As we have discussed, the spin $\pm \sfrak$ components are in fact related to each other by purely differential relations--the \emph{Teukolsky-Starobinsky identities} (TSI) \cite{TeuPress1974III,starobinsky1973amplification}. The covariant form of these identities is derived in \cite{aksteiner2019new}. These identities are of fundamental importance in our analysis for both of the spin $\pm\sfrak$ components in this paper.

\begin{lemma}
\label{lem:TSI:gene}
\begin{enumerate}
\item
There are the following TSI for the spin $\pm 1$ components\begin{subequations}
\label{eq:TSI:spin1:system}
\begin{align}
\label{eq:TSI:simpleform}
(\edthR'-ia\sin\theta\Lxi)^2\psiplus
={}&\Delta\VR^2(\Delta\psiminus),\\
\label{eq:otherTSI:simpleform}
(\edthR+ia\sin\theta\Lxi)^2 \psiminus={}&
Y^2 \psiplus.
\end{align}
\end{subequations}
Further, equation \eqref{eq:TSI:simpleform} can be written as
\begin{align}
\label{eq:TSI:simpleform:v2}
(\edthR'-ia\sin\theta\Lxi)^2 \Phiplus={}\Phiminus{2}+a^2\Phiminus{0}=\dotPhiminus{2}.
\end{align}

\item
There are the following TSI for the spin $\pm 2$ components of the linearized gravity:
\begin{subequations}
\label{eq:TSIspin2}
\begin{align}
\label{eq:TSIspin2:V-2}
 (\edthR'-ia\sin\theta\Lxi)^4 \psiplustwo  -12 M \overline{\Lxi\psiplustwo}={}&\Delta^2\VR^4 (\Delta^2 \psiminustwo),\\
\label{eq:TSIspin2:Y+2}
(\edthR+ia\sin\theta\Lxi)^4 \psiminustwo +12 M\overline{\Lxi\psiminustwo}={}&Y^4 (\psiplustwo) .
\end{align}
\end{subequations}
Further, equation \eqref{eq:TSIspin2:V-2} can be written as
\begin{align}
\label{eq:TSIspin2:Phiminustwo4}
 (\edthR'-ia\sin\theta\Lxi)^4 \Phiplustwo  -12 M \overline{\Lxi\Phiplustwo}={}&\Phiminustwo{4}+10a^2\Phiminustwo{2}+9a^4\Phiminustwo{0}=\dotPhiminustwo{4}.
\end{align}
\end{enumerate}
\end{lemma}

\begin{remark}
We remark that these TSI will be projected on spin-weighted spherical harmonic modes and, because of the spin-weighted spherical harmonic modes coupling, the obtained equations  are different from the original TSI in \cite{TeuPress1974III} in which a projection on spin-weighted spheroidal harmonic modes is applied and no mode coupling is present.
\end{remark}

\begin{proof}
The TSI \eqref{eq:TSI:spin1:system} and \eqref{eq:TSIspin2} can be derived from the covariant form \cite{aksteiner2019new}, or, following the same way as in \cite{TeuPress1974III,starobinsky1973amplification}. In particular, one notes that these equations \eqref{eq:TSI:simpleform}, \eqref{eq:TSI:simpleform}, \eqref{eq:TSIspin2:V-2} and \eqref{eq:TSIspin2:Y+2} are the physical space version of equations (3.9)--(3.10), (3.15)--(3.16), (3.21)--(3.22) and (3.27)--(3.28) of \cite{TeuPress1974III} in the frequency space, respectively.

To show formula \eqref{eq:TSI:simpleform:v2}, we substitute $\Delta \psiminus=\sqrt{\R}\Phiminus{0}$ and $\Delta^{-1}\psiplus=(\R)^{-3/2}\Phiplus$ into equation \eqref{eq:TSI:simpleform} and find that the RHS equals
\begin{align}
&\Delta\VR^2 (\sqrt{\R}\Phiminus{0})
=\Delta\VR \bigg(\frac{r}{\sqrt{\R}}\Phiminus{0}+\frac{1}{\sqrt{\R}}\Phiminus{1}\bigg)=\frac{\Delta}{(\R)^{\frac{3}{2}}}(\Phiminus{2}+a^2\Phiminus{0}).
\end{align}
This thus proves \eqref{eq:TSI:simpleform:v2}. Equation \eqref{eq:TSIspin2:Phiminustwo4} is similarly proven by plugging $\Delta^2 \psiminustwo=(\R)^{3/2}\Phiminustwo{0}$ and $\Delta^{-2}\psiplustwo=(\R)^{-5/2}\Phiplustwo$ into equation \eqref{eq:TSIspin2:V-2}.
\end{proof}


\section{Almost sharp decay estimates}
\label{sect:APL}


In this section, we show the almost sharp decay for the spin $\pm \sfrak$ components in a subextreme Kerr spacetime under a conditional assumption of a \emph{basic energy and Morawetz (BEAM) estimate}  (also known as \emph{integrated local energy decay estimates}) for an inhomogeneous TME. This BEAM estimate assumption is introduced in Section \ref{sect:BEAM} and we apply it to achieve the resulting BEAM estimates for the spin $\pm \sfrak$ components as well as for their modes in a subextreme Kerr. We then prove $r^p$ estimates for an inhomogeneous spin-weighted wave equation and an inhomogeneous transport equation in Section \ref{subsect:rp} and make use of these $r^p$ estimates together with the BEAM estimates to prove energy decay for both of the spin $\pm \sfrak$ components in Section \ref{subsect:EDE} and their modes in Section \ref{sect:ED:modes}. In the end, these energy decay estimates are utilized in Section \ref{subsect:APL:general} to prove the almost sharp decay.

\subsection{Assumptions on the BEAM estimates}
\label{sect:BEAM}

To properly state the BEAM estimate assumption, we first define the energies and spacetime Morawetz integrals of spin $s$ scalars.

\begin{definition}
Let  $\reg\geq \sfrak+1$, let $\varsigma\in (0,\half)$,  and let $\delta >0$ be a small constant. Let $\varphi_s$ be an arbitrary spin $s$ scalar in a subextreme Kerr spacetime $(\mathcal{M}, g_{M,a})$. Let $\chitrap$ be a smooth real-valued function which equals $0$ in the trapping region and $1$ a bit away from the trapping region.
Define the following energies
\begin{align}
E^{\reg}_{\Sigmatb}(\varphi_{+\sfrak})\doteq&{}\sum\limits_{\abs{\mathbf{a}}\leq \reg-\sfrak-1}\bigg(\sum_{0\leq i\leq \sfrak-1}
\norm{\PDeri^{\mathbf{a}}(r^{-\varsigma}Y^i\varphi_{+\sfrak})}^2_{W_{-2}^1(\Sigmatb)}
+\norm{\PDeri^{\mathbf{a}}Y^{\sfrak}\varphi_{+\sfrak}}^2_{W_{-2}^1(\Sigmatb)}
\bigg),\\
 E^{\reg}_{\Sigmatb}(\varphi_{-\sfrak})\doteq&{}\sum_{i=0}^{\sfrak}\sum_{\abs{\mathbf{a}}\leq \reg-\sfrak-1}
\norm{\PDeri^{\mathbf{a}}(r^2 V)^i\varphi_{-\sfrak}}^2_{W_{-2}^1(\Sigmaone)}
\end{align}
and the following spacetime Morawetz integrals for any $\tb_2>\tb_1\geq \tb_0$
\begin{align}
M^{\reg}_{\Donetwo}(\varphi_{+\sfrak})\doteq{}&\sum_{\abs{\mathbf{a}}\leq \reg-\sfrak-1}\bigg(\sum_{i=0}^{\sfrak-1}\norm{\PDeri^{\mathbf{a}}(r^{-\varsigma}Y^{i}\varphi_{+\sfrak})}^2_{W_{-3-\delta}^0(\Donetwo)}
+\norm{\PDeri^{\mathbf{a}}Y^{\sfrak}\varphi_{+\sfrak}}^2_{W_{-3-\delta}^0(\Donetwo)}\notag\\
&
+\sum_{i=0}^{\sfrak-1}\norm{\PDeri^{\mathbf{a}}\PSDeri(r^{-\varsigma}Y^i\varphi_{+\sfrak})}^2_{W_{-3-\delta}^0(\Donetwo)}
+\norm{\chitrap\PDeri^{\mathbf{a}}\PSDeri(Y^{\sfrak}\varphi_{+\sfrak})}^2_{W_{-3-\delta}^0(\Donetwo)}\notag\\
&+\norm{\PDeri^{\mathbf{a}}\partial_{r^*}(Y^{\sfrak}\varphi_{+\sfrak})}^2_{W_{-3-\delta}^0(\Donetwo)}\bigg),
\\
M^{\reg}_{\Donetwo}(\varphi_{-\sfrak})\doteq{}&\sum_{i=0}^{\sfrak}\sum_{\abs{\mathbf{a}}\leq \reg-\sfrak-1}\Big(
\norm{\PDeri^{\mathbf{a}}\curlV^i\varphi_{-\sfrak}}^2_{W_{-3-\delta}^0(\Donetwo)}
+\norm{\chitrap\PDeri^{\mathbf{a}}\PSDeri(\curlV^i\varphi_{-\sfrak})}^2_{W_{-3-\delta}^0(\Donetwo)}\notag\\
&\qquad\qquad\qquad+\norm{\PDeri^{\mathbf{a}}\partial_{r^*}(\curlV^i\varphi_{-\sfrak})}^2_{W_{-3-\delta}^0(\Donetwo)}
\Big).
\end{align}
\end{definition}

We can now state our main assumption on the BEAM estimates for an inhomogeneous TME.

\begin{assump}[Assumption on the BEAM estimates for inhomogeneous TME]
\label{ass:BEAM:inhomogeneous}Let $s\in \{0,\pm1, \pm2\}$. Let $M>0$ and $\abs{s}<M$. Let $\varphi_{s}$ and $N[\varphi_{s}]$ be spin $s$ scalars and let $\varphi_s$ satisfy the following inhomogeneous TME on a subextreme Kerr background:
\begin{align}
\label{eq:TME:inhomo}
&\Boxhat_s\varphi_s+2s((r-M)Y-2r\Lxi)\varphi_s
+\frac{2ar}{\R}\Leta\varphi_s\notag\\
&\qquad
+\bigg(\frac{2s r(r-M)}{\R}-\frac{2Mr^3+a^2r^2-4a^2Mr+a^4}{(\R)^2}\bigg)\varphi_s={}N[\varphi_s].
\end{align}
We say that the BEAM estimates assumption for this inhomogeneous TME is satisfied on a Kerr background $(\mathcal{M}, g_{M,a})$ if there exists  $\varsigma\in (0,\half)$ such that given any $0<\delta<1/2$ and any $\sfrak+1\leq \reg\in \mathbb{N}^+$, there exist universal constants $\regl\geq 0$ and $C=C(M, a, \delta, \reg)$\footnote{This constant depends on the hyperboloidal foliation via the function $\hhyp=\hhyp(r)$. For simplicity, we shall suppress this dependence for this universal constant throughout this work as one can fix this function once for all.} such that the following
BEAM estimates are valid in the region $\Donetwo$ for any $\tb_2>\tb_1\geq \tb_0$:
\begin{subequations}\label{eq:BEAM:inhomo1}
\begin{align}
\label{eq:BEAM:-s:inhomo}
\hspace{4ex}&\hspace{-4ex}
E^{\reg}_{\Sigmatwo}(\varphi_{-\sfrak})
+M^{\reg}_{\Donetwo}(\varphi_{-\sfrak})\notag\\
\leq {}&C\Big(E^{\reg}_{\Sigmaone}(\varphi_{-\sfrak})
+\sum_{\tb'\in\{\tb_1,\tb_2\}}E^{\reg+\regl}_{\Sigma_{\tb'}}(N[\varphi_{-\sfrak}])
+\sum_{i_0=0,1}\sum_{i=0}^{\sfrak}\norm{\Lxi^{i_0}\curlV^i N[\varphi_{-\sfrak}]}^2_{W_{-3+\delta}^{\reg+\regl}(\Donetwo)}
\Big),\\
\label{eq:BEAM:+s:inhomo}
\hspace{4ex}&\hspace{-4ex}E^{\reg}_{\Sigmatwo}(\varphi_{+\sfrak})
+M^{\reg}_{\Donetwo}(\varphi_{+\sfrak})\notag\\
\leq {}&C\Big(E^{\reg}_{\Sigmaone}(\varphi_{+\sfrak})
+\sum_{\tb'\in\{\tb_1,\tb_2\}}E^{\reg+\regl}_{\Sigma_{\tb'}}(N[\varphi_{+\sfrak}])
+\sum_{i_0=0,1}\norm{\Lxi^{i_0}N[\varphi_{+\sfrak}]}^2_{W_{-3+\delta}^{\reg+\regl}(\Donetwo)}\Big).
\end{align}
\end{subequations}
\end{assump}

\begin{remark}
The requirement that we need to impose bounds over extra $\regl$-order derivatives of the inhomogeneous term is due to the well-known trapping phenomenon which causes a loss of regularity in the Morawetz estimates. In fact, as can be seen from the proof in Remark \ref{rem:BEAM:pf:slowly}, $\regl=1$ is sufficient.
\end{remark}

\begin{remark}
\label{rem:BEAM:pf:slowly}
The BEAM estimates for the TME with vanishing inhomogeneous term   are proven for $s=0$ in \cite{dafermos2016decay} on any subextreme Kerr, $s=\pm 1$ in \cite{Ma2017Maxwell} on slowly rotating Kerr and $s=\pm 2$ in \cite{Ma17spin2Kerr} on slowly rotating Kerr, and the proof can be easily adapted to show this  BEAM estimate assumption \ref{ass:BEAM:inhomogeneous} in these cases. Consider only  $s=-\sfrak$ case, the case $s=+\sfrak$ being similarly treated. The general approach in these works is to consider the wave systems of $\{\curlVR^{i}(\mu^{\sfrak}\varphi_{-\sfrak})\}_{i=0,1,\ldots, 2\sfrak}$ (hence with inhomogeneous terms $\{\curlVR^{i}(\mu^{\sfrak}N[\varphi_{-\sfrak}])\}_{i=0,1,\ldots, 2\sfrak}$), therefore it suffices to bound the following integral
\begin{align}
\sum_{k_0=0}^{\reg}\sum_{i=0}^{2\sfrak}\bigg|\int_{\Donetwo}\Sigma^{-1}\Re\big(\partial^{k_0}\curlVR^{i}(\mu^{\sfrak}N[\varphi_{-\sfrak}]) \overline{X \partial^{k_0}\curlVR^{i}(\mu^{\sfrak}\varphi_{-\sfrak})}\big)\di^4\mu\bigg|
\end{align}
 by the last two terms in \eqref{eq:BEAM:-s:inhomo}, with $X\varphi=(O(1)\Lxi +O(r^{-1})\Leta + O(1)Y +O(r^{-1}))\varphi$. The integral outside the trapping region and the integral supported in the trapping region but arising from either the $r$-derivative part or no derivative part of $X$ can all be estimated using Cauchy--Schwarz, and it remains to bound the integral of $O(1)\Sigma^{-1}\Re\big(\partial^{k_0}\curlVR^{i}(\mu^{\sfrak}N[\varphi_{-\sfrak}]) \overline{X\partial^{k_0}\curlVR^{i}(\mu^{\sfrak}\varphi_{-\sfrak})}\big)$ with $X=\Lxi, \Leta$ in the trapping region.  By an integration by parts in $X$, we then bound these integrals by the last two terms in \eqref{eq:BEAM:-s:inhomo}, thereby proving  the estimate \eqref{eq:BEAM:-s:inhomo}.
 \end{remark}

 \emph{We shall emphasis that this assumption on a subextreme Kerr background with a fixed parameter $\varsigma\in (0,\half)$ and a suitably large regularity parameter $\reg$ is assumed throughout the rest of this paper.}

In the case that we are considering the TME of the spin $\pm \sfrak$ components with vanishing inhomogeneous term, we immediately arrive at:

\begin{lemma}[BEAM estimates for the spin $\pm \sfrak$ components on a subextreme Kerr] \label{ass:BEAM}In the DOC of a subextreme Kerr spacetime, given any  $0<\delta<1/2$ and $\sfrak+1\leq \reg\in \mathbb{N}^+$, there exist universal constants $\regl>0$ and $C=C(M,\reg)$ such that the following
BEAM estimates are valid in the region $\Donetwo$ for any $\tb_0\leq \tb_1<\tb_2$:
\begin{subequations}\label{eq:BEAM}
\begin{align}
\label{eq:BEAM:-s}
E^{\reg}_{\Sigmatwo}(\Psiminuss)
+M^{\reg}_{\Donetwo}(\Psiminuss)
\leq {}&CE^{\reg}_{\Sigmaone}(\Psiminuss),\\
\label{eq:BEAM:+s}
E^{\reg}_{\Sigmatwo}(\Psipluss)
+M^{\reg}_{\Donetwo}(\Psipluss)\leq {}&CE^{\reg}_{\Sigmaone}(\Psipluss).
\end{align}
\end{subequations}
The above also hold if replacing $\Psi_s$ by $\Lxi^j\Psi_s$ $(j\in \mathbb{N})$ everywhere since $\Lxi^j$ commutes with the TME.
\end{lemma}

However, for each $\ell$ mode of the spin $\pm \sfrak$ components, because of the coupling with the other modes, each $\ell$ mode of the spin $\pm \sfrak$ components satisfies an inhomogeneous TME, and this leads to a different BEAM estimate for a fixed mode.

\begin{lemma}[BEAM estimates for a fixed mode of the spin $\pm \sfrak$ components on a subextreme Kerr]\label{ass:BEAM:mode}
Let  $\ell\geq \sfrak$. In the DOC of a subextreme Kerr spacetime, given any $0<\delta<1/2$ and $\sfrak+1\leq \reg\in \mathbb{N}^+$, there exist universal constants $\regl>0$ and $C=C(M,\delta,\reg)$ such that the following
BEAM estimates are valid in the region $\Donetwo$ for any $\tb_0\leq \tb_1<\tb_2$:
\begin{subequations}\label{eq:BEAM:mode}
\begin{align}
\label{eq:BEAM:-s:mode}
\hspace{4ex}&\hspace{-4ex}
E^{\reg}_{\Sigmatwo}(\ellmode{\Psiminuss}{\ell})
+M^{\reg}_{\Donetwo}(\ellmode{\Psiminuss}{\ell})\notag\\
\leq {}&C\Big(E^{\reg}_{\Sigmaone}(\ellmode{\Psiminuss}{\ell})
+\sum_{\tb'=\tb_1,\tb_2}\sum_{i=0}^{\sfrak}\norm{\Lxi\curlV^i{\Psiminuss}}^2_{W_{-2}^{\reg+\regl}(\Sigma_{\tb'})}
+\sum_{i=0}^{\sfrak}\norm{\Lxi\curlV^i{\Psiminuss}}^2_{W_{-3+\delta}^{\reg+\regl}(\Donetwo)}\Big),\\
\label{eq:BEAM:+s:mode}
\hspace{4ex}&\hspace{-4ex}
E^{\reg}_{\Sigmatwo}(\ellmode{\Psipluss}{\ell})
+M^{\reg}_{\Donetwo}(\ellmode{\Psipluss}{\ell})\notag\\
\leq {}&C\Big(E^{\reg}_{\Sigmaone}(\ellmode{\Psipluss}{\ell})
+\sum_{\tb'=\tb_1,\tb_2}\norm{\Lxi{\Psipluss}}^2_{W_{-2}^{\reg+\regl}(\Sigma_{\tb'})}
+\norm{ \Lxi{\Psipluss}}^2_{W_{-3+\delta}^{\reg+\regl}(\Donetwo)}\Big).
\end{align}
\end{subequations}
The above also hold if replacing $\ellmode{\Psi_s}{\ell}$ by $\Lxi^j\ellmode{\Psi_s}{\ell}$  everywhere for any $j\in \mathbb{N}$.
Meanwhile, the above estimates hold also for $\geq \ell$ modes, i.e. they are valid if we replace $\Lxi^j\ellmode{\Psi_s}{\ell}$  by $\Lxi^j\ellmode{\Psi_s}{\geq \ell}$, respectively.
\end{lemma}

\begin{proof}
By projecting the TME onto an $\ell$ mode and in view of the expression \eqref{eq:squareShat}
of $\Boxhat_{s}$, we achieve
\begin{align}
\label{eq:TME:radfield:mode}
&\Boxhat_s\ellmode{\Psi_s}{\ell}+2s((r-M)Y-2r\Lxi)\ellmode{\Psi_s}{\ell}
\notag\\
&+\frac{2ar}{\R}\Leta\ellmode{\Psi_s}{\ell}
+\bigg(\frac{2s r(r-M)}{\R}-\frac{2Mr^3+a^2r^2-4a^2Mr+a^4}{(\R)^2}\bigg)\ellmode{\Psi_s}{\ell}\notag\\
&={}N[\ellmode{\Psi_s}{\ell}]=\Lxi\Comm{\ell}{s}{\Psi_s}.
\end{align}
The assumed BEAM estimates for an inhomogeneous TME then apply and yield
\begin{subequations}\label{eq:BEA:mode}
\begin{align}
\label{eq:BEA:-s:mode}
\hspace{4ex}&\hspace{-4ex}
E^{\reg}_{\Sigmatwo}(\ellmode{\Psiminuss}{\ell})
+M^{\reg}_{\Donetwo}(\ellmode{\Psiminuss}{\ell})\notag\\
\leq {}&C\Big(E^{\reg}_{\Sigmaone}(\ellmode{\Psiminuss}{\ell})
+\sum_{\tb'=\tb_1,\tb_2}\sum_{i=0}^{\sfrak}\norm{\Lxi\curlV^i\Comm{\ell}{-\sfrak}{\Psiminuss}}^2_{W_{-2}^{\reg+\regl}(\Sigma_{\tb'})}
+\sum_{i_0=0,1}\sum_{i=0}^{\sfrak}\norm{ \Lxi^{i_0}\Lxi\curlV^i\Comm{\ell}{-\sfrak}{\Psiminuss}}^2_{W_{-3+\delta}^{\reg+\regl}(\Donetwo)}\Big),\\
\label{eq:BEA:+s:mode}
\hspace{4ex}&\hspace{-4ex}
E^{\reg}_{\Sigmatwo}(\ellmode{\Psipluss}{\ell})
+M^{\reg}_{\Donetwo}(\ellmode{\Psipluss}{\ell})\notag\\
\leq {}&C\Big(E^{\reg}_{\Sigmaone}(\ellmode{\Psipluss}{\ell})
+\sum_{\tb'=\tb_1,\tb_2}\norm{\Lxi\Comm{\ell}{+\sfrak}{\Psipluss}}^2_{W_{-2}^{\reg+\regl}(\Sigma_{\tb'})}
+\sum_{i_0=0,1}\norm{\Lxi^{i_0} \Lxi\Comm{\ell}{+\sfrak}{\Psipluss}}^2_{W_{-3+\delta}^{\reg+\regl}(\Donetwo)}\Big).
\end{align}
\end{subequations}
In view of Definition \ref{def:Commells} and Proposition \ref{prop:modeprojection:1}, the
desired estimates \eqref{eq:BEAM:mode} then follow. The same argument applies to $\geq \ell$ modes.
\end{proof}

\subsection{General $r^p$ lemmas}
\label{subsect:rp}

We present $r^p$ estimates for an inhomogeneous spin-weighted wave equation (which are taken from  \cite{andersson2019stability}) as well as an $r^p$ estimate for an inhomogeneous transport equation.

To start with, we define a class of inhomogeneous spin-weighted wave equations and inhomogeneous transport equations to which the $r^p$ estimates in Lemma \ref{prop:wave:rp} can be applied.
\begin{definition}
Let $\varphi$ and $\vartheta$ be spin $s$ scalars.\footnote{For simplicity, we have dropped the subscript $s$ and write $\varphi_s$ and $\vartheta_s$ as $\varphi$ and $\vartheta$ respectively.}
\begin{enumerate}
\item
We shall write the governing equation of $\varphi$ as
\begin{align}
\label{eq:wave:r^p:short}
\Boxhat_{s,G}\varphi=\vartheta
\end{align}
 if $\varphi$ is supported on $\geq \ell_0$ modes and satisfies an inhomogeneous spin-weighted wave equation
 \begin{align}
\label{eq:wave:rp}
\Boxhat_{s}\varphi -b_V V\varphi -b_{\phi}\Leta\varphi -b_0\varphi=\vartheta.
\end{align}
with  $b_{V}$, $b_{\phi}$ and $b_0$ being smooth real functions of $r$ and $\sin\theta$ such that
\begin{itemize}
\item $\exists b_{V,-1}\geq 0$ such that $b_V=b_{V,-1} r +O(1)$,
\item $b_{\phi}=O(r^{-1})$, and
\item $\exists b_{0,0}\in \mathbb{R}$ such that $b_0=b_{0,0}+O(r^{-1})$ and $b_{0,0}+(\ell_0+s)(\ell_0-s+1)\geq 0$.
\end{itemize}
 \item
We shall write the governing equation of $\varphi$ as
\begin{align}
\label{eq:transport:rp:short}
\mu Y_{G}\varphi=\vartheta
\end{align}
 if $\varphi$ satisfies an inhomogeneous transport equation
\begin{equation}
\label{eq:r^p:transport:plarge}
\mu Y \varphi +(b_0+2r^{-1})\varphi={}\vartheta
\end{equation}
where $b_0=b_{0,0}r^{-1}+b_{0,\text{rem}}$ with $b_{0,0}\in \mathbb{R}^+\cup\{0\}$ and $b_{0,\text{rem}}$ being an $O(r^{-2})$ function independent of $\theta, \pb$.
 \end{enumerate}
\end{definition}

\begin{lemma}[$r^p$ lemma]
\label{prop:wave:rp}
Let $\reg\in \mathbb{N}$, $\sfrak=\abs{s}\leq 2$,\footnote{The statements in this lemma  actually apply to general $s$ with $s\in \half \mathbb{Z}$.} and $\ell_0\geq \sfrak$.
\begin{enumerate}

\item\label{pt:2:prop:wave:rp}[$r^p$ estimate for an inhomogeneous spin-weighted wave equation].
 Let $\varphi$ (supported on $\geq \ell_0$ modes) and $\vartheta$ be spin $s$ scalars satisfying the inhomogeneous spin-weighted wave equation \eqref{eq:wave:r^p:short}.
Then there are constants $\hat{R}_0=\hat{R}_0(\ell_0, p,\reg, b_0,b_{\phi},b_V)$ and $C=C(\ell_0, p,\hat{R}_0, \reg, b_0,b_{\phi},b_V)$ such that for all $R_0\geq \hat{R}_0$ and $\tb_2>\tb_1\geq \tb_0$,
 for $p\in (0, 2)$,
    \begin{align}\label{eq:rp:less2}
\hspace{2ex}&\hspace{-2ex}
\norm{rV\varphi}^2_{W_{p-2}^\reg(\Sigmatwo^{\geq R_0})}
+\norm{\varphi}^2_{W_{-2}^{\reg+1}(\Sigmatwo^{\geq R_0})}
+\norm{\varphi}^2_{W_{p-3}^{\reg+1}(\Donetwo^{\geq R_0})}
+\norm{Y\varphi}^2_{W_{-1-\varsigma}^{\reg}(\Donetwo^{\geq R_0})}
\notag\\
&\lesssim_{[R_0-M,R_0]} {}C\Big(\norm{rV\varphi}^2_{W_{p-2}^\reg(\Sigmaone^{\geq R_0})}
+\norm{\varphi}^2_{W_{-2}^{\reg+1}(\Sigmaone^{\geq R_0})}
+\norm{\vartheta}^2_{W_{p-3}^{\reg}(\Donetwo^{\geq R_0-M})}\Big);
\end{align}

\item\label{pt:2:prop:transport:rp:mode}[$r^p$ estimate for an inhomogeneous transport equation].
Assume $\varphi$ and $\vartheta$ be spin $s$ scalars satisfying the inhomogeneous transport equation \eqref{eq:transport:rp:short}.
Then for any $
\delta\in (0,\half)$  and any $\veps\in (0,1/2)$, there are constants $\hat{R}_0=\hat{R}_0(\ell_0, p,\reg, b_0)$, $C_0= C_0(\ell_0, p,\hat{R}_0,\reg, b_0)$ and $C_1=C_1(\ell_0, p,\veps,\reg, \hat{R}_0,b_0)$ such that for all $R_0\geq \hat{R}_0$, $\tb_2>\tb_1\geq \tb_0$ and $p\geq\delta>0$, both of the following estimates hold:
\begin{subequations}
\label{eq:rp:pless5:2:mode}
\begin{align}\label{eq:rp:pleq4:2:mode}
\hspace{2ex}&\hspace{-2ex}
\norm{\varphi}^2_{W_{p-4}^\reg(\Sigmatwo^{\geq R_0})}
+\norm{\varphi}^2_{W_{p-5}^{\reg}(\Donetwo^{\geq R_0})}
\notag\\
&\lesssim_{[R_0-M,R_0]} {}C_0\Big(\norm{\varphi}^2_{W_{p-4}^\reg(\Sigmaone^{\geq R_0})}
+\norm{\vartheta}^2_{W_{p-3}^{\reg}(\Donetwo^{\geq R_0-M})}\Big);\\
\label{eq:rp:pgeq4:2:mode}
\hspace{2ex}&\hspace{-2ex}
\norm{\varphi}^2_{W_{p-4}^\reg(\Sigma_{\tb_2}^{\geq R_0})}
+\norm{\varphi}^2_{W_{p-5}^{\reg}(\DOC_{\tb_1,\tb_2}^{\geq R_0})}
\notag\\
&\lesssim_{[R_0-M,R_0]} {}C_1\Big(\norm{\varphi}^2_{W_{p-4}^\reg(\Sigmaone^{\geq R_0})}
+\int_{\tb_1}^{\tb_2}\tb^{1+\veps}\norm{\vartheta}^2_{W_{p-4}^{\reg}(\Sigmatb^{\geq R_0-M})}
\di \tb\Big).
\end{align}
\end{subequations}
\end{enumerate}

In all the above estimates, we have implicitly included in the symbol $\lesssim_{[R_0-M, R_0]}$ the  integral terms $\norm{\varphi}^2_{W_{0}^{\reg+1}(\Sigmatwo^{R_0-M,R_0})}
+\norm{\varphi}^2_{W_{0}^{\reg+1}(\Sigmaone^{R_0-M,R_0})} +\norm{\varphi}^2_{W_{0}^{\reg+1}(\Donetwo^{R_0-M,R_0})}$ supported on $[R_0-M, R_0]$.
\end{lemma}

\begin{proof}

Point (\ref{pt:2:prop:wave:rp}) for $p\in(0,2)$ has been proven in \cite[Lemmas 5.5 and 5.6]{andersson2019stability}. Notice that there is a sign difference between the operator $\Boxhat_{s}$ in this work with the one in \cite{andersson2019stability}, and this also causes some sign changes in equation \eqref{eq:wave:rp}.

It remains to prove point (\ref{pt:2:prop:transport:rp:mode}). Let $\chi_x(r)$ be a cutoff function such that it equals $1$ for $r\geq x$ and vanishes for $r\leq x-M$. By multiplying equation \eqref{eq:r^p:transport:plarge} by $2\chi_{R_0}r^{p-4}\bar{\varphi}$, taking the real part and integrating by parts, one arrives at
\begin{align}
\label{eq:rp:transport:k=0}
Y (\chi_{R_0}\mu r^{p-4}\abs{\varphi}^2)
+\big(\partial_r (\chi_{R_0}\mu r^{p-4})
+4\chi_{R_0}r^{p-5} +2\chi_{R_0}b_0 r^{p-4}\big)\abs{\varphi}^2={}\Re(2\chi_{R_0}r^{p-4}\vartheta\bar{\varphi}).
\end{align}
The coefficient of $\abs{\varphi}^2$ term is equal to $(p+2b_{0,0})\chi_{R_0}r^{p-5}+r^{p-6}((p-4)\chi_{R_0}r(\mu-1)+r^2\partial_r (\chi_{R_0}\mu) +2 \chi_{R_0}r^2 (b_0-b_{0,0}r^{-1}))$, and by assumption, it is greater than $\frac{p}{2}r^{p-5}$ in region $r\geq \hat{R}_0$ for $\hat{R}_0$ large enough. Thus, by applying a Cauchy--Schwarz to the RHS of \eqref{eq:rp:transport:k=0} and integrating equation \eqref{eq:rp:transport:k=0} in $\Donetwo^{\geq R_0-M}$ with $R_0\geq \hat{R}_0$, we obtain the estimate \eqref{eq:rp:pleq4:2:mode} in the case of $\reg=0$.  On the other hand, we can also utilize the H\"{o}lder's inequality  to bound the RHS of \eqref{eq:rp:transport:k=0} by $\veps_0 p\chi_{R_0}r^{p-4}\abs{\varphi}^2+ \chi_{R_]}\frac{1}{\veps_0 p}r^{p-4}\abs{\vartheta}^2$, then integrating over $\Donetwo^{\geq R_0-M}$ with $R_0\geq \hat{R}_0$ yield
\begin{align*}
\hspace{2ex}&\hspace{-2ex}
\norm{\varphi}^2_{W_{p-4}^0(\Sigma_{\tb'}^{\geq R_0})}
+\norm{\varphi}^2_{W_{p-5}^{0}(\DOC_{\tb_1,\tb'}^{\geq R_0})}
\notag\\
&\lesssim_{[R_0-M,R_0]} {}\norm{\varphi}^2_{W_{p-4}^0(\Sigmaone^{\geq R_0})}
+\veps_0\int_{\tb_1}^{\tb'}\frac{1}{\tb^{1+\veps}}\norm{\varphi}^2_{W_{p-4}^0(\Sigmatb^{\geq R_0})}\di \tb
+\frac{1}{\veps_0}\int_{\tb_1}^{\tb'}\tb^{1+\veps}\norm{\vartheta}^2_{W_{p-4}^{0}(\Sigmatb^{\geq R_0-M})}
\di \tb.
\end{align*}
By taking a supreme over $\tb'\in [\tb_1, \tb_2]$, the second term in the last line can be absorbed by the LHS, and we thus obtain the estimate \eqref{eq:rp:pgeq4:2:mode} in the case of $\reg=0$.

We next commute the transport equation \eqref{eq:r^p:transport:plarge} with $r\tildeV$; in view of the commutator \eqref{comm:muYandtildeVandvv}, this gives
\begin{align}\label{eq:rtildeVvarphi:transport:03}
\mu Y (r\tildeV \varphi) + ((b_0+\mu r^{-1})+2r^{-1})r\tildeV\varphi=r\tildeV \vartheta-\mu r \partial_r (b_0+2r^{-1})\varphi\doteq\vartheta_{r\tildeV}.
\end{align}
This equation can again be put into the form of the transport equation \eqref{eq:r^p:transport:plarge} and the assumptions are all satisfied. Thus, the estimate \eqref{eq:rp:pleq4:2:mode} with $\reg=0$ holds by replacing $\varphi$ and $\vartheta$ by $r\tildeV\varphi$ and $\vartheta_{r\tilde{V}}$ respectively. Note that $\norm{\vartheta_{r\tilde{V}}}^2_{W_{p-3}^{0}(\Donetwo^{\geq R_0-M})}\lesssim \norm{r\tildeV\vartheta}^2_{W_{p-3}^{0}(\Donetwo^{\geq R_0-M})} + \norm{\varphi}^2_{W_{p-5}^{0}(\Donetwo^{\geq R_0-M})}$, and the term $\norm{\varphi}^2_{W_{p-5}^{0}(\Donetwo^{\geq R_0})}$ is already bounded in the previous discussions. One can thus inductively show that for any $\reg\geq 0$,
\begin{align}
\label{eq:rp:transport:klarge:93}
\hspace{2ex}&\hspace{-2ex}
\sum_{i_1\leq \reg}\Big(\norm{(r\tildeV)^{i_1}\varphi}^2_{W_{p-4}^0(\Sigmatwo^{\geq R_0})}
+\norm{(r\tildeV)^{i_1}\varphi}^2_{W_{p-5}^{0}(\Donetwo^{\geq R_0})}\Big)
\notag\\
&\lesssim_{[R_0-M,R_0]} {}\sum_{{i_1}\leq \reg}\Big(\norm{(r\tildeV)^{i_1}\varphi}^2_{W_{p-4}^0(\Sigmaone^{\geq R_0})}
+\norm{(r\tildeV)^{i_1}\vartheta}^2_{W_{p-3}^{0}(\Donetwo^{\geq R_0-M})}\Big).
\end{align}
Since $\Lxi$,  $\Leta$ $\edthR$ and $\edthR'$ commute with the transport  equation \eqref{eq:r^p:transport:plarge}, the above estimate \eqref{eq:rp:transport:klarge:93} manifestly hold with $\varphi$ and $\vartheta$ replaced by $\Lxi^{i_2}\Leta^{i_3}\edthR^{i_4}(\edthR')^{i_5}\varphi$ and $\Lxi^{i_2}\Leta^{i_3}\edthR^{i_4}(\edthR')^{i_5}\vartheta$ respectively for any $i_2,i_3, i_4,i_5\in \mathbb{N}$. In the end, in view of the fact that the operators in the set $\CDeri$ can be expanded in terms of $\{r\tildeV, \Lxi, \Leta,\edthR,\edthR'\}$ with $O(1)$ coefficients in the region $r\geq R_0-M$, the estimate \eqref{eq:rp:pleq4:2:mode} is therefore valid. The other estimate \eqref{eq:rp:pgeq4:2:mode} for general $\reg\geq 0$ can be proven in a completely analogous manner.
\end{proof}

\subsection{Energy decay estimates for the entire spin $\pm \sfrak$ components}
\label{subsect:EDE}

Recall that we have made the BEAM estimate assumption \ref{ass:BEAM:inhomogeneous}, hence the BEAM estimates in Lemma \ref{ass:BEAM}  for  the spin $\pm \sfrak$ component are valid.

We first define a few $r$-weighted energies for the spin $\pm \sfrak$ components.

\begin{definition}
\label{def:Ffts:Phiplusminuss:-1to2}
For any $j\in\mathbb{N}$, define
\begin{align}
\PsiminussHigh{j}\doteq ((\R)V)^j \Psiminuss.
\end{align}
Define for the spin $+\sfrak$ component the energies $F^{(i)}(\reg,p,\tb,\Psipluss)$ as follows\footnote{These energies, as well as the other energies $F^{(i)}(\reg,p,\tb,\varphi_s)$ defined for a spin $s$ scalar in this section, actually correspond to the energy $F(\reg, p,\tb)$ in Lemma \ref{lem:hierarchyImpliesDecay:73} and satisfy the assumptions \eqref{assump:HierarchyToDecay(1)} and \eqref{assump:HierarchyToDecay(2)} of Lemma \ref{lem:hierarchyImpliesDecay:73}.}
\begin{subequations}
\label{eq:def:Ffts:Phiplus:-1to2}
\begin{align}
\label{def:Ffts:Phiplus:1:p-10}
F^{(0)}(\reg,p,\tb,\Psipluss)={}&0, \qquad\text{for } p\in [-1,\delta),\\
\label{def:Ffts:Phiplus:1:p02}
F^{(0)}(\reg,p,\tb,\Psipluss)={}& \norm{rV\Psipluss}^2_{W_{p-2}^{\reg-\sfrak-2}(\Sigmatb)}
+\norm{\Psipluss}^2_{W_{-2}^{\reg-\sfrak-1}(\Sigmatb)}
+E^{\reg}_{\Sigmatb}(\Psipluss),
\,\, \text{for } p\in [\delta,2).
\end{align}
\end{subequations}
Let $l(j,\sfrak)=\max\{0,j-\sfrak\}$, and  for any $i\in [\sfrak,2\sfrak]$, define for the spin $-\sfrak$ component the energies \begin{subequations}
\begin{align}
F^{(i)}(\reg,p,\tb,\Psiminuss)={}&0, \quad \text{for } p\in [-1,\delta),\\
F^{(i)}(\reg,p,\tb,\Psiminuss)={}&\sum_{j=0}^{i}
\Big(\norm{rV\PsiminussHigh{j}}^2_{W_{p-2}^{\reg-\sfrak-1-l(j,\sfrak)}(\Sigmatb)}
+\norm{\PsiminussHigh{j}}^2_{W_{-2}^{\reg-\sfrak-l(j,\sfrak)}(\Sigmatb)}\Big), \,\, \text{for } p\in [\delta,2).
\end{align}
\end{subequations}
Additionally, for any $j\in \mathbb{N}$, we define  $F^{(i)}(\reg,p,\tb,\Lxi^j\Psiminuss)$   by simply replacing  $\Psiminuss$ and $\PsiminussHigh{j}$ in $F^{(i)}(\reg,p,\tb,\Psiminuss)$ by $\Lxi^j\Psiminuss$ and $\Lxi^j \PsiminussHigh{j}$. Similarly, we define $F^{(0)}(\reg,p,\tb,\Lxi^j\mellmode{\Psiminuss}{m}{\ell})$ for an $(m,\ell)$ mode of the spin $-\sfrak$ component, $F^{(0)}(\reg,p,\tb,\Lxi^j\ellmode{\Psiminuss}{\ell})$ for an $\ell$ mode, $F^{(0)}(\reg,p,\tb,\Lxi^j\ellmode{\Psiminuss}{\geq \ell})$ for $\geq \ell$ modes,  and the analogues for the spin $+\sfrak$ component and its modes.
\end{definition}

In order to employ the statement in point (\ref{pt:2:prop:wave:rp}) of  Lemma \ref{prop:wave:rp} to derive the $r^p$ estimates for the spin $\pm\sfrak$ components, it is manifest that equation \eqref{eq:wave:PhisHighi:general} can be put into the form of \eqref{eq:wave:r^p:short} as long as $i\leq \sfrak-s$; therefore, we conclude:

\begin{lemma}
For the spin $+\sfrak$ component, we have
\begin{align}
\label{eq:wave:shortform:Phipluss0:847}
\Boxhat_{+\sfrak,G}{\PhiplussHigh{0}}={}&0.
\end{align}
For the spin $-\sfrak$ component, we have for $0\leq i\leq \sfrak-1$,
\begin{subequations}
\label{eq:wave:shortform:Phiminussi:36}
\begin{align}
\Boxhat_{-\sfrak,G}{\Phiminuss{i}}={}&\vartheta(\Phiminuss{i}) =O(r^{-1}){\Phiminuss{i+1}}+\sum_{0\leq i'<i}\sum_{n=0,1}O(1)\Leta^n{\Phiminuss{i'}}
\end{align}
and for $\sfrak\leq i\leq 2\sfrak$,
\begin{align}
\Boxhat_{-\sfrak,G}{\Phiminuss{i}}={}&\vartheta(\Phiminuss{i}) = \sum_{0\leq i'<i}\sum_{n=0,1}O(1)\Leta^n{\Phiminuss{i'}}.
\end{align}
\end{subequations}
\end{lemma}

We shall now obtain global $r^p$ estimates for the spin $\pm \sfrak$ components.
\begin{prop}
\label{prop:rpplusglobal}
Let $\reg$ be suitably large.
Then for any $\tb_2>\tb_1\geq \tb_0$ and $p\in [\delta,2-\delta]$,
\begin{align}\label{eq:rpplussglobal:less2}
F^{(0)}(\reg,p,\tb_2,\Psipluss)
+\norm{\Psipluss}^2_{W_{p-3}^{\reg-\sfrak-1}(\Donetwo)}
\lesssim_{\reg,p} {}&F^{(0)}(\reg,p,\tb_1,\Psipluss),
\end{align}
and for any $\tb_2>\tb_1\geq \tb_0$, $i\in [\sfrak,2\sfrak]$ and $p\in [\delta,2-\delta]$,
\begin{align}
\label{eq:rpminussglobal:less2}
F^{(i)}(\reg,p,\tb_2,\Psiminuss)
+\sum_{j=0}^i\norm{\PsiminussHigh{j}}^2_{W_{p-3}^{\reg-\sfrak-1-l(j,s)}(\Donetwo)}
\lesssim_{\reg,p,i}  F^{(i)}(\reg,p,\tb_1,\Psiminuss).
\end{align}
\end{prop}

\begin{proof}
The $\sfrak=0$, $\sfrak=1$ and $\sfrak=2$ cases have been addressed in \cite{dafermos2009new,Ma20almost,andersson2019stability} respectively. We outline the basic idea here.

Consider the spin $+\sfrak$ component. We apply the $r^p$ estimate \eqref{eq:rp:less2} with $\varphi=\PhiplussHigh{0}$ and $\vartheta=0$ to equation \eqref{eq:wave:shortform:Phipluss0:847}, and by adding this estimate to the BEAM estimate \eqref{eq:BEAM:+s} for the spin $+\sfrak$ component, we prove the global $r^p$ estimate \eqref{eq:rpplussglobal:less2} for $p\in [\delta,2-\delta]$.

We next turn to the spin $-\sfrak$ component, and to illustrate the approach in proving the desired estimates, it suffices to consider the most complicated case $\sfrak=2$ and the other cases $\sfrak=0,1$ can be dealt with in a same (but simpler) manner.

First, consider the wave system consisting of the first three subequations of \eqref{eq:basicwavesys:-2}. Each of these subequations can be put into the form of \eqref{eq:wave:r^p:short} with the corresponding inhomogeneous terms
\begin{subequations}
\label{expr:varthetaPhiminussi}
\begin{align}
\vartheta(\Phiminustwo{0})={}&-\frac{4(r^3-3Mr^2 +a^2 r+a^2 M)}{(\R)^2} \Phiminustwo{1}=O(r^{-1})\Phiminustwo{1},\\
\vartheta(\Phiminustwo{1})={}&-\frac{2(r^3-3Mr^2 +a^2 r+a^2 M)}{(\R)^2} \Phiminustwo{2}
-\frac{6a(r^2 -a^2)}{\R}\Leta\Phiminustwo{0}
-\frac{6Mr^4 - 6a^2 r^3 - 18 a^2 Mr^2 -6a^4 r}{(\R)^2}\Phiminustwo{0}\notag\\
={}&O(r^{-1})\Phiminustwo{2}
+O(1)\Phiminustwo{0}
+O(1)\Leta\Phiminustwo{0},\\
\vartheta(\Phiminustwo{2})={}&\frac{20a^2(\PR)}{(\R)^2}\Phiminustwo{1}-\frac{8a(r^2 -a^2)}{\R}\Leta\Phiminustwo{1}
\notag\\
&-\frac{24a^3r}{\R}\Leta \Phiminustwo{0}-\frac{6a^2 (r^4+10Mr^3-6a^2 Mr-a^4)}{(\R)^2}\Phiminustwo{0}
\notag\\
={}&O(1)\Leta\Phiminustwo{1}+
O(r^{-1})\Phiminustwo{1}+O(r^{-1})\Leta\Phiminustwo{0}+O(1)\Phiminustwo{0}.
\end{align}
Thus, for each $\varphi=\Phiminustwo{i}$ and $\vartheta=\vartheta(\Phiminustwo{i})$, $i\in [0,1,2]$, we can apply point (\ref{pt:2:prop:wave:rp}) of Lemma \ref{prop:wave:rp} to achieve its corresponding $r^p$ estimate \eqref{eq:rp:less2}. It remains to estimate the last term $\norm{\vartheta(\Phiminustwo{i}}^2_{W_{p-3}^{\reg}(\Donetwo^{\geq R_0-M})}$ on the RHS of \eqref{eq:rp:less2}, which is naturally bounded by $\norm{\vartheta(\Phiminustwo{i}}^2_{W_{p-3}^{\reg}(\Donetwo^{\geq R_0})}+\sum\limits_{i=0,1,2}\norm{\vartheta(\Phiminustwo{i}}^2_{W_{p-3}^{\reg}(\Donetwo^{[R_0-M, R_0]})}$. By adding $A_i$ multiple of the estimate \eqref{eq:rp:less2} for $\varphi=\Phiminustwo{i}$ and summing over $i=0,1,2$, and by further taking $A_0\gg A_1\gg A_2$, one finds that $\sum_{i=0,1,2}A_i\norm{\vartheta(\Phiminustwo{i}}^2_{W_{p-3}^{\reg}(\Donetwo^{\geq R_0})}$ is absorbed. In the end, adding in the BEAM estimate \eqref{eq:BEAM:-s} yields the desired estimate \eqref{eq:rpminussglobal:less2} in the case of $i=2$.

For the case $i=3$, we again put the equation \eqref{eq:waveeq:Phiminustwo3} of $\Phiminustwo{3}$ into the form of \eqref{eq:wave:r^p:short}  with $\varphi=\Phiminustwo{3}$ and
\begin{align}
\vartheta(\Phiminustwo{3})=\text{the last three lines of equation } \eqref{eq:waveeq:Phiminustwo3}
\end{align}
and apply the $r^p$ estimate \eqref{eq:rp:less2}. The term $\norm{\vartheta(\Phiminustwo{3}}^2_{W_{p-3}^{\reg}(\Donetwo^{\geq R_0-M})}$ on the RHS of the estimate \eqref{eq:rp:less2} for $\varphi=\Phiminustwo{3}$ is bounded by the spacetime integral on the LHS of the estimate \eqref{eq:rpminussglobal:less2} for $i=2$, hence this proves the estimate \eqref{eq:rpminussglobal:less2} for $i=3$. For $i=4$, equation \eqref{eq:waveeq:Phiminustwo4} of $\Phiminustwo{4}$ is put into the form of \eqref{eq:wave:r^p:short}  with $\varphi=\Phiminustwo{4}$ and
\begin{align}
\vartheta(\Phiminustwo{4})=\text{the last four lines of equation } \eqref{eq:waveeq:Phiminustwo4},
\end{align}
\end{subequations}
and the remaining steps in  the case $i=4$ are the same as the ones in  the case $i=3$.
\end{proof}

\begin{lemma}
In the region $r\geq 4M$, we have for a spin $s$ scalar $\varphi$ that
\begin{align}
\label{eq:r2LxiV:form:3}
r^2\Lxi V\varphi={}& O(1)\Boxhat_{s,G}\varphi
+\sum_{\abs{\mathbf{a}}\leq 2} O(1)\CDeri^{\mathbf{a}} \varphi.
\end{align}
\end{lemma}

\begin{proof}
In the expression \eqref{eq:squareShat} of the wave operator $\Boxhat_{s}$, we use $Y=\frac{\R}{\Delta}\big(2\Lxi+\frac{2a}{\R}\Leta-V\big)$ away from horizon to rewrite $r^2V\Lxi \varphi$ as the desired form.
\end{proof}

\begin{prop}
\label{prop:BEDC:Phiplusminuss:less2:1}
Let $j\in \mathbb{N}$, $i\in [\sfrak,2\sfrak]$, and let $\reg$ be suitably large. There are constants $\regl(j)$ and $C_j$ such that for any $p\in [\delta,2-\delta]$ and any $\tb\geq\tb_0$,
\begin{subequations}
\begin{align}
\label{eq:BED:Psipluss:less2}
\hspace{7ex}&\hspace{-7ex}
F^{(0)}(\reg,p,\tb_2,\Lxi^j\Psipluss)
+\norm{\Lxi^j\Psipluss}^2_{W_{p-3}^{\reg}(\Dtwoinfty)}\notag\\
\lesssim_{p,j,\reg} {}&\ \langle\tb_2-\tb_1\rangle^{-2-2j+C_j\delta+p}F^{(0)}(\reg+\regl(j),2-\delta,\tb_1,\Psipluss),\\
\label{eq:BED:Psiminuss:less2}
\hspace{7ex}&\hspace{-7ex}
F^{(i)}(\reg,p,\tb_2,\Lxi^j\Psiminuss)
+\sum_{i'=0}^i\norm{\Lxi^j\PsiminussHigh{i'}}^2_{W_{p-3}^{\reg}(\Dtwoinfty)}
\notag\\
\lesssim_{p,j,\reg} {}&\ \langle\tb_2-\tb_1\rangle^{-2-2(j+2\sfrak-i)+C_{j}\delta+p}F^{(2\sfrak)}(\reg+\regl(j),2-\delta,\tb_1,\Psiminuss).
\end{align}
\end{subequations}
\end{prop}

\begin{proof}
Note that for any $\reg$ suitably large, we have for any $p\in [\delta,2-\delta]$ that $\norm{\Psipluss}^2_{W_{p-3}^{\reg-\sfrak-2}(\Donetwo)}\gtrsim_{\reg,p} \int_{\tb_1}^{\tb_2}F^{(0)}(\reg-\regl,p-1,\tb,\Psipluss)\di\tb
$ for a finite $\regl$ by a simple application of Hardy's inequality,  thus the estimate \eqref{eq:rpplussglobal:less2} yields
\begin{align}\label{eq:rpplussglobal:less2:v5}
F^{(0)}(\reg,p,\tb_2,\Psipluss)
+\int_{\tb_1}^{\tb_2}F^{(0)}(\reg-\regl,p-1,\tb,\Psipluss)\di\tb\lesssim_{\reg,p} {}&F^{(0)}(\reg,p,\tb_1,\Psipluss).
\end{align}
An application of Lemma \ref{lem:hierarchyImpliesDecay:73} to this estimate  then implies that for any $p\in [\delta,2-\delta]$,
\begin{align}
\label{eq:rpplussglobal:less2:v6}
F^{(0)}(\reg-\regl,p,\tb_2,\Psipluss)\lesssim {}&\langle \tb_2-\tb_1\rangle^{-2+\delta+p}F^{(0)}(\reg,2-\delta,\tb_1,\Psipluss).
\end{align}
This proves the estimate \eqref{eq:BED:Psipluss:less2} for $j=0$.

To show the general $j\in \mathbb{N}$ case for the estimate \eqref{eq:BED:Psipluss:less2}, we prove it by induction. Assume it holds for $j$, and we prove the $j+1$ case. Recall equation \eqref{eq:wave:shortform:Phipluss0:847} satisfied by $\PhiplussHigh{0}$, and in view of the formula \eqref{eq:r2LxiV:form:3}, we have in the region $r\geq 4M$ that $r^2\Lxi V \PhiplussHigh{0}=\sum_{\abs{\mathbf{a}}\leq 2} O(1)\CDeri^{\mathbf{a}} \PhiplussHigh{0}$. Therefore, for any $\tb>\tb_1\geq \tb_0$,
\begin{align}
\label{eq:8947283}
\hspace{4ex}&\hspace{-4ex}
F^{(0)}(\reg,2-\delta,\tb,\Lxi^{j+1}\Psipluss)\notag\\
={}&
\norm{rV\Lxi^{j+1}\Psipluss}^2_{W_{2-\delta-2}^{\reg-\sfrak-2}(\Sigmatb)}
+\norm{\Lxi^{j+1}\Psipluss}^2_{W_{-2}^{\reg-\sfrak-1}(\Sigmatb)}
+E^{\reg}_{\Sigmatb}(\Lxi^{j+1}\Psiplus)\notag\\
\lesssim_{\reg,\delta}{}&\norm{r^2\Lxi V(\Lxi^{j}\Psipluss)}^2_{W_{-\delta-2}^{\reg-\sfrak-2}(\Sigmatb)}
+\norm{\Lxi^{j+1}\Psipluss}^2_{W_{-2}^{\reg-\sfrak-1}(\Sigmatb)}
+E^{\reg}_{\Sigmatb}(\Lxi^{j+1}\Psipluss)\notag\\
\lesssim_{\reg,\delta}{}&F^{(0)}(\reg+2,\delta,\tb,\Lxi^{j}\Psipluss)\notag\\
\lesssim_{\reg,\delta,j}{}&\langle \tb-\tb_1\rangle^{-2-2j+C_j\delta}F^{(0)}(\reg+\regl(j),2-\delta,\tb_1,\Psipluss)
\end{align}
where in the last step we have used the base assumption.
Further, since $\Lxi$ commutes with the TME, the estimates \eqref{eq:rpplussglobal:less2:v5} and \eqref{eq:rpplussglobal:less2:v6} are valid if replacing $\Psipluss$ by $\Lxi^{j+1}\Psipluss$. This together with the above estimate yields
\begin{align}
\label{eq:rpplussglobal:less2:v8}
F^{(0)}(\reg,p,\tb_2,\Lxi^{j+1}\Psipluss)\lesssim_{\reg,\delta} {}&\langle \tb_2-\tb_1\rangle^{-2+\delta+p}F(\reg+\regl,2-\delta,\tb_1+\frac{\tb_2-\tb_1}{2},\Lxi^{j+1}\Psipluss)\notag\\
\lesssim_{\reg,\delta,j} {}& \langle \tb_2-\tb_1\rangle^{-2-2(j+1)+p+C_j\delta}  F^{(0)}(\reg+\regl(j),2-\delta,\tb_1,\Psipluss),
\end{align}
which thus completes the induction and proves \eqref{eq:BED:Psipluss:less2}.

We proceed by proving the estimate \eqref{eq:BED:Psiminuss:less2} for the spin $-\sfrak$ component. By definition, one has $\sum_{j'=0}^i\norm{\PsiminussHigh{j'}}^2_{W_{p-3}^{\reg-\sfrak-1-l(j,s)}(\Donetwo)}\gtrsim_{\reg,p} \int_{\tb_1}^{\tb_2}F^{(i)}(\reg-1,p-1,\tb,\Psiminuss)$, hence we have from \eqref{eq:rpminussglobal:less2} that
\begin{align}
\label{eq:rpminussglobal:less2:v9}
F^{(i)}(\reg,p,\tb_2,\Psiminuss)
+\int_{\tb_1}^{\tb_2}F^{(i)}(\reg-1,p-1,\tb,\Psiminuss)\di \tb
\lesssim_{\reg, p} F^{(i)}(\reg,p,\tb_1,\Psiminuss).
\end{align}
An application of Lemma \ref{lem:hierarchyImpliesDecay:73} to \eqref{eq:rpminussglobal:less2:v9}
yields that for any $p\in [\delta,2-\delta]$ and $i\in [\sfrak,2\sfrak]$,
\begin{align}\label{eq:rpminussglobal:less2:v11}
F^{(i)}(\reg,p,\tb_2,\Psiminuss)
\lesssim_{\reg,p} \langle \tb_2-\tb_1\rangle^{-2+\delta +p}F^{(i)}(\reg+\regl,2-\delta,\tb_1,\Psiminuss).
\end{align}
By definition, we have $\Psiminuss^{(j'+1)}=(\R)V\Psiminuss^{(j')}$, hence for any $i\in [\sfrak+1,2\sfrak]$,
\begin{align}\label{eq:rpminussglobal:less2:v13}
F^{(i)}(\reg,2-\delta,\tb,\Psiminuss)\gtrsim_{\reg} F^{(i-1)}(\reg,\delta,\tb,\Psiminuss).
\end{align}
The above two estimates together then prove \eqref{eq:BED:Psiminuss:less2} for $j=0$.  The general $j\in\mathbb{N}$ cases are proven in a same manner as the above one in proving the general $j\in \mathbb{N}$ cases for the spin $+\sfrak$ component together with an application of
\begin{align}
\label{eq:749720}
F^{(2\sfrak)}(\reg,2-\delta,\tb,\Lxi^{j+1}\Psiminuss)
\lesssim {}F^{(2\sfrak)}(\reg+\regl,\delta,\tb,\Lxi^{j}\Psiminuss)\end{align}
that is similar to \eqref{eq:8947283} for the spin $+\sfrak$ component.
\end{proof}

\subsection{Energy decay estimates for the modes of the spin $\pm \sfrak$ components}
\label{sect:ED:modes}

Since the BEAM estimates \eqref{eq:BEAM:mode} for the modes will be frequently used, we shall estimate the last two terms of the RHS in each subequation of \eqref{eq:BEAM:mode} and deduce an alternative form of the BEAM estimates for the modes of the spin $\pm \sfrak$ components. This is provided in the following lemma.

\begin{lemma}[Alternative form of BEAM estimates for the modes]\label{lem:esti:error:BEAM:modes:836}
Let $j,\reg\in \mathbb{N}$. For any $\tilde{\ell}\in \{\sfrak,\sfrak+1,\geq \sfrak+2\}$, there exists a constant $\regl>0$ such that
\begin{subequations}\label{eq:BEAM:mode:s}
\begin{align}
\label{eq:BEAM:-s:mode:s}
E^{\reg}_{\Sigmatwo}(\Lxi^j\ellmode{\Psiminuss}{\tilde\ell})
+M^{\reg}_{\Donetwo}(\Lxi^j\ellmode{\Psiminuss}{\tilde\ell})
\lesssim_{\reg,\delta}{}&E^{\reg}_{\Sigmaone}(\Lxi^j\ellmode{\Psiminuss}{\tilde\ell})
+F^{(\sfrak)}(\reg+\regl,\delta,\tb_1,\Lxi^{j+1}\Psiminuss)
,\\
\label{eq:BEAM:+s:mode:s}
E^{\reg}_{\Sigmatwo}(\Lxi^j\ellmode{\Psipluss}{\tilde\ell})
+M^{\reg}_{\Donetwo}(\Lxi^j\ellmode{\Psipluss}{\tilde\ell})
\lesssim_{\reg,\delta} {}&E^{\reg}_{\Sigmaone}(\Lxi^j\ellmode{\Psipluss}{\tilde\ell})
+F^{(0)}(\reg+\regl,\delta,\tb_1,\Lxi^{j+1}\Psipluss).
\end{align}
\end{subequations}
\end{lemma}

\begin{proof}
This follows from BEAM estimates \eqref{eq:BEAM:mode} and the estimates in Proposition \ref{prop:rpplusglobal} with $p=\delta$.
\end{proof}

 We then derive the wave equations of the modes of the scalars $\hatPhisHigh{i}$ and put them into the form of  \eqref{eq:wave:r^p:short} such that the $r^p$ estimates in Lemma \ref{prop:wave:rp} can be applied.

\begin{lemma}
For any $\ell\geq \sfrak$ and $\sfrak-s\leq i\leq \ell-s$,
the scalars $\ellmode{\hatPhisHigh{i}}{\ell}$, the $\ell$ mode of $\hatPhisHigh{i}$ defined in \eqref{ansatz:hatPhisHigh}, and the scalars $\ellmode{\hatPhisHigh{i}}{\geq \ell}$
satisfy the following spin-weighted wave equations
\begin{subequations}
\label{eq:wave:hatPhisHighi:modeandmodegtr:rpwaveform}
\begin{align}
\label{eq:wave:hatPhisHighi:mode:rpwaveform}
\Boxhat_{s,G}\ellmode{\hatPhisHigh{i}}{\ell}
={}&\vartheta(\ellmode{\hatPhisHigh{i}}{\ell})
=\sum_{n\leq d(i)}\sum_{\sfrak-s\leq j\leq i}O(r^{-1})\Leta^n\ellmode{\hatPhisHigh{i}}{\ell}
+\Lxi\Comm{\ell}{s}{\hatPhisHigh{i}}\\
\label{eq:wave:hatPhisHighi:modegtr:rpwaveform}
\Boxhat_{s,G}\ellmode{\hatPhisHigh{i}}{\geq\ell}
={}&\vartheta(\ellmode{\hatPhisHigh{i}}{\geq\ell})
=\sum_{n\leq d(i)}\sum_{\sfrak-s\leq j\leq i}O(r^{-1})\Leta^n\ellmode{\hatPhisHigh{i}}{\geq\ell}
-\sum_{\sfrak\leq \ell'\leq \ell-1}\Lxi\Comm{\ell'}{s}{\hatPhisHigh{i}}
\end{align}
\end{subequations}
with $d(i)$ a constant depending only on $i$.

Further, for $0\leq i\leq \sfrak-1$,
\begin{subequations}
\label{eq:wave:shortform:Phiminussi:mode:loweri}
\begin{align}
\label{eq:wave:shortform:Phiminussi:mode:loweri:1}
\Boxhat_{-\sfrak,G}\ellmode{\Phiminuss{i}}{\ell}
={}& \vartheta(\ellmode{\Phiminuss{i}}{\ell})\notag\\
={}&\Proj{\ell}(\vartheta(\Phiminuss{i}) )
+\Lxi\Comm{\ell}{s}{\Phiminuss{i}},\\
\label{eq:wave:shortform:Phiminussi:modegtr:loweri:1}
\Boxhat_{-\sfrak,G}\ellmode{\Phiminuss{i}}{\geq\ell}
={}& \vartheta(\ellmode{\Phiminuss{i}}{\geq\ell})\notag\\
={}&\Proj{\geq\ell}(\vartheta(\Phiminuss{i}) )
-\sum_{\sfrak\leq \ell'\leq \ell-1}\Lxi\Comm{\ell'}{s}{\Phiminuss{i}}
\end{align}
and for $\sfrak\leq i\leq 2\sfrak$,
\begin{align}
\label{eq:wave:shortform:Phiminussi:mode:loweri:3}
\Boxhat_{-\sfrak,G}\ellmode{\Phiminuss{i}}{\ell}={}& \vartheta(\ellmode{\Phiminuss{i}}{\ell})\notag\\
={}&
\Proj{\ell}(\vartheta(\Phiminuss{i}) )
+\Lxi\Comm{\ell}{-\sfrak}{\Phiminuss{i}},\\
\label{eq:wave:shortform:Phiminussi:modegtr:loweri:3}
\Boxhat_{-\sfrak,G}\ellmode{\Phiminuss{i}}{\geq\ell}={}& \vartheta(\ellmode{\Phiminuss{i}}{\geq\ell})\notag\\
={}&
 \Proj{\geq \ell}(\vartheta(\Phiminuss{i}) 
-\sum_{\sfrak\leq \ell'\leq \ell-1}\Lxi\Comm{\ell'}{-\sfrak}{\Phiminuss{i}}.
\end{align}
\end{subequations}
\end{lemma}

\begin{proof}
We put equation \eqref{eq:wave:hatPhisHighi:an:ellmode} into the form of \eqref{eq:wave:rp}  and find the assumptions in point (\ref{pt:2:prop:wave:rp}) of Lemma \ref{prop:wave:rp} are all satisfied for $\sfrak-s\leq i\leq \ell-s$ in view of \eqref{eq:l=l0mode:eigenvalue}; hence, we arrive at
\begin{align}
\label{eq:wave:hatPhisHighi:an:ellmode:89}
\Boxhat_{s,G}\ellmode{\hatPhisHigh{i}}{\ell}
={}&\ellmode{\hat{H}_{s,i}}{\ell}
+\Lxi\Comm{\ell}{s}{\hatPhisHigh{i}}
\end{align}
with $\ellmode{\hat{H}_{s,i}}{\ell}$ being the $\ell$ mode of $\hat{H}_{s,i}$ defined in \eqref{def:hatHsi:error:47}. By the definition of $\hat{H}_{s,i}$ in equation \eqref{def:hatHsi:error:47} and using also the expression \eqref{ansatz:hatPhisHigh}, one has
\begin{align}\label{def:hatHsi:error:53}
\ellmode{\hat{H}_{s,i}}{\ell}={}&\sum_{n\leq d(i)}\sum_{\sfrak-s\leq j\leq i}O(r^{-1})\Leta^n\hatPhisHigh{j}.\end{align}
This together with \eqref{eq:wave:hatPhisHighi:an:ellmode:89} proves equation \eqref{eq:wave:hatPhisHighi:mode:rpwaveform}. Equation \eqref{eq:wave:hatPhisHighi:modegtr:rpwaveform} follows easily from \eqref{eq:wave:hatPhisHighi:mode:rpwaveform} and \eqref{eq:wave:hatPhisHighi:an} and using \eqref{eq:Commells:sum}.

The derivation of equations \eqref{eq:wave:shortform:Phiminussi:mode:loweri} is direct by applying $\PJ_{\ell}^{-\sfrak}$ or $\PJ_{\geq \ell}^{-\sfrak}$ to equations \eqref{eq:wave:shortform:Phiminussi:36} and making use of the commutator formula \eqref{comm:BoxhatsandPJ}.
\end{proof}

To apply point  (\ref{pt:2:prop:wave:rp}) of the $r^p$ lemma \ref{prop:wave:rp}, we have to  first estimate the commutator $\Comm{\ell}{s}{\varphi_s}$ for a general spin $s$ scalar $\varphi_s$. It follows from  formula \eqref{def:eq:Commells} and Proposition \ref{prop:modeprojection:1} that
\begin{align}
\label{eq:comm:ell:s:general:85}
\Comm{\ell}{s}{\varphi_s}= \sum_{\max\{\sfrak,\ell-2\}\leq \ell'\leq \ell+2}O(1)\Lxi\ellmode{\varphi_s}{\ell'}
+ \sum_{\max\{\sfrak,\ell-1\}\leq \ell'\leq \ell+1}O(1)\ellmode{\varphi_s}{\ell'}.
\end{align}
As a consequence, we obtain

\begin{lemma}
Let $\reg\in\mathbb{N}$. There exists a universal constant $\regl$ such that for any $\tb_2>\tb_1\geq \tb_0$
\begin{subequations}
\label{eq:error:Lxicomm:modes:all}
\begin{align}
\label{eq:error:Lxicomm:modes:s}
\norm{\Lxi\Comm{\tilde{\ell}}{s}{\hatPhisHigh{\sfrak-s}}}^2_{W_{p-3}^{\reg}(\Donetwo^{\geq R_0-M})}
\lesssim_{\reg,p}{}&
\norm{\Lxi{\hatPhisHigh{\sfrak-s}}}^2_{W_{p-3}^{\reg+\regl}(\Donetwo^{\geq R_0-M})}, \,\, \forall \tilde{\ell}\in \{\sfrak,\sfrak+1,\geq \sfrak+2\},\\
\label{eq:error:Lxicomm:modes:s+1}
\norm{\Lxi\Comm{\tilde{\ell}}{s}{\hatPhisHigh{\sfrak-s+1}}}^2_{W_{p-3}^{\reg}(\Donetwo^{\geq R_0-M})}
\lesssim_{\reg,p}{}&\norm{{\hatPhisHigh{\sfrak-s}}}^2_{W_{p-3}^{\reg+\regl}(\Donetwo^{\geq R_0-M})}, \,\, \forall \tilde{\ell}\in \{\sfrak+1,\geq \sfrak+2\},\\
\label{eq:error:Lxicomm:modes:geqs+2}
\norm{\Lxi\Comm{\geq\sfrak+2}{s}{\hatPhisHigh{\sfrak-s+2}}}^2_{W_{p-3}^{\reg}(\Donetwo^{\geq R_0-M})}
\lesssim_{\reg,p}{}&
\norm{ \ellmode{\hatPhisHigh{\sfrak-s+1}}{\sfrak+1}
}^2_{W_{p-3}^{\reg+\regl}(\Donetwo^{\geq R_0-M})}
+\norm{ \ellmode{\hatPhisHigh{\sfrak-s+1}}{\geq\sfrak+2}
}^2_{W_{p-3}^{\reg+\regl}(\Donetwo^{\geq R_0-M})}
\notag\\
&+\norm{{\hatPhisHigh{\sfrak-s}}}^2_{W_{p-3}^{\reg+\regl}(\Donetwo^{\geq R_0-M})}
+\norm{\tildePhisHighell{s}{\sfrak}}^2_{W_{p-3}^{\reg+\regl}(\Donetwo^{\geq R_0-M})}.
\end{align}
\end{subequations}
\end{lemma}

\begin{proof}
Since from Proposition \ref{prop:wavesys:hatPhisHighi}, the governing equations of $\hatPhiplussHigh{i}$ and $\hatPhiminuss{2\sfrak+i}$ are of the same form, it suffices to prove only for the case $s=+\sfrak$, and a similar argument holds for the $s=-\sfrak$ case.

 By \eqref{eq:comm:ell:s:general:85}, one has
\begin{subequations}
\begin{align}
\label{eq:LxiComm:mode:s0:gene:80}
\Lxi\Comm{{\ell}}{+\sfrak}{\hatPhiplussHigh{0}}={}&\sum_{\ell'=\sfrak}^{\ell+2}O(1)\Lxi^2\ellmode{\hatPhiplussHigh{0}}{\ell'}
+\sum_{\ell'=\sfrak}^{\ell+1}O(1)\Lxi\ellmode{\hatPhiplussHigh{0}}{\ell'}, \quad \forall {\ell}\in \{\sfrak,\sfrak+1\}, \\
\label{eq:LxiComm:mode:s1:gene:80}
\Lxi\Comm{{\ell}}{+\sfrak}{\hatPhiplussHigh{1}}={}&\sum_{\ell'=\sfrak}^{\ell+2}O(1)\Lxi^2\ellmode{\hatPhiplussHigh{1}}{\ell'}
+\sum_{\ell'=\sfrak}^{\ell+1}O(1)\Lxi\ellmode{\hatPhiplussHigh{1}}{\ell'}
, \quad \forall {\ell}\in \{\sfrak,\sfrak+1\},\\
\label{eq:LxiComm:mode:s2:gene:80}
\Lxi\Comm{\geq\sfrak+2}{+\sfrak}{\hatPhiplussHigh{2}}={}&
-\sum_{\sfrak\leq\ell\leq \sfrak+1}\Lxi\Comm{\ell}{+\sfrak}{\hatPhiplussHigh{2}}=\sum_{\ell'=\sfrak}^{\sfrak+3}O(1)\Lxi^2\ellmode{\hatPhiplussHigh{2}}{\ell'}
+\sum_{\ell'=\sfrak}^{\sfrak+2}O(1)\Lxi\ellmode{\hatPhiplussHigh{2}}{\ell'}.
\end{align}
\end{subequations}
In view of $\Comm{\geq \sfrak+2}{+\sfrak}{\varphi_{+\sfrak}}=-\Comm{\sfrak}{+\sfrak}{\varphi_{+\sfrak}}-\Comm{\sfrak+1}{+\sfrak}{\varphi_{+\sfrak}}$, we have
\begin{subequations}
\begin{align}
\label{eq:LxiComm:mode:s0:gene:88}
\Lxi\Comm{\geq \sfrak+2}{+\sfrak}{\hatPhiplussHigh{0}}={}&-\sum_{\ell=\sfrak,\sfrak+1}\Lxi\Comm{{\ell}}{+\sfrak}{\hatPhiplussHigh{0}},\\
\label{eq:LxiComm:mode:s1:gene:88}
\Lxi\Comm{\geq \sfrak+2}{+\sfrak}{\hatPhiplussHigh{1}}={}&-\sum_{\ell=\sfrak,\sfrak+1}\Lxi\Comm{{\ell}}{+\sfrak}{\hatPhiplussHigh{1}}.
\end{align}
\end{subequations}
The estimate \eqref{eq:error:Lxicomm:modes:s} follows from \eqref{eq:LxiComm:mode:s0:gene:80} and \eqref{eq:LxiComm:mode:s0:gene:88}.

By the formula \eqref{ansatz:hatPhisHigh} of $\hatPhiplussHigh{1}$, we have $\Lxi\ellmode{\hatPhiplussHigh{1}}{\ell}=\Lxi\curlVR\ellmode{\hatPhiplussHigh{0}}{\ell}
+\sum_{n\leq c}O(1)\Leta^n\Lxi \ellmode{\hatPhiplussHigh{0}}{\ell}$, and using the formula \eqref{eq:r2LxiV:form:3} together with the wave equation \eqref{eq:wave:hatPhisHighi:mode:rpwaveform} of $\ellmode{\hatPhiplussHigh{0}}{\ell}$, we find that there exists a $n>0$ such that for any $\ell\geq \sfrak$,
\begin{subequations}
\begin{align}
\label{eq:LxihatPhipluss1:s:>378}
\Lxi\ellmode{\hatPhiplussHigh{1}}{\ell}={}&\sum_{\abs{\mathbf{a}}\leq n}O(1)\CDeri^{\mathbf{a}} \ellmode{\hatPhiplussHigh{0}}{\ell}
+\Lxi\Comm{{\ell}}{+\sfrak}{\hatPhiplussHigh{0}}\notag\\
={}&\sum_{\abs{\mathbf{a}}\leq n}O(1)\CDeri^{\mathbf{a}} \ellmode{\hatPhiplussHigh{0}}{\ell}
+\sum_{\ell'=\max\{\sfrak,\ell-2\}}^{\ell+2}O(1)\Lxi^2\ellmode{\hatPhiplussHigh{0}}{\ell'}
+\sum_{\ell'=\max\{\sfrak, \ell-1\}}^{\ell+1}O(1)\Lxi\ellmode{\hatPhiplussHigh{0}}{\ell}.
\end{align}
\end{subequations}
 In a similar manner, we conclude that there exists a constant $n>0$ such that for any $\ell\in \{\sfrak,\sfrak+1\}$,
\begin{align}
\Lxi\Comm{\ell}{+\sfrak}{\hatPhiplussHigh{1}}={}&\sum_{\ell'=\sfrak}^{\ell+2}O(1)\Lxi^2\ellmode{\hatPhiplussHigh{1}}{\ell'}
+\sum_{\ell'=\sfrak}^{\ell+1}O(1)\Lxi\ellmode{\hatPhiplussHigh{1}}{\ell'}
=\sum_{\sfrak\leq \ell'\leq \ell+2}\sum_{\abs{\mathbf{a}}\leq n}O(1)\CDeri^{\mathbf{a}} \ellmode{\hatPhiplussHigh{0}}{\ell'}.
\end{align}
This together with \eqref{eq:LxiComm:mode:s1:gene:88} then yields the estimate  \eqref{eq:error:Lxicomm:modes:s+1}.

 Finally, by the formula \eqref{ansatz:hatPhisHigh} of $\hatPhiplussHigh{2}$, the formula \eqref{eq:r2LxiV:form:3}, and the wave equation \eqref{eq:wave:hatPhisHighi:mode:rpwaveform} of $\ellmode{\hatPhiplussHigh{1}}{\sfrak}$,
 \begin{align}
 \Lxi\ellmode{\hatPhiplussHigh{2}}{\sfrak}={}&\Lxi \curlVR\ellmode{\hatPhiplussHigh{1}}{\sfrak}
 +\sum_{n\leq n_1}\Lxi \Leta^n\curlVR\ellmode{\hatPhiplussHigh{0}}{\sfrak}
 +\sum_{i=0,1}\sum_{n\leq n_2}O(1) \Lxi\Leta^n\ellmode{\hatPhiplussHigh{i}}{\sfrak}\notag\\
 ={}&\sum_{\abs{\mathbf{a}}\leq n_1}O(1)\CDeri^{\mathbf{a}} \ellmode{\curlVR\hatPhiplussHigh{0}}{\sfrak}
 +\sum_{\abs{\mathbf{a}}\leq n_2}O(1)\CDeri^{\mathbf{a}} \ellmode{\hatPhiplussHigh{0}}{\sfrak}
 +\sum_{\ell=\sfrak+1}^{\sfrak+3}\sum_{\abs{\mathbf{a}}\leq n_3}O(1)\CDeri^{\mathbf{a}} \ellmode{\hatPhiplussHigh{0}}{\ell}.
 \end{align}
 This way of arguing can also be employed to eventually achieve
 \begin{align}
\Lxi\Comm{\geq\sfrak+2}{+\sfrak}{\hatPhiplussHigh{2}}={}&
\sum_{\ell'=\sfrak}^{\sfrak+3}O(1)\Lxi^2\ellmode{\hatPhiplussHigh{2}}{\ell'}
+\sum_{\ell'=\sfrak}^{\sfrak+2}O(1)\Lxi\ellmode{\hatPhiplussHigh{2}}{\ell'}\notag\\
={}&\sum_{\abs{\mathbf{a}}\leq n_1}O(1)\CDeri^{\mathbf{a}} \tildePhisHighell{+\sfrak}{\sfrak}
 +\sum_{\ell=\sfrak}^{\sfrak+3}\sum_{\abs{\mathbf{a}}\leq n_2}O(1)\CDeri^{\mathbf{a}} \ellmode{\hatPhiplussHigh{0}}{\ell}
+\sum_{\ell=\sfrak+1}^{\sfrak+2}\sum_{\abs{\mathbf{a}}\leq n_3}O(1)\CDeri^{\mathbf{a}} \ellmode{\hatPhiplussHigh{1}}{\ell}.
 \end{align}
The estimate  \eqref{eq:error:Lxicomm:modes:geqs+2} then holds.
\end{proof}

In addition, we shall utilize equation \eqref{eq:wave:tildePhisHighi:an:ellmode} to derive further energy decay for the modes. This is realized by applying the statement in point (\ref{pt:2:prop:transport:rp:mode}) of Lemma \ref{prop:wave:rp} to equation \eqref{eq:wave:tildePhisHighi:an:ellmode} for an extended range of $p$. Consequently, we shall estimate the integral term $\norm{\vartheta}^2_{W_{p-3}^{\reg}(\Donetwo^{\geq R_0-M})}$ or $\int_{\tb_1}^{\tb_2}\tb^{1+\delta}\norm{\vartheta}^2_{W_{p-4}^{\reg}(\Sigmatb^{\geq R_0-M})}\di \tb$ (by taking $\veps=\delta$) in the estimate \eqref{eq:rp:pless5:2:mode} but now with $\vartheta=\tilde{H}_{s,\ell}$ that is of the form \eqref{eq:tildeHsell-s}. The following lemma is to bound these integral terms.

\begin{lemma}
\label{lem:esti:tildeHsell:plarge}
For $p\in [\delta, 4-\delta]$,
\begin{subequations}
\begin{align}
\label{eq:tildeHsell:p-3norm:s:4}
\norm{\tilde{H}_{s,\sfrak}}^2_{W_{p-3}^{\reg}(\Donetwo^{\geq R_0-M})}\lesssim_{\reg,p}{}&\norm{{\hatPhisHigh{\sfrak-s}}
}^2_{W_{p-5}^{\reg}(\Donetwo^{\geq R_0-M})},\\
\label{eq:tildeHsell:p-3norm:s+1:4}
\norm{\tilde{H}_{s,\sfrak+1}}^2_{W_{p-3}^{\reg}(\Donetwo^{\geq R_0-M})}\lesssim_{\reg,p}{}&
\norm{\ellmode{\hatPhisHigh{\sfrak-s+1}}{\sfrak+1}
}^2_{W_{p-5}^{\reg+\regl}(\Donetwo^{\geq R_0-M})}
+\norm{\ellmode{\hatPhisHigh{\sfrak-s+1}}{\geq \sfrak+2}
}^2_{W_{p-5}^{\reg+\regl}(\Donetwo^{\geq R_0-M})}\notag\\
&
+\norm{{\hatPhisHigh{\sfrak-s}}
}^2_{W_{p-5}^{\reg+\regl}(\Donetwo^{\geq R_0-M})}
+\norm{\tildePhisHighell{s}{\sfrak}
}^2_{W_{p-5}^{\reg+\regl}(\Donetwo^{\geq R_0-M})},
\end{align}
\end{subequations}
and for $p\in [4+\delta, 5-\delta]$,
\begin{align}
\label{eq:tildeHsell:p-3norm:s:5}
\int_{\tb_1}^{\tb_2}\tb^{1+\delta}\norm{\tilde{H}_{s,\sfrak}}^2_{W_{p-4}^{\reg}(\Sigmatb^{\geq R_0-M})}\di \tb
\lesssim_{\reg,p}{}&
\int_{\tb_1}^{\tb_2}\tb^{1+\delta}\norm{{\hatPhisHigh{\sfrak-s}}
}^2_{W_{p-6}^{\reg+\regl}(\Sigmatb^{\geq R_0-M})}\di\tb.
\end{align}
\end{lemma}

\begin{proof}
 By the expression of $\tilde{H}_{s,\ell}$ in formula \eqref{eq:tildeHsell-s}, the estimates \eqref{eq:tildeHsell:p-3norm:s:4} and \eqref{eq:tildeHsell:p-3norm:s:5} follow immediately and we have in addition
 \begin{align}
 \label{eq:tildeHss+1:p:730}
 \norm{\tilde{H}_{s,\sfrak+1}}^2_{W_{p-3}^{\reg}(\Donetwo^{\geq R_0-M})}\lesssim_{\reg, p}{}&
\norm{\ellmode{\hatPhisHigh{\sfrak-s+1}}{\sfrak+1}
}^2_{W_{p-5}^{\reg+\regl}(\Donetwo^{\geq R_0-M})}
+\norm{\ellmode{\hatPhisHigh{\sfrak-s+1}}{\geq \sfrak+2}
}^2_{W_{p-5}^{\reg+\regl}(\Donetwo^{\geq R_0-M})}\notag\\
&
+\norm{\ellmode{\hatPhisHigh{\sfrak-s+1}}{\sfrak}
}^2_{W_{p-5}^{\reg+\regl}(\Donetwo^{\geq R_0-M})}
+\norm{{\hatPhisHigh{\sfrak-s}}
}^2_{W_{p-5}^{\reg+\regl}(\Donetwo^{\geq R_0-M})}.
\end{align}
Note that by definition of $\tildePhisHighell{s}{\sfrak}$ in \eqref{ansatz:tildePhisHigh:ellmode} and definition of $\hatPhisHigh{\sfrak-s+1}$ in \eqref{ansatz:hatPhisHigh}, one has
\begin{align}
\ellmode{\hatPhisHigh{\sfrak-s+1}}{\sfrak}=O(1)\tildePhisHighell{s}{\sfrak}+\sum_{\sfrak\leq\ell\leq \sfrak+2}\sum_{\abs{\mathbf{a}}\leq n}O(1) \CDeri^{\mathbf{a}}\ellmode{\hatPhisHigh{\sfrak-s}}{\ell}.
\end{align}
Substituting this back into \eqref{eq:tildeHss+1:p:730} then proves the estimate \eqref{eq:tildeHsell:p-3norm:s+1:4}.
\end{proof}

Recall from Definition \ref{def:Ffts:Phiplusminuss:-1to2} the formulas of the $r$-weighted energies
$F^{(0)}(\reg,p,\tb,\Lxi^j\ellmode{\Psipluss}{\ell})$
 and $F^{(i)}(\reg,p,\tb,\ellmode{\Psiminuss}{\ell})$ for an $\ell$ mode and $F^{(0)}(\reg,p,\tb,\Lxi^j\ellmode{\Psipluss}{\geq\ell})$
 and $F^{(i)}(\reg,p,\tb,\ellmode{\Psiminuss}{\geq\ell})$ for $\geq\ell$ modes of the spin $\pm\sfrak$ components, with $p\in [-1, 2-\delta]$. For our purpose of deriving extended $r^p$ hierarchy, we define the following $r$-weighted energies with an enlarged range of the parameter $p$.

 \begin{definition}
\label{def:Ffts:Phiplusminuss:modes:-1to2}
For the spin $+\sfrak$ component, define
\begin{subequations}
\label{eq:def:Ffts:Phiplus:modes:0:2to5}
\begin{align}
\label{def:Ffts:Phiplus:modes:0:p2to2delta}
F^{(0)}(\reg,p,\tb,\Lxi^j\ellmode{\Psipluss}{\sfrak})={}&0, \qquad\text{for } p\in (2-\delta,2+\delta),\\
\label{def:Ffts:Phiplus:modes:0:p25}
F^{(0)}(\reg,p,\tb,\Lxi^j\ellmode{\Psipluss}{\sfrak})={}&\norm{\Lxi^j\tildePhisHighell{+\sfrak}{\sfrak}}^2_{W_{p-4}^{\reg-1}(\Sigmatb^{\geq 4M})}
+ F^{(0)}(\reg,2-\delta,\tb,\Lxi^j\ellmode{\Psipluss}{\sfrak}),
\,\, \text{for } p\in [2+\delta, 5-\delta].
\end{align}
\end{subequations}
Define
\begin{subequations}
\label{eq:def:Ffts:Phiplus:modes:1:-1to2}
\begin{align}
\label{def:Ffts:Phiplus:modes:1:p-10}
F^{(1)}(\reg,p,\tb,\Lxi^j\ellmode{\Psipluss}{\sfrak+1})={}&0, \qquad\text{for } p\in [-1,\delta)\cup (2-\delta, 2+\delta),\\
\label{def:Ffts:Phiplus:modes:1:p02}
F^{(1)}(\reg,p,\tb,\Lxi^j\ellmode{\Psipluss}{\sfrak+1})={}&\norm{rV\Lxi^j\ellmode{\hatPhiplussHigh{1}}{\sfrak+1}}^2_{W_{p-2}^{\reg-1}(\Sigmatb^{\geq 4M})}
+\norm{\Lxi^j\ellmode{\hatPhiplussHigh{1}}{\sfrak+1}}^2_{W_{-2}^{\reg}(\Sigmatb^{\geq 4M})}\notag\\
&+ F^{(0)}(\reg,2-\delta,\tb,\Lxi^j\ellmode{\Psipluss}{\sfrak+1}),
\quad \text{for } p\in [\delta,2-\delta],\\
\label{def:Ffts:Phiplus:modes:1:p25}
F^{(1)}(\reg,p,\tb,\Lxi^j\ellmode{\Psipluss}{\sfrak+1})={}&\norm{\Lxi^j\tildePhisHighell{+\sfrak}{\sfrak+1}}^2_{W_{p-4}^{\reg-1}(\Sigmatb^{\geq 4M})}
+ F^{(1)}(\reg,2-\delta,\tb,\Lxi^j\ellmode{\Psipluss}{\sfrak+1}),
\,\, \text{for } p\in [2+\delta,4-\delta].
\end{align}
\end{subequations}
Define $F^{(1)}(\reg,p,\tb,\Lxi^j\ellmode{\Psipluss}{\geq \sfrak+2})$ for $p\in[-1,2-\delta]$ in the same way as in \eqref{eq:def:Ffts:Phiplus:modes:1:-1to2}. Further, define
\begin{subequations}
\label{eq:def:Ffts:Phiplus:modes:2:-1to2}
\begin{align}
\label{def:Ffts:Phiplus:modes:2:p-10}
F^{(2)}(\reg,p,\tb,\Lxi^j\ellmode{\Psipluss}{\geq \sfrak+2})={}&0, \qquad\text{for } p\in [-1,\delta),\\
\label{def:Ffts:Phiplus:modes:2:p02}
F^{(2)}(\reg,p,\tb,\Lxi^j\ellmode{\Psipluss}{\geq \sfrak+2})={}&\norm{rV\Lxi^j\ellmode{\hatPhiplussHigh{2}}{\geq\sfrak+2}}^2_{W_{p-2}^{\reg-1}(\Sigmatb^{\geq 4M})}
+\norm{\Lxi^j\ellmode{\hatPhiplussHigh{2}}{\geq\sfrak+2}}^2_{W_{-2}^{\reg}(\Sigmatb^{\geq 4M})}\notag\\
&+ F^{(1)}(\reg,p,\tb,\Lxi^j\ellmode{\Psipluss}{\geq \sfrak+2}),
\quad \text{for } p\in [\delta,2-\delta].
\end{align}
\end{subequations}

For the spin $-\sfrak$ component, define
\begin{subequations}
\label{eq:def:Ffts:Phiminus:modes:0:2to5}
\begin{align}
\label{def:Ffts:Phiminus:modes:0:p2to2delta}
F^{(2\sfrak)}(\reg,p,\tb,\Lxi^j\ellmode{\Psiminuss}{\sfrak})={}&0, \qquad\text{for } p\in (2-\delta,2+\delta),\\
\label{def:Ffts:Phiminus:modes:0:p25}
F^{(2\sfrak)}(\reg,p,\tb,\Lxi^j\ellmode{\Psiminuss}{\sfrak})={}&\norm{\Lxi^j\tildePhisHighell{-\sfrak}{\sfrak}}^2_{W_{p-4}^{\reg-1}(\Sigmatb^{\geq 4M})}
+ F^{(2\sfrak)}(\reg,2-\delta,\tb,\Lxi^j\ellmode{\Psiminuss}{\sfrak}),
\,\,  \text{for } p\in [2+\delta, 5-\delta].
\end{align}
\end{subequations}
Define
\begin{subequations}
\label{eq:def:Ffts:Phiplus:modes:2:-1to4:6382}
\begin{align}
F^{(2\sfrak+1)}(\reg,p,\tb,\Lxi^j\ellmode{\Psiminuss}{\sfrak+1})={}&0, \quad \text{for } p\in [-1,\delta)\cup (2-\delta, 2+\delta),\\
F^{(2\sfrak+1)}(\reg,p,\tb,\Lxi^j\ellmode{\Psiminuss}{\sfrak+1})={}&
\norm{rV\Lxi^j\ellmode{\hatPhiminuss{2\sfrak+1}}{\sfrak+1}}^2_{W_{p-2}^{\reg-1}(\Sigmatb^{\geq 4M})}
+\norm{\Lxi^j\ellmode{\hatPhiminuss{2\sfrak+1}}{\sfrak+1}}^2_{W_{-2}^{\reg}(\Sigmatb^{\geq 4M})}\notag\\
&
+\sum_{i=\sfrak}^{2\sfrak}F^{(i)}(\reg,2-\delta,\tb,\Lxi^j\ellmode{\Psiminuss}{\sfrak+1})
, \quad \text{for } p\in [\delta,2),\\
\label{def:Ffts:Phiminus:modes:1:p25}
F^{(2\sfrak+1)}(\reg,p,\tb,\Lxi^j\ellmode{\Psiminuss}{\sfrak+1})={}&\norm{\Lxi^j\tildePhisHighell{-\sfrak}{\sfrak+1}}^2_{W_{p-4}^{\reg-1}(\Sigmatb^{\geq 4M})}
+ F^{(2\sfrak+1)}(\reg,2-\delta,\tb,\Lxi^j\ellmode{\Psiminuss}{\sfrak+1}),
\,\,  \text{for } p\in [2+\delta,4-\delta].
\end{align}
\end{subequations}
Define $F^{(2\sfrak+1)}(\reg,p,\tb,\Lxi^j\ellmode{\Psiminuss}{\geq\sfrak+2})$ for $p\in [-1,2-\delta]$ in the same way as in \eqref{eq:def:Ffts:Phiplus:modes:2:-1to4:6382}. Further, define
\begin{subequations}
\begin{align}
F^{(2\sfrak+2)}(\reg,p,\tb,\Lxi^j\ellmode{\Psiminuss}{\geq\sfrak+2})={}&0, \quad \text{for } p\in [-1,\delta),\\
F^{(2\sfrak+2)}(\reg,p,\tb,\Lxi^j\ellmode{\Psiminuss}{\geq\sfrak+2})={}&
\norm{rV\Lxi^j\ellmode{\hatPhiminuss{2\sfrak+2}}{\geq\sfrak+2}}^2_{W_{p-2}^{\reg-1}(\Sigmatb^{\geq 4M})}
+\norm{\Lxi^j\ellmode{\hatPhiminuss{2\sfrak+2}}{\geq\sfrak+2}}^2_{W_{-2}^{\reg}(\Sigmatb^{\geq 4M})}\notag\\
&
+F^{(2\sfrak+1)}(\reg,2-\delta,\tb,\Lxi^j\ellmode{\Psiminuss}{\geq\sfrak+2})
, \quad \text{for } p\in [\delta,2).
\end{align}
\end{subequations}
\end{definition}

\begin{remark}
In defining the energies $F^{(0)}(\reg,p,\tb,\Lxi^j\ellmode{\Psipluss}{\sfrak})$ and $F^{(2\sfrak)}(\reg,p,\tb,\Lxi^j\ellmode{\Psiminuss}{\sfrak})$ for $p\in [2+\delta, 5-\delta]$,  their expressions are dependent  not only on the $\sfrak$ mode of the spin $s$ component but also on the other modes in view of the definition \eqref{ansatz:tildePhisHigh:ellmode} of $\tildePhisHighell{s}{\sfrak}$. Similarly for the energies $F^{(1)}(\reg,p,\tb,\Lxi^j\ellmode{\Psipluss}{\sfrak+1})$ and $F^{(2\sfrak+1)}(\reg,p,\tb,\Lxi^j\ellmode{\Psiminuss}{\sfrak+1})$ for $p\in [2+\delta, 4-\delta]$.
\end{remark}

Our first goal is to derive global $r^p$ estimates for the modes of $\{\Phiminuss{i}\}_{i\leq 2\sfrak}$, which are analogues of the estimates \eqref{eq:rpminussglobal:less2} in Proposition \ref{prop:rpplusglobal} but at the mode level.

\begin{cor}
\label{cor:globalrp:modes:loweri:39}
Let $\reg\in \mathbb{N}$.  For any $\tb_2>\tb_1\geq \tb_0$, $i\in [\sfrak,2\sfrak]$, $p\in [\delta,2-\delta]$, and  $\tilde{\ell}\in \{\sfrak,\sfrak+1, \geq \sfrak+2\}$,
\begin{align}
\label{eq:rpminussglobal:mode:less2}
\hspace{6ex}&\hspace{-6ex}
F^{(i)}(\reg,p,\tb_2,\Lxi^j\ellmode{\Psiminuss}{\tilde{\ell}})
+\sum_{i'=0}^i\norm{\Lxi^j\ellmode{\PsiminussHigh{i'}}{\tilde{\ell}}}^2_{W_{p-3}^{\reg-\sfrak-1-l(j,s)}(\Donetwo)}\notag\\
\lesssim_{\reg,\delta} {}&F^{(i)}(\reg,p,\tb_1,\Lxi^j\ellmode{\Psiminuss}{\tilde{\ell}})
+F^{(\sfrak)}(\reg+\regl,\delta,\tb_1,\Lxi^{j+1}\Psiminuss)
+\sum_{i'=0}^i\norm{\Lxi^{j+1}{\PsiminussHigh{i'}}}^2_{W_{p-3}^{\reg+\regl}(\Donetwo)}.
\end{align}
\end{cor}

\begin{proof}
The proof is adapted from the one of Proposition \ref{prop:rpplusglobal}. The only difference lies in the extra coupling terms with the other modes. It suffices to consider $j=0$ case, since $\Lxi^j$ commutes with the TME.
Equations of $\ellmode{\Phiminuss{i'}}{\tilde{\ell}}$ $(i'=0,1,\ldots, 2\sfrak$) in system \eqref{eq:wave:shortform:Phiminussi:mode:loweri} are the same  as the governing equations \eqref{eq:wave:shortform:Phiminussi:36} of $\Phiminuss{i'}$ except that on the RHS of the wave equations for  $\ellmode{\Phiminuss{i'}}{\tilde{\ell}}$ in system \eqref{eq:wave:shortform:Phiminussi:mode:loweri}, there is an additional term $\Lxi\Comm{{\tilde\ell}}{-\sfrak}{\Phiminuss{i'}}$. Thus, in applying the $r^p$ estimate \eqref{eq:rp:less2} to each subequation of $\ellmode{\Phiminuss{i'}}{\tilde{\ell}}$, we have one additional integral term $\norm{\Lxi\Comm{{\tilde\ell}}{-\sfrak}{\Phiminuss{i'}}}^2_{W_{p-3}^{\reg}(\Donetwo^{\geq R_0-M})}\lesssim_{\reg,p} \norm{\Lxi\Phiminuss{i'}}^2_{W_{p-3}^{\reg}(\Donetwo^{\geq R_0-M})}$. In the end, we combine the obtained $r^p$ estimates for modes with the BEAM estimates \eqref{eq:BEAM:mode:s} for modes to conclude the global $r^p$ estimate \eqref{eq:rpminussglobal:mode:less2}.
\end{proof}

We then derive the global $r^p$ estimates for a larger ranger of $p$ weight. This is achieved in the following two corollaries.

 \begin{cor}[Global $r^p$ estimates for $p\in (0,2)$]
 There exists a constant $\regl$ such that for any $\tb_2>\tb_1\geq\tb_0$, the following global $r^p$ estimates for $p\in[\delta, 2-\delta]$ hold:
  \begin{itemize}
 \item for any $\tilde{\ell}\in \{\sfrak, \sfrak+1,\geq \sfrak+2\}$,
 \begin{subequations}
\begin{align}
\label{eq:rp:modes:Psipluss:0:p02}
\hspace{5.5ex}&\hspace{-5.5ex}
F^{(0)}(\reg,p,\tb_2,\Lxi^j\ellmode{\Psipluss}{\tilde{\ell}})
+\norm{\Lxi^j\ellmode{\Psipluss}{\tilde{\ell}}}^2_{W_{p-3}^{\reg}(\Donetwo)}\notag\\
\lesssim_{\reg,p} {}&F^{(0)}(\reg+\regl,p,\tb_1,\Lxi^j\ellmode{\Psipluss}{\tilde{\ell}})
+F^{(0)}(\reg+\regl,\delta,\tb_1,\Lxi^{j+1}\Psipluss)
+ \norm{\Lxi^{j+1}{\hatPhiplussHigh{0}}}^2_{W_{p-3}^{\reg+\regl}(\Donetwo^{\geq 4M})},\\
\label{eq:rp:modes:Psiminuss:0:p02:89}
\hspace{5.5ex}&\hspace{-5.5ex}
F^{(2\sfrak)}(\reg,p,\tb_2,\Lxi^j\ellmode{\Psiminuss}{\tilde{\ell}})
+\sum_{i=0}^{2\sfrak}\norm{\Lxi^j\ellmode{\PsiminussHigh{i}}{\tilde{\ell}}}^2_{W_{p-3}^{\reg}(\Donetwo)}\notag\\
\lesssim_{\reg,p} {}&F^{(2\sfrak)}(\reg+\regl,p,\tb_1,\Lxi^j\ellmode{\Psiminuss}{\tilde{\ell}})
+F^{(\sfrak)}(\reg+\regl,\delta,\tb_1,\Lxi^{j+1}\Psiminuss)
+ \norm{\Lxi^{j+1}{\hatPhiminuss{2\sfrak}}}^2_{W_{p-3}^{\reg+\regl}(\Donetwo^{\geq 4M})};
\end{align}

\item for any $\tilde{\ell}\in \{\sfrak+1,\geq \sfrak+2\}$,
\begin{align}
\label{eq:rp:modes:Psipluss:1:p02}
\hspace{5.5ex}&\hspace{-5.5ex}
F^{(1)}(\reg,p,\tb_2,\Lxi^j\ellmode{\Psipluss}{\tilde{\ell}})
+\norm{\Lxi^j\ellmode{\hatPhiplussHigh{1}}{\tilde{\ell}}}^2_{W_{p-3}^{\reg}(\Donetwo^{\geq 4M})}
+\norm{\Lxi^j\ellmode{\Psipluss}{\tilde{\ell}}}^2_{W_{-3-\delta}^{\reg}(\Donetwo)}
\notag\\
\lesssim_{\reg,p}  {}&F^{(1)}(\reg+\regl,p,\tb_1,\Lxi^j\ellmode{\Psipluss}{\tilde{\ell}})
+F^{(0)}(\reg+\regl,\delta,\tb_1,\Lxi^{j+1}\Psipluss)
+ \norm{\Lxi^{j}{\Psipluss}}^2_{W_{p-3}^{\reg+\regl}(\Donetwo)},\\
\label{eq:rp:modes:Psiminuss:1:p02:89}
\hspace{5.5ex}&\hspace{-5.5ex}
F^{(2\sfrak+1)}(\reg,p,\tb_2,\Lxi^j\ellmode{\Psiminuss}{\tilde{\ell}})
+\norm{\Lxi^j\ellmode{\hatPhiminuss{2\sfrak+1}}{\tilde{\ell}}}^2_{W_{p-3}^{\reg}(\Donetwo^{\geq 4M})}
+\sum_{i=0}^{2\sfrak}\norm{\Lxi^j\ellmode{\PsiminussHigh{i}}{\tilde{\ell}}}^2_{W_{-3-\delta}^{\reg}(\Donetwo)}
\notag\\
\lesssim_{\reg,p}  {}&F^{(2\sfrak+1)}(\reg+\regl,p,\tb_1,\Lxi^j\ellmode{\Psiminuss}{\tilde{\ell}})
+F^{(\sfrak)}(\reg+\regl,\delta,\tb_1,\Lxi^{j+1}\Psiminuss)
+ \sum_{i=0}^{2\sfrak}\norm{\Lxi^{j}{\PsiminussHigh{i}}}^2_{W_{p-3}^{\reg+\regl}(\Donetwo)};
\end{align}

\item for $\geq \sfrak+2$ modes,
\begin{align}
\label{eq:rp:modes:Psipluss:2:p02}
\hspace{5.5ex}&\hspace{-5.5ex}
F^{(2)}(\reg,p,\tb_2,\Lxi^j\ellmode{\Psipluss}{\geq\sfrak+2})
+\norm{\Lxi^j\ellmode{\hatPhiplussHigh{2}}{\geq\sfrak+2}}^2_{W_{p-3}^{\reg}(\Donetwo^{\geq 4M})}
+\norm{\Lxi^j\ellmode{\Psipluss}{\geq \sfrak+2}}^2_{W_{-3-\delta}^{\reg}(\Donetwo)}
\notag\\
\lesssim_{\reg,p}  {}&F^{(2)}(\reg+\regl,p,\tb_1,\Lxi^j\ellmode{\Psipluss}{\geq\sfrak+2}) +F^{(0)}(\reg+\regl,\delta,\tb_1,\Lxi^{j+1}\Psipluss)
+ \norm{\Lxi^{j}\Psipluss}^2_{W_{p-3}^{\reg+\regl}(\Donetwo)}
\notag\\
&
+\norm{ \ellmode{\hatPhiplussHigh{1}}{\sfrak+1}
}^2_{W_{p-3}^{\reg+\regl}(\Donetwo^{\geq R_0-M})}+\norm{ \ellmode{\hatPhiplussHigh{1}}{\geq\sfrak+2}
}^2_{W_{p-3}^{\reg+\regl}(\Donetwo^{\geq R_0-M})}
+\norm{\tildePhisHighell{+\sfrak}{\sfrak}}^2_{W_{p-3}^{\reg+\regl}(\Donetwo^{\geq R_0-M})},\\
\label{eq:rp:modes:Psiminuss:2:p02:89}
\hspace{5.5ex}&\hspace{-5.5ex}
F^{(2\sfrak+2)}(\reg,p,\tb_2,\Lxi^j\ellmode{\Psiminuss}{\geq\sfrak+2})
+\norm{\Lxi^j\ellmode{\hatPhiminuss{2\sfrak+2}}{\geq\sfrak+2}}^2_{W_{p-3}^{\reg}(\Donetwo^{\geq 4M})}
+\sum_{i=0}^{2\sfrak}\norm{\Lxi^j\ellmode{\PsiminussHigh{i}}{\geq \sfrak+2}}^2_{W_{-3-\delta}^{\reg}(\Donetwo)}
\notag\\
\lesssim_{\reg,p}  {}&F^{(2\sfrak+2)}(\reg+\regl,p,\tb_1,\Lxi^j\ellmode{\Psiminuss}{\geq\sfrak+2})
+F^{(\sfrak)}(\reg+\regl,\delta,\tb_1,\Lxi^{j+1}\Psipluss)
+ \sum_{i=0}^{2\sfrak}\norm{\Lxi^{j}\PsiminussHigh{i}}^2_{W_{p-3}^{\reg+\regl}(\Donetwo)}
\notag\\
&
+\norm{ \ellmode{\hatPhiminuss{2\sfrak+1}}{\sfrak+1}
}^2_{W_{p-3}^{\reg+\regl}(\Donetwo^{\geq R_0-M})}+\norm{ \ellmode{\hatPhiminuss{2\sfrak+1}}{\geq\sfrak+2}
}^2_{W_{p-3}^{\reg+\regl}(\Donetwo^{\geq R_0-M})}
+\norm{\tildePhisHighell{-\sfrak}{\sfrak}}^2_{W_{p-3}^{\reg+\regl}(\Donetwo^{\geq R_0-M})}.
\end{align}
\end{subequations}
\end{itemize}
\end{cor}

\begin{proof}
We take the case $\tilde\ell=\sfrak$ of the estimate \eqref{eq:rp:modes:Psipluss:0:p02} as an example to illustrate the general idea. By applying the estimate of point (\ref{pt:2:prop:wave:rp}) in Lemma \ref{prop:wave:rp} to equation \eqref{eq:wave:hatPhisHighi:mode:rpwaveform} with $s=+\sfrak$ and $\ell=\sfrak$ and  adding in a sufficiently large multiple of the BEAM estimate \eqref{eq:BEAM:+s:mode:s} for $\ell=\sfrak$ such that the error terms supported on $[R_0-M,R_0]$ in the $r^p$ estimate are absorbed, we arrive at
\begin{align*}
\hspace{5.5ex}&\hspace{-5.5ex}
F^{(0)}(\reg,p,\tb_2,\Lxi^j\ellmode{\Psipluss}{\sfrak})
+\norm{\Lxi^j\ellmode{\Psipluss}{\sfrak}}^2_{W_{p-3}^{\reg}(\Donetwo)}\notag\\
\lesssim_{\reg,p}  {}&F^{(0)}(\reg+\regl,p,\tb_1,\Lxi^j\ellmode{\Psipluss}{\sfrak})
+\norm{\Lxi^{j}\Lxi\Comm{{\ell}}{+\sfrak}{\hatPhiplussHigh{0}}}^2_{W_{p-3}^{\reg}(\Donetwo^{\geq R_0-M})}.
\end{align*}
Note that in the derivation of the above estimate,  the error terms arising from the terms with $O(r^{-1})$ coefficients on the RHS of \eqref{eq:wave:hatPhisHighi:mode:rpwaveform} are bounded by $\norm{\Lxi^j{\Psipluss}}^2_{W_{-3-\delta}^{\reg+\regl}(\Donetwo)}$ which has been already controlled in the BEAM estimate. We then make use of the estimate \eqref{eq:error:Lxicomm:modes:s} to estimate the last term, thus the estimate \eqref{eq:rp:modes:Psipluss:0:p02} with $\tilde\ell=\sfrak$  follows. The remaining estimates for the modes of the spin $+\sfrak$ component hold by arguing in the same manner by applying the estimate of point (\ref{pt:2:prop:wave:rp}) in Lemma \ref{prop:wave:rp} to equation \eqref{eq:wave:hatPhisHighi:modeandmodegtr:rpwaveform}, adding in the BEAM estimate \eqref{eq:BEAM:+s:mode:s} and making use of the estimate \eqref{eq:error:Lxicomm:modes:all}.

As can be seen from Proposition \ref{eq:wave:hatPhisHighi:an:ellmode}, the scalar $\ellmode{\hatPhiminuss{2\sfrak+i}}{\tilde{\ell}}$ satisfies basically the same wave equation as the one of the scalar $\ellmode{\hatPhiplussHigh{i}}{\tilde{\ell}}$. Therefore, the above discussions for the spin $+\sfrak$ component can be applied to prove the desired estimates for the modes of the spin $-\sfrak$ component with the only difference that we shall now add in the BEAM estimate  \eqref{eq:BEAM:-s:mode:s}  instead.
\end{proof}

\begin{cor}[Global $r^p$ estimates for an extended range of $p$]
 Let $j\in \mathbb{N}$. There exists a constant $\regl=\regl$ such that for any $\tb_2>\tb_1\geq\tb_0$,
 \begin{itemize}
 \item for any $p\in [2+\delta, 4)$,
 \begin{subequations}
 \label{eq:rp:modes:Psipms:0:p24:56}
\begin{align}
\label{eq:rp:modes:Psipluss:0:p24:56}
\hspace{5.5ex}&\hspace{-5.5ex}
F^{(0)}(\reg,p,\tb_2,\Lxi^j\ellmode{\Psipluss}{\sfrak})
+\norm{\Lxi^j\tildePhisHighell{+\sfrak}{\sfrak}}^2_{W_{p-5}^{\reg}(\Donetwo^{\geq 4M})}
+\norm{\Lxi^j\ellmode{\Psipluss}{\sfrak}}^2_{W_{-1-\delta}^{\reg}(\Donetwo)}\notag\\
\lesssim_{\reg,p}  {}&F^{(0)}(\reg+\regl,p,\tb_1,\Lxi^j\ellmode{\Psipluss}{\sfrak})
+F^{(0)}(\reg+\regl,\delta,\tb_1,\Lxi^{j+1}\Psipluss)\notag\\
&
+\norm{\Lxi^j{\Psipluss}
}^2_{W_{p-5}^{\reg+\regl}(\Donetwo)}
+ \norm{\Lxi^{j+1}{\hatPhiplussHigh{0}}}^2_{W_{-1-\delta}^{\reg+\regl}(\Donetwo^{\geq 4M})}
,\\
\label{eq:rp:modes:Psiminuss:0:p24:56}
\hspace{5.5ex}&\hspace{-5.5ex}
F^{(2\sfrak)}(\reg,p,\tb_2,\Lxi^j\ellmode{\Psiminuss}{\sfrak})
+\norm{\Lxi^j\tildePhisHighell{-\sfrak}{\sfrak}}^2_{W_{p-5}^{\reg}(\Donetwo^{\geq 4M})}
+\norm{\Lxi^j\ellmode{\Psiminuss}{\sfrak}}^2_{W_{-1-\delta}^{\reg}(\Donetwo)}\notag\\
\lesssim_{\reg,p}  {}&F^{(2\sfrak)}(\reg+\regl,p,\tb_1,\Lxi^j\ellmode{\Psiminuss}{\sfrak})
+\sum_{i=0}^{2\sfrak}\norm{\Lxi^j\ellmode{\PsiminussHigh{i}}{{\sfrak}}}^2_{W_{-3-\delta}^{\reg+\regl}(\Donetwo)}
+ \norm{\Lxi^{j+1}{\hatPhiminuss{2\sfrak}}}^2_{W_{-1-\delta}^{\reg+\regl}(\Donetwo^{\geq 4M})}\notag\\
&
+F^{(\sfrak)}(\reg+\regl,\delta,\tb_1,\Lxi^{j+1}\Psiminuss)
+\norm{\Lxi^j{\hatPhiminuss{2\sfrak}}
}^2_{W_{p-5}^{\reg+\regl}(\Donetwo^{\geq R_0-M})},
\end{align}
\end{subequations}
and for $p \in [4, 5-\delta]$,
\begin{subequations}
\label{eq:rp:modes:Psipms:0:p45:56}
\begin{align}
\label{eq:rp:modes:Psipluss:0:p45:56}
\hspace{5.5ex}&\hspace{-5.5ex}
F^{(0)}(\reg,p,\tb_2,\Lxi^j\ellmode{\Psipluss}{\sfrak})
+\norm{\Lxi^j\tildePhisHighell{+\sfrak}{\sfrak}}^2_{W_{p-5}^{\reg}(\Donetwo^{\geq 4M})}
+\norm{\Lxi^j\ellmode{\Psipluss}{\sfrak}}^2_{W_{-1-\delta}^{\reg}(\Donetwo)}\notag\\
\lesssim_{\reg,p}  {}&F^{(0)}(\reg+\regl,p,\tb_1,\Lxi^j\ellmode{\Psipluss}{\sfrak})
+F^{(0)}(\reg+\regl,\delta,\tb_1,\Lxi^{j+1}\Psipluss)
+\norm{\Lxi^j\ellmode{\Psipluss}{{\sfrak}}}^2_{W_{-3-\delta}^{\reg+\regl}(\Donetwo)}\notag\\
&
+ \norm{\Lxi^{j+1}\Psipluss}^2_{W_{-1-\delta}^{\reg+\regl}(\Donetwo)}
+\int_{\tb_1}^{\tb_2}\tb^{1+\delta}\norm{\Lxi^j{\hatPhiplussHigh{0}}
}^2_{W_{p-6}^{\reg+\regl}(\Sigmatb^{\geq R_0-M})}\di\tb,\\
\label{eq:rp:modes:Psiminuss:0:p45:56}
\hspace{5.5ex}&\hspace{-5.5ex}
F^{(2\sfrak)}(\reg,p,\tb_2,\Lxi^j\ellmode{\Psiminuss}{\sfrak})
+\norm{\Lxi^j\tildePhisHighell{-\sfrak}{\sfrak}}^2_{W_{p-5}^{\reg}(\Donetwo^{\geq 4M})}
+\norm{\Lxi^j\ellmode{\Psiminuss}{\sfrak}}^2_{W_{-1-\delta}^{\reg}(\Donetwo)}\notag\\
\lesssim_{\reg,p}  {}&F^{(2\sfrak)}(\reg+\regl,p,\tb_1,\Lxi^j\ellmode{\Psiminuss}{\sfrak})
+\sum_{i=0}^{2\sfrak}\norm{\Lxi^j\ellmode{\PsiminussHigh{i}}{{\sfrak}}}^2_{W_{-3-\delta}^{\reg+\regl}(\Donetwo)}
+F^{(\sfrak)}(\reg+\regl,\delta,\tb_1,\Lxi^{j+1}\Psiminuss)\notag\\
&
+ \norm{\Lxi^{j+1}{\hatPhiminuss{2\sfrak}}}^2_{W_{-1-\delta}^{\reg+\regl}(\Donetwo^{\geq 4M})}
+\int_{\tb_1}^{\tb_2}\tb^{1+\delta}\norm{{\Lxi^j\hatPhiminuss{2\sfrak}}
}^2_{W_{p-6}^{\reg+\regl}(\Sigmatb^{\geq R_0-M})}\di\tb;
\end{align}
\end{subequations}

\item for any $p\in [\delta, 2-\delta]\cup [2+\delta, 4-\delta]$,
\begin{subequations}
\label{eq:rp:modes:Psipms:1:p24:56}
\begin{align}
\label{eq:rp:modes:Psipluss:1:p24:56}
\hspace{5.5ex}&\hspace{-5.5ex}
F^{(1)}(\reg,p,\tb_2,\Lxi^j\ellmode{\Psipluss}{\sfrak+1})
+\norm{\Lxi^j\tildePhisHighell{+\sfrak}{\sfrak+1}}^2_{W_{p-5}^{\reg}(\Donetwo^{\geq 4M})}
+\norm{\Lxi^j\ellmode{\Psipluss}{\sfrak+1}}^2_{W_{-3-\delta}^{\reg}(\Donetwo)}
\notag\\
\lesssim_{\reg,p}
{}&F^{(1)}(\reg+\regl,p,\tb_1,\Lxi^j\ellmode{\Psipluss}{\sfrak+1})
+F^{(0)}(\reg+\regl,\delta,\tb_1,\Lxi^{j+1}\Psipluss)
+ \norm{\Lxi^{j}{\Psipluss}}^2_{W_{p-5}^{\reg+\regl}(\Donetwo)}\notag\\
&
+\norm{\Lxi^{j}\ellmode{\hatPhiplussHigh{1}}{\sfrak+1}
}^2_{W_{p-5}^{\reg+\regl}(\Donetwo^{\geq R_0-M})}
+\norm{\Lxi^{j}\ellmode{\hatPhiplussHigh{1}}{\geq \sfrak+2}
}^2_{W_{p-5}^{\reg+\regl}(\Donetwo^{\geq R_0-M})}
+\norm{\Lxi^{j}\tildePhisHighell{+\sfrak}{\sfrak}
}^2_{W_{p-5}^{\reg+\regl}(\Donetwo^{\geq R_0-M})},\\
\label{eq:rp:modes:Psiminuss:1:p24:56}
\hspace{5.5ex}&\hspace{-5.5ex}
F^{(2\sfrak+1)}(\reg,p,\tb_2,\Lxi^j\ellmode{\Psiminuss}{\sfrak+1})
+\norm{\Lxi^j\ellmode{\hatPhiminuss{2\sfrak+1}}{\sfrak+1}}^2_{W_{p-5}^{\reg}(\Donetwo^{\geq 4M})}
+\norm{\Lxi^j\tildePhisHighell{-\sfrak}{\sfrak+1}}^2_{W_{-5-\delta}^{\reg}(\Donetwo)}
\notag\\
\lesssim_{\reg,p}  {}&F^{(2\sfrak+1)}(\reg+\regl,p,\tb_1,\Lxi^j\ellmode{\Psiminuss}{\sfrak+1})
+F^{(\sfrak)}(\reg+\regl,\delta,\tb_1,\Lxi^{j+1}\Psiminuss)
+ \sum_{0\leq i'\leq 2\sfrak}\norm{\Lxi^{j}{\PsiminussHigh{i'}}}^2_{W_{p-5}^{\reg+\regl}(\Donetwo)}\notag\\
&+\norm{\Lxi^{j}\ellmode{\hatPhiminuss{2\sfrak+1}}{\sfrak+1}
}^2_{W_{p-5}^{\reg+\regl}(\Donetwo^{\geq R_0-M})}
+\norm{\Lxi^{j}\ellmode{\hatPhiminuss{2\sfrak+1}}{\geq \sfrak+2}
}^2_{W_{p-5}^{\reg+\regl}(\Donetwo^{\geq R_0-M})}
+\norm{\Lxi^{j}\tildePhisHighell{-\sfrak}{\sfrak}
}^2_{W_{p-5}^{\reg+\regl}(\Donetwo^{\geq R_0-M})}.
\end{align}
\end{subequations}
\end{itemize}
\end{cor}

\begin{proof}
Note that equation \eqref{eq:wave:tildePhisHighi:an:ellmode} for $\tildePhisHighell{+\sfrak}{\ell}$ can be put into the form of
\begin{align}
\mu Y_G \tildePhisHighell{+\sfrak}{\ell}= \vartheta(\tildePhisHighell{+\sfrak}{\ell})=\tilde{H}_{+\sfrak,\ell}.
\end{align}
The proof is based on applying the statement in point (\ref{pt:2:prop:transport:rp:mode}) of Lemma \ref{prop:wave:rp} to this inhomogeneous transport equation of $\tildePhisHighell{+\sfrak}{\ell}$ for an extended range of $p$. Consider first $\ell=\sfrak$. We apply the $r^p$ estimate \eqref{eq:rp:pleq4:2:mode} for $p\in [2+\delta, 4)$ to equation \eqref{eq:wave:tildePhisHighi:an:ellmode}  of $\tildePhisHighell{+\sfrak}{\sfrak}$. Note that $\vartheta(\tildePhisHighell{+\sfrak}{\sfrak})=\tilde{H}_{+\sfrak,\sfrak}$ and  that $\norm{\tilde{H}_{+\sfrak,\sfrak}}^2_{W_{p-3}^{\reg}(\Donetwo^{\geq R_0-M})}$  has been estimated in  \eqref{eq:tildeHsell:p-3norm:s:4}, then the estimate \eqref{eq:rp:modes:Psipluss:0:p24:56} follows by adding in the estimate \eqref{eq:rp:modes:Psipluss:0:p02} with $\tilde{\ell}=\sfrak$ and $p=2-\delta$.  This also works for $\sfrak+1$ mode and yields the estimate \eqref{eq:rp:modes:Psipluss:1:p24:56}. To show the estimate \eqref{eq:rp:modes:Psipluss:0:p45:56} for $p\in [4, 5-\delta]$, the only difference from proving \eqref{eq:rp:modes:Psipluss:0:p24:56} for $p\in [2+\delta, 4-\delta]$ is that we utilize the $r^p$ estimate \eqref{eq:rp:pgeq4:2:mode} to equation \eqref{eq:wave:tildePhisHighi:an:ellmode}  of $\tildePhisHighell{+\sfrak}{\sfrak}$ and use the estimate \eqref{eq:tildeHsell:p-3norm:s:5} to bound the error term $\int_{\tb_1}^{\tb_2}\tb^{1+\delta}\norm{\vartheta(\tildePhisHighell{+\sfrak}{\sfrak+1})}^2_{W_{p-4}^{\reg}(\Donetwo^{\geq R_0-M})}\di\tb=\int_{\tb_1}^{\tb_2}\tb^{1+\delta}\norm{\tilde{H}_{+\sfrak,\sfrak+1}}^2_{W_{p-4}^{\reg}(\Donetwo^{\geq R_0-M})}\di\tb$.
\end{proof}

The above three corollaries on the global $r^p$ estimates for different modes can be combined together to yield suitable decay for the $r$-weighted energies of the modes.

\begin{definition}\label{defineenergy-initial:111111}
For any $\reg$ suitably large and $\delta\in (0,\half)$ small, define two energies for the spin $\pm \sfrak$ components respectively:
\begin{align}
\textbf{I}^{\reg,\delta}_{\text{total}, \tb}[\Psipluss]\doteq{}&
F^{(0)}(\reg,5-\delta,\tb,\ellmode{\Psipluss}{\sfrak})\notag\\
&
+F^{(1)}(\reg,4-\delta,\tb,\ellmode{\Psipluss}{\sfrak+1}) +F^{(2)}(\reg,2-\delta,\tb,\ellmode{\Psipluss}{\geq\sfrak+2}),\\
\textbf{I}^{\reg,\delta}_{\text{total}, \tb}[\Psiminuss]\doteq{}&
F^{(2\sfrak)}(\reg,5-\delta,\tb,\ellmode{\Psiminuss}{\sfrak})\notag\\
&
+F^{(2\sfrak+1)}(\reg,4-\delta,\tb,\ellmode{\Psiminuss}{\sfrak+1}) +F^{(2\sfrak+2)}(\reg,2-\delta,\tb,\ellmode{\Psiminuss}{\geq\sfrak+2}).
\end{align}
Similarly define $\textbf{I}^{\reg,\delta}_{\text{total}, \tb}[\Lxi^j\Psipluss]$ and $\textbf{I}^{ \reg,\delta}_{\text{total}, \tb}[\Lxi^j\Psiminuss]$ by simply replacing $\Psipluss$ and $\Psiminuss$ by $\Lxi^j \Psipluss$ and $\Lxi^j \Psiminuss$ respectively everywhere. Finally, define
\begin{align}
\IE{\reg,\delta}{\tb}\doteq\textbf{I}^{\reg,\delta}_{\text{total}, \tb}[\Psipluss]+\textbf{I}^{\reg,\delta}_{\text{total}, \tb}[\Psiminuss].
\end{align}
\end{definition}

\begin{prop}[Energy decay for the modes]
\label{prop:energydecay:full:modes:pmsfrak}
Let $j\in \mathbb{N}$.
For the spin $+\sfrak$ component,  we have for any $p\in [\delta, 2-\delta]$,
\begin{subequations}
\label{eq:ED:S6:modes:Psipluss:012:p02:174}
\begin{align}
\label{eq:ED:S6:modes:Psipluss:1:p02:174}
F^{(0)}(\reg,p,\tb_2,\Lxi^j\ellmode{\Psipluss}{\sfrak+1})
+\norm{\Lxi^j\ellmode{\Psipluss}{\sfrak+1}}^2_{W_{p-3}^{\reg}(\DOC_{\tb_2,\infty})}
\lesssim_{\reg,\delta, j} {}&\langle \tb_2-\tb_1\rangle^{-6-2j+p+C_j\delta}
\textbf{I}^{\reg+\regl(j),\delta}_{\text{total}, \tb_1}[\Psipluss],\\
\label{eq:ED:S6:modes:Psipluss:2:p02:113}
F^{(0)}(\reg,p,\tb_2,\Lxi^j\ellmode{\Psipluss}{\geq \sfrak+2})
+\norm{\Lxi^j\ellmode{\Psipluss}{\geq \sfrak+2}}^2_{W_{p-3}^{\reg}(\DOC_{\tb_2,\infty})}
\lesssim_{\reg,\delta, j} {}&\langle \tb_2-\tb_1\rangle^{-6-2j+p+C_j\delta}
\textbf{I}^{\reg+\regl(j),\delta}_{\text{total}, \tb_1}[\Psipluss],\\
\label{eq:ED:S2:modes:Psipluss:0:p02:493}
F^{(0)}(\reg,p,\tb_2,\Lxi^j\ellmode{\Psipluss}{\sfrak})
+\norm{\Lxi^j\ellmode{\Psipluss}{ \sfrak}}^2_{W_{p-3}^{\reg}(\DOC_{\tb_2,\infty})}
\lesssim_{\reg,\delta, j}  {}&\langle \tb_2-\tb_1\rangle^{-5-2j+p+C_j\delta}
\textbf{I}^{\reg+\regl(j),\delta}_{\text{total}, \tb_1}[\Psipluss].
\end{align}
Meanwhile, for any $p\in [2+\delta, 5-\delta]$,
 \begin{align}
\label{eq:ED:S2:modes:Psipluss:0:p05:493}
F^{(0)}(\reg,p,\tb_2,\Lxi^j\ellmode{\Psipluss}{\sfrak})
+\norm{\Lxi^j\tildePhisHighell{+\sfrak}{\sfrak}}^2_{W_{p-5}^{\reg}(\DOC_{\tb_2,\infty}^{\geq 4M})}
+\norm{\Lxi^j\ellmode{\Psipluss}{ \sfrak}}^2_{W_{-3}^{\reg}(\DOC_{\tb_2,\infty})}
\lesssim_{\reg,\delta, j} {}&\langle \tb_2-\tb_1\rangle^{-5-2j+p+C_j\delta}
\textbf{I}^{\reg+\regl(j),\delta}_{\text{total}, \tb_1}[\Psipluss].
\end{align}
\end{subequations}

For the spin $-\sfrak$ component, we have for any $p\in [\delta, 2-\delta]$,
\begin{subequations}
\label{eq:ED:S6:modes:Psiminuss:01234:p02:389}
\begin{align}
\label{eq:ED:S6:modes:Psiminuss:1:p02:174}
F^{(\sfrak)}(\reg,p,\tb_2,\Lxi^j\ellmode{\Psiminuss}{\sfrak+1})
+\sum_{i=0}^\sfrak\norm{\Lxi^j\ellmode{\PsiminussHigh{i}}{\sfrak+1}}^2_{W_{p-3}^{\reg}(\DOC_{\tb_2,\infty})}
\lesssim_{\reg,\delta, j} {}&\langle \tb_2-\tb_1\rangle^{-6-2\sfrak-2j+p+C_j\delta}
\textbf{I}^{\reg+\regl(j),\delta}_{\text{total}, \tb_1}[\Psiminuss],\\
\label{eq:ED:S6:modes:Psiminuss:2:p02:113}
F^{(\sfrak)}(\reg,p,\tb_2,\Lxi^j\ellmode{\Psiminuss}{\geq \sfrak+2})
+\sum_{i=0}^\sfrak\norm{\Lxi^j\ellmode{\PsiminussHigh{i}}{\geq \sfrak+2}}^2_{W_{p-3}^{\reg}(\DOC_{\tb_2,\infty})}
\lesssim_{\reg,\delta, j} {}&\langle \tb_2-\tb_1\rangle^{-6-2\sfrak-2j+p+C_j\delta}
\textbf{I}^{\reg+\regl(j),\delta}_{\text{total}, \tb_1}[\Psiminuss],\\
\label{eq:ED:S2:modes:Psiminuss:0:p02:493}
F^{(\sfrak)}(\reg,p,\tb_2,\Lxi^j\ellmode{\Psiminuss}{\sfrak})
+\sum_{i=0}^\sfrak\norm{\Lxi^j\ellmode{\PsiminussHigh{i}}{\sfrak}}^2_{W_{p-3}^{\reg}(\DOC_{\tb_2,\infty})}
\lesssim_{\reg,\delta, j} {}&\langle \tb_2-\tb_1\rangle^{-5-2\sfrak-2j+p+C_j\delta}
\textbf{I}^{\reg+\regl(j),\delta}_{\text{total}, \tb_1}[\Psiminuss].
\end{align}
Meanwhile, for any $p\in [2+\delta, 5-\delta]$,
 \begin{align}
\label{eq:ED:S2:modes:Psiminuss:0:p05:493}
\hspace{6.5ex}&\hspace{-6.5ex}
F^{(2\sfrak)}(\reg,p,\tb_2,\Lxi^j\ellmode{\Psiminuss}{\sfrak})
+\norm{\Lxi^j\tildePhisHighell{-\sfrak}{\sfrak}}^2_{W_{p-5}^{\reg}(\DOC_{\tb_2,\infty}^{\geq 4M})}
+\sum_{i=0}^{2\sfrak}\norm{\Lxi^j\ellmode{\PsiminussHigh{i}}{\sfrak}}^2_{W_{-3}^{\reg}(\DOC_{\tb_2,\infty})}
\notag\\
\lesssim_{\reg,\delta, j} {}&\langle \tb_2-\tb_1\rangle^{-5-2j+p+C_j\delta}
\textbf{I}^{\reg+\regl(j),\delta}_{\text{total}, \tb_1}[\Psiminuss].
\end{align}
\end{subequations}
\end{prop}

\begin{proof}
We shall first make use of the global $r^p$ estimates \eqref{eq:rpminussglobal:mode:less2} for $\ellmode{\Phiminuss{i}}{\tilde{\ell}}$ $(\tilde{\ell}\in \{\sfrak,\sfrak+1, \geq \sfrak+2\}, i=0,1,\ldots, 2\sfrak)$ to show some weak decay for the modes of the spin $-\sfrak$ component. An application of Lemma \ref{lem:hierarchyImpliesDecay:73} to the global $r^p$ estimate \eqref{eq:rpminussglobal:mode:less2} for $i\in [\sfrak,2\sfrak]$ yields
\begin{align}
\hspace{5.5ex}&\hspace{-5.5ex}
F^{(i)}(\reg,p,\tb_2,\Lxi^j\ellmode{\Psiminuss}{\tilde{\ell}})
+\sum_{n=0}^i\norm{\Lxi^j\ellmode{\PsiminussHigh{n}}{\tilde{\ell}}}^2_{W_{p-3}^{\reg-\sfrak-1-l(j,s)}(\DOC_{\tb_2,\infty})}
+F^{(i)}(\reg,p,\tb_2,\Lxi^{j+1}\Psiminuss)\notag\\
\lesssim_{\reg,\delta} {}&\langle \tb_2-\tb_1'\rangle^{-2+\delta +p}
\Big(F^{(i)}(\reg+\regl,2-\delta,\tb_1',\Lxi^j\ellmode{\Psiminuss}{\tilde{\ell}})
+F^{(i)}(\reg+\regl,2-\delta,\tb_1',\Lxi^{j+1}\Psiminuss)\Big).
\end{align}
Here, we have made use of the estimate \eqref{eq:BED:Psiminuss:less2} such that we can add $F^{(i)}(\reg,p,\tb_2,\Lxi^{j+1}\Psiminuss)$ to the LHS.
In addition, we have for any $i\in [\sfrak+1,2\sfrak]$,
\begin{align}
\hspace{6ex}&\hspace{-6ex}
F^{(i-1)}(\reg,2-\delta,\tb, \Lxi^j\ellmode{\Psiminuss}{\tilde{\ell}})
+F^{(i-1)}(\reg,2-\delta,\tb_2,\Lxi^{j+1}\Psiminuss)\notag\\
\lesssim_{\reg,\delta} {}&F^{(i)}(\reg+\regl,\delta,\tb,\Lxi^j\ellmode{\Psiminuss}{\tilde{\ell}})
+F^{(i)}(\reg+\regl,\delta,\tb_2,\Lxi^{j+1}\Psiminuss),
\end{align}
hence, we utilize these estimates together to obtain
\begin{align}
\hspace{5.5ex}&\hspace{-5.5ex}
F^{(\sfrak)}(\reg,p,\tb_2,\Lxi^j\ellmode{\Psiminuss}{\tilde{\ell}})
+\sum_{n=0}^\sfrak\norm{\Lxi^j\ellmode{\PsiminussHigh{n}}{\tilde{\ell}}}^2_{W_{p-3}^{\reg-\sfrak-1-l(j,s)}(\DOC_{\tb_2,\infty})}
+F^{(\sfrak)}(\reg,p,\tb_2,\Lxi^{j+1}\Psiminuss)\notag\\
\lesssim_{\reg,\delta} {}&\langle \tb_2-\tb_1'\rangle^{-(2-2\delta)(\sfrak+1) +p-\delta}
\big(F^{(2\sfrak)}(\reg+\regl,2-\delta,\tb_1',\Lxi^j\ellmode{\Psiminuss}{\tilde{\ell}})
+F^{(2\sfrak)}(\reg+\regl,2-\delta,\tb_1',\Lxi^{j+1}\Psiminuss)\big)\notag\\
\lesssim_{\reg,\delta} {}&\langle \tb_2-\tb_1'\rangle^{-(2-2\delta)(\sfrak+1) +p-\delta}
\big(F^{(2\sfrak)}(\reg+\regl,2-\delta,\tb_1',\Lxi^j\ellmode{\Psiminuss}{\tilde{\ell}})
+F^{(2\sfrak)}(\reg+\regl,\delta,\tb_1',\Lxi^{j}\Psiminuss)\big),
\end{align}
where in the second step we have utilized \eqref{eq:749720}.
In a similar manner as proving the general $j$ case in Proposition \ref{prop:BEDC:Phiplusminuss:less2:1}, it holds
\begin{align}
\hspace{7ex}&\hspace{-7ex}
F^{(2\sfrak)}(\reg,2-\delta,\tb_1',\Lxi^j\ellmode{\Psiminuss}{\tilde{\ell}})
+F^{(2\sfrak)}(\reg,\delta,\tb_1',\Lxi^{j}\Psiminuss)\notag\\
\lesssim_{\reg,\delta,j}{}&
\langle \tb_1'-\tb_1\rangle^{-(2-2\delta)j}
(F^{(2\sfrak)}(\reg+\regl(j),2-\delta,\tb_1,\ellmode{\Psiminuss}{\tilde{\ell}})
+F^{(2\sfrak)}(\reg+\regl(j),\delta,\tb_1,\Psiminuss)),
\end{align}
thus combining the above two estimates with $\tb_1'=\tb_1 +\frac{\tb_2-\tb_1}{2}$ then yields for any $p\in [\delta, 2-\delta]$ and $\tilde{\ell}\in \{\sfrak,\sfrak+1, \geq \sfrak+2\}$ that
\begin{align}
\label{eq:ED:modes:Psiminuss:0:0to2:96}
\hspace{7ex}&\hspace{-7ex}
F^{(\sfrak)}(\reg,p,\tb_2,\Lxi^j\ellmode{\Psiminuss}{\tilde{\ell}})
+\sum_{n=0}^\sfrak\norm{\Lxi^j\ellmode{\PsiminussHigh{n}}{\tilde{\ell}}}^2_{W_{p-3}^{\reg-\sfrak-1-l(j,s)}(\DOC_{\tb_2,\infty})}\notag\\
\lesssim_{\reg,\delta, j} {}&\langle \tb_2-\tb_1\rangle^{-(2-2\delta)(\sfrak+1+j) +p-\delta}
(F^{(2\sfrak)}(\reg+\regl(j),2-\delta,\tb_1,\ellmode{\Psiminuss}{\tilde{\ell}})
+F^{(2\sfrak)}(\reg+\regl(j),\delta,\tb_1,\Psiminuss)).
\end{align}
By the same argument, we have for any $p\in [\delta, 2-\delta]$ and $\tilde{\ell}\in \{\sfrak,\sfrak+1, \geq \sfrak+2\}$ that
\begin{align}
\label{eq:ED:modes:Psipluss:0:0to296}
\hspace{7ex}&\hspace{-7ex}
F^{(0)}(\reg,p,\tb_2,\Lxi^j\ellmode{\Psipluss}{\tilde{\ell}})
+\norm{\Lxi^j\ellmode{\Psipluss}{\tilde{\ell}}}^2_{W_{p-3}^{\reg}(\DOC_{\tb_2,\infty})}\notag\\
\lesssim_{\reg,\delta, j} {}&\langle \tb_2-\tb_1\rangle^{-(2-2\delta)(1+j) +p-\delta}\big(F^{(0)}(\reg+\regl(j),2-\delta,\tb_1,\ellmode{\Psipluss}{\tilde{\ell}}) +F^{(0)}(\reg+\regl(j),\delta,\tb_1,\Psipluss) \big).
\end{align}

Next, we consider further decay for the $\sfrak$ mode of the spin $\pm \sfrak$ components. Recall the global $r^p$ estimate \eqref{eq:rp:modes:Psipms:0:p24:56} for the $\sfrak$ mode. Consider the case for the spin $+\sfrak$ component. The estimate \eqref{eq:ED:modes:Psipluss:0:0to296} just proven yields that the last three terms on the RHS of \eqref{eq:rp:modes:Psipluss:0:p24:56} are bounded by $\langle \tb_1-\tb_1'\rangle^{-(2-2\delta)}\big(F^{(0)}(\reg+\regl(j),2-\delta,\tb_1',\Lxi^j\ellmode{\Psipluss}{\sfrak}) +F^{(0)}(\reg+\regl(j),2-\delta,\tb_1',\Lxi^j\Psipluss) \big)$ where $\tb_1'\in [\tb_0,\tb_1]$ being arbitrary, thus an application of Lemma \ref{lem:hierarchyImpliesDecay:73} to the estimate implies for any $p\in [2+\delta, 4-\delta]$,
\begin{align}
\label{eq:ED:S1:modes:Psipluss:0:p24:56}
\hspace{7ex}&\hspace{-7ex}
F^{(0)}(\reg,p,\tb_2,\Lxi^j\ellmode{\Psipluss}{\sfrak})
+\norm{\Lxi^j\tildePhisHighell{+\sfrak}{\sfrak}}^2_{W_{p-5}^{\reg}(\DOC_{\tb_2,\infty}^{\geq 4M})}
+F^{(0)}(\reg,p-2,\tb_2,\Lxi^j\Psipluss)
\notag\\
\lesssim_{\reg,\delta, j} {}&\langle \tb_2-\tb_1\rangle^{-4+\delta +p}\big(F^{(0)}(\reg+\regl,4-\delta,\tb_1,\Lxi^j\ellmode{\Psipluss}{\sfrak})
  +F^{(0)}(\reg+\regl,2-\delta,\tb_1,\Lxi^j\Psipluss) \big).
\end{align}
Further, because of
\begin{align}
\hspace{5.5ex}&\hspace{-5.5ex}
F^{(0)}(\reg,4-\delta,\tb,\Lxi^{j+1}\ellmode{\Psipluss}{\sfrak})
+F^{(0)}(\reg,2-\delta,\tb,\Lxi^{j+1}\Psipluss)\notag\\
\lesssim_{\reg,\delta}{}& F^{(0)}(\reg,2+\delta,\tb,\Lxi^j\ellmode{\Psipluss}{\sfrak})
+F^{(0)}(\reg,\delta,\tb,\Lxi^j\Psipluss),
\end{align}
by repeating the proof for the general $j$ case, we obtain for any $p\in [2+\delta, 4-\delta]$,
\begin{align}
\label{eq:ED:S2:modes:Psipluss:0:p24:56}
\hspace{7ex}&\hspace{-7ex}
F^{(0)}(\reg,p,\tb_2,\Lxi^j\ellmode{\Psipluss}{\sfrak})
+F^{(0)}(\reg,p-2,\tb_2,\Lxi^j\Psipluss)
+\norm{\Lxi^j\tildePhisHighell{+\sfrak}{\sfrak}}^2_{W_{p-5}^{\reg}(\DOC_{\tb_2,\infty}^{\geq 4M})}
\notag\\
\lesssim_{\reg,\delta, j} {}&\langle \tb_2-\tb_1\rangle^{-(2-2\delta)(1+j) +p-2+C_j\delta}\big(F^{(0)}(\reg+\regl(j),4-\delta,\tb_1,\ellmode{\Psipluss}{\sfrak}) +F^{(0)}(\reg+\regl(j),2-\delta,\tb_1,\Psipluss) \big).
\end{align}
By definition \eqref{ansatz:tildePhisHigh:ellmode} of $\tildePhisHighell{+\sfrak}{\sfrak}$, there exists a $\regl>0$ such that
\begin{align}
F^{(0)}(\reg,2+\delta,\tb,\ellmode{\Psipluss}{\sfrak})
+F^{(0)}(\reg,\delta,\tb,\Psipluss)
\gtrsim_{\delta,\reg} F^{(0)}(\reg-\regl,2-\delta,\tb,\ellmode{\Psipluss}{\tilde{\ell}}) +F^{(0)}(\reg-\regl,\delta,\tb,\Psipluss) ,
\end{align}
therefore, the above  estimate \eqref{eq:ED:S2:modes:Psipluss:0:p24:56} together with the previously proven estimate \eqref{eq:ED:modes:Psipluss:0:0to296} with $\tilde{\ell}=\sfrak$ implies for any $p\in [\delta, 2-\delta]$,
\begin{align}
\label{eq:ED:S3:modes:Psipluss:0:0to4:96}
\hspace{7ex}&\hspace{-7ex}
F^{(0)}(\reg,p,\tb_2,\Lxi^j\ellmode{\Psipluss}{\sfrak})
+\norm{\Lxi^j\ellmode{\Psipluss}{\sfrak}}^2_{W_{p-3}^{\reg}(\DOC_{\tb_2,\infty}^{\geq 4M})}\notag\\
\lesssim_{\reg,\delta, j} {}&\langle \tb_2-\tb_1\rangle^{-(2-2\delta)(2+j) +p+C_j\delta}\big(F^{(0)}(\reg+\regl(j),4-\delta,\tb_1,\ellmode{\Psipluss}{\sfrak})+F^{(0)}(\reg+\regl(j),2-\delta,\tb_1,\Psipluss) \big).
\end{align}
Following the same argument, we have for the $\sfrak$ mode of the spin $-\sfrak$ component that for any $p\in [2+\delta, 4-\delta]$,
\begin{align}
\label{eq:ED:S2:modes:Psiminuss:0:p24:56}
\hspace{7ex}&\hspace{-7ex}
F^{(2\sfrak)}(\reg,p,\tb_2,\Lxi^j\ellmode{\Psiminuss}{\sfrak})
+\norm{\Lxi^j\tildePhisHighell{-\sfrak}{\sfrak}}^2_{W_{p-5}^{\reg}(\DOC_{\tb_2,\infty}^{\geq 4M})}
+F^{(2\sfrak)}(\reg,p-2,\tb_2,\Psiminuss)
\notag\\
\lesssim_{\reg,\delta, j} {}&\langle \tb_2-\tb_1\rangle^{-(2-2\delta)(1+j) +p-2+C_j\delta}\big(F^{(2\sfrak)}(\reg+\regl(j),4-\delta,\tb_1,\ellmode{\Psiminuss}{\sfrak})+F^{(2\sfrak)}(\reg+\regl(j),2-\delta,\tb_1,\Psiminuss) \big),
\end{align}
and, together with \eqref{eq:ED:modes:Psiminuss:0:0to2:96}, we have for any $p\in [\delta, 2-\delta]$ and $i\in [\sfrak, 2\sfrak]$,
\begin{align}
\label{eq:ED:S3:modes:Psiminuss:0:0to4:96}
\hspace{7ex}&\hspace{-7ex}
F^{(i)}(\reg,p,\tb_2,\Lxi^j\ellmode{\Psiminuss}{\sfrak})
+\sum_{i'=0}^{i}\norm{\Lxi^j\ellmode{\PsiminussHigh{i}}{\sfrak}}^2_{W_{-1-\delta}^{\reg}(\DOC_{\tb_2,\infty})}\notag\\
\lesssim _{\reg,\delta, j}{}&\langle \tb_2-\tb_1\rangle^{-(2-2\delta)(2+2\sfrak-i+j) +p+C_j\delta}\big(F^{(2\sfrak)}(\reg+\regl(j),4-\delta,\tb_1,\ellmode{\Psiminuss}{\sfrak}) +F^{(2\sfrak)}(\reg+\regl(j),2-\delta,\tb_1,\Psiminuss) \big).
\end{align}

Turn to the $\sfrak+1$ and $\geq \sfrak+2$ modes of the spin $\pm \sfrak$ components. Let $\tilde{\ell}\in \{\sfrak+1, \geq \sfrak+2\}$. In the estimate \eqref{eq:rp:modes:Psipluss:1:p02} for $\tilde{\ell}\in \{\sfrak+1, \geq \sfrak+2\}$, the last two terms on the RHS are bounded by $\langle \tb_1-\tb_1'\rangle^{-(2-2\delta) +p-\delta}\big(F^{(0)}(\reg+\regl(j),2-\delta,\tb_1',\Lxi^j\ellmode{\Psipluss}{\tilde{\ell}}) +F^{(0)}(\reg+\regl(j),2-\delta,\tb_1',\Lxi^j\Psipluss) \big) $ in view of the proven estimate \eqref{eq:ED:modes:Psipluss:0:0to296}, hence we achieve from the estimate \eqref{eq:rp:modes:Psipluss:1:p02} that for any $p\in [\delta, 2-\delta]$,
\begin{align}
\label{eq:ED:S1:modes:Psipluss:1and2:p02:75}
\hspace{7ex}&\hspace{-7ex}
F^{(1)}(\reg,p,\tb_2,\Lxi^j\ellmode{\Psipluss}{\tilde{\ell}})
+\norm{\Lxi^j\ellmode{\hatPhiplussHigh{1}}{\tilde{\ell}}}^2_{W_{p-3}^{\reg}(\DOC_{\tb_2,\infty}^{\geq 4M})}
+\norm{\Lxi^j\ellmode{\Psipluss}{\tilde{\ell}}}^2_{W_{-3-\delta}^{\reg}(\DOC_{\tb_2,\infty}^{\geq 4M})}
\notag\\
\lesssim_{\reg,\delta, j} {}&\langle \tb_2-\tb_1\rangle^{-(2-2\delta)(1+j) +p+C_j\delta}
\big(F^{(1)}(\reg+\regl(j),2-\delta,\tb_1,\ellmode{\Psipluss}{\tilde{\ell}})   +F^{(0)}(\reg+\regl(j),2-\delta,\tb_1,\Psipluss)\big).
\end{align}
We can also add freely $F^{(0)}(\reg,\delta,\tb_2,\Lxi^j\Psipluss)$ to the LHS because of the estimate \eqref{eq:BED:Psipluss:less2}.
Since the relation $F^{(1)}(\reg,\delta,\tb,\Lxi^j\ellmode{\Psipluss}{\tilde{\ell}})+F^{(0)}(\reg,\delta,\tb_2,\Lxi^j\Psipluss)
\gtrsim F^{(0)}(\reg-\regl,2-\delta,\tb,\Lxi^j\ellmode{\Psipluss}{\tilde{\ell}})+F^{(0)}(\reg-\regl,\delta,\tb_2,\Lxi^j\Psipluss)
$ holds true,
this energy decay estimate and the decay estimate \eqref{eq:ED:modes:Psipluss:0:0to296} together imply that for any $\tilde{\ell}\in \{\sfrak+1, \geq \sfrak+2\}$ and $p\in [\delta, 2-\delta]$,
\begin{align}
\label{eq:ED:S2:modes:Psipluss:1and2:p02:109}
\hspace{7ex}&\hspace{-7ex}
F^{(0)}(\reg,p,\tb_2,\Lxi^j\ellmode{\Psipluss}{\tilde{\ell}})
+\norm{\Lxi^j\ellmode{\Psipluss}{\tilde{\ell}}}^2_{W_{p-3}^{\reg}(\DOC_{\tb_2,\infty})}
\notag\\
\lesssim_{\reg,\delta, j} {}&\langle \tb_2-\tb_1\rangle^{-(2-2\delta)(2+j) +p+C_j\delta}
\Big(
F^{(1)}(\reg+\regl(j),2-\delta,\tb_1,\ellmode{\Psipluss}{\tilde{\ell}})
+F^{(0)}(\reg+\regl(j),2-\delta,\tb_1,\Psipluss)\Big).
\end{align}

One can see from the above estimates \eqref{eq:ED:S3:modes:Psipluss:0:0to4:96} and \eqref{eq:ED:S2:modes:Psipluss:1and2:p02:109} that we have achieved the same energy decay for the energy $F^{(0)}(\reg,p,\tb_2,\Lxi^j\ellmode{\Psipluss}{\tilde{\ell}})$ for $\tilde\ell\in \{\sfrak,\sfrak+1, \geq \sfrak+2\}$. In particular, using the estimates \eqref{eq:ED:S3:modes:Psipluss:0:0to4:96} with $p=2-\delta$ and \eqref{eq:ED:S1:modes:Psipluss:1and2:p02:75} with $p=\delta$ and adding them together, we have
\begin{align}
\label{eq:ED:S2:modes:Psipluss:01and2:137}
\hspace{7ex}&\hspace{-7ex}
F^{(0)}(\reg,2-\delta,\tb_2,\Lxi^j\ellmode{\Psipluss}{\sfrak})
+F^{(1)}(\reg,\delta,\tb_2,\ellmode{\Lxi^j\Psipluss}{\geq\sfrak+1})
+F^{(0)}(\reg,2-\delta,\tb_2,\Lxi^j{\Psipluss})\notag\\
\lesssim_{\reg,\delta, j}{}& \langle \tb_2-\tb_1\rangle^{-(2-2\delta)(1+j) +C_j\delta}
\Big(F^{(0)}(\reg+\regl(j),4-\delta,\tb_1,\Lxi^j\ellmode{\Psipluss}{\sfrak})\notag\\
&
+\sum_{\tilde{\ell}\in \{\sfrak+1, \geq \sfrak+2\}}F^{(1)}(\reg+\regl(j),2-\delta,\tb_1,\ellmode{\Psipluss}{\tilde{\ell}})
+F^{(0)}(\reg+\regl(j),2-\delta,\tb_1,\Psipluss)\Big).
\end{align}
The reason that we can add $F^{(0)}(\reg,2-\delta,\tb_2,\Lxi^j{\Psipluss})$ to the LHS is by a simple fact that $F^{(0)}(\reg,2-\delta,\tb_2,\Lxi^j{\Psipluss})\lesssim F^{(0)}(\reg+\regl,2-\delta,\tb_2,\Lxi^j\ellmode{\Psipluss}{\sfrak})
+F^{(1)}(\reg+\regl,\delta,\tb_2,\ellmode{\Lxi^j\Psipluss}{\geq\sfrak+1})$.

Our next goal is to further refine these energy decay for the $\sfrak$, $\sfrak+1$ and $\geq \sfrak+2$ modes in different ways.

For $\geq \sfrak+2$ modes, we utilize the global $r^p$ estimate \eqref{eq:rp:modes:Psipluss:2:p02} with $p\in [\delta, 2-\delta]$. We utilize the estimates \eqref{eq:BED:Psipluss:less2} for the last fifth and fourth terms, \eqref{eq:ED:S1:modes:Psipluss:1and2:p02:75} for the last third and second terms and \eqref{eq:ED:S1:modes:Psipluss:0:p24:56} for the last term  on the RHS and bound these last five terms by
\begin{align}
C_{\reg,\delta}\tb_1^{-2+\delta +p}
\big(&F^{(0)}(\reg+\regl,4-\delta,\tb_1,\Lxi^j\ellmode{\Psipluss}{\sfrak})
+F^{(1)}(\reg+\regl,2-\delta,\tb_1,\Lxi^j\ellmode{\Psipluss}{\sfrak+1}) \notag\\
&
+F^{(1)}(\reg+\regl,2-\delta,\tb_1,\Lxi^j\ellmode{\Psipluss}{\geq \sfrak+2})
+F^{(0)}(\reg+\regl,2-\delta,\tb_1,\Lxi^j\Psipluss)\big).
\end{align}
Plugging this estimate back to the global $r^p$ estimate \eqref{eq:rp:modes:Psipluss:2:p02}, and using the estimate \eqref{eq:ED:S2:modes:Psipluss:01and2:137}, we conclude for any $p\in [\delta, 2-\delta]$,
\begin{align}
\label{eq:ED:S2:modes:Psipluss:2:p24:48}
\hspace{7ex}&\hspace{-7ex}
F^{(2)}(\reg,p,\tb_2,\Lxi^j\ellmode{\Psipluss}{\geq\sfrak+2})
+F^{(0)}(\reg,2-\delta,\tb_2,\Lxi^j\ellmode{\Psipluss}{\sfrak})
+F^{(1)}(\reg,\delta,\tb_2,\ellmode{\Lxi^j\Psipluss}{\geq\sfrak+1})+\norm{\Lxi^j\ellmode{\hatPhiplussHigh{2}}{\geq\sfrak+2}}^2_{W_{p-3}^{\reg}(\DOC_{\tb_2,\infty}^{\geq 4M})}
\notag\\
\lesssim_{\reg,\delta, j} {}&\langle \tb_2-\tb_1\rangle^{-(2-2\delta)(1+j) +p+C_j\delta}
\textbf{I}^{\reg+\regl(j)}_{\text{total}, \tb_1}[\Psipluss].
\end{align}
 Since there exists a universal constant $\regl$ such that
 \begin{align}
 &F^{(2)}(\reg,\delta,\tb,\Lxi^j\ellmode{\Psipluss}{\geq\sfrak+2})
+F^{(0)}(\reg,2-\delta,\tb,\Lxi^j\ellmode{\Psipluss}{\sfrak})
+F^{(1)}(\reg,\delta,\tb,\ellmode{\Lxi^j\Psipluss}{\geq\sfrak+1})\notag\\
&\gtrsim_{\reg,\delta} F^{(1)}(\reg-\regl,2-\delta,\tb,\Lxi^j\ellmode{\Psipluss}{\geq\sfrak+2})
+F^{(0)}(\reg-\regl,2-\delta,\tb,\Lxi^j\ellmode{\Psipluss}{\sfrak})
+F^{(1)}(\reg-\regl,\delta,\tb,\ellmode{\Lxi^j\Psipluss}{\geq\sfrak+1}),
\end{align}
then, by using the above energy decay estimate \eqref{eq:ED:S2:modes:Psipluss:2:p24:48} and the estimate \eqref{eq:ED:S2:modes:Psipluss:1and2:p02:109} with $\tilde{\ell}$ taking $\geq \sfrak+2$, we arrive at the estimate \eqref{eq:ED:S6:modes:Psipluss:2:p02:113}.

We proceed to the $\sfrak+1$ mode. In the global $r^p$ estimate \eqref{eq:rp:modes:Psipluss:1:p24:56} for $p\in [2+\delta, 4-\delta]$, we use again the estimates \eqref{eq:BED:Psipluss:less2}, \eqref{eq:ED:S1:modes:Psipluss:1and2:p02:75} and \eqref{eq:ED:S1:modes:Psipluss:0:p24:56} and find that the last five terms are bounded by
\begin{align*}
\tb_1^{-2+\delta +p}
\big(&F^{(0)}(\reg+\regl(j),4-\delta,\tb_1,\Lxi^j\ellmode{\Psipluss}{\sfrak})
+F^{(1)}(\reg+\regl(j),2-\delta,\tb_1,\Lxi^j\ellmode{\Psipluss}{\sfrak+1}) \notag\\
&
+F^{(1)}(\reg+\regl(j),2-\delta,\tb_1,\Lxi^j\ellmode{\Psipluss}{\geq \sfrak+2})
+F^{(0)}(\reg+\regl(j),2-\delta,\tb_1,\Lxi^j\Psipluss)\big).
\end{align*}
Hence, the same argument applies and yields for any $p\in [2+\delta, 4-\delta]$,
\begin{align}
\label{eq:ED:S2:modes:Psipluss:1:p24:48}
\hspace{7ex}&\hspace{-7ex}
F^{(1)}(\reg,p,\tb_2,\Lxi^j\ellmode{\Psipluss}{\sfrak+1})
+F^{(0)}(\reg,2+\delta,\tb_2,\Lxi^j\ellmode{\Psipluss}{\sfrak})
+F^{(1)}(\reg,\delta,\tb_2,\ellmode{\Lxi^j\Psipluss}{\geq\sfrak+1})
+\norm{\Lxi^j\ellmode{\hatPhiplussHigh{1}}{\sfrak+1}}^2_{W_{p-3}^{\reg}(\DOC_{\tb_2,\infty}^{\geq 4M})}
\notag\\
\lesssim_{\reg,\delta, j} {}&\langle \tb_2-\tb_1\rangle^{-(2-2\delta)(1+j) +p-2+C_j\delta}
\textbf{I}^{\reg+\regl(j),\delta}_{\text{total}, \tb_1}[\Psipluss].
\end{align}
Again, this estimate and the estimate \eqref{eq:ED:S2:modes:Psipluss:1and2:p02:109} with $\tilde{\ell}=\sfrak+1$ yields the estimate \eqref{eq:ED:S6:modes:Psipluss:1:p02:174}.

Last, we consider the $\sfrak$ mode. We utilize the global $r^p$ estimate \eqref{eq:rp:modes:Psipluss:0:p45:56} with $p\in [4,5-\delta)$. Using the estimates \eqref{eq:BED:Psipluss:less2} for the last fourth, third and second terms and \eqref{eq:ED:S2:modes:Psipluss:1and2:p02:109} for the last term, the last four terms on the RHS are bounded by
\begin{align*}
\tb_1^{-2+C\delta +p-4}
\Big(&F^{(0)}(\reg+\regl(j),4-\delta,\tb_1,\Lxi^j\ellmode{\Psipluss}{\sfrak})
+F^{(1)}(\reg+\regl(j),2-\delta,\tb_1,\Lxi^j\ellmode{\Psipluss}{\sfrak+1}) \notag\\
&
+F^{(1)}(\reg+\regl(j),2-\delta,\tb_1,\Lxi^j\ellmode{\Psipluss}{\geq \sfrak+2})
+F^{(0)}(\reg+\regl(j),2-\delta,\tb_1,\Lxi^j\Psipluss)\Big),
\end{align*}
therefore, we obtain for any $p\in [4, 5-\delta]$ that
\begin{align}
\label{eq:ED:S2:modes:Psipluss:0:p45:48}
\hspace{7ex}&\hspace{-7ex}
F^{(0)}(\reg,p,\tb_2,\Lxi^j\ellmode{\Psipluss}{\sfrak})
+F^{(0)}(\reg,2-\delta,\tb_2,\Lxi^j{\Psipluss})
+\norm{\tildePhisHighell{+\sfrak}{\sfrak}}^2_{W_{p-5}^{\reg}(\DOC_{\tb_2,\infty}^{\geq 4M})}
\notag\\
\lesssim_{\reg,\delta, j} {}&\langle \tb_2-\tb_1\rangle^{-5-(2-2\delta)j+C\delta +p}
\textbf{I}^{\reg+\regl(j),\delta}_{\text{total}, \tb_1}[\Psipluss],
\end{align}
where we have utilized the estimate \eqref{eq:ED:S2:modes:Psipluss:01and2:137} to include the term $F^{(0)}(\reg,2-\delta,\tb_2,\Lxi^j{\Psipluss})$ on the LHS.
Together with the estimate \eqref{eq:ED:S1:modes:Psipluss:0:p24:56}, we achieve the estimate \eqref{eq:ED:S2:modes:Psipluss:0:p05:493} for any $p\in [2+\delta, 5-\delta]$ and the estimate \eqref{eq:ED:S2:modes:Psipluss:0:p02:493} for any $p\in [\delta, 2-\delta]$.

In the end, we consider the modes of the spin $-\sfrak$ component. Note from Proposition \ref{eq:wave:hatPhisHighi:an:ellmode} that the scalar $\ellmode{\hatPhiminuss{2\sfrak+i}}{\tilde{\ell}}$ satisfies the same wave equation as the one of the scalar $\ellmode{\hatPhiplussHigh{i}}{\tilde{\ell}}$ and from Proposition \ref{prop:wavesys:tildePhisHighi:ellmode} that the scalar
$\tildePhisHighell{-\sfrak}{\ell}$ satisfies the same equation as  the one of scalar $\tildePhisHighell{+\sfrak}{\ell}$. As a consequence,  we have analogous estimates for the modes of the spin $-\sfrak$ component as the estimates \eqref{eq:ED:S6:modes:Psipluss:012:p02:174} for the modes of the spin $+\sfrak$ component. That is, we have for any $\tb_1>\tb_1'\geq\tb_0$ and any $p\in [\delta, 2-\delta]$ that
\begin{align*}
F^{(2\sfrak)}(\reg,p,\tb_1,\Lxi^j\ellmode{\Psiminuss}{\sfrak+1})
\lesssim_{\reg,\delta, j} {}&\langle \tb_1-\tb_1'\rangle^{-6-2j+p+C_j\delta}
\textbf{I}^{\reg+\regl(j),\delta}_{\text{total}, \tb_1'}[\Psiminuss]\\
F^{(2\sfrak)}(\reg,p,\tb_1,\Lxi^j\ellmode{\Psiminuss}{\geq \sfrak+2})
\lesssim_{\reg,\delta, j} {}&\langle \tb_1-\tb_1'\rangle^{-6-2j+p+C_j\delta}
\textbf{I}^{\reg+\regl(j),\delta}_{\text{total}, \tb_1'}[\Psiminuss]\\
F^{(2\sfrak)}(\reg,p,\tb_1,\Lxi^j\ellmode{\Psiminuss}{\sfrak})
\lesssim_{\reg,\delta, j} {}&\langle \tb_1-\tb_1'\rangle^{-5-2j+p+C_j\delta}
\textbf{I}^{\reg+\regl(j),\delta}_{\text{total}, \tb_1'}[\Psiminuss],
\end{align*}
and the estimate \eqref{eq:ED:S2:modes:Psiminuss:0:p05:493} for any $p\in [2+\delta, 5-\delta]$ holds. We take  $p=2-\delta$ in the above estimates to attain energy decay for $F^{(2\sfrak)}(\reg,p,\tb_1,\Lxi^j\ellmode{\Psiminuss}{\tilde{\ell}})$ for $\tilde{\ell}\in \{\sfrak,\sfrak+1,\geq \sfrak+2\}$
(specifically, $\langle \tb_1-\tb_1'\rangle^{-4-2j+C_j\delta}\textbf{I}^{\reg+\regl(j),\delta}_{\text{total}, \tb_1'}[\Psiminuss]$ for $\tilde{\ell}\in \{\sfrak+1,\geq \sfrak+2\}$ and $\langle \tb_1-\tb_1'\rangle^{-3-2j+C_j\delta}\textbf{I}^{\reg+\regl(j),\delta}_{\text{total}, \tb_1'}[\Psiminuss]$ for $\tilde{\ell}=\sfrak$), and taking $p=\delta$ in the above estimates and summing up together yields $\langle \tb_1-\tb_1'\rangle^{-5-2j+C_j\delta}
\textbf{I}^{\reg+\regl(j),\delta}_{\text{total}, \tb_1'}[\Psiminuss]$ decay for $F^{(2\sfrak)}(\reg,\delta,\tb_1,\Psiminuss)$, thus we plug  these two energy decay estimates back to \eqref{eq:ED:modes:Psiminuss:0:0to2:96} with $\tb_1=\frac{\tb_2+\tb_1'}{2}$ to conclude the rest estimates in \eqref{eq:ED:S6:modes:Psiminuss:01234:p02:389}.
\end{proof}

\subsection{Almost sharp decay for the spin $\pm \sfrak$ components}
\label{subsect:APL:general}
We derive the almost sharp pointwise decay estimates for the spin $\pm \sfrak$ components in this subsection.

To begin with, we make use of the energy decay estimates in Proposition \ref{prop:energydecay:full:modes:pmsfrak} to derive some weaker (than almost sharp) pointwise decay for the spin $\pm \sfrak$ components.

\begin{cor}
\label{cor:PWD:pms:012}
For the spin $+\sfrak$ component, we have
\begin{subequations}
\label{eq:PWD:pluss:012:428}
\begin{align}
\label{eq:PWD:pluss:0:428}
\absCDeri{\Lxi^j(r^{-1}\ellmode{\Psipluss}{\sfrak})}{\reg}\lesssim_{j,\reg,\delta} v^{-1}\tb^{-2-j+C_j \delta} \textbf{I}^{\reg+\regl(j),\delta}_{\text{total}, \tb_0}[\Psipluss],\\
\label{eq:PWD:pluss:12:428}
\absCDeri{\Lxi^j(r^{-1}\ellmode{\Psipluss}{\geq \sfrak+1})}{\reg}\lesssim_{j,\reg,\delta} v^{-1}\tb^{-\frac{5}{2}-j+C_j \delta} \textbf{I}^{\reg+\regl(j),\delta}_{\text{total}, \tb_0}[\Psipluss].
\end{align}
\end{subequations}

For the spin $-\sfrak$ component, we have
\begin{subequations}
\label{eq:PWD:minuss:012:428}
\begin{align}
\label{eq:PWD:minuss:0:428}
\sum_{i=0}^{\sfrak}\absCDeri{\Lxi^j(r^{-1}\ellmode{\PsiminussHigh{i}}{\sfrak})}{\reg}\lesssim_{j,\reg,\delta} v^{-1}\tb^{-2-\sfrak-j+C_j \delta} \textbf{I}^{\reg+\regl(j),\delta}_{\text{total}, \tb_0}[\Psiminuss],\\
\label{eq:PWD:minuss:12:428}
\sum_{i=0}^{\sfrak}\absCDeri{\Lxi^j(r^{-1}\ellmode{\PsiminussHigh{i}}{\geq \sfrak+1})}{\reg}\lesssim_{j,\reg,\delta} v^{-1}\tb^{-\frac{5}{2}-\sfrak-j+C_j \delta} \textbf{I}^{\reg+\regl(j),\delta}_{\text{total}, \tb_0}[\Psiminuss].
\end{align}
\end{subequations}

Further, we have for $\rb\geq 3M$,
\begin{subequations}
\begin{align}
\label{eq:PWD:pluss:0:tildePhi:428}
\absCDeri{\Lxi^j(\tildePhisHighell{+\sfrak}{\sfrak})}{\reg}\lesssim_{j,\reg,\delta}{}& v^{-1+C_j \delta}\tb^{-j} \textbf{I}^{\reg+\regl(j),\delta}_{\text{total}, \tb_0}[\Psipluss],\\
\label{eq:PWD:minuss:0:tildePhi:428}
\absCDeri{\Lxi^j(\tildePhisHighell{-\sfrak}{\sfrak})}{\reg}\lesssim_{j,\reg,\delta}{}& v^{-1+C_j \delta}\tb^{-j} \textbf{I}^{\reg+\regl(j),\delta}_{\text{total}, \tb_0}[\Psiminuss].
\end{align}
\end{subequations}
\end{cor}

\begin{proof}
Note from equations \eqref{eq:ED:S6:modes:Psipluss:012:p02:174} that for  $\sfrak+1$ and $\geq \sfrak+2$ modes, the energies and the spacetime integrals have the same decay, hence we arrive at the same decay estimates for $\geq \sfrak+1$. Then,  an application of the Sobolev inequality \eqref{eq:Sobolev:2} with $\alpha=\delta$ yields
\begin{align}
\absCDeri{\Lxi^j\ellmode{\Psipluss}{\geq \sfrak+1}}{\reg}\lesssim_{j,\reg,\delta} \tb^{-\frac{5}{2}-j+C_j \delta} \textbf{I}^{\reg+\regl(j),\delta}_{\text{total}, \tb_0}[\Psipluss],
\end{align}
and applying the other Sobolev inequality \eqref{eq:Sobolev:3} yields
\begin{align}
\absCDeri{r^{-1}\Lxi^j\ellmode{\Psipluss}{\geq \sfrak+1}}{\reg}\lesssim_{j,\reg,\delta} \tb^{-\frac{7}{2}-j+C_j \delta} \textbf{I}^{\reg+\regl(j),\delta}_{\text{total}, \tb_0}[\Psipluss].
\end{align}
The above two estimates  then prove the pointwise decay estimate \eqref{eq:PWD:pluss:12:428} in regions $\{r\geq \tb\}$ and $\{r\leq \tb\}$ respectively. The rest estimates are proven in the same manner and we omit the proof.
\end{proof}

In the following two subsubsections, we will refine these pointwise decay estimates \eqref{eq:PWD:pluss:012:428} and \eqref{eq:PWD:minuss:012:428} in the exterior region $\{\rb\geq \tb\}$ and the interior region $\{\rb\leq \tb\}$, respectively, such that the decay estimates for the $\sfrak$ mode are close to the sharp decay (i.e. the Price's law decay), and the decay of the $\sfrak+1$ and $\geq \sfrak+2$ modes are faster than the Price's law for the entire spin $\pm \sfrak$ components but slower than the expected Price's law of the modes themselves.

We state the almost sharp decay estimates for the spin $\pm \sfrak$ components here.

\begin{prop}[Almost sharp pointwise decay estimates for the spin $\pm \sfrak$ components]
\label{prop:AP:pms:012} Let $j,\reg\in \mathbb{N}$.
For the spin $+\sfrak$ component, we have
\begin{subequations}
\label{eq:AP:pluss:012:428}
\begin{align}
\label{eq:AP:pluss:0:428}
\absCDeri{\Lxi^j(r^{-2\sfrak}\ellmode{\psipluss}{\sfrak})}{\reg}\lesssim_{j,\reg,\delta} v^{-1-2\sfrak}\tb^{-2-j+C_j \delta}  \IE{\reg+\regl(j),\delta}{\tb_0},
\end{align}
and for $\geq \sfrak+1$ modes,
\begin{align}
\label{eq:AP:pluss:12:428}
\absCDeri{\Lxi^j(r^{-2\sfrak}\ellmode{\psipluss}{\geq \sfrak+1})}{\reg}\lesssim_{j,\reg,\delta} v^{-1-2\sfrak}\tb^{-\frac{5}{2}-j+C_j \delta}   \IE{\reg+\regl(j),\delta}{\tb_0}.
\end{align}
\end{subequations}

For the spin $-\sfrak$ component, we have for the $\sfrak$ mode that
\begin{subequations}
\label{eq:AP:minuss:012:428}
\begin{align}
\label{eq:AP:minuss:0:428}
\absCDeri{\Lxi^j(\ellmode{\psiminuss}{\sfrak})}{\reg}\lesssim_{j,\reg,\delta} v^{-1}\tb^{-2-2\sfrak-j+C_j \delta}   \IE{\reg+\regl(j),\delta}{\tb_0},
\end{align}
and for $\geq \sfrak+1$ modes that
\begin{align}
\label{eq:AP:minuss:12:428}
\absCDeri{\Lxi^j(\ellmode{\psiminuss}{\geq \sfrak+1})}{\reg}\lesssim_{j,\reg,\delta} v^{-1}\tb^{-\frac{5}{2}-2\sfrak-j+C_j \delta}   \IE{\reg+\regl(j),\delta}{\tb_0}.
\end{align}
\end{subequations}

Moreover, in the interior region $\{\rb\leq\tb\}$, we have for $\prb\ellmode{\psiminuss}{\sfrak}$, the radial derivative of the $\sfrak$ mode of the spin $-\sfrak$ component, the following decay:
\begin{align}
\label{eq:AP:prbminuss:0:428}
\abs{\Lxi^j(\mu r \prb)^\reg(\prb\ellmode{\psiminuss}{\sfrak})}\lesssim_{j,\reg,\delta} v^{-1}\tb^{-3-2\sfrak-j+C_j \delta}  \IE{\reg+\regl(j),\delta}{\tb_0}.
\end{align}
\end{prop}

This proposition will be proven in the following two subsubsections in the exterior and interior regions respectively. We shall remark that in both regions, the TSI in Section \ref{sect:TSI} will be of crucial importance in deriving the decay estimates for one spin component from the ones of the other spin component, an observation been already made in \cite{MaZhang21PriceSchw}.

\subsubsection{Proof of Proposition \ref{prop:AP:pms:012} in the exterior region $\{\rb\geq \tb\}$}
\label{sect:AP:pms:ext:83}

Note first that in the exterior region $\{\rb\geq \tb\}$, it holds $r\gtrsim v$, hence the estimates \eqref{eq:AP:pluss:012:428}   for the spin $+\sfrak$ component are valid.

It remains to show the estimates \eqref{eq:AP:minuss:012:428}    for the spin $-\sfrak$ component, and this is achieved by make using of the estimates \eqref{eq:AP:pluss:012:428}   for the spin $+\sfrak$ component together with the TSI \eqref{eq:otherTSI:simpleform} and \eqref{eq:TSIspin2:Y+2}.

 Consider only the more complicated $\sfrak=2$ case (because of the presence of an extra term $12 M\overline{\Lxi\psiminustwo}$ in \eqref{eq:TSIspin2:Y+2}), and the simpler case $\sfrak=1$ can be treated in the same way.
Recall the TSI \eqref{eq:TSIspin2:Y+2}. Commuting $j$ times with the Killing vector $\Lxi$  and using the formula $Y=\mu^{-1}(2\Lxi +\frac{2a}{\R}\Leta - r^{-1} rV)$, it can be rewritten as
\begin{align}
\label{eq:TSI:AP:spin-2:ext}
\edthR^4 \Lxi^j\psiminustwo ={}&\sum_{0\leq j_1+j_2+j_3\leq 4}O(1)(r^{-1}\Leta)^{j_1}(rV)^{j_2}(\Lxi)^{j_3} (r^{-4}\psiplustwo)\notag\\
&
-12 M\overline{\Lxi^{j+1}\psiminustwo}
+\sum_{1\leq j_1\leq 4, \,\,j_1+j_2\leq 4}O(1)\Lxi^{j_1}\edthR^{j_2}\Lxi^j\psiminustwo.
\end{align}
The $\absCDeri{\cdot}{\reg}$ norms of the first line of the RHS is bounded by $C_{j,\delta,\reg}v^{-1}\tb^{-2-4-j+C_j\delta}\textbf{I}^{\reg+\regl(j),\delta}_{\text{total}, \tb_0}[\Psiplustwo]$ from \eqref{eq:PWD:pluss:012:428}, and the ones of the second line is bounded by $C_{j,\delta,\reg}v^{-1}\tb^{-2-2-1-j+C_j\delta}\textbf{I}^{\reg+\regl(j),\delta}_{\text{total}, \tb_0}[\Psiminustwo]$ from \eqref{eq:PWD:minuss:012:428}, hence
\begin{align}
\absCDeri{\edthR^4 \Lxi^j\psiminustwo }{\reg}\lesssim_{j,\delta,\reg} v^{-1}\tb^{-5-j+C_j \delta} \IE{\reg+\regl(j),\delta}{\tb_0}.
\end{align}
Since by \eqref{eq:ellipticop:eigenvalue:fixedmode} there is a trivial kernel for the operator $\edthR^4$ when acting on  spin $-2$ scalars, we can thus apply elliptic estimates to the LHS and conclude
\begin{align}
\label{eq:TSI:spin-2:ext:inter:94}
\absCDeri{ \Lxi^j\psiminustwo }{\reg}\lesssim_{j,\delta,\reg}  v^{-1}\tb^{-5-j+C_j \delta} \IE{\reg+\regl(j),\delta}{\tb_0}.
\end{align}
Now we have obtained an extra $\tb^{-1}$ decay for $\Lxi^j \psiminustwo$ compared to the decay estimate \eqref{eq:PWD:minuss:012:428}, and we can run the above argument again except that we now use \eqref{eq:TSI:spin-2:ext:inter:94} instead of the decay estimates \eqref{eq:PWD:minuss:012:428} to estimate the second line of \eqref{eq:TSI:AP:spin-2:ext}.
This allows us to achieve
\begin{align}
\label{eq:TSI:spin-2:ext:inter:64}
\absCDeri{ \Lxi^j\psiminustwo }{\reg}\lesssim_{j,\delta,\reg}  v^{-1}\tb^{-6-j+C_j \delta} \IE{\reg+\regl(j),\delta}{\tb_0}.
\end{align}
In particular, the TSI \eqref{eq:TSI:AP:spin-2:ext}
can now be written as
\begin{align}
\label{eq:TSI:AP:spin-2:ext:alter}
\edthR^4 \Lxi^j\psiminustwo - Y^4 (\psiplustwo)
={}&
-12 M\overline{\Lxi^{j+1}\psiminustwo}
+\sum_{1\leq j_1\leq 4, \,\,j_1+j_2\leq 4}O(1)\Lxi^{j_1}\edthR^{j_2}\Lxi^j\psiminustwo ,
\end{align}
the absolute value of the RHS of which is bounded by $C_{j,\delta,\reg}v^{-1}\tb^{-7-j+C_j \delta} \IE{\reg+\regl(j),\delta}{\tb_0}$.

Our next step is to first project the TSI  \eqref{eq:TSI:AP:spin-2:ext:alter} onto the $\sfrak$ mode and $\geq\sfrak+1$ modes, and this leads to the following TSI in the mode level:
\begin{subequations}
\label{eq:TSI:AP:spin-2:ext:mode}
\begin{align}
\label{eq:TSI:AP:spin-2:ext:mode:s}
&\edthR^4 \Lxi^j\ellmode{\psiminustwo}{\sfrak} - \ellmode{Y^4 (\Lxi^j\psiplustwo)}{\sfrak}
={}
\PJ_\sfrak\Big(-12 M\overline{\Lxi^{j+1}\psiminustwo}
+\sum_{1\leq j_1\leq 4, \,\,j_1+j_2\leq 4}O(1)\Lxi^{j_1}\edthR^{j_2}\Lxi^j\psiminustwo\Big),\\
\label{eq:TSI:AP:spin-2:ext:mode:geqs+1}
&\edthR^4 \Lxi^j\ellmode{\psiminustwo}{\geq \sfrak+1} - \ellmode{Y^4 (\Lxi^j\psiplustwo)}{\geq\sfrak+1}
={}
\PJ_{\geq \sfrak+1}\Big(-12 M\overline{\Lxi^{j+1}\psiminustwo}
+\sum_{1\leq j_1\leq 4, \,\,j_1+j_2\leq 4}O(1)\Lxi^{j_1}\edthR^{j_2}\Lxi^j\psiminustwo\Big).
\end{align}
\end{subequations}
The $\absCDeri{\cdot}{\reg}$ norms of the RHS of both \eqref{eq:TSI:AP:spin-2:ext:mode:s} and \eqref{eq:TSI:AP:spin-2:ext:mode:geqs+1} are bounded by $C_{j,\delta,\reg}v^{-1}\tb^{-7-j+C_j \delta} \IE{\reg+\regl(j),\delta}{\tb_0}$, and by the estimates \eqref{eq:AP:pluss:012:428}, we have
\begin{subequations}
\begin{align}
\absCDeri{ \ellmode{Y^4 (\Lxi^j\psiplustwo)}{\sfrak}
}{\reg}\lesssim_{j,\delta,\reg}  {}&v^{-1}\tb^{-6-j+C_j \delta} \IE{\reg+\regl(j),\delta}{\tb_0},\\
\absCDeri{ \ellmode{Y^4 (\Lxi^j\psiplustwo)}{\geq \sfrak+1}
}{\reg}\lesssim_{j,\delta,\reg}  {}&v^{-1}\tb^{-\frac{13}{2}-j+C_j \delta} \IE{\reg+\regl(j),\delta}{\tb_0}.
\end{align}
\end{subequations}
Therefore, by an elliptic estimate (which is again due to the trivial kernel of $\edthR^4$ when acting on a spin $-2$ scalar), we prove the decay estimates \eqref{eq:AP:minuss:012:428} for the spin $-2$ component in the exterior region $\{\rb\geq\tb\}$.
\qed

\subsubsection{Proof of Proposition \ref{prop:AP:pms:012} in the interior region $\{\rb\leq \tb\}$}

Before passing to the detailed proof, we provide an outline of the proof. The proof of Proposition \ref{prop:AP:pms:012} in the interior region $\{\rb\leq \tb\}$ is divided into four steps. The first two steps are to obtain different types of elliptic estimates for the spin $-\sfrak$ component: the first step is to make use of subsystems of \eqref{eq:basicwavesys:-1} for $\sfrak=1$ or of \eqref{eq:basicwavesys:-2} for $\sfrak=2$, isolate out the spin-weighted angular elliptic parts, and apply elliptic estimates to achieve faster $r^{-\sfrak}$ decay for the spin $-\sfrak$ component than the decay estimates in Corollary \ref{cor:PWD:pms:012}; while the second step is to write the TME \eqref{eq:TME:psis:hypercoords} for the spin $-\sfrak$ component as a three dimensional elliptic (but only in a region a bit far away from horizon) equation in space and, nevertheless, achieve elliptic estimates such that we can improve the above $r^{-\sfrak}$ decay to $\tb^{-\sfrak}$ decay, thus proving the almost sharp decay \eqref{eq:AP:minuss:012:428} for the spin $-\sfrak$ component. As a byproduct, we obtain in the third step that the radial derivative of the $\sfrak$ mode of the spin $-\sfrak$ component has extra $\tb^{-1}$ decay. In the last step, we utilize these almost sharp decay for the spin $-\sfrak$ component together with the TSI and the proven estimates for the spin $+\sfrak$ component in Corollary \ref{cor:PWD:pms:012} to deduce the almost sharp decay for the spin $+\sfrak$ component.

\textbf{Step 1}. Our first step is to derive elliptic estimates for subsystems of \eqref{eq:basicwavesys:-1} for $\sfrak=1$ and of \eqref{eq:basicwavesys:-2} for $\sfrak=2$ to achieve further $r^{-\sfrak}$ decay for the modes of the spin $-\sfrak$ component. The main estimates we shall prove in the interior region $\{\rb\leq\tb\}$ are as follows:
for the $\sfrak$ mode,
\begin{subequations}
\label{eq:AP:minuss:012:r:368}
\begin{align}
\label{eq:AP:minuss:12:r:368:sfrak}
\absCDeri{\Lxi^j(\ellmode{\psiminuss}{\sfrak})}{\reg}\lesssim_{j,\delta,\reg}  r^{-\sfrak}v^{-1}\tb^{-2-\sfrak-j+C_j \delta}   \IE{\reg+\regl(j),\delta}{\tb_0},
\end{align}
and for $\geq \sfrak+1$ modes,
\begin{align}
\label{eq:AP:minuss:12:r:368}
\absCDeri{\Lxi^j(\ellmode{\psiminuss}{\geq \sfrak+1})}{\reg}\lesssim_{j,\delta,\reg}  r^{-\sfrak}v^{-1}\tb^{-\frac{5}{2}-\sfrak-j+C_j \delta}   \IE{\reg+\regl(j),\delta}{\tb_0}.
\end{align}
\end{subequations}

The above estimates \eqref{eq:AP:minuss:012:r:368} in the case $\sfrak=0$ are already contained in the estimates \eqref{eq:PWD:minuss:012:428}. We shall prove only $\sfrak=1$ and $\sfrak=2$ cases.

Let us first consider the case $\sfrak=1$. By the expression \eqref{eq:squareShat} of $\Boxhat_s$, we can recast the first subequation of \eqref{eq:basicwavesys:-1} in the region $\{3M\leq \rb\leq \tb\}$ as
\begin{align}
(\edthR\edthR'-2)\Phiminus{0}
={}&((\R)YV
-2a\Lxi\Leta-a^2 \sin^2 \theta\Lxi^2
-2ia\cos\theta \Lxi)\Phiminus{0}\notag\\
&-\frac{2(r^3-3Mr^2 +a^2 r+a^2 M)}{(\R)^2}\Phiminus{1}
+\frac{4ar}{\R}\Leta \Phiminus{0}
-\frac{a^2\Delta}{(\R)^4}\Phiminus{0}\notag\\
={}&O(1)\Lxi\PsiminusHigh{1}
+O(r^{-1})rV\PsiminusHigh{1}+O(r^{-2})\Leta\PsiminusHigh{1}+O(r^{-1})\PsiminusHigh{1}+O(1)\Leta\Lxi\PsiminusHigh{0}
\notag\\
&+O(r^{-1})\PsiminusHigh{1}
+O(r^{-1})\Leta \PsiminusHigh{0}
+O(r^{-2})\PsiminusHigh{0}\notag\\
&-\mu (a^2 \sin^2 \theta\Lxi^2
+2ia\cos\theta \Lxi)\PsiminusHigh{0},
\end{align}
where in the second step we have used the definition $\Phiminus{1}=\mu^{-1}(\R)V\Phiminus{0}$ and all the $O(\cdot)$ coefficients are $\theta$-independent. By projecting this equation onto $\sfrak$ mode and $\geq \sfrak+1$ modes and applying elliptic estimates on sphere, and noticing that the terms on the RHS either are with $r^{-1}$ decay coefficient or  contain $\Lxi$ derivative that yields an extra $\tb^{-1}$ decay (thus extra $r^{-1}$ decay since $r\leq \tb$) by Corollary \ref{cor:PWD:pms:012}, the estimates \eqref{eq:AP:minuss:012:r:368}  follow.

Then consider $\sfrak=2$. Again, in the region $\{3M\leq \rb\leq \tb\}$, we use the expression \eqref{eq:squareShat} of $\Boxhat_s$ and the definition $\Phiminus{i+1}=\mu^{-1}(\R)V\Phiminus{i}$ to rewrite the first two subequations of \eqref{eq:basicwavesys:-2} into
\begin{subequations}
\label{eq:AP:minuss:ell:629}
\begin{align}
\label{eq:AP:minuss:ell:0:629}
\hspace{4ex}&\hspace{-4ex}
(\edthR\edthR'-4)\Phiminustwo{0}\notag\\
={}&((\R)YV
-2a\Lxi\Leta-a^2 \sin^2 \theta\Lxi^2
-4ia\cos\theta \Lxi)\Phiminustwo{0}\notag\\
&+O(r^{-1})\Phiminustwo{1}
+O(r^{-1})\Leta \Phiminustwo{0}
+O(r^{-1})\Phiminustwo{0}\notag\\
={}&O(1)\Lxi\PsiminustwoHigh{1}
+O(r^{-1}) rV\PsiminustwoHigh{1}
+O(r^{-2})\Leta \PsiminustwoHigh{1}
+O(r^{-1})\PsiminustwoHigh{1}
+O(1)\Lxi\Leta\Psiminustwo\notag\\
&+O(r^{-1})\PsiminustwoHigh{1}
+O(r^{-1})\Leta \Psiminustwo
+O(r^{-1})\Psiminustwo\notag\\
&-\mu^2(a^2 \sin^2 \theta\Lxi^2
+4ia\cos\theta \Lxi)\Psiminustwo,\\
\hspace{4ex}&\hspace{-4ex}(\edthR\edthR'-6)\Phiminustwo{1}+6M\Phiminustwo{0}+6a\Leta\Phiminustwo{0}\notag\\
={}&((\R)YV
-2a\Lxi\Leta-a^2 \sin^2 \theta\Lxi^2
-4ia\cos\theta \Lxi)\Phiminustwo{1}\notag\\
&+O(r^{-1})\Phiminustwo{2}
+O(r^{-2})\Leta\Phiminustwo{1}
+O(r^{-1})\Phiminustwo{1}
+O(r^{-1})\Leta \Phiminustwo{0}
+O(r^{-1})\Phiminustwo{0}\notag\\
={}&O(1)\Lxi\PsiminustwoHigh{2}
+O(r^{-1}) rV \PsiminustwoHigh{2}+O(r^{-2})\Leta\PsiminustwoHigh{2}
+O(1)\Lxi\Leta\PsiminustwoHigh{1}\notag\\
&+O(r^{-1})\PsiminustwoHigh{2}
+O(r^{-2})\Leta\PsiminustwoHigh{1}
+O(r^{-1})\PsiminustwoHigh{1}
+O(r^{-1})\Leta \Psiminustwo
+O(r^{-1})\Psiminustwo\notag\\
&-(a^2 \sin^2 \theta\Lxi^2
+4ia\cos\theta \Lxi)\Phiminustwo{1}.
\end{align}
\end{subequations}
The LHS of the system \eqref{eq:AP:minuss:ell:629} can be written as $\big(\begin{smallmatrix}
  \edthR\edthR'-4 & 0\\
  6M+6a\Leta & \edthR\edthR'-6
\end{smallmatrix}\big)\big(\begin{smallmatrix}
 \Phiminustwo{0} \\
\Phiminustwo{1}
\end{smallmatrix}\big)$, and this $2\times 2$ matrix is lower triangular and has nonzero eigenvalues. Therefore, by the same argument of projecting this equation onto $\sfrak$ mode and $\geq \sfrak+1$ modes, applying elliptic estimates on sphere in Section \ref{sect:decompIntoModes} and noticing that the terms on the RHS either are with $r^{-1}$ decay coefficient or  contain $\Lxi$ derivative that yields an extra $\tb^{-1}$ decay (thus extra $r^{-1}$ decay since $r\leq \tb$) by Corollary \ref{cor:PWD:pms:012}, we achieve extra $r^{-1}$ decay compared to the ones in \eqref{eq:AP:minuss:012:428}. That is, the following holds for $\sfrak=2$:
\begin{subequations}
\label{eq:AP:minuss:012:r2:250}
\begin{align}
\label{eq:AP:minuss:0:r2:250}
\sum_{i=0,1}\absCDeri{\Lxi^j(\ellmode{r^{-1}\PsiminussHigh{i}}{\sfrak})}{\reg}\lesssim_{j,\delta,\reg}  r^{-1}v^{-1}\tb^{-2-2\sfrak-j+C_j \delta}   \IE{\reg+\regl(j),\delta}{\tb_0},\\
\label{eq:AP:minuss:12:r2:250}
\sum_{i=0,1}\absCDeri{\Lxi^j(\ellmode{r^{-1}\PsiminussHigh{i}}{\geq \sfrak+1})}{\reg}\lesssim_{j,\delta,\reg}  r^{-1}v^{-1}\tb^{-\frac{5}{2}-2\sfrak-j+C_j \delta}  \IE{\reg+\regl(j),\delta}{\tb_0}.
\end{align}
\end{subequations}
Given these estimates, we now apply the same argument to the single equation \eqref{eq:AP:minuss:ell:0:629}, and for the same reason, we can derive extra $r^{-1}$ decay for $\ellmode{r^{-1}\PsiminussHigh{0}}{\sfrak}$ compared to the ones in \eqref{eq:AP:minuss:012:r2:250}, hence completing the proof of the estimates \eqref{eq:AP:minuss:012:r:368} in the case $j=0$.  Commuting the equations used in this step with $\Lxi^j$ then proves the estimates \eqref{eq:AP:minuss:012:r:368} for general $j\in \mathbb{N}$.

\textbf{Step 2}. This second step is to prove the almost sharp decay estimates \eqref{eq:AP:minuss:012:428} for the spin $-\sfrak$ component in the interior region by a different type of elliptic estimate. This other type of elliptic estimates in $3$-dimensional space allows us to trade the achieved extra $r^{-\sfrak}$ decay in the previous step for extra $\tb^{-\sfrak}$ decay.

Our main estimates to show in this step are as follows:
\begin{subequations}
\label{eq:int:AP:induct:reg:beta2s+1:ss+1:391}
\begin{align}
\label{eq:int:AP:induct:reg:beta2s+1:s:391}
&\int_{\rb\leq \tb}r^{-1+2\delta}(\absCDeri{r\prb\Lxi^j\ellmode{\psiminuss}{\sfrak}}{\reg}^2
+\absCDeri{\Lxi^j\ellmode{\psiminuss}{\sfrak}}{\reg}^2)\di \rb
\lesssim_{j,\delta,\reg}  {}\tb^{-6-4\sfrak-2j+C_j\delta}\IE{\reg+\regl(j),\delta}{\tb_0} ,\\
\label{eq:int:AP:induct:reg:beta2s+1:s+1:391}
&\int_{\Sigmatb^{\leq \tb}}r^{-1+2\delta}\big(\absCDeri{r\prb\Lxi^j\ellmode{\psiminuss}{\geq \sfrak+1}}{\reg}^2+\absCDeri{\edthR'\Lxi^j\ellmode{\psiminuss}{\geq \sfrak+1}}{\reg}^2+\absCDeri{\Lxi^j\ellmode{\psiminuss}{\geq \sfrak+1}}{\reg}^2\big)\di^3\mu\notag\\
&\quad
\lesssim_{j,\delta,\reg}  {}\tb^{-7-4\sfrak-2j+C_j\delta} \IE{\reg+\regl(j),\delta}{\tb_0} .
\end{align}
\end{subequations}
The pointwise decay estimates \eqref{eq:AP:minuss:012:428}  then follow easily from the Sobolev inequality \eqref{eq:Sobolev:1} applied to these energy decay estimates. As a result, the remaining discussions in this step are devoted to proving the estimates \eqref{eq:int:AP:induct:reg:beta2s+1:ss+1:391}.

Recall equation \eqref{eq:TME:psis:hypercoords}. We take $s=-\sfrak$ in equation \eqref{eq:TME:psis:hypercoords}, commute with $\Lxi^j$
and project onto $\geq \sfrak+1$ modes,
arriving at
\begin{align}
\prb(\Delta^{\sfrak+1}\prb\Lxi^j\ellmode{\psiminuss}{\geq \sfrak+1})+2a\Leta\Delta^{\sfrak}\prb\Lxi^j\ellmode{\psiminuss}{\geq \sfrak+1}
+\Delta^{\sfrak}\edthR\edthR'\Lxi^j\ellmode{\psiminuss}{\geq \sfrak+1}
=\Delta^{\sfrak}\Lxi^{j+1} \Proj{\geq \sfrak+1}H[\psiminuss].
\end{align}
For ease of notation, we denote $\varphi_{\geq \sfrak+1}=\ellmode{\psiminuss}{\geq \sfrak+1}$ and $H_{\geq \sfrak+1}=\Proj{\geq \sfrak+1}H[\psiminuss]$. The above equation then becomes
\begin{align}
\label{eq:int:AP:geqs+1:479}
\prb(\Delta^{\sfrak+1}\prb\Lxi^j\varphi_{\geq \sfrak+1})+2a\Leta\Delta^{\sfrak}\prb\Lxi^j\varphi_{\geq \sfrak+1}+\Delta^{\sfrak}\edthR\edthR'\Lxi^j\varphi_{\geq \sfrak+1}=\Delta^{\sfrak}\Lxi^{j+1} H_{\geq \sfrak+1}.
\end{align}
We multiply $2f\Delta^{\sfrak+1} \overline{\prb\Lxi^j\varphi_{\geq \sfrak+1}}$ on both sides and take the real part, then by Leibniz's rule, we obtain
\begin{align}
\label{eq:int:ellip:generalmulti:high:843}
\hspace{4ex}&\hspace{-4ex}
\prb(f\abs{\Delta^{\sfrak+1}\prb\Lxi^j\varphi_{\geq \sfrak+1}}^2-f\Delta^{2\sfrak+1}\abs{\edthR'\Lxi^j\varphi_{\geq \sfrak+1}}^2)
-\partial_r f \abs{\Delta^{\sfrak+1}\prb\Lxi^j\varphi_{\geq \sfrak+1}}^2
+ \partial_{r} (f\Delta^{2\sfrak+1})\abs{\edthR'\Lxi^j\varphi_{\geq \sfrak+1}}^2
\notag\\
\hspace{4ex}&\hspace{-4ex}
+\Re(\edthR(2f\Delta^{2\sfrak+1} \edthR' \Lxi^j\varphi_{\geq \sfrak+1}\overline{\prb\Lxi^j\varphi_{\geq \sfrak+1}}))
+\Leta(2af \Delta^{2\sfrak+1}\abs{\prb\Lxi^j\varphi_{\geq \sfrak+1}}^2)
\notag\\
={}&\Re(2f\Delta^{2\sfrak+1}\Lxi^{j+1} H_{\geq \sfrak+1}\cdot\prb \overline{\Lxi^j\varphi_{\geq \sfrak+1}}).
\end{align}
We then take $f=\mu^{-2\sfrak-1} (\R)^{-\beta}$ with $0<\beta< 2\sfrak+1$ in the above formula and integrate the formula in $\Sigmatb^{\leq \tb}$.  Note that the boundary term at $\rb=r_+$ vanishes since
\begin{align*}
\hspace{3ex}&\hspace{-3ex}
\big(f\abs{\Delta^{\sfrak+1}\prb\Lxi^j\varphi_{\geq \sfrak+1}}^2-f\Delta^{2\sfrak+1}\abs{\edthR'\Lxi^j\varphi_{\geq \sfrak+1}}^2\big)\vert_{\rb=r_+}\notag\\
={}&(\mu (\R)^{2\sfrak-\beta+2}\abs{\prb\Lxi^j\varphi_{\geq \sfrak+1}}^2
-\mu (\R)^{2\sfrak-\beta+1}\abs{\edthR'\Lxi^j\varphi_{\geq \sfrak+1}}^2)\vert_{\rb=r_+}=0,
\end{align*}
and the integral of the second line vanishes.
Further,
\begin{subequations}
\begin{align}
&-\partial_r f={}(2\sfrak+1)\partial_r \mu \mu^{-2\sfrak-2} (\R)^{-\beta} +2\beta \mu^{-2\sfrak-1}r(\R)^{-\beta-1}\gtrsim_{\beta} \mu^{-2\sfrak-2}r^{-2\beta -1},\\
&\partial_{r} (f\Delta^{2\sfrak+1})={}2(2\sfrak-\beta+1)r(\R)^{2\sfrak-\beta}\gtrsim_{\beta}  r^{4\sfrak-2\beta+1},\\
&\int_{S^2}\abs{\edthR'\Lxi^j\varphi_{\geq \sfrak+1}}^2\di^2\mu\geq {}2(\sfrak+1)\int_{S^2}\abs{\Lxi^j\varphi_{\geq \sfrak+1}}^2\di^2 \mu,
\end{align}
\end{subequations}
where the last inequality follows from \eqref{eq:ellip:highermodes}.
Hence, an application of Cauchy--Schwarz to the integral of the RHS of \eqref{eq:int:ellip:generalmulti:high:843}  then yields for any $0<\beta< 2\sfrak+1$,
\begin{align}
\label{eq:int:AP:induct:s+1:473}
\hspace{5ex}&\hspace{-5ex}
\int_{\Sigmatb^{\leq \tb}}r^{4\sfrak-2\beta+1}\big(\abs{r\prb\Lxi^j\varphi_{\geq \sfrak+1}}^2+\abs{\edthR'\Lxi^j\varphi_{\geq \sfrak+1}}^2+\abs{\Lxi^j\varphi_{\geq \sfrak+1}}^2\big)\di^3\mu\notag\\
\lesssim_{\beta} {}&\int_{\Sigmatb^{\leq \tb}}r^{4\sfrak-2\beta+1}\abs{\Lxi^{j+1} H_{\geq \sfrak+1}}^2\di^3\mu
+\Big(\int_{S^2}r^{4\sfrak-2\beta+2}\abs{\edthR'\Lxi^j\varphi_{\geq \sfrak+1}}^2\di^2\mu\Big)
\Big\vert_{\rb=\tb} .
\end{align}

We can also treat the $\sfrak$ mode in an exactly same way. Taking $s=-\sfrak$ in equation \eqref{eq:TME:psis:hypercoords}, commuting with $\Lxi^j$
and projecting onto an $(m,\sfrak)$ mode,
 we arrive at
\begin{align}
\prb(\Delta^{\sfrak+1}\prb\Lxi^j\mellmode{\psiminuss}{m}{\sfrak})+2iam\Delta^{\sfrak}\prb\Lxi^j\mellmode{\psiminuss}{m}{\sfrak}
=\Delta^{\sfrak}\Lxi^{j+1} \Proj{m, \sfrak}H[\psiminuss].
\end{align}
For ease of notation, we denote $\varphi_{m,\sfrak}=\mellmode{\psiminuss}{m}{\sfrak}$, $H_{m,\sfrak}=\Proj{m,\sfrak}H[\psiminuss]$, $\varphi_{\sfrak}=\ellmode{\psiminuss}{\sfrak}$ and $H_{\sfrak}=\Proj{\sfrak}H[\psiminuss]$, and recast the above equation as
\begin{align}
\prb(\Delta^{\sfrak+1}\prb\Lxi^j\varphi_{m,\sfrak})+2iam\Delta^{\sfrak}\prb\Lxi^j\varphi_{m,\sfrak}
=\Delta^{\sfrak}\Lxi^{j+1} H_{m,\sfrak}.
\end{align}
The only difference between this equation and equation \eqref{eq:int:AP:geqs+1:479} lies in the angular derivative term.
With the same discussions, one achieves for any $0<\beta<2\sfrak+1$,
\begin{align}
\label{eq:int:AP:induct:s:473}
\hspace{5ex}&\hspace{-5ex}
\Big(\int_{S^2}\mu r^{4\sfrak-2\beta+2}\abs{r\prb\Lxi^j\varphi_{\sfrak}}^2\di^2\mu\Big)
\Big\vert_{\rb=\tb} +\int_{\Sigmatb^{\leq \tb}}r^{4\sfrak-2\beta+1}(\abs{r\prb\Lxi^j\varphi_{\sfrak}}^2
+\abs{\Lxi^j\varphi_{\sfrak}}^2)\di \rb\notag\\
\lesssim_{\beta} {}&\int_{\Sigmatb^{\leq \tb}}r^{4\sfrak-2\beta+1}\abs{\Lxi^{j+1} H_{\sfrak}}^2\di\rb
+\Big(\int_{S^2}r^{4\sfrak-2\beta+2}\abs{\Lxi^j\varphi_{\sfrak}}^2\di^2\mu\Big)
\Big\vert_{\rb=\tb} .
\end{align}
Here, we have summed over $m$ with $\abs{m}\leq \sfrak$ and used the Hardy's inequality  \eqref{eq:HardyIneqRHS}.

By the expression \eqref{def:Hpsis} of $H[\psi_s]$, we have
\begin{subequations}
\label{eq:AP:minuss:Int:ellip:ss+1:837}
\begin{align}
\label{eq:AP:minuss:Int:ellip:s:837}
\absCDeri{\Lxi^{j+1}H_{\sfrak}}{\reg}^2\lesssim_{\reg} {} &\absCDeri{\Lxi^{j+1}(r\varphi_{\sfrak})}{\reg+1}^2
+\absCDeri{\Lxi^{j+1}\varphi_{\sfrak}}{\reg+1}^2,\\
\label{eq:AP:minuss:Int:ellip:s+1:837}
\absCDeri{\Lxi^{j+1}H_{\geq \sfrak+1}}{\reg}^2\lesssim_{\reg} {} &
\absCDeri{\Lxi^{j+1}(r\varphi_{\geq \sfrak+1})}{\reg+1}^2
+\absCDeri{\Lxi^{j+1}\varphi_{\sfrak}}{\reg+1}^2.
\end{align}
\end{subequations}

We first take $\beta=\sfrak+1-\delta$ in both \eqref{eq:int:AP:induct:s+1:473} and \eqref{eq:int:AP:induct:s:473}. In view of the estimate \eqref{eq:AP:minuss:Int:ellip:s:837} and the pointwise estimates \eqref{eq:AP:minuss:012:r:368}, the RHS of \eqref{eq:int:AP:induct:s:473} is bounded by $C_{j,\delta}\tb^{-6-2\sfrak-2j+C_j\delta} \IE{\regl(j),\delta}{\tb_0}$, thus arriving at
\begin{subequations}
\label{eq:int:AP:induct:betas+1:ss+1:391}
\begin{align}
\label{eq:int:AP:induct:betas+1:s:391}
\int_{\Sigmatb^{\leq \tb}}r^{2\sfrak-1+2\delta}(\abs{r\prb\Lxi^j\varphi_{\sfrak}}^2
+\abs{\Lxi^j\varphi_{\sfrak}}^2)\di ^3\mu
\lesssim_{j,\delta}  {}&\tb^{-6-2\sfrak-2j+C_j\delta} \IE{\regl(j),\delta}{\tb_0}.
\end{align}
We can now utilize this estimate, the estimate \eqref{eq:AP:minuss:Int:ellip:s+1:837} and the pointwise estimates \eqref{eq:AP:minuss:012:r:368} to find the  the RHS of \eqref{eq:int:AP:induct:s+1:473} is bounded by $C_{j,\delta}\tb^{-7-2\sfrak-2j+C_j\delta} \IE{\regl(j),\delta}{\tb_0}$, which yields
\begin{align}
\label{eq:int:AP:induct:betas+1:s+1:391}
\int_{\Sigmatb^{\leq \tb}}r^{2\sfrak-1+2\delta}\big(\abs{r\prb\Lxi^j\varphi_{\geq \sfrak+1}}^2+\abs{\edthR'\Lxi^j\varphi_{\geq \sfrak+1}}^2+\abs{\Lxi^j\varphi_{\geq \sfrak+1}}^2\big)\di^3\mu
\lesssim_{j,\delta}  {}&\tb^{-7-2\sfrak-2j+C_j\delta} \IE{\reg+\regl(j),\delta}{\tb_0} .
\end{align}
\end{subequations}

Next, we take $\beta=\sfrak+2-\delta$ in both \eqref{eq:int:AP:induct:s+1:473} and \eqref{eq:int:AP:induct:s:473}. The same argument applies except that we shall use \eqref{eq:int:AP:induct:betas+1:ss+1:391} instead of \eqref{eq:AP:minuss:012:r:368} to control the RHS of \eqref{eq:int:AP:induct:s:473}; we will achieve \begin{subequations}
\label{eq:int:AP:induct:betas+2:ss+1:391}
\begin{align}
\label{eq:int:AP:induct:betas+2:s:391}
\int_{\Sigmatb^{\leq \tb}}r^{2\sfrak-3+2\delta}(\abs{r\prb\Lxi^j\varphi_{\sfrak}}^2
+\abs{\Lxi^j\varphi_{\sfrak}}^2)\di^3\mu
\lesssim_{j,\delta} {}&\tb^{-8-2\sfrak-2j+C_j\delta} \IE{\reg+\regl(j),\delta}{\tb_0}.
\end{align}
Moreover, using this estimate together with the estimate \eqref{eq:AP:minuss:Int:ellip:s+1:837} and the pointwise estimates \eqref{eq:int:AP:induct:betas+1:ss+1:391} to control the RHS of \eqref{eq:int:AP:induct:s+1:473}, one finds
\begin{align}
\label{eq:int:AP:induct:betas+2:s+1:391}
\int_{\Sigmatb^{\leq \tb}}r^{2\sfrak-3+2\delta}\big(\abs{r\prb\Lxi^j\varphi_{\geq \sfrak+1}}^2+\abs{\edthR'\Lxi^j\varphi_{\geq \sfrak+1}}^2+\abs{\Lxi^j\varphi_{\geq \sfrak+1}}^2\big)\di^3\mu
\lesssim_{j,\delta} {}&\tb^{-9-2\sfrak-2j+C_j\delta} \IE{\reg+\regl(j),\delta}{\tb_0} .
\end{align}
\end{subequations}
Note that the improvement of \eqref{eq:int:AP:induct:betas+2:ss+1:391} compared to \eqref{eq:int:AP:induct:betas+1:ss+1:391} lies in the fact that we have traded the $r$ weights inside the integral on the LHS for the same amount of $\tb$ decay.  This argument can be inductively applied until we reach the final choice $\beta=2\sfrak+1-\delta$, and we eventually conclude the estimate \eqref{eq:int:AP:induct:reg:beta2s+1:s:391}.
Further, using this estimate together with the estimate \eqref{eq:AP:minuss:Int:ellip:s+1:837} and the pointwise estimates \eqref{eq:int:AP:induct:betas+1:ss+1:391} to control the RHS of \eqref{eq:int:AP:induct:s+1:473}, the estimate \eqref{eq:int:AP:induct:reg:beta2s+1:s+1:391} with $\reg=0$ follows.

We then proceed to general $\reg\in\mathbb{N}$ case. Since $\Lxi$ and $\Leta$ commute with equation \eqref{eq:TME:psis:hypercoords}, and since $\edthR\edthR'$ commutes with the LHS of equation \eqref{eq:TME:psis:hypercoords} and the obtained RHS enjoys the same kind of estimates as the ones in \eqref{eq:AP:minuss:Int:ellip:ss+1:837} (with the only difference that the RHS of \eqref{eq:AP:minuss:Int:ellip:ss+1:837} requires  higher order regularity norms),  we achieve the estimates with $\CDeri$ replaced by $\{\Lxi, \Leta, \edthR, \edthR'\}$.

Based on the above discussions, It remains to prove  the estimates \eqref{eq:int:AP:induct:reg:beta2s+1:ss+1:391}  with $\CDeri$ replaced by $\{\rb\prb\}$. We prove it by induction in $\reg$, that is,
assuming it holds for $\reg=n-1$, $n\in \mathbb{N}^+$,   we prove for $\reg=n$. We multiply both sides of equation \eqref{eq:TME:psis:hypercoords} by $\mu^{-\sfrak}$ to get a rewritten form of  equation \eqref{eq:TME:psis:hypercoords}:
\begin{align}
\label{eq:TME:psis:hypercoords:rewrite:minuss}
\mu^{-\sfrak}\prb(\Delta^{\sfrak+1}\prb\psiminuss)+2a(\R)^{\sfrak}\Leta\prb\psiminuss
+(\R)^{\sfrak}\edthR\edthR'\psiminuss
=(\R)^{\sfrak}\Lxi H[\psiminuss].
\end{align}
We then  commute this equation with $\rb \prb$, and since
\begin{align}
r\prb\Big(\mu^{-\sfrak}\prb(\Delta^{1+\sfrak} \prb\psiminuss)\Big)={}&
\mu^{-(\sfrak+1)}\prb\Big(\mu \Delta^{\sfrak+1} \prb (r\prb\psiminuss)\Big)\notag\\
&+\big(O_{\infty}(1)\mu r^{2\sfrak+2}\prb^2+O_{\infty}(1)r^{2\sfrak+1}\prb +O_{\infty}(1)r^{2\sfrak}\big)\psiminuss,
\end{align}
where $O_{\infty}(1)$ are $O(1)$ functions and smooth everywhere in $\rb\in [r_+,\infty)$,
we obtain
for any $n\in \mathbb{N}^+$,
\begin{align}
\label{eq:almost:int:m:highorder:ge:7598}
\hspace{2ex}&\hspace{-2ex}\mu^{-(\sfrak+n)}\prb\big((\R)^{\sfrak+1}\mu^{1+(\sfrak+n)} \prb ((r\prb)^n\psiminuss)\big)+2a(\R)^{\sfrak}\Leta\prb((r\prb)^n\psiminuss)
+(\R)^{\sfrak}\edthR\edthR'((r\prb)^n\psiminuss)\notag\\
={}&\Lxi (r\prb)^nH[\psiminuss]
+r^{2\sfrak}\Big(\sum_{i_1=0}^n\sum_{i_2\leq 1}\sum_{i_3\leq 1} O_{\infty}(1)(r\prb)^{i_1}\Leta^{i_2}(\edthR\edthR')^{i_3}\psiminuss+O_{\infty}(1)\mu (r\prb)^{n+1}\psiminuss\Big).
\end{align}
We can achieve elliptic estimates for this equation of $(r\prb)^n\psiminuss$ in a similar way of treating equation \eqref{eq:TME:psis:hypercoords:rewrite:minuss} (or equivalently, equation \eqref{eq:TME:psis:hypercoords}). More specifically, by projecting the above equation onto an $(m,\sfrak)$ mode (resp. $\geq \sfrak+1$ modes), we multiply both sides of the obtained equation by $2\mu^{\sfrak+n}f(\R)^{\sfrak+1}\mu^{1+\sfrak+n}\overline{\prb ((r\prb)^n\psiminuss)}$, with $f=\mu^{-2(\sfrak+n)-1}(\R)^{-\beta}$, and integrate over $r_+\leq \rb\leq \tb$ (resp. $\Sigmatb^{\leq \tb}$).
The integral arising from the last term of \eqref{eq:almost:int:m:highorder:ge:7598} can be estimated by the assumption in the induction together with the proven estimate \eqref{eq:int:AP:induct:reg:beta2s+1:ss+1:391} but with $\CDeri$ replaced by $\{\Lxi, \Leta, \edthR, \edthR'\}$, thus the same  argument as the one in treating $\reg=0$ case applies and yields the estimate \eqref{eq:int:AP:induct:reg:beta2s+1:ss+1:391} for $\reg=n$.

\textbf{Step 3}. This third step is to prove the estimate \eqref{eq:AP:prbminuss:0:428} which encodes further decay for the radial derivative of the $(m, \sfrak)$ mode of the spin $-\sfrak$ component, i.e. for $\prb\mellmode{\psiminuss}{m}{\sfrak}$, in the interior region $\{\rb\leq \tb\}$.

We shall need the following lemma  that is immediate from Proposition \ref{prop:goodsecondorderODE:varphis}.

\begin{lemma}
Let
\begin{align}
\label{def:Scalingw(r)}
\Scaling=\Scaling(a,M,r,m)=e^{\int_{r_+}^r \frac{2iam}{\Delta(r')}d r' }.
\end{align}
 The $(m,\sfrak)$ mode $\mellmode{\psiminuss}{m}{\sfrak}$ satisfies
\begin{align}
\label{eq:TME:psiminuss:hypercoords}
\prb(\Scaling\Delta^{\sfrak+1}\prb\mellmode{\psiminuss}{m}{\sfrak})
={}&\Scaling \Delta^{\sfrak}\Lxi \big(H[\mellmode{\psiminuss}{m}{\sfrak}]- \Comm{m,\sfrak}{-\sfrak}{\psiminuss}
\big)
\end{align}
with
the term $H[\mellmode{\psiminuss}{m}{\sfrak}]$  on the RHS satisfying
\begin{align}
H[\mellmode{\psiminuss}{m}{\sfrak}]={}&O(1)r\prb (\sqrt{\R}\mellmode{\psiminuss}{m}{\sfrak})
+O(1)\Lxi\mellmode{\psiminuss}{m}{\sfrak}
+O(1)\mellmode{\psiminuss}{m}{\sfrak}.
\end{align}
\end{lemma}

\begin{proof}
By Proposition \ref{prop:goodsecondorderODE:varphis} with $s=-\sfrak$, and in view of the facts that $\Hhyp=O(r^{-2})$ and $(r-M)\mu^{-1} -r =O(r^{-1})$, we have for the spin $-\sfrak$ component that
\begin{align}
\prb(\Delta^{\sfrak+1}\prb\psiminuss)+2a\Delta^{\sfrak}\Leta\prb\psiminuss
+\Delta^{\sfrak}\edthR\edthR'\psiminuss
=\Delta^{\sfrak}\Lxi H[\psiminuss],
\end{align}
with
the term $H[\psiminuss]$  on the RHS satisfying
\begin{align*}
H[\psiminuss]={}&O(1)r\prb (\sqrt{\R}\psiminuss)
+O(1)\Lxi\psiminuss
+O(1)\Leta\psiminuss
+O(1)  \psiminuss
+a^2\sin^2\theta \Lxi\psiminuss-2ias\cos\theta\psiminuss.
\end{align*}
We project this equation onto the $(m,\sfrak)$ mode and, noticing from \eqref{eq:l=l0mode:eigenvalue} that $\edthR\edthR'(\mellmode{\psiminuss}{m}{\sfrak}Y_{m,\sfrak}^{-\sfrak}(\cos\theta)e^{im\pb})=0$, we  conclude
\begin{align}
\prb(\Delta^{\sfrak+1}\prb\mellmode{\psiminuss}{m}{\sfrak})+2iam\Delta^{\sfrak}\prb\mellmode{\psiminuss}{m}{\sfrak}={}&\Delta^{\sfrak}\Lxi \big(H[\mellmode{\psiminuss}{m}{\sfrak}]- \Comm{m,\sfrak}{-\sfrak}{\psiminuss}\big).
\end{align}
A simple rescaling then yields the desired equation.
\end{proof}

The above equation \eqref{eq:TME:psiminuss:hypercoords} can be integrated from horizon to yield a refined decay estimate for $\prb\mellmode{\psiminuss}{m}{\sfrak}$ in the interior region $\{\rb\leq \tb\}$.

\underline{\textit{Proof of the estimate \eqref{eq:AP:prbminuss:0:428}}}:
For any point $(\tb,\rb')$ satisfying $\rb'\leq \tb$, we integrate equation \eqref{eq:TME:psiminuss:hypercoords} from horizon and obtain
\begin{align}
\big(\Scaling\Delta^{\sfrak+1}\prb\mellmode{\psiminuss}{m}{\sfrak} \big)(\tb,\rb')= \int_{r_+}^{\rb'}\Scaling \Delta^{\sfrak}\Lxi \big(H[\mellmode{\psiminuss}{m}{\sfrak}]- \Comm{m,\sfrak}{-\sfrak}{\psiminuss}
\big)\di\rb.
\end{align}
By Definition \ref{def:Commells} for $\Comm{m,\sfrak}{-\sfrak}{\psiminuss}$ and the decay estimates \eqref{eq:AP:minuss:012:428}, the absolute value of the RHS is bounded by $C_{\delta} \IE{\regl,\delta}{\tb_0}(\Delta^{\sfrak+1}v^{-1}\tb^{-3-2\sfrak+C \delta})  (\tb, \rb')$, which thus yields \eqref{eq:AP:prbminuss:0:428} for $\reg=j=0$.

We next apply $\prb(\mu r\cdot)$ on both sides of equation \eqref{eq:TME:psiminuss:hypercoords}  and integrate this new obtained equation from horizon. The above proof still works and implies that
\begin{align}
\abs{\prb(\mu r(\prb\ellmode{\psiminuss}{\sfrak}))}\lesssim_{\delta} v^{-1}\tb^{-3-2\sfrak+C \delta}  \IE{\regl,\delta}{\tb_0}.
\end{align}
This together with the estimate in the previous step completes the proof of \eqref{eq:AP:prbminuss:0:428} in the case $(\reg,j)=(1,0)$. The same argument applies to the general $(\reg\in \mathbb{N}, j=0)$ case. In the end, it is manifest that $\Lxi^j$ commutes with equation \eqref{eq:TME:psiminuss:hypercoords} and from the decay estimates \eqref{eq:AP:minuss:012:428}, $\Lxi^j$ acting on the RHS of \eqref{eq:TME:psiminuss:hypercoords}  has extra $\tb^{-j}$ decay, hence the above argument applies and  completes the proof in the general $(\reg,j)\in \mathbb{N}\times \mathbb{N}$ cases.

\textbf{Step 4}. Our last step is to show the decay estimates \eqref{eq:AP:pluss:012:428} for the spin $+\sfrak$ component via  the TSI together with the proven almost sharp decay estimates \eqref{eq:AP:minuss:012:428} for the spin $-\sfrak$ component.

The proof is in fact in the same spirit of the one in Section \ref{sect:AP:pms:ext:83} where the almost sharp decay estimates for the spin $-\sfrak$ component  in the exterior region are proven via the almost sharp decay of the spin $+\sfrak$ component and an application of (the other) TSI. Again, we consider only the more complicated $\sfrak=2$ case, and the simpler case $\sfrak=1$ can be similarly treated.

Recall the TSI \eqref{eq:TSIspin2:V-2}. Commuting with $\Lxi^j$ and multiplying by $(\R)^{-2}$, it can be written as
\begin{align}
\label{eq:TSIspin2:V-2:rescaled:49}
 (\edthR'-ia\sin\theta\Lxi)^4 (\Lxi^j((\R)^{-2}\psiplustwo))  -12 M \overline{\Lxi^{j+1}((\R)^{-2}\psiplustwo)}={}&\mu^2\VR^4 (\Delta^2 \Lxi^j(\psiminustwo)).
 \end{align}
 By the decay estimates \eqref{eq:AP:minuss:012:428} for the modes of the spin $-\sfrak$ component, we find that
if projecting this equation onto the $\sfrak$ mode, the $\absCDeri{\cdot}{\reg}$ norm of the RHS is bounded by $C_{j, \delta, \reg}v^{-1}\tb^{-2-2\sfrak-j+C_j \delta}   \IE{\reg+\regl(j),\delta}{\tb_0}$; instead, if projecting this equation onto the $\geq \sfrak+1$ modes, the $\absCDeri{\cdot}{\reg}$ norm of the RHS has decay $C_{j, \delta, \reg}v^{-1}\tb^{-\frac{5}{2}-2\sfrak-j+C_j \delta}  \IE{\reg+\regl(j),\delta}{\tb_0}$.  The remaining discussions are exactly the same as the ones in Section \ref{sect:AP:pms:ext:83} and will be dropped; these will prove the estimates \eqref{eq:AP:pluss:012:428} but with the factor $v^{-1-2\sfrak}$ on the RHS replaced by $v^{-1}\tb^{-2\sfrak}$. However, in the interior region $\{\rb\leq \tb\}$, we have $\tb\gtrsim v$, hence the estimates \eqref{eq:AP:pluss:012:428} hold.
\qed

\section{Global sharp decay of the spin $\pm \sfrak$ components}
\label{sect:PL}

In this section, we will prove the sharp decay for the spin $\pm\sfrak$ components using the almost sharp decay estimates proven in the previous section.  In Section \ref{subsection:conservationlaw}, we deduce  for the $(m,\sfrak)$-mode of the spin $+\sfrak$ component a global conservation law, which allows us to calculate the integral of its radiation field along the future null infinity.  This conservation law is then utilized in Section \ref{subsection:proof:sharpdecay} to derive the precise asymptotic profile of this mode in separate regions $\{r\geq v^\alpha\}$ and $\{r\leq v^\alpha\}$ for some $\alpha\in (\half,1)$.

\textit{Throughout this section, the BEAM estimates assumption \ref{ass:BEAM:inhomogeneous} for an inhomogeneous TEM is  always assumed.} Therefore, in view of Remark \ref{rem:BEAM:pf:slowly}, all the estimates in Section \ref{sect:APL} are valid for $\sfrak=0$ in any subextreme Kerr and $\sfrak=1,2$ in slowly rotating Kerr with $\abs{a}/M$ sufficiently small, and are valid for $\sfrak=1,2$ in any subextreme Kerr under  Assumption \ref{ass:BEAM:inhomogeneous}.

\subsection{Global conservation law}
\label{subsection:conservationlaw}

The main result of this subsection is to compute the integral of the radiation field of any $(m,\sfrak)$ mode of the spin $+\sfrak$ component on future null infinity with respect to the initial data. This is achieve by a global conservation law for the TME of this mode.


Recall equation \eqref{eq:TME:varphipluss:hypercoords} of $\varphi_{+\sfrak}=\Delta^{-\sfrak}\psipluss$ in Corollary \ref{cor:varphipluss:eq}. By  projecting this  equation  onto an $(m,\sfrak)$ mode,  we obtain
\begin{align}
\label{eq:TME:varphipluss:hypercoords:PL}
\prb(\Delta^{\sfrak+1}\prb\varphi+2iam\Delta^{\sfrak}\varphi)
=\ptb \Proj{m,\sfrak}(H[\psipluss])
\end{align}
where we have used equation \eqref{eq:eigenvalueSWSHO} and denoted $\varphi=\ellmode{\varphipluss}{m,\sfrak}=\Delta^{-\sfrak}\ellmode{\psipluss}{m,\sfrak}$.  For further analysis, we expand $\Proj{m,\sfrak}(H[\psipluss])$ as follows:
\begin{align}\label{sourceterm:psipluss}\begin{split}
& \Proj{m,\sfrak}(H[\psipluss])\\
=&-2\sqrt{\R}(\mu \Hhyp-1)\prb (\sqrt{\R}\ellmode{\psipluss}{m,\sfrak})\\
&
-(\R)\mu\Hhyp
(\Hhyp -2\mu^{-1})\Lxi\ellmode{\psipluss}{m,\sfrak}
+2iam\big[1+(\Hhyp-2\mu^{-1})\big]\ellmode{\psipluss}{m,\sfrak}\\
&
-\big[(\R)\partial_r (\mu \Hhyp)+2s((r-M)(2\mu^{-1}-\Hhyp)-2r)
\big]\ellmode{\psipluss}{m,\sfrak}\\
&-\Proj{m,\sfrak}(a^2\sin^2\theta\Lxi\psipluss)+\Proj{m,\sfrak}(2ia\sfrak\cos\theta\psipluss)\\
=&2\mu^{\sfrak}\frac{1}{\sqrt{\R}}\tildePhisHighell{+\sfrak}{m,\sfrak}-2\mu \Hhyp\sqrt{\R}\prb (\sqrt{\R}\ellmode{\psipluss}{m,\sfrak})\\
&
-\mu(\R)(\Hhyp)^2\Lxi\ellmode{\psipluss}{m,\sfrak}
+[2iam\Hhyp-(\R)\partial_r (\mu \Hhyp)]\ellmode{\psipluss}{m,\sfrak},
\end{split}\end{align}
with $\tildePhisHighell{+\sfrak}{m,\sfrak}=\Proj{m,\sfrak}\big(\curlVR\PhiplussHigh{0}
-\half\big(2a\Leta\PhiplussHigh{0}
+a^2\sin^2 \theta\Lxi \PhiplussHigh{0}
-2ia\sfrak\cos\theta \PhiplussHigh{0}\big)\big)$ by the definition in Proposition \eqref{prop:wavesys:tildePhisHighi:ellmode}.
Further,  from formula \eqref{sourceterm:psipluss}, one finds  that $\Proj{m,\sfrak}(H[\psipluss])$ is smooth up to and including horizon and it holds as $\rb\to \infty$,
\begin{align}\label{formula:sourceterm:F}
\Proj{m,\sfrak}(H[\psipluss])= \frac{2\mu^{\sfrak}}{\sqrt{\R}}\tildePhisHighell{+\sfrak}{m,\sfrak}+ O(r^{-2}) r\prb\ellmode{\Psipluss}{m,\sfrak}+O(r^{-2}) \ellmode{\Psipluss}{m,\sfrak}+O(r^{-2}) \Lxi\ellmode{\Psipluss}{m,\sfrak}.
\end{align}

Before stating the conservation law, we introduce some notations and calculate the both sides of the TSI \eqref{eq:TSI:simpleform} and \eqref{eq:TSIspin2:V-2} in the following lemma.

\begin{lemma}\label{lemma:expand:TST:righthand:V:short}

\begin{itemize}
\item For $\sfrak=1,2$,  we have for the RHS of the TSI \eqref{eq:TSI:simpleform} and \eqref{eq:TSIspin2:V-2} that
\begin{align}
   \label{eq:TSIRHS:expand:short}
 \Delta^{\sfrak} \VR^{2\sfrak}(\Delta^{\sfrak}\psiminuss)={}&\sum_{j=0}^{2\sfrak}\tilde{u}_{-\sfrak,j}(r) V^j\psiminuss\notag\\
 \doteq{}& \sum_{j=1}^{2\sfrak}(C_{2\sfrak}^j-\delta_{-\sfrak,j})(\R)^{j}(r-M)^{2\sfrak-j} u_{-\sfrak, j}(r) V^{j} \psiminuss
 +(\R)^{\sfrak}u_{-\sfrak,0}(r)\psiminuss,
 \end{align}
 where $\delta_{-\sfrak,j}=0$ except that $\delta_{-2,2}=10$  and  $\delta_{-2,1}=20$, and
 \begin{align}\begin{split}
& u_{-\sfrak,j}(r)=1+u_{-\sfrak,j}^{(1)}(r)\mu +u_{-\sfrak,j}^{(2)}(r)\mu^2,\ \quad \text{for}\ \ j=1,\cdots,2\sfrak,\\
& u_{-\sfrak,0}(r)=(2\sfrak)!\mu^{\sfrak},
 \end{split}\end{align}
 with $u_{-\sfrak,j}^{(1)}(r)=O(1)$ and $u_{-\sfrak,j}^{(2)}(r)=O(1)$ being smooth functions up to and including horizon.
 \item On $\Horizon$,  we can expand $\Delta^{\sfrak} \VR^{2\sfrak}(\Delta^{\sfrak}\psiminuss)\vert_{\Horizon}$ as follow:
 \begin{align}
 \label{eq:TSI:RHS:Horizon:exp}
 \Delta^{\sfrak} \VR^{2\sfrak}(\Delta^{\sfrak}\psiminuss)\big|_{\Horizon}
 =\sum_{n=1}^{2\sfrak}(C_{2\sfrak}^n-\delta_{-\sfrak,n})(r_+-M)^{2\sfrak-n}
 (2a\Leta)^n\psiminuss+
 \sum_{j+k\leq2\sfrak-1}\mathfrak{b}_{-\sfrak,j,k}\Lxi^{j+1}\Leta^k\psiminuss
 \end{align}
 where $\mathfrak{b}_{\sfrak,j,k}=(C_{2\sfrak}^{j+1+k}-\delta_{-\sfrak,j+1+k})C_{j+1+k}^{j+1}
 2^{2j+2+k}a^k(Mr_+)^{j+1}(r_+-M)^{2\sfrak-1-j-k}$.
 \item
 For $\sfrak= 1,2$ and $|m|\leq\sfrak$,
let $\mathfrak{a}_{m,\sfrak}$ be the unique differential operator such that
   \begin{align}
   \label{definition:psipluss:tau:0:r+}
   \Lxi\mathfrak{a}_{m,\sfrak}(\psipluss)
     =\Proj{m,\sfrak}^{-\sfrak}\big((\edthR'-ia\sin\theta\Lxi)^{2\sfrak}\psipluss-12M(\sfrak-1)\Lxi\overline{\psipluss}\big) -\Proj{m,\sfrak}^{-\sfrak}\big((\edthR')^{2\sfrak}\psipluss\big).
     \end{align}
     Then the $(m,\sfrak)$ mode projection form of the TSI \eqref{eq:TSI:simpleform} and \eqref{eq:TSIspin2:V-2} for $\sfrak=1, 2$ becomes
     \begin{align}
     \label{eq:TSI:proj:8493}
     \Proj{m,\sfrak}^{-\sfrak}\big((\edthR')^{2\sfrak}\psipluss+ \Lxi\mathfrak{a}_{m,\sfrak}(\psipluss)\big)  = \Proj{m,\sfrak}\big( \Delta^{\sfrak} \VR^{2\sfrak}(\Delta^{\sfrak}\psiminuss)\big).
     \end{align}
 \end{itemize}
\end{lemma}

\begin{proof}
In fact, one can expand out $\Delta^{\sfrak} \VR^{2\sfrak}(\Delta^{\sfrak}\psiminuss)$ and obtain
\begin{subequations}
 \begin{align}
 \Delta\VR^2(\Delta\psiminus)
 =(\R)^2 V^2\psiminus +(2(r-M) (\R)+2r\Delta)V\psiminus  +2\Delta \psiminus
 \end{align}
 and
\begin{align}
 &\Delta^2\VR^4 (\Delta^2 \psiminustwo)\notag\\
 =&(\R)^4 V^4 \psiminustwo
 +\Big[4(r-M)(\R)^3+12r(\R)^2\Delta \Big] V^3 \psiminustwo\notag\\
 &+\Big[-4(r-M)^2(\R)^2+(40r(r-M)+16)(\R)\Delta+
 (20r^2+8a^2)\Delta^2\Big]V^2\psiminustwo\notag\\
 &+\Big[-16(r-M)^3(\R)+(40(\R)+16r(r-M))(r-M)\Delta+(56r-20M)\Delta^2
 \Big]V\psiminustwo\notag\\
 &+24\Delta^2  \psiminustwo.
\end{align}
  \end{subequations}
Formula \eqref{eq:TSIRHS:expand:short} then follows.  By restricting these equations on $\Horizon$ and using $V\big|_{\Horizon}=2\Lxi+\frac{a}{M r_+}\Leta$, one immediately  achieves equation \eqref{eq:TSI:RHS:Horizon:exp}.
\end{proof}

In the theorem below, a global conservation law is derived, and using this conservation law, the integral of the radiation field of an $(m,\sfrak)$ mode of the spin $+\sfrak$ component along future null infinity is calculated in terms of the initial data of the spin $\pm\sfrak$ components on $\Sigmazero$.

\begin{thm}[Global conservation law]
Assume $\IE{\reg,\delta}{\tb_0}<+\infty$ for a sufficiently small $\delta>0$ and some suitably large integer $\reg$.  Then, we have for $\sfrak=0,1,2$ and $|m|\leq\sfrak$ the following conservation law
\begin{align}\label{eq:conservation:PL}\begin{split}
\hspace{4ex}&\hspace{-4ex}(2\sfrak+1)\int_{\tau_0}^{\infty}
\lim\limits_{\rho\to\infty}\ellmode{\Phipluss}{m,\sfrak}(\tau,\rho)\di\tau\\
=&-[2iam -2\sfrak(r_+-M)]\int_{\tb_0}^{\infty}\ellmode{\psipluss}{m,\sfrak}\vert_{\Horizon}\di \tb
+\int_{r_+}^{+\infty} \Proj{m,\sfrak}(H[\psipluss])(\rho,\tb_0)\di\rb
\end{split}\end{align}
and the following expression of the value of the integral of $\ellmode{\Phipluss}{m,\sfrak}$ along future null infinity\footnote{Note that for $\sfrak=0$ (hence $(m=0)$), the above formula \eqref{eq:conservation:PL} already provides the value of integral of $\ellmode{\Phi_0}{0,0}$ along future null infinity in terms of the initial hypersurface integral $\int_{r_+}^{+\infty} \Proj{0,0}(H[\psi_0])(\rho,\tb_0)\di\rb$.}
\begin{align}\label{nullinfity:integral:ellmode:psipluss}
\hspace{4ex}&\hspace{-4ex}(2\sfrak+1)\int_{\tau_0}^{\infty}
\lim\limits_{\rho\to\infty}\ellmode{\Phipluss}{m,\sfrak}(\tau,\rho)\di\tau\notag\\
=&\int_{r_+}^{+\infty}\Proj{m,\sfrak}(H[\psipluss])(\rho,\tb_0)\di\rb\notag\\
 &\quad-\frac{2iam -2\sfrak(r_+-M)}{(2\sfrak)!}\Big\{\mathfrak{a}_{m,\sfrak}(\psipluss)(\tau_0,r_+)
-\sum_{j+k\leq2\sfrak-1}(im)^k\mathfrak{b}_{-\sfrak,j,k}\Lxi^j\ellmode{\psiminuss}{m,\sfrak}(\tau_0,r_+)\notag\\
&\qquad\qquad+\mathfrak{c}_{m,\sfrak}\int_{r_+}^{+\infty}\Delta^{-\sfrak-1}(r)
\Scaling^{-1}(r)
\int_{r_+}^r\Scaling(r')\Delta^{\sfrak}(r')\Proj{m,\sfrak}(H[\psiminuss])(r',\tau_0)\di r'\di r\Big\},
\end{align}
where $\mathfrak{c}_{m,\sfrak}=\sum_{n=1}^{2\sfrak}(C_{2\sfrak}^n-\delta_{-\sfrak,n})(r_+-M)^{2\sfrak-n}
 (2iam)^n$ for $\sfrak=1,2$ and  $\mathfrak{c}_{0,0}=1$,  $\mathfrak{a}_{m,\sfrak}(\psipluss)$ and $\mathfrak{b}_{-\sfrak,j,k}$ are defined as in Lemma \ref{lemma:expand:TST:righthand:V:short} for $\sfrak=1,2$ and $\mathfrak{a}_{0,0}(\psi_0)=0$ for $\sfrak=0$, and $\Scaling(r)=e^{\int_{r_+}^r \frac{2iam}{\Delta(r')}d r' }$ as defined in \eqref{def:Scalingw(r)}.
\end{thm}
\begin{proof}

{\it Step 1. Conservation law.}
By assumption and the estimates \eqref{eq:PWD:pluss:0:428} and \eqref{eq:PWD:pluss:0:tildePhi:428}, there exits a small $\delta$ such that  $|\tildePhisHighell{+\sfrak}{m,\sfrak}|\lesssim v^{-1+\delta}(\IE{\regl}{\tau_0})^{\half}$ and $\Lxi^j(\rho\partial_\rho)^i\ellmode{\Psipluss}{m,\sfrak}\lesssim \tau^{-2-j+\delta}(\IE{\regl}{\tau_0})^{\half}$ for $i,j\leq1$, which suggest
\begin{align}
\label{eq:esti:Proj:ms:H:639}
 |\Proj{m,\sfrak}(H[\psipluss])|\lesssim_{\delta}(\rho^{-1-\veps}v^{-1+\delta+\veps}
 +\rho^{-2}\tau^{-2+\delta})(\IE{\regl}{\tau_0})^{\half}
 \end{align}
  by formula \eqref{sourceterm:psipluss}.

 We integrate  equation \eqref{eq:TME:varphipluss:hypercoords:PL} in $\mathcal{D}(\tb_0,\tb',r')=\{ (\tb,\rho)|\tb_0\leq \tb\leq \tb', r_+\leq\rho\leq r'\}$ and obtain
\begin{align}\label{eq:conservation:PL-1}
\hspace{4ex}&\hspace{-4ex}
\int_{r_+}^{r'} \Proj{m,\sfrak}(H[\psipluss])(\rho,\tb')\di\rb - \int_{r_+}^{r'} \Proj{m,\sfrak}(H[\psipluss])(\rho,\tb_0)\di\rb +\int_{\tb_0}^{\tb'} (\Delta^{\sfrak+1}\prb \varphi+2iam\Delta^{\sfrak}\varphi)\vert_{\Horizon}\di \tb  \notag\\
={}&\int_{\tb_0}^{\tb'} (\Delta^{\sfrak+1}\prb \varphi+2iam\Delta^{\sfrak}\varphi)(\tau,r')\di \tb\notag\\
={}&\int_{\tb_0}^{\tb'} \bigg(\frac{\Delta^{\sfrak+1}}{(\R)^{\sfrak+\frac{3}{2}}}(\R)\prb \mellmode{\Phipluss}{m}{\sfrak} +\bigg(\Delta^{\sfrak+1}\partial_r ((\R)^{-\sfrak-\half})+
\frac{2iam\Delta^{\sfrak}}{(\R)^{\sfrak+\half}}\bigg)\mellmode{\Phipluss}{m}{\sfrak}\bigg)(\tau,r')\di \tb.
\end{align}
The first term on the LHS is bounded by $C_{\delta}(\tau')^{-1+\delta+\veps}(\IE{\regl}{\tau_0})^{\half}$ in view of the above bound \eqref{eq:esti:Proj:ms:H:639} for $\Proj{m,\sfrak}(H[\psipluss])$.
Further, taking $r'\to\infty$, and by the boundedness  of both $|\ellmode{\Phipluss^{(1)}}{m,\sfrak}|$ and $|\ellmode{\Phipluss}{m,\sfrak}|$, the RHS equals $(2\sfrak+1)\int_{\tau_0}^{\tau'}
\lim\limits_{\rho\to\infty}\ellmode{\Phipluss}{m,\sfrak}(\tau,\rho)\di\tau$, and the last term in the first line equals $-\int_{\tb_0}^{\tb'}[2iam -2\sfrak(r_+-M)]\ellmode{\psipluss}{m,\sfrak}\vert_{\Horizon}\di \tb$. In total,  we achieve
\begin{align}
\hspace{4ex}&\hspace{-4ex}(2\sfrak+1)\int_{\tau_0}^{\tau'}
\lim\limits_{\rho\to\infty}\ellmode{\Phipluss}{m,\sfrak}(\tau,\rho)\di\tau\notag\\
=&\int_{r_+}^{+\infty} \Proj{m,\sfrak}(H[\psipluss])(\rho,\tb_0)\di\rb-\int_{\tb_0}^{\tb'}[2iam -2\sfrak(r_+-M)]\ellmode{\psipluss}{m,\sfrak}\vert_{\Horizon}\di \tb
+O((\tau')^{-1+\delta+\veps}).
\end{align}
The almost sharp decay estimate \eqref{eq:AP:pluss:0:428} ensures that we can directly take  the limit  $\tau'\to\infty$, and this yields the global conservation law \eqref{eq:conservation:PL} in the black hole exterior region.

{\it Step 2. Calculating the integral along future null infinity in terms of the initial data.}
Now, we are going to compute the first term in the last line of \eqref{eq:conservation:PL}. That is to say, we shall calculate the integral of $\ellmode{\psipluss}{m,\sfrak}$ along the event horizon.
For $\sfrak=0$, we can directly calculate the integral expression of $\ellmode{\psi_0}{0,0}$, while for $\sfrak=1,2$, we should instead first calculate the integral expression of $\ellmode{\psiminuss}{m,\sfrak}$ and then utilize the TSI to determine the value of the horizon integral of $\ellmode{\psipluss}{m,\sfrak}$.

For $\sfrak=1,2$, we first use the TSI to calculate the horizon integral of $\ellmode{\psipluss}{m,\sfrak}$ from the horizon integral of $\ellmode{\psiminuss}{m,\sfrak}$. Recall the mode projection form \eqref{eq:TSI:proj:8493} of the TSI. By restricting  \eqref{eq:TSI:proj:8493}  on $\Horizon$ and using \eqref{eq:TSI:RHS:Horizon:exp}, and by \eqref{eq:ellipticop:eigenvalue:fixedmode} that indicates $ \Proj{m,\sfrak}^{-\sfrak}((\edthR')^{2\sfrak}\psipluss)=(2\sfrak)!\ellmode{\psipluss}{m,\sfrak}$, we have
\begin{align}\label{ellmode:TSI:onHorizon}\begin{split}
\hspace{4ex}&\hspace{-4ex}(2\sfrak)!\ellmode{\psipluss}{m,\sfrak}+\Lxi\mathfrak{a}_{m,\sfrak}(\psipluss)\\
={}&\sum_{n=1}^{2\sfrak}(C_{2\sfrak}^n-\delta_{-\sfrak,n})(r_+-M)^{2\sfrak-n}
 (2iam)^n\ellmode{\psiminuss}{m,\sfrak}
+\Lxi\Big(\sum_{j+k\leq2\sfrak-1}(im)^k\mathfrak{b}_{-\sfrak,j,k}\Lxi^j\ellmode{\psiminuss}{m,\sfrak}\Big).
\end{split}\end{align}
Integrating this equation along $\Horizon$, we get
\begin{align}\label{horizon:integral:ellmode:psipluss}
(2\sfrak)!\int_{\tau_0}^{+\infty}\ellmode{\psipluss}{m,\sfrak}|_{\Horizon}\di\tau
={}&\mathfrak{a}_{m,\sfrak}(\psipluss)(\tau_0,r_+)
-\sum_{j+k\leq2\sfrak-1}(im)^k\mathfrak{b}_{-\sfrak,j,k}\Lxi^j\ellmode{\psiminuss}{m,\sfrak}
(\tau_0,r_+)\notag\\
&+\sum_{n=0}^{2\sfrak}(C_{2\sfrak}^n-\delta_{-\sfrak,n})(r_+-M)^{2\sfrak-n}
 (2iam)^n
\int_{\tau_0}^\infty\ellmode{\psiminuss}{m,\sfrak}\vert_{\Horizon}\di\tau.
\end{align}

It remains to calculate the last term on the RHS of \eqref{horizon:integral:ellmode:psipluss}, i.e. the horizon integral of  $\ellmode{\psiminuss}{m,\sfrak}$. By \eqref{eq:TME:psiminuss:hypercoords}, for $\rho>r_+$ and  $\sfrak=0,1,2$, we have
\begin{align}\label{integral:formula:psiminuss}
\ellmode{\psiminuss}{m,\sfrak}(\rho,\tau)=-\int_{\rho}^{+\infty}\Delta^{-\sfrak-1}(r)\Scaling^{-1}(r)
\int_{r_+}^r\Scaling(r')\Delta^{\sfrak}(r')\Lxi (\Proj{m,\sfrak}H[\psiminuss])(r',\tau)\di r'\di r,
\end{align}
the integral on the RHS of which is well-defined for any fixed $\tau\geq\tau_0$ since  $H[\psiminuss]=O(\rho^{-1})$ as $\rho$ sufficiently large. Further, it is easy to show that the integral in \eqref{integral:formula:psiminuss}  is continuous up to and including horizon,
hence, formula \eqref{integral:formula:psiminuss} holds on $\rb=r_+$ as well. By integrating \eqref{integral:formula:psiminuss} in $\tau$ on $\Horizon$,  we conclude for $\sfrak=0,1,2$,
\begin{align}\label{horizon:integral:ellmode:psiminuss}
\int_{\tau_0}^\infty\ellmode{\psiminuss}{m,\sfrak}\vert_{\Horizon}\di\tau
=\int_{r_+}^{+\infty}\Delta^{-\sfrak-1}(r)\Scaling^{-1}(r)
\int_{r_+}^r\Scaling(r')\Delta^{\sfrak}(r')\Proj{m,\sfrak}(H[\psiminuss])(r',\tau_0)\di r'\di r,
\end{align}
since the value as $\tb\to \infty$ vanishes by the estimate \eqref{eq:esti:Proj:ms:H:639}.

In the end, for $\sfrak=1,2$, we substitute \eqref{horizon:integral:ellmode:psipluss} and \eqref{horizon:integral:ellmode:psiminuss} into \eqref{eq:conservation:PL} to achieve \eqref{nullinfity:integral:ellmode:psipluss}, while for $\sfrak=0$, it suffices to substitute only \eqref{horizon:integral:ellmode:psiminuss} with $\sfrak=0$ into \eqref{eq:conservation:PL}.
\end{proof}

Additionally, we are also able to compute the integrals of $\ellmode{\PhiplussHigh{j}}{m,\ell}$, $\ell> \sfrak$ and $0\leq j<\ell-\sfrak$, on future null infinity.

\begin{lemma}
Let $\sfrak=0, 1, 2$.  Assume $\IE{\reg,\delta}{\tau_0}<+\infty$ for a sufficiently small $\delta>0$  and  some suitably large integer $k$. Then,
for $\ell>\sfrak$ and $0\leq j<\ell-\sfrak$, we have
\begin{align}\label{nullinfity:integral:ellmode:psipluss:high}\begin{split}
\int_{\tau_0}^{+\infty}\lim\limits_{\rho\to\infty}
\ellmode{\PhiplussHigh{j}}{m,\ell}(\tau,\rho)\di\tau
={}&\frac{2}{(\ell-\sfrak-j)(\ell+\sfrak+j+1)}
\lim\limits_{\rho\to\infty}\ellmode{\tildePhiplussHigh{j}}{m,\ell}(\tau_0,\rho)\\
&-\sum_{j'=0}^{j-1}\sum_{n=0}^{j-j'}\frac{2(im)^n x_{\sfrak,j,j',n}}{(\ell-\sfrak-j')(\ell+\sfrak+j'+1)}\lim\limits_{\rho\to\infty}
\ellmode{\tildePhiplussHigh{j'}}{m,\ell}(\tau_0,\rho),
\end{split}\end{align}
where
\begin{align}\label{ansatz:tildePhisHigh:ellmode:i}
\ellmode{\tildePhiplussHigh{j}}{m,\ell}=\Proj{m,\ell}^{+\sfrak}\Big(\curlVR\hatPhiplussHigh{j}
-\half\big(2a\Leta\hatPhiplussHigh{j}
+a^2\sin^2 \theta\Lxi \hatPhiplussHigh{j}
-2ia\sfrak\cos\theta \hatPhiplussHigh{j}\big)\Big).
\end{align}
\end{lemma}
\begin{proof}
Similar to the proof in Proposition \ref{prop:wavesys:tildePhisHighi:ellmode}, we  rewrite \eqref{eq:wave:hatPhisHighi:an:ellmode} as
\begin{align}
\mu Y\ellmode{\tildePhiplussHigh{i}}{m,\ell}
+[(\ell+\sfrak)(\ell-\sfrak+1)-(2\sfrak+i)(i+1)]\ellmode{\hatPhiplussHigh{i}}{m,\ell}
+O(r^{-1})=0.
\end{align}
Since $\mu Y(\ellmode{\tildePhiplussHigh{i}}{m,\ell})=2\Lxi \ellmode{\tildePhiplussHigh{i}}{m,\ell}+O(r^{-1})rV\ellmode{\tildePhiplussHigh{i}}{m,\ell}+O(r^{-2})\Leta\ellmode{\tildePhiplussHigh{i}}{m,\ell}$,
by integrating the above equation from $\tb_0$ to $\tb'$ and taking $\rho\to\infty$, we achieve
\begin{align}
\int_{\tb_0}^{\tb'}(\ell-\sfrak-i)(\ell+\sfrak+i+1)\lim\limits_{\rho\to\infty}
\ellmode{\hatPhiplussHigh{i}}{m,\ell}(\tau,\rho)\di\tb={}&
-2\Big(\lim\limits_{\rho\to\infty}\ellmode{\tildePhiplussHigh{i}}{m,\ell}(\tau,\rho)\Big)\Big|_{\tb_0}^{\tb'}.
\end{align}
We then take $\tb'\to +\infty$, and  since $\ellmode{\tildePhiplussHigh{i}}{m,\ell}(\tau,\rho)$ decays in $\tb$, we get
\begin{align}
\int_{\tau_0}^{+\infty}\lim\limits_{\rho\to\infty}
\ellmode{\hatPhiplussHigh{i}}{m,\ell}(\tau,\rho)\di\tau
=\frac{2}{(\ell-\sfrak-i)(\ell+\sfrak+i+1)}
\lim\limits_{\rho\to\infty}\ellmode{\tildePhiplussHigh{i}}{m,\ell}(\tau_0,\rho).
\end{align}
In the end,  in view of
the definition of $\hatPhiplussHigh{j}$ in Proposition \ref{prop:wavesys:hatPhisHighi} which reads
\begin{align}
\ellmode{\hatPhiplussHigh{j}}{m,\ell}={}&\ellmode{\PhiplussHigh{j}}{m,\ell}+\sum_{j'=0}^{j-1} \sum_{n=0}^{j-j'}(im)^n x_{\sfrak,j,j',n}\ellmode{\hatPhiplussHigh{j'}}{m,\ell},
\end{align}
formula \eqref{nullinfity:integral:ellmode:psipluss:high} then follows.
\end{proof}

\begin{remark}
In particular, if the initial data on $\Sigma_{\tau_0}$ are compactly supported or decay sufficiently faster as $\rb\to +\infty$,  then equality \eqref{nullinfity:integral:ellmode:psipluss:high} actually implies $\int_{\tau_0}^{+\infty}\lim\limits_{\rho\to\infty}
\ellmode{\PhiplussHigh{j}}{m,\ell}(\tau,\rho)\di\tau=0$ for any $\ell>\sfrak$ and $0\leq j<\ell-\sfrak$.
\end{remark}


\subsection{Proof of the sharp decay}
\label{subsection:proof:sharpdecay}

To show the sharp decay (i.e. the Price's law), we will frequently use the coordinates $(u,v,\theta,\tilde{\phi})$, and the partial derivatives $\partial_u$ and $\partial_v$ shall be understood in this coordinate system. In this $(u, v, \theta,\pb)$ coordinate system, we can express $\partial_u$ and $\partial_v$ as
\begin{align}
\label{exp:pupv}
\partial_u=\half \mu Y,\quad \partial_v=\half \tildeV=\half \Big(V-\frac{2a}{\R}\Leta\Big).
\end{align}

The following lemma lists some useful relations and estimates among $u$, $v$, $r$, and $\tb$ that are utilized in different regions in our proof for sharp decay estimates. The proof is simple and omitted.

\begin{lemma}
For any $\alpha\in (\half, 1)$, let $\gamma_{\alpha}=\{r=v^\alpha\}$. For any $u$ and $v$, let $u_{\gamma_\alpha}(v)$ and $v_{\gamma_{\alpha}}(u)$ be such that $(u_{\gamma_\alpha}(v),v), (u, v_{\gamma_{\alpha}}(u))\in \gamma_{\alpha}$.
In the region $r\geq v^{\alpha}$,
\begin{subequations}
\begin{align}
\label{eq:rela:a}
&r\gtrsim{} v^{\alpha}+u^{\alpha},\\
\label{eq:rela:b}
&\abs{u-v_{\gamma_{\alpha}}(u)}\lesssim{}u^{\alpha},\\
\label{eq:rela:c}
&\abs{2r-(v-u)}\lesssim{}\log (r-r_+);
\end{align}
in the region $\{r\geq v^\alpha\}\cap\{r\geq\frac{v}{4}\}$,
\begin{align}
\label{eq:rela:d}
&v+u\lesssim r\lesssim v;
\end{align}
in the region $\{r\geq v^{\alpha}\}\cap\{r\leq\frac{v}{4}\}$,
\begin{align}
\label{eq:rela:e}
&u \sim v, \quad r\gtrsim v^{\alpha};
\end{align}
in the region $\{r\leq v^{\alpha}\}$,
\begin{align}
\label{eq:rela:f}
&v\sim \tb.
\end{align}
\end{subequations}
On $\Sigmazero$, for $r$ large,
\begin{align}\label{asymp-factor}
\bigg|r^{-1}v- 2-\frac{4Mr^{-1}}{r_--r_+}\log\frac{(r-r_-)^{r_-}}{(r-r_+)^{r_+}}\bigg|\lesssim r^{-1}.
\end{align}
\end{lemma}



Our analysis starts from deriving the precise asymptotic profile of the $(m,\sfrak)$ mode of the spin $+\sfrak$ component. We first make an assumption on the initial data of this mode towards $\rb\to +\infty$.

\begin{assump}[Initial data assumption to order $i$] \label{assump:initialdata:nonvanishing:pm1} Let $\sfrak=0, 1, 2$,  let $i\in \mathbb{N}$, and let $\abs{m}\leq \sfrak$. Let $\tildePhisHighell{+\sfrak}{m,\sfrak}$ be defined as in Proposition \ref{prop:wavesys:tildePhisHighi:ellmode}. Assume on $\Sigma_{\tau_0}$ that  there are constants $\ql_{m,\sfrak}\in \mathbb{R}\setminus\{0\}$, $\beta\in (0,\half)$ and $0\leq D_0<\infty$ such that for all $0\leq i'\leq i$ and $\rb \geq 10M$,
\begin{align}
\Big|\prb^{i'}\big(r^{-2\sfrak-2}\tildePhisHighell{+\sfrak}{m,\sfrak}
-r^{-2\sfrak-3}\ql_{m,\sfrak}\big)(\tau_0,\rho)\Big|\lesssim D_0\rb^{-2\sfrak-3-\beta-i'}.
\end{align}
\end{assump}

\subsubsection{Sharp decay for $\ellmode{\Phipluss}{m,\sfrak}$ in $\{r\geq v^{\alpha}\}$}
\label{sect:esti:tildePhisHighell:zero}

To being with, we utilize equation \eqref{ansatz:tildePhisHigh:mellmode:2} for $\ell=\sfrak$, $s=\sfrak$ which reads
\begin{align}
\label{eq:wave:tildePhisHighi:an:ellmode:ell=sfrak}
-\mu Y\tildePhisHighell{+\sfrak}{m,\sfrak}
-\frac{2(\sfrak+1)(r^3-3Mr^2 +a^2 r+a^2 M)}{(\R)^2}\tildePhisHighell{+\sfrak}{m,\sfrak}
={}&(\R)^{-\half}\G,
\end{align}
and, a simple scaling for the above equation \eqref{eq:wave:tildePhisHighi:an:ellmode:ell=sfrak} yields
\begin{align}
\label{eq:wave:tildePhisHighi:an:ellmode:zero}
-\mu Y ( \mu^{\sfrak+1} (\R)^{-\sfrak-1}\tildePhisHighell{+\sfrak}{m,\sfrak})
= {}&\mu^{\sfrak+1} (\R)^{-\sfrak-\frac{3}{2}}\G.
\end{align}
Here,
\begin{align}\label{definition:righthand:tildephi:G}
\G
=&(2(\sfrak+1)(2\sfrak+1)M-2iam\sfrak)\ellmode{\Phipluss}{m,\sfrak}+(\sfrak+1)
\Proj{m,\sfrak}(a^2\sin\theta\Lxi\Phipluss-2ia\sfrak\cos\theta\Phipluss)\notag\\
&\quad+\half (rV+O(r^{-1}))\big(
\Proj{m,\sfrak}(a^2\sin\theta\Lxi\Phipluss-2ia\sfrak\cos\theta\Phipluss)\big)
+O(r^{-1})\ellmode{\Phipluss}{m,\sfrak}
\end{align}
which follows from $\eqref{definition:tildeH:2345}$ and \eqref{eq:wave:PhisHigh0}.

For future applications, we rewrite $\G$ into a different form. First, the definition of $\tildePhisHighell{+\sfrak}{m,\sfrak}$ in Proposition \ref{prop:wavesys:tildePhisHighi:ellmode} implies
\begin{align}\label{vpHi:tildePhi:123}
V\ellmode{\Phipluss}{m,\sfrak}\sim {r^{-2}}\big(\tildePhisHighell{+\sfrak}{m,\sfrak}
+O(1)\Lxi \ellmode{\Phipluss}{m,\ell\leq \sfrak+2}+O(1)\ellmode{\Phipluss}{m,\ell\leq \sfrak+1}\big).
\end{align}
Combining \eqref{vpHi:tildePhi:123}, Proposition \ref{prop:modeprojection:1} and the definition of $\hatPhiplussHigh{i}$ in Definition \ref{def:Phiminusi}, we have
\begin{align}
\G=&(2(\sfrak+1)(2\sfrak+1)M-2iam\sfrak)\ellmode{\Phipluss}{m,\sfrak}\notag\\
&+(\sfrak+1)a^2\sum_{\sfrak\leq\ell\leq \sfrak+2}c_{m,\ell}^{\sfrak}\Lxi\ellmode{\Phipluss}{m,\ell}
-2ia\sfrak(\sfrak+1)\sum_{\sfrak\leq\ell\leq\sfrak+1}b_{m,\ell}^\sfrak\ellmode{\Phipluss}{m,\ell}\notag\\
&+O(1)\Lxi (rV)\ellmode{\Phipluss}{m,\ell\leq \sfrak+2}+O(1)(rV)\ellmode{\Phipluss}{m,\ell\leq \sfrak+1}+O(r^{-1})\Lxi^{\leq1}\ellmode{\Phipluss}{m,\ell\leq \sfrak+2}\notag\\
={}&(2(\sfrak+1)(2\sfrak+1)M-2iam\sfrak)\ellmode{\Phipluss}{m,\sfrak}\notag\\
&+(\sfrak+1)a^2\sum_{\sfrak\leq\ell\leq \sfrak+2}c_{m,\ell}^{\sfrak}\Lxi\ellmode{\Phipluss}{m,\ell}
-2ia\sfrak(\sfrak+1)\sum_{\sfrak\leq\ell\leq\sfrak+1}b_{m,\ell}^\sfrak\ellmode{\Phipluss}{m,\ell}\notag\\
&+O(r^{-1})\Big(\sum_{\ell=\sfrak+1}^{\sfrak+2}\Lxi^{\leq 1}\ellmode{\hatPhiplussHigh{1}}{m,\ell}
+\Lxi^{\leq1}\tildePhisHighell{+\sfrak}{m,\sfrak}+
\Lxi^{\leq1}\ellmode{\Phipluss}{m,\ell\leq \sfrak+2}\Big).
\end{align}
 Further, we have from the above formula that for any $j\in \Integers^+$,
\begin{align}\label{definition:righthand:tildephi:highG}
V^j \G\sim r^{-1-j}\Big(\Lxi^{\leq1}(rV)^{\leq j}\tildePhisHighell{+\sfrak}{m,\sfrak}
+\sum_{\ell=\sfrak+1}^{\sfrak+2}\Lxi^{\leq1}(rV)^{\leq j}\ellmode{\hatPhiplussHigh{1}}{m,\ell}+
\Lxi^{\leq1}\ellmode{\Phipluss}{m,\ell\leq \sfrak+2}\Big).
\end{align}

In view of \eqref{nullinfity:integral:ellmode:psipluss:high},
we have
\begin{align}
\hspace{4ex}&\hspace{-4ex}
-2ia\sfrak(\sfrak+1)\int_{\tb_0}^{+\infty}\sum_{\sfrak\leq\ell\leq\sfrak+1}b_{m,\ell}^{+\sfrak}\lim\limits_{\rho\to+\infty}\ellmode{\Phipluss}{m,\ell}(\tb,\rb)\di\tb\notag\\
={}&-2ia\sfrak(\sfrak+1)\sum_{\sfrak\leq\ell\leq\sfrak+1}b_{m,\ell}^{+\sfrak}\frac{2}{(\ell-\sfrak)(\ell+\sfrak+1)}
\lim\limits_{\rho\to\infty}\ellmode{\tildePhiplussHigh{0}}{m,\ell}(\tau_0,\rho)\notag\\
={}&
-2ia\sfrak(\sfrak+1)\sum_{\sfrak\leq\ell\leq\sfrak+1}b_{m,\ell}^{+\sfrak}\frac{2}{(\ell-\sfrak)(\ell+\sfrak+1)}\notag\\
&\qquad\times
\lim\limits_{\rho\to\infty}
\Proj{m,\ell}^{+\sfrak}\Big(\curlVR\Phipluss
-\half\big(2a\Leta\Phipluss
+a^2\sin^2 \theta\Lxi \Phipluss
-2ia\sfrak\cos\theta \Phipluss\big)\Big)(\tb_0,\rb),
\end{align}
hence we are able to calculate the integral of $\G$ along future null infinity by \eqref{nullinfity:integral:ellmode:psipluss:high}:
\begin{align}
\label{eq:integralofG:scriplus}
\hspace{4ex}&\hspace{-4ex}
\int_{\tau_0}^{+\infty}
 \lim\limits_{\rho\to+\infty}\G(\tau,\rho)\di \tau\notag\\
 ={}&
 (2(\sfrak+1)(2\sfrak+1)M-2iam\sfrak)\int_{\tau_0}^{+\infty} \lim\limits_{\rho\to+\infty}\ellmode{\Phipluss}{m,\sfrak}(\tau,\rho)\di \tau\notag\\
 &- (\sfrak+1)a^2\sum_{\sfrak\leq\ell\leq \sfrak+2}c_{m,\ell}^{+\sfrak} \lim\limits_{\rho\to+\infty}\ellmode{\Phipluss}{m,\ell}(\tb_0,\rb)\notag\\
 &
-2ia\sfrak(\sfrak+1)\sum_{\sfrak\leq\ell\leq\sfrak+1}b_{m,\ell}^{+\sfrak}\frac{2}{(\ell-\sfrak)(\ell+\sfrak+1)}\notag\\
&\qquad\times
\lim\limits_{\rho\to\infty}
\Proj{m,\ell}^{+\sfrak}\Big(\curlVR\Phipluss
-\half\big(2a\Leta\Phipluss
+a^2\sin^2 \theta\Lxi \Phipluss
-2ia\sfrak\cos\theta \Phipluss\big)\Big)(\tb_0,\rb).
\end{align}

\begin{lemma}
\label{Kerr:prop:VtildePhil-1:nearinf:pm1}
Let $\sfrak=0, 1, 2$.  Assume the initial data assumption \ref{assump:initialdata:nonvanishing:pm1}  holds to order $0$, and the initial energy $\IE{\regl,\delta}{\tau_0}<+\infty$ for a sufficiently small $\delta>0$  and  some suitably large integer $\regl$. Then for  $\alpha$ sufficiently close to $1$ and $\delta=\delta(\alpha)$ sufficiently small, there exists an $\veps=\veps(\alpha, \delta)>0$  such that in the region $r\geq v^\alpha$,
\begin{align}\begin{split}
\label{esti:PL:nv:p:ell-sfrak}
\bigg|
(v-u)^{-2\sfrak-1}\ellmode{\Phipluss}{m,\sfrak}
 -4\MQ_{m,\sfrak}\frac{v+(2\sfrak+1)u}{(2\sfrak+2)(2\sfrak+1)v^{2\sfrak+2}u^2}
 \bigg|
 \lesssim_{\delta,\alpha}{}
(v-u)^{-2\sfrak-1}u^{-2-\veps}((\IE{\regl,\delta}{\tau_0})^{\half}+D_0).
\end{split}\end{align}
Here,
\begin{align}
\MQ_{m,\sfrak}=\ql_{m,\sfrak}
 -\half\int_{\tau_0}^{+\infty}
 \lim\limits_{\rho\to+\infty}\G(\tau,\rho)\di \tau,
\end{align}
where $\ql_{m,\sfrak}$ is determined in the initial data assumption \ref{assump:initialdata:nonvanishing:pm1}, $\int_{\tau_0}^{+\infty}
 \lim\limits_{\rho\to+\infty}\G(\tau,\rho)\di \tau$ is calculated in \eqref{eq:integralofG:scriplus} with $\Proj{m,\ell}^{+\sfrak}(\sin^2\theta\varphi_{+\sfrak})$ and $\Proj{m,\ell}^{+\sfrak}(\cos\theta\varphi_{+\sfrak})$ for a spin $+\sfrak$ scalar $\varphi_{+\sfrak}$ and the constants $b_{m,\ell}^{+\sfrak}$ and $c_{m,\ell}^{+\sfrak}$ defined as in Proposition \ref{prop:modeprojection:1}, and the integral $\int_{\tau_0}^{+\infty} \lim\limits_{\rho\to+\infty}\ellmode{\Phipluss}{m,\sfrak}(\tau,\rho)\di \tau$ is calculated in \eqref{nullinfity:integral:ellmode:psipluss}.
\end{lemma}

\begin{proof}
{\it Step 1. Asymptotics of $\tildePhisHighell{+\sfrak}{m,\sfrak}$.}
We integrate  equation \eqref{eq:wave:tildePhisHighi:an:ellmode:zero} along constant $v$ starting from $\Sigmazero$, and by \eqref{exp:pupv}, we obtain
\begin{align}
\label{eq:integrateconstv:123}
\hspace{4ex}&\hspace{-4ex}
\big(v^{2\sfrak+3} \mu^{\sfrak+1} (\R)^{-\sfrak-1}\tildePhisHighell{+\sfrak}{m,\sfrak}\big) (u,v) -
\big(v^{2\sfrak+3} \mu^{\sfrak+1} (\R)^{-\sfrak-1}\tildePhisHighell{+\sfrak}{m,\sfrak}\big) (u_{\Sigmazero}(v),v)\notag\\
={}&-\half\int_{u_{\Sigmazero}(v)}^u \mu^{\sfrak+1} \Big(\frac{v^2}{\R}\Big)^{\sfrak+\frac{3}{2}}\G (u',v)\di u'\notag\\
={}&-\half\int_{u_{\Sigmazero}(v')}^u \mu^{\sfrak+1} \bigg(\frac{v^2}{\R}\bigg)^{\sfrak+\frac{3}{2}}\G (u',v')\di u'\notag\\
&
+\half\int_{v}^{v'}\int_{u_{\Sigmazero}(v'')}^u  \Big(\half V-\frac{iam}{\R}\Big)\bigg(\mu^{\sfrak+1} \bigg(\frac{v^2}{\R}\bigg)^{\sfrak+\frac{3}{2}}\G\bigg) (u',v'')\di u'\di v''\notag\\
&+\half\int_{\Sigmazero\cap\{
\rho\geq \rho(\tau_0,v)\}}\half \mu\Hhyp\cdot\mu^{\sfrak+1} \bigg(\frac{v^2}{\R}\bigg)^{\sfrak+\frac{3}{2}}\G(\tau_0,\rho)
\di \rho
\end{align}
for any $v'>v$, and then we  take $v'\to + \infty$.

Next, we focus on analyzing the RHS of \eqref{eq:integrateconstv:123}.
The first integral in the third last line is equal to
\begin{align}
\hspace{4ex}&\hspace{-4ex}
-2^{2\sfrak+2}\int_{u_{\Sigmazero}(+\infty)}^u\lim\limits_{v'\to+\infty}\G(u',v')\di u'\notag\\
=&
-2^{2\sfrak+2}\int_{u_{\Sigmazero}(+\infty)}^{+\infty}\lim\limits_{v'\to+\infty}\G(u',v')\di u'
+2^{2\sfrak+2}\int_{u}^{+\infty}\lim\limits_{v'\to+\infty}\G(u',v')\di u'\notag\\
 ={}& -2^{2\sfrak+2}\int_{u_{\Sigmazero}(+\infty)}^{+\infty}\lim\limits_{v'\to+\infty}\G(u',v')\di u'+ O(u^{-1+C\delta})(\IE{\regl,\delta}{\tau_0})^{\half},
\end{align}
 where we have used the decay estimates in Corollary \ref{cor:PWD:pms:012} and Proposition \ref{prop:AP:pms:012} that imply
 \begin{align}
 \label{eq:G:esti:reg:392}
\absCDeri{\G}{\reg}\lesssim_{\delta, k} u^{-2+C\delta}(\IE{k+\regl,\delta}{\tau_0})^{\half}
 \end{align}
 for $\delta$ sufficiently small. And the second term in the third last line is bounded by $ Cv^{(2\sfrak+3)-\alpha(2\sfrak+4)+\delta}(\IE{k+\regl}{\tau_0})^{\half}$ by using the estimates \eqref{eq:PWD:pluss:0:tildePhi:428}.

For the integral in the second  last line of equation \eqref{eq:integrateconstv:123}, its absolute value is  bounded by
\begin{align}\label{prove:tildePhi:decay:345}
&\int_{v}^{+\infty}\int_{u_{\Sigmazero}(v')}^{u}\Big\{v^{2\sfrak+2}r^{-2\sfrak-4}
(vr^{-1}+\log(r-r_+)+u)|\G|
+v^{2\sfrak+3}r^{-2\sfrak-3}|V \G|\Big\}\di u'\di v'.
\end{align}
The first part of \eqref{prove:tildePhi:decay:345} can be estimated by using the decay estimates in Corollary \ref{cor:PWD:pms:012} and Proposition \ref{prop:AP:pms:012}, that is,
\begin{align}
&\int_{v}^{+\infty}\int_{u_{\Sigmazero}(v')}^{u}v^{2\sfrak+2}r^{-2\sfrak-4}
\big(vr^{-1}+\log(r-r_+)+u\big)|\G|
\di u'\di v'\notag\\
&\lesssim \int_{v}^{+\infty}\int_{u_{\Sigmazero}(v')}^{u}v^{2\sfrak+2-\alpha(2\sfrak+3)}
(v^{1-\alpha}+u)r^{-1}|\G|\di u'\di v'\notag\\
&\lesssim_{\delta} (\IE{\regl,\delta}{\tau_0})^{\half}\int_{v}^{+\infty}\int_{u_{\Sigmazero}(v')}^{u}v^{2\sfrak+1-\alpha(2\sfrak+3)}
(v^{1-\alpha}+u)u^{-2+\delta}\di u'\di v'\notag\\
&\lesssim_{\alpha,\delta} {}(v^{3+2\sfrak-\alpha(2\sfrak+4)}+v^{2\sfrak+2-\alpha(2\sfrak+3)+2\delta})(\IE{\regl,\delta}{\tau_0})^{\half}.
\end{align}
To estimate the remaing part in \eqref{prove:tildePhi:decay:345},
by applying Sobolev inequality \eqref{eq:Sobolev:2} to \eqref{eq:ED:S1:modes:Psipluss:1and2:p02:75}, we get
\begin{align}\label{pointwise:decay:r:hatPhi:1}
\sum_{\ell=\sfrak+1}^{\sfrak+2}\absCDeri{\Lxi^j(\ellmode{\hatPhiplussHigh{1}}{m,\ell})}{\reg}
\lesssim_{k,\delta}\tau^{-\half+C\delta-j}(\IE{\reg+\regl,\delta}{\tau_0})^{\half}.
\end{align}
 Thus, combining with \eqref{definition:righthand:tildephi:highG}, \eqref{pointwise:decay:r:hatPhi:1} and Corollary \ref{cor:PWD:pms:012}, we obtain
\begin{align}\label{prove:sharpdecay:tildephi}
&\int_{v}^{+\infty}\int_{u_{\Sigmazero}(v')}^{u}
v^{2\sfrak+3}r^{-2\sfrak-3}|V \G|\di u'\di v'\notag\\
&\lesssim \int_{v}^{+\infty}\int_{u_{\Sigmazero}(v')}^{u}
v^{2\sfrak+3}r^{-2\sfrak-5}\Big(\Lxi^{\leq1}(rV)^{\leq 1}\tildePhisHighell{\sfrak}{m,\sfrak}
+\sum_{\ell=\sfrak+1}^{\sfrak+2}\Lxi^{\leq1}(rV)^{\leq 1}\ellmode{\hatPhiplussHigh{1}}{m,\ell}+
\Lxi^{\leq1}\ellmode{\Phipluss}{m,\ell\leq \sfrak+2}\Big)\di u'\di v'\notag\\
&\lesssim_{\delta} (\IE{\regl,\reg}{\tau_0})^{\half}\int_{v}^{+\infty}\int_{u_{\Sigmazero}(v')}^{u}
v^{2\sfrak+3}r^{-2\sfrak-5}\Big(v^{-1+\delta}+u^{-\half+C\delta}
+u^{-2+\delta}\Big)\di u'\di v'\notag\\
&\lesssim_{\alpha,\delta} {}\Big(v^{2\sfrak+4-\alpha(2\sfrak+5)+2\delta}+
v^{2\sfrak+\frac{9}{2}-\alpha(2\sfrak+\frac{11}{2}+C\delta)}+v^{2\sfrak+4-\alpha(2\sfrak+5)}\Big)(\IE{\regl,\delta}{\tau_0})^{\half}.
\end{align}
 In summary, by taking $\delta$ sufficiently small and $\alpha$ (depending on the value of $\delta)$ sufficiently close to $1$, the integral in the second  last line of equation \eqref{eq:integrateconstv:123} is bounded by $v^{-\veps}$ for some small $\veps$.

For the integral in the last line of equation \eqref{eq:integrateconstv:123}, by the estimate \eqref{eq:G:esti:reg:392} and inequality \eqref{asymp-factor}, it is bounded by
\begin{align}
&C\int_{\Sigmazero, \rho(\tau_0,v)}r^{-2}(\log(r-r_+))^{2\sfrak+3}|\G|(\tau_0,\rho)\di\rho\notag\\
&\lesssim {} \rho(\tau_0,v)^{-1+\epsilon}(\IE{\regl,\delta}{\tau_0})^{\half}\lesssim v^{\alpha(-1+\epsilon)}(\IE{\regl,\delta}{\tau_0})^{\half}
\end{align}
for any $\epsilon\in(0,1)$.

Last, for the second term in the first line of \eqref{eq:integrateconstv:123}, by initial data assumption and \eqref{asymp-factor}, we have
\begin{align}
\big|\big(v^{2\sfrak+3} \mu^{\sfrak+1} (\R)^{-\sfrak-1}\tildePhisHighell{+\sfrak}{m,\sfrak}\big) (\tau_0,v)-2^{2\sfrak+3}\ql_{m,\sfrak}\big|\lesssim D_0 v^{-\veps}.
\end{align}
Combined with the above discussions, we achieve for $\delta$ sufficiently small and $\alpha$ sufficiently close to $1$, it holds in the region $r\geq v^\alpha$ that
\begin{align}
\label{Kerr:eq:VPhiplusHigh1:generalell:largeregion}
\Big|
(\R)^{-1}\mu\tildePhisHighell{+\sfrak}{m,\sfrak}
-(\R)^{\sfrak} v^{-2\sfrak-3}2^{2\sfrak+3}\MQ_{m,+\sfrak}
\Big|
\lesssim_{\delta, \alpha}{}v^{-3}(v^{-\veps}+u^{-1+\delta})((\IE{\regl,\delta}{\tau_0})^{\half}+D_0),
\end{align}
with
\begin{align}
\MQ_{m,\sfrak}=\ql_{m,\sfrak}
 -\half\int_{\tau_0}^{+\infty}
 \lim\limits_{\rho\to+\infty}\G(\tau,\rho)\di \tau.
\end{align}

{\it Step 2. Asymptotics of $\ellmode{\Phipluss}{m,\sfrak}$.}
We first recall the definition of $\tildePhisHighell{+\sfrak}{m,\sfrak}$ in Proposition \ref{prop:wavesys:tildePhisHighi:ellmode}:
\begin{align}\label{kerr:ansatz:tildePhisHigh:ellmode}
\frac{1}{\R}\mu\tildePhisHighell{+\sfrak}{m,\sfrak}=&V\ellmode{\Phipluss}{m,\sfrak}
-\frac{1}{2(\R)}\mu\Big(2iam\ellmode{\Phipluss}{m,\sfrak}\notag\\
&\qquad+a^2\sum_{\sfrak\leq\ell\leq \sfrak+2}c_{m,\ell}^{\sfrak}\Lxi\ellmode{\Phipluss}{m,\ell}
-2ia\sfrak\sum_{\sfrak\leq\ell\leq\sfrak+1}b_{m,\ell}^\sfrak\ellmode{\Phipluss}{m,\ell}\Big),
\end{align}
where we have used the mode projection Proposition \ref{prop:modeprojection:1}.
Together with \eqref{Kerr:eq:VPhiplusHigh1:generalell:largeregion}  and the almost sharp decay estimates in Proposition \ref{prop:AP:pms:012}, this yields
\begin{align}
\label{eq:wave:VRPhipluss:ell=sfrak:zero}\begin{split}
\Big|V\ellmode{\Phipluss}{m,\sfrak} - 2^{2\sfrak+3}\MQ_{m,\sfrak}\frac{(\R)^{\sfrak}}{v^{2\sfrak+3}}\Big|
\lesssim_{\delta,\alpha}{}\big(v^{-3}(v^{-\veps}+u^{-1+\delta})+ v^{-1-\alpha}u^{-2+\delta}\big)(\IE{\regl,\delta}{\tau_0})^{\half}.
\end{split}\end{align}

We then derive the asymptotic profile of $\ellmode{\Phipluss}{m,\sfrak}$. To obtain the asymptotics for $\ellmode{\Phipluss}{m,\sfrak}$,
one integrates along $u=const$ and utilizes \eqref{exp:pupv} to obtain
\begin{align}\label{eq-v-hatPhips}
\hspace{4ex}&\hspace{-4ex}
\ellmode{\Phipluss}{m,\sfrak}(u,v)\notag\\
={}&\ellmode{\Phipluss}{m,\sfrak}(u,v_{\gamma_\alpha}(u))
+\half \int_{v_{\gamma_\alpha}(u)}^{v}
 (V-\frac{2iam}{\R})\ellmode{\Phipluss}{m,\sfrak} (u,v')\di v'\notag\\
 ={}&\ellmode{\Phipluss}{m,\sfrak}(u,v_{\gamma_\alpha}(u))
+\half \int_{v_{\gamma_\alpha}(u)}^{v}
 \Big((V-\frac{2iam}{\R})\ellmode{\Phipluss}{m,\sfrak}
- 2^{2\sfrak+3}\MQ_{m,\sfrak}\frac{(\R)^{\sfrak}}{v^{2\sfrak+3}}\Big)(u,v')\di v'\notag\\
&
+2^{2\sfrak+2}\MQ_{m,\sfrak}\int_{v_{\gamma_\alpha}(u)}^{v}
 \frac{(\R)^{\sfrak}}{v^{2\sfrak+3}}(u,v')\di v'.
 \end{align}

 For the last line of \eqref{eq-v-hatPhips}, one has by \eqref{eq:rela:c} that
\begin{align}
\hspace{4ex}&\hspace{-4ex}\int_{v_{\gamma_\alpha}(u)}^{v}
 \frac{(\R)^{\sfrak}}{v^{2\sfrak+3}}(u,v')\di v'\notag\\
=& 2^{-2\sfrak}\int_{v_{\gamma_\alpha}(u)}^{v}
 \frac{(v-u)^{2\sfrak}}{v^{2\sfrak+3}}(u,v')\di v'+O(1)\int_{v_{\gamma_\alpha}(u)}^{v}
 \frac{r^{2\sfrak-1}\log r}{v^{2\sfrak+3}}(u,v')\di v',
\end{align}
and a simple calculation yields
\begin{align}
\int\frac{(v-u)^{2\sfrak}}{v^{2\sfrak+3}}(u,v)\di v
=&\sum_{j=0}^{2\sfrak}C_{2\sfrak}^j\frac{1}{j-2\sfrak-2}v^{j-2\sfrak-2}(-u)^{2\sfrak-j}\notag\\
=&\frac{1}{(2\sfrak+2)(2\sfrak+1)}
\sum_{j=0}^{2\sfrak}C_{2\sfrak+2}^j(j-2\sfrak-1)v^{j-2\sfrak-2}(-u)^{2\sfrak-j}\notag\\
=&\frac{1}{(2\sfrak+2)(2\sfrak+1)}
\partial_v\Big(\sum_{j=0}^{2\sfrak}C_{2\sfrak+2}^jv^{j-2\sfrak-1}(-u)^{2\sfrak-j}\Big)\notag\\
=&\frac{1}{(2\sfrak+2)(2\sfrak+1)}u^{-2}
\partial_v\Big(\sum_{j=0}^{2\sfrak}C_{2\sfrak+2}^j(\frac{-u}{v})^{2\sfrak+2-j}v\Big)\notag\\
=&\frac{1}{(2\sfrak+2)(2\sfrak+1)}u^{-2}
\partial_v\Big(v(\frac{v-u}{v})^{2\sfrak+2}+\frac{u}{2\sfrak+2}-v\Big)\notag\\
=&\frac{1}{(2\sfrak+2)(2\sfrak+1)}\Big( (\frac{v-u}{v})^{2\sfrak+2}-1\Big)\frac{1}{u^2}+
\frac{1}{2\sfrak+1}\frac{(v-u)^{2\sfrak+1}}{v^{2\sfrak+2}u}.
\end{align}
Thus, we conclude
\begin{align}\begin{split}
&\bigg|\int_{v_{\gamma_\alpha}(u)}^{v}
 \frac{(\R)^{\sfrak}}{v^{2\sfrak+3}}(u,v')\di v'-2^{-2\sfrak}
 \bigg(\frac{1}{(2\sfrak+2)(2\sfrak+1)}\frac{(v-u)^{2\sfrak+2}}{v^{2\sfrak+2}u^2}+
\frac{1}{2\sfrak+1}\frac{(v-u)^{2\sfrak+1}}{v^{2\sfrak+2}u}\bigg)\bigg|\\
&\lesssim{} (v^{(\alpha-1)(2\sfrak+2)}u^{-2}+v^{(\alpha-1)(2\sfrak+1)-1}u^{-1}
+v^{-3+\veps})|_{v_{\gamma_\alpha}(u)}\\
&\lesssim {}u^{-2}v^{(\alpha-1)(2\sfrak+1)}.
\end{split}\end{align}

By \eqref{eq:wave:VRPhipluss:ell=sfrak:zero},  the second last integral on the RHS of
\eqref{eq-v-hatPhips} is bounded by
\begin{align}
\big(v^{-2}(v^{-\veps}+u^{-1+\delta})+v^{-\alpha}u^{-2+\delta}+v^{-2\alpha+1}u^{-2+\delta}\big)
\big|_{\gamma_\alpha}(\IE{\regl,\delta}{\tau_0})^{\half}\lesssim_{\delta,\alpha} u^{-2-\veps}(\IE{\regl,\delta}{\tau_0})^{\half},
\end{align}
for $\delta$ sufficiently small.
For the first term on the RHS of
\eqref{eq-v-hatPhips}, by using Proposition \ref{prop:AP:pms:012},
\begin{align}
\abs{\ellmode{\Phipluss}{m,\sfrak}(u,v_{\gamma_\alpha}(u))}\lesssim{}& r^{2\sfrak+1} v^{-1-2\sfrak} u^{-2+\delta'}|_{\gamma_\alpha}(\IE{\regl,\delta'}{\tau_0})^{\half}\notag\\
\lesssim{}& v^{(\alpha-1)(2\sfrak+1)+2{\delta'}}u^{-2-{\delta'}}(\IE{\regl,\delta'}{\tau_0})^{\half}\notag\\
\lesssim {}&u^{-2-\delta'}(\IE{\regl,\delta'}{\tau_0})^{\half}
\end{align}
by taking $\delta'$ (depending on the value of $1-\alpha$) sufficiently small.
In summary, by letting $\alpha$ sufficiently close to $1$ and $\delta=\delta(\alpha)$ sufficiently small, there exists an $\veps>0$ such that
\begin{align}
\bigg|
\ellmode{\Phipluss}{m,\sfrak}
 -4\MQ_{m,\sfrak}\bigg(\frac{1}{(2\sfrak+2)(2\sfrak+1)}\frac{(v-u)^{2\sfrak+2}}{v^{2\sfrak+2}u^2}+
\frac{1}{2\sfrak+1}\frac{(v-u)^{2\sfrak+1}}{v^{2\sfrak+2}u}\bigg)
 \bigg|
 \lesssim_{\alpha,\delta}{}&
u^{-2-\veps}(\IE{\regl,\delta}{\tau_0})^{\half}.
\end{align}
Thus, we complete the proof.
\end{proof}

\subsubsection{Sharp decay for derivatives of $\ellmode{\Phipluss}{m,\sfrak}$ in $\{r\geq v^{\alpha}\}$.}

We proceed to derive the asymptotic profiles of the derivatives of $\ellmode{\Phipluss}{m,\sfrak}$ in $\{r\geq v^{\alpha}\}$.

\begin{lemma}
\label{Kerr:prop:asymp:tildePhiplussHighest}
Let $\sfrak=0, 1, 2$, $\abs{m}\leq \sfrak$, and  $j\in \mathbb{N}$. Let the initial data assumption \ref{assump:initialdata:nonvanishing:pm1} to order $j$ hold true, and let $j_1, j_2,j_3\in \mathbb{N}$ with $j_1+j_2+j_3\leq  j$. Let $\MQ_{m,\sfrak}$ be defined as in Lemma \ref{Kerr:prop:VtildePhil-1:nearinf:pm1}. Assume $\IE{\reg, \delta}{\tau_0}<+\infty$ for a sufficiently small $\delta>0$ and  some suitably large integer $k$ depending on $j$. Then for $\alpha\in (\half,1)$ sufficiently close to $1$ and $\delta=\delta(\alpha)>0$ sufficiently small, there exists an $\veps=\veps(\alpha,\delta)>0$ such that in the region $\{r \geq v^{\alpha}\}$,
\begin{align}\label{high:estimates:out:1234}
&\bigg|\Lxi^{j_1}\pv^{j_2}\pu^{j_3}\Big\{
(v-u)^{-2\sfrak-1}\ellmode{\Phipluss}{m,\sfrak}
 -4\MQ_{m,\sfrak}\frac{v+(2\sfrak+1)u}{(2\sfrak+1)(2\sfrak+2)v^{2\sfrak+2}u^2}\Big\}
 \bigg|\notag\\
 &\lesssim_{j,\alpha,\delta}{}\sum_{n=0}^{j}(v-u)^{-2\sfrak-1-j+n}u^{-2-n-\veps}((\IE{\reg,\delta}{\tau_0})^{\half}+D_0).
\end{align}
\end{lemma}

\begin{proof}
We divide the proof into four steps.

{\it Step 1.  Asymptotics of $\tildeV$ derivatives of $\tildePhisHighell{+\sfrak}{m,\sfrak}$.}
By commuting equation \eqref{eq:wave:tildePhisHighi:an:ellmode:zero} with $\tildeV^i$, and because of the commutators \eqref{eq:comms} and formula \eqref{definition:righthand:tildephi:highG}, we have
\begin{align}
\label{eq:wave:tildePhisHighi:an:ellmode:859}
\hspace{3.5ex}&\hspace{-3.5ex}-\mu Y \big( \tildeV^i (\mu^{\sfrak+1} (\R)^{-\sfrak-1}\tildePhisHighell{+\sfrak}{m,\sfrak})\big)\notag\\
={}&\tildeV^i(\mu^{\sfrak+1} (\R)^{-\sfrak-\frac{3}{2}}\G)\notag\\
={}&(-1)^i(\R)^{-\sfrak-\frac{3}{2}-\frac{i}{2}}\bigg(\frac{(2\sfrak+3+i-1)!}{(2\sfrak+2)!}\G
+O(1)\sum_{n=1}^i (rV)^n \G\bigg)\notag\\
={}&(-1)^i(\R)^{-\sfrak-\frac{3}{2}-\frac{i}{2}}\bigg\{\frac{(2\sfrak+3+i-1)!}{(2\sfrak+2)!}\G\notag\\
&+O(r^{-1})\Lxi^{\leq1}(rV)^{\leq i}\tildePhisHighell{+\sfrak}{m,\sfrak}
+O(r^{-1})\sum_{\ell=\sfrak+1}^{\sfrak+2}\Lxi^{\leq1}(rV)^{\leq i}\ellmode{\hatPhiplussHigh{1}}{m,\ell}+
O(r^{-1})\Lxi^{\leq1}\ellmode{\Phipluss}{m,\ell\leq \sfrak+2}\bigg\}.
\end{align}
Notice that the terms in the last line have faster decay in $r$ than the terms in the last second line by  \eqref{pointwise:decay:r:hatPhi:1} and Corollary \ref{cor:PWD:pms:012}.

Multiply on both sides of \eqref{eq:wave:tildePhisHighi:an:ellmode:859} by $v^{2\sfrak+3+i}$ and integrate along constant $v$ from the initial hypersurface $\Sigmazero$. We apply the same steps used in Step 1 of the proof of Lemma \ref{Kerr:prop:VtildePhil-1:nearinf:pm1} and arrive at
\begin{align}
\hspace{4ex}&\hspace{-4ex}\big(v^{2\sfrak+3+i} \tildeV^i(\mu^{\sfrak+1} (\R)^{-\sfrak-1}\tildePhisHighell{+\sfrak}{m,\sfrak})\big) (u,v)\notag\\
\hspace{4ex}&\hspace{-4ex}-\big(v^{2\sfrak+3+i} \tildeV^i(\mu^{\sfrak+1} (\R)^{-\sfrak-1}\tildePhisHighell{+\sfrak}{m,\sfrak})\big) (u_{\Sigmazero}(v),v)\notag\\
={}&
(-1)^{i+1}2^{2\sfrak+2+i}\frac{(2\sfrak+2+i)!}{(2\sfrak+2)!}\int_{\tau_0}^{+\infty}
\lim\limits_{\rho\to+\infty}\G(\tau,\rho)\di \tau+ (O(u^{-1+\delta})+O(v^{-\veps}))(\IE{\regl(i),\delta}{\tau_0})^{\half}
\end{align}
for $\regl=\regl(i)$ large enough and $\delta>0$ small enough.
Further, by the initial data assumption, we achieve for any $i\in \mathbb{N}$ that
\begin{align}\label{eq:tildePhisHighell:uv:value:higher:explicit}\begin{split}
\hspace{2ex}&\hspace{-2ex}
\Big|\tildeV^i(\mu^{\sfrak+1} (\R)^{-\sfrak-1}\tildePhisHighell{+\sfrak}{m,\sfrak})\big) (u,v)
-\partial_v^i(v^{-2\sfrak-3})2^{2\sfrak+3+i}\MQ_{m,\sfrak}\Big|\\
&\lesssim_{i,\delta,\alpha} {} v^{-2\sfrak-3-i}(v^{-\veps}+u^{-1+\delta})((\IE{\regl(i),\delta}{\tau_0})^{\half}+D_0).
\end{split}\end{align}

{\it Step 2. Asymptotics of $\partial_v^i\ellmode{\Phipluss}{m,\sfrak}$.}
We substitute \eqref{kerr:ansatz:tildePhisHigh:ellmode} to \eqref{eq:tildePhisHighell:uv:value:higher:explicit} with $i=0$. Combined with the basic calculation
\begin{align}
\partial_v\Big\{\frac{1}{(2\sfrak+2)(2\sfrak+1)}\frac{(v-u)^{2\sfrak+2}}{v^{2\sfrak+2}u^2}+
\frac{1}{2\sfrak+1}\frac{(v-u)^{2\sfrak+1}}{v^{2\sfrak+2}u}\Big\}
=\frac{(v-u)^{2\sfrak}}{v^{2\sfrak+3}},
\end{align}
the estimate \eqref{esti:PL:nv:p:ell-sfrak} and the expression $\partial_v=\half V-\frac{a}{\R}\Leta$ by \eqref{exp:pupv}, we achieve
\begin{align}\begin{split}
&\bigg|\partial_v\Big\{
(v-u)^{-2\sfrak-1}\ellmode{\Phipluss}{m,\sfrak}
 -4\MQ_{m,\sfrak}\frac{v+(2\sfrak+1)u}{(2\sfrak+2)(2\sfrak+1)v^{2\sfrak+2}u^2}\Big\}
 \bigg|\\
& \lesssim_{\alpha,\delta}{}\bigg(
(v-u)^{-2\sfrak-2}v^{-\veps}u^{-2-\delta}+(v-u)^{-1}v^{-2\sfrak-3}(v^{-\veps}+u^{-1+\delta})
\bigg)((\IE{\regl,\delta}{\tau_0})^{\half}+D_0).
\end{split}\end{align}
Further, by \eqref{eq:tildePhisHighell:uv:value:higher:explicit}, we have
\begin{align}
&(v-u)\bigg|
\tilde{V}^i\partial_v\Big\{
(v-u)^{-2\sfrak-1}\ellmode{\Phipluss}{m,\sfrak}
 -4\MQ_{m,\sfrak}\frac{v+(2\sfrak+1)u}{(2\sfrak+2)(2\sfrak+1)v^{2\sfrak+2}u^2}\Big\}
 \bigg|\notag\\
 &\lesssim \Big|
\tilde{V}^i\Big\{
(v-u)^{-2\sfrak-1}\ellmode{\Phipluss}{m,\sfrak}
 -4\MQ_{m,\sfrak}\frac{v+(2\sfrak+1)u}{(2\sfrak+2)(2\sfrak+1)v^{2\sfrak+2}u^2}\Big\}
 \Big|\notag\\
 &\quad
 +\big|\tilde{V}^i\big(O(r^{-2})\Lxi^{\leq1}\ellmode{\Phipluss}{m,\ell\leq\sfrak+2}\big)\big|
 +v^{-2\sfrak-3-i}(v^{-\veps}+u^{-1+\delta})((\IE{\regl,\delta}{\tau_0})^{\half}+D_0),
\end{align}
hence, we obtain via  a simple iteration that
\begin{align}
&\bigg|\partial_v^i\Big\{
(v-u)^{-2\sfrak-1}\ellmode{\Phipluss}{m,\sfrak}
 -4\MQ_{m,\sfrak}\frac{v+(2\sfrak+1)u}{(2\sfrak+2)(2\sfrak+1)v^{2\sfrak+2}u^2}\Big\}
 \bigg|\notag\\
 &\lesssim_{i,\alpha,\delta}{}
\Big((v-u)^{-2\sfrak-i-1}v^{-\veps}u^{-2-\delta}
+\sum_{j=1}^i(v-u)^{-j}v^{-2\sfrak-3-i+j}(v^{-\veps}+u^{-1+\delta})\Big)((\IE{\regl,\delta}{\tau_0})^{\half}+D_0).
\end{align}

{\it Step 3. Asymptotics of $\Lxi^i\ellmode{\Phipluss}{m,\sfrak}$.} Combining the estimate \eqref{eq:tildePhisHighell:uv:value:higher:explicit} and equation  \eqref{eq:wave:tildePhisHighi:an:ellmode:859}, and by   $2\Lxi=\mu Y + \tildeV$, we get
 \begin{align}\begin{split}
&\Big|2\Lxi\tildeV^{i-1}(\mu^{\sfrak+1} (\R)^{-\sfrak-1}\tildePhisHighell{+\sfrak}{m,\sfrak}) (u,v)
-\partial_v^{i}(v^{-2\sfrak-3})2^{2\sfrak+3+i}\MQ_{m,\sfrak}\Big|\\
&\lesssim_{i,\alpha,\delta} {}\bigg(v^{-2\sfrak-3-i}(v^{-\veps}+u^{-1+\delta})
+r^{-2\sfrak-2-i}u^{-2+\delta}\bigg)((\IE{\regl(i),\delta}{\tau_0})^{\half}+D_0).
\end{split}\end{align}
Repeating the above process yields
\begin{align}\label{kerr:high:derive:lxi:v-1}\begin{split}
&\Big|\Lxi^i\Big\{(\R)^{-\sfrak}\big[(\R)^{-1}\mu\tildePhisHighell{\sfrak}{m,\sfrak} (u,v)
-\frac{2^{2\sfrak+3}(\R)^{\sfrak}}{v^{2\sfrak+3}}\MQ_{m,\sfrak}\big]\Big\}\Big|\\
&\lesssim_{i,\alpha,\delta} {}\bigg(v^{-2\sfrak-3-i}(v^{-\veps}+u^{-1+\delta})
+\sum_{j=0}^{i-1}r^{-2\sfrak-3-j}u^{-2-(i-1-j)+\delta}\bigg)((\IE{\regl,\delta}{\tau_0})^{\half}+D_0),
\end{split}\end{align}
which is equivalent to
\begin{align}
&\bigg|\Lxi^i \bigg((\R)^{-1}\mu\tildePhisHighell{\sfrak}{m,\sfrak} - \frac{2^{2\sfrak+3}(\R)^{\sfrak}}{v^{2\sfrak+3}}\MQ_{m,\sfrak}\bigg)\bigg|\notag\\
&\lesssim_{i,\alpha,\delta}\bigg( r^{2\sfrak}v^{-2\sfrak-3-i}(v^{-\veps}+u^{-1+\delta})
+\sum_{j=0}^{i-1}r^{-3-j}u^{-2-(i-1-j)+\delta}\bigg)((\IE{\regl,\delta}{\tau_0})^{\half}+D_0).
\end{align}
Similar to Step 2, we combine the estimate  \eqref{kerr:ansatz:tildePhisHigh:ellmode} and the almost sharp pointwise decay estimates in Proposition \ref{prop:AP:pms:012} together to obtain
\begin{align}\begin{split}
&\bigg|\tildeV \Lxi^i\Big\{
\ellmode{\Phipluss}{m,\sfrak}
 -4\MQ_{m,\sfrak}\Big(\frac{1}{(2\sfrak+2)(2\sfrak+1)}\frac{(v-u)^{2\sfrak+2}}{v^{2\sfrak+2}u^2}+
\frac{1}{2\sfrak+1}\frac{(v-u)^{2\sfrak+1}}{v^{2\sfrak+2}u}\Big)\Big\}
 \bigg|\\
& \lesssim_{i,\alpha,\delta}{}
\bigg(v^{-3-i}(v^{-\veps}+u^{-1+\delta})+ v^{-1-\alpha}u^{-2-i+\delta}
+\sum_{j=0}^{i-1}v^{-\alpha(3+j)}u^{-1-i+j+\delta}\bigg)((\IE{\regl(i),\delta}{\tau_0})^{\half}+D_0).
\end{split}\end{align}
By integrating the above inequality along $u$-constant hypersurface from $\gamma_{\alpha}$, one has
\begin{align}\begin{split}
&\bigg|\Lxi^i\Big\{
\ellmode{\Phipluss}{m,\sfrak}
 -4\MQ_{m,\sfrak}\Big(\frac{1}{(2\sfrak+2)(2\sfrak+1)}\frac{(v-u)^{2\sfrak+2}}{v^{2\sfrak+2}u^2}+
\frac{1}{2\sfrak+1}\frac{(v-u)^{2\sfrak+1}}{v^{2\sfrak+2}u}\Big)\Big\}
 \bigg|\\
 &\lesssim_{i,\alpha,\delta} {}\Big(r^{2\sfrak+1}v^{-2\sfrak-1}u^{-2-i+\delta}|_{\gamma_{\alpha}}
 +\sum_{j=0}^i(v-u)^{2\sfrak+1} u^{-1-j}v^{-2\sfrak-2+j-i}|_{\gamma_{\alpha}}\\
 {}&
 \quad+v^{-2-i}(v^{-\veps}+u^{-1+\delta})+ v^{-\alpha}u^{-2-i+\delta}
+\sum_{j=0}^{i-1}v^{-\alpha(3+j)+1}u^{-1-i+j+\delta}\Big|_{\gamma_{\alpha}}\Big)((\IE{\regl(i),\delta}{\tau_0})^{\half}+D_0)\\
&\lesssim_{i,\alpha,\delta} {}u^{-2-i-\veps}((\IE{\regl(i),\delta}{\tau_0})^{\half}+D_0).
\end{split}\end{align}
Therefore, we achieve
\begin{align}\begin{split}
&\bigg|\Lxi^i\Big\{
(v-u)^{-2\sfrak-1}\ellmode{\Phipluss}{m,\sfrak}
 -4\MQ_{m,\sfrak}\frac{v+(2\sfrak+1)u}{(2\sfrak+2)(2\sfrak+1)v^{2\sfrak+2}u^2}\Big\}
 \bigg|\\
& \lesssim_{i,\alpha,\delta}{}(v-u)^{-2\sfrak-1}u^{-2-i-\veps}((\IE{\regl(i),\delta}{\tau_0})^{\half}+D_0).
\end{split}\end{align}

{\it Step 4. Asymptotics of $\Lxi^i\pv^j\pu^k\ellmode{\Phipluss}{m,\sfrak}$.} Similar to proving \eqref{kerr:high:derive:lxi:v-1}, we can derive the asymptotics for $\Lxi^j\tilde{V}^k$ derivatives, and these imply \eqref{high:estimates:out:1234} for $\Lxi^i\pv^j$. Finally, using $\pu=\Lxi-\pv$, we complete the proof.
\end{proof}


\subsubsection{Sharp decay for the spin $\pm\sfrak$ components in $\{r\geq v^{\alpha'}\}$.}

Given the above asymptotics for the spin $+\sfrak$ component, one can derive the asymptotics for the spin $-\sfrak$ component via the TSI in Section \ref{sect:TSI}. We state the asymptotics of both of the spin $\pm \sfrak$ components in region $\{r\geq v^{\alpha'}\}$, for some $\alpha'\in (\half, 1)$,  in the following theorem.

\begin{thm}[Asymptotics of the spin $\pm \sfrak$ component in $\{r\geq v^{\alpha'}\}$]\label{thm:PL:extregion}
Let $\sfrak=0, 1, 2$ and let $\abs{m}\leq \sfrak$. Let $\abs{a}/M<1$ for $\sfrak=0$ and let $\abs{a}/M\ll 1$ sufficiently small.  Let $j\in\mathbb{N}$ and $\abs{\mathbf{a}}= j$ and $\PriceDeri=\{\Lxi, \pu,\pv\}$. Let $\MQ_{m,\sfrak}$ be defined as in Lemma \ref{Kerr:prop:VtildePhil-1:nearinf:pm1}. Assume the initial data condition \ref{assump:initialdata:nonvanishing:pm1} to order $j+2\sfrak$ hold true, and $\IE{\reg,\delta}{\tau_0}<+\infty$ for a sufficiently small $\delta>0$  and  some suitably large integer $k$ depending on $j+2\sfrak$. Then, there exists an $\alpha'\in (\half, 1)$ sufficiently close to $1$ and an $\veps=\veps(\alpha',\delta)>0$ sufficiently small such that
in the region  $\{r \geq v^{\alpha'}\}$,
\begin{align}\label{thm:PL:anymodes:away:positive}
&\bigg|\PriceDeri^{\mathbf{a}}\bigg((\R)^{-\sfrak}\psipluss
-\frac{2^{2\sfrak+3}}{(2\sfrak+1)(2\sfrak+2)}\frac{v+(2\sfrak+1)\tb}{v^{2\sfrak+2}\tb^2}
\sum_{|m|\leq\sfrak}\MQ_{m,\sfrak}Y_{m,\sfrak}^{+\sfrak}(\cos\theta)e^{im\pb}\bigg)\bigg|\notag\\
&\lesssim_{j,\alpha',\delta} {} v^{-2\sfrak-1}\tb^{-2-j-\veps}((\IE{\reg,\delta}{\tau_0})^{\half}+D_0)
\end{align}
and
\begin{align}
\label{eq:asymp:psiminus:away:pm1}
&\bigg|\PriceDeri^{\mathbf{a}}\bigg({\psiminuss}
-\frac{2^{2\sfrak+3}}{(2\sfrak+1)(2\sfrak+2)}
\frac{\tb+(2\sfrak+1)v
}{\tb^{2\sfrak+2}v^2}
\sum_{|m|\leq\sfrak}\MQ_{m,\sfrak}Y_{m,\sfrak}^{-\sfrak}(\cos\theta)e^{im\pb}\bigg)\bigg|\notag\\
&\lesssim_{j,\alpha',\delta} {} v^{-1}\tb^{-2-j-2\sfrak-\veps}((\IE{\reg,\delta}{\tau_0})^{\half}+D_0).
\end{align}

Moreover, the above statement holds for $\abs{a}/M<1$ in the cases $\sfrak=1,2$ under the BEAM estimates assumption \ref{ass:BEAM:inhomogeneous}.
\end{thm}

\begin{proof}
Take $\alpha'\in(\alpha,1)$ to be determined. First, in the region $\{r\geq v^{\alpha'}\}\cap\{r\geq \frac{v}{4}\}$,  we have
\begin{align}
\sum_{k=0}^{j}(v-u)^{-2\sfrak-1-j+k}u^{-2-k-\veps}
\lesssim \sum_{k=0}^{j}v^{-2\sfrak-1-j+k}u^{-2-k-\veps}\lesssim v^{-2\sfrak-1}u^{-2-j-\veps}.
\end{align}
Next, in the region $\{r\geq v^{\alpha'}\}\cap\{r\leq \frac{v}{4}\}$, there exists an $\veps'=\veps'(\alpha')>0$ such that
\begin{align}
\sum_{k=0}^{j}(v-u)^{-2\sfrak-1-j+k}u^{-2-k-\veps}
\lesssim \sum_{k=0}^{j}v^{\alpha'(-2\sfrak-1-j+k)}v^{-2-k-\veps}\lesssim v^{-2\sfrak-3-j-\veps'},
\end{align}
by taking $\alpha'$ sufficiently close to $1$. Together with \eqref{high:estimates:out:1234} for the asymptotics of each $(m,\sfrak)$ mode and the pointwise decay estimates \eqref{eq:PWD:pluss:12:428} for $\geq \sfrak+1$ modes, the estimate  \eqref{thm:PL:anymodes:away:positive} follows.

It remains to consider the spin $-\sfrak$ component for $\sfrak=1,2$. As mentioned already, the asymptotics of the spin $-\sfrak$ component can be calculated explicitly from the  TSI \eqref{eq:otherTSI:simpleform} and \eqref{eq:TSIspin2:Y+2} and the already proven asymptotics of the spin $+\sfrak$ component.
The TSI \eqref{eq:otherTSI:simpleform} and \eqref{eq:TSIspin2:Y+2} for $\sfrak=1,2$ can be written as
\begin{align}
(\edthR+ia\sin\theta\Lxi)^{2\sfrak} \psiminuss +12 M(\sfrak-1)\overline{\Lxi\psiminuss}={}&Y^{2\sfrak} (\psipluss),
\end{align}
which can further be expanded and rewritten in the following form
\begin{align}\label{kerr:TSI:Y:1-2}
(\edthR)^{2\sfrak}\psiminuss={}&Y^{2\sfrak} (\psipluss)+\sum_{i_1\geq 1, i_1+i_2\leq 2\sfrak}O(1)\edthR^{i_2}\Lxi^{i_1}\psiminuss - 12M(\sfrak-1)\overline{\Lxi\psiminuss}.
\end{align}
The last two terms in TSI \eqref{kerr:TSI:Y:1-2} are with $\Lxi$-derivative and hence  have (at least) faster $\tb^{-1+\veps}$ decay than $\psiminuss$. Meanwhile,
one can expand $Y^{2\sfrak}\psipluss= 2^{2\sfrak}\pu^{2\sfrak}\psipluss+O(r^{-1})\pu^{2\sfrak}\psipluss+\sum\limits_{i\leq2\sfrak-1}O(r^{-2})\pu^{i}\psipluss$ by $\mu Y=2\partial_u$ from \eqref{exp:pupv} and the terms $O(r^{-1})\pu^{2\sfrak}\psipluss+\sum\limits_{i\leq2\sfrak-1}O(r^{-2})\pu^{i}\psipluss$ clearly have faster decay than the term $2^{2\sfrak}\pu^{2\sfrak}\psipluss$ in the region $\{r \geq v^{\alpha'}\}$. As a result,  by projecting  the above TSI \eqref{kerr:TSI:Y:1-2} onto an $(m,\sfrak)$ mode, one finds
\begin{align}
&\bigg|\PriceDeri^{\mathbf{a}}\ellmode{\psiminuss}{m,\sfrak}
-\frac{2^{2\sfrak}}{(2\sfrak)!}\PriceDeri^{\mathbf{a}}(\pu^{2\sfrak}
\ellmode{\psipluss}{m,\sfrak})\bigg|\notag\\
&={}\bigg|\PriceDeri^{\mathbf{a}}\ellmode{\psiminuss}{m,\sfrak}
-\frac{2^{2\sfrak}}{(2\sfrak)!}\PriceDeri^{\mathbf{a}}\pu^{2\sfrak}
((v-u)^{2\sfrak}(v-u)^{-2\sfrak}\ellmode{\psipluss}{m,\sfrak})\bigg|\notag\\
&\lesssim_{j,\alpha',\delta} v^{-1}u^{-2-j-2\sfrak-\veps}((\IE{\reg(j),\delta}{\tau_0})^{\half}+D_0).
\end{align}
In view of the estimate \eqref{thm:PL:anymodes:away:positive} and the pointwise decay estimates \eqref{eq:AP:minuss:12:428} for $\geq \sfrak+1$ modes of the spin $-\sfrak$ component,
this yields
\begin{align}
\label{eq:asymp:psiminus:away:pm1:compli}
&\bigg|\PriceDeri^{\mathbf{a}}{\psiminuss}
-\frac{2^{2\sfrak+3}}{(2\sfrak)!}
\PriceDeri^{\mathbf{a}}\pu^{2\sfrak}\Big(\frac{(v-u)^{2\sfrak}(v+(2\sfrak+1)u)
}{(2\sfrak+2)(2\sfrak+1)v^{2\sfrak+2}u^2}\Big)
\sum_{|m|\leq\sfrak}\MQ_{m,\sfrak}Y_{m,\sfrak}^{-\sfrak}(\cos\theta)e^{im\pb}\bigg|\notag\\
&\lesssim_{j,\alpha',\delta} {} v^{-1}u^{-2-j-2\sfrak-\veps}((\IE{\reg,\delta}{\tau_0})^{\half}+D_0).
\end{align}
In the end, by elementary calculations, one has
\begin{align}
\pu^{2\sfrak}\Big(\frac{(v-u)^{2\sfrak}(v+(2\sfrak+1)u)
}{v^{2\sfrak+2}u^2}\Big)={}&\frac{(2\sfrak)! ((2\sfrak+1)v +u)}{v^2 u^{2+2\sfrak}},
\end{align}
then substituting this into \eqref{eq:asymp:psiminus:away:pm1:compli}
proves \eqref{eq:asymp:psiminus:away:pm1}.
\end{proof}

\subsubsection{Sharp decay for the spin $\pm\sfrak$ components in the region $\{r\leq v^{\alpha'}\}$}
\label{subsect:nv:int:PL}

In contrast to the approach in the region $\{r\geq v^{\alpha'}\}$ that the asymptotics for the spin $+\sfrak$ component are first derived and the ones for the spin $-\sfrak$ component then follow from the TSI,  our argument begins with deriving the asymptotics for
 the spin $-\sfrak$ component, and these yield the asymptotics for the spin $+\sfrak$ component via the other TSI of Section \ref{sect:TSI}.

The asymtotics of the spin $\pm \sfrak$ components in the region $\{r\leq v^{\alpha'}\}$ are provided in the following theorem.

\begin{thm}[Asymptotics of the spin $\pm \sfrak$ component in $\{r\leq v^{\alpha'}\}$]
\label{thm:PL:anymodes:near:pm1}
Let $j\in\mathbb{N}$, and $\sfrak=0, 1, 2$. Let $\MQ_{m,\sfrak}$ be defined as in Lemma \ref{Kerr:prop:VtildePhil-1:nearinf:pm1}.  Let $\alpha'$ be chosen as in Theorem \ref{thm:PL:extregion}.  Assume for each $m$ with $\abs{m}\leq \sfrak$, the initial data assumption \ref{assump:initialdata:nonvanishing:pm1} to order $j+2\sfrak$ hold true, and $\IE{\reg,\delta}{\tau_0}<+\infty$ for a sufficiently small $\delta>0$ and a suitably large integer $k$ depending on $j+2\sfrak$.
Then, there exists an $\veps>0$ such that
in the region  $\{r \leq v^{\alpha'}\}$,
\begin{subequations}
\begin{align}
\label{eq:asymp:psiminus:near:int:pm1:d}
\hspace{3ex}&\hspace{-3ex}
\bigg|\Lxi^j\bigg(\psiminuss
-\sum_{|m|\leq\sfrak}\frac{2^{2\sfrak+3}}{(2\sfrak+1)}\MQ_{m,\sfrak}Y_{m,\sfrak}^{-\sfrak}(\cos\theta)e^{im\pb}
\tau^{-3-2\sfrak-j}\bigg)\bigg|\notag\\
&\lesssim_{j,\delta,\alpha'} {}\tau^{-3-2\sfrak-j-\veps}((\IE{\reg,\delta}{\tau_0})^{\half}+D_0),\\
\label{eq:asymp:psiplus:near:int:pm1:d}
\hspace{3ex}&\hspace{-3ex}\bigg|\Lxi^j\bigg((\R)^{-\sfrak}\psipluss
-\sum_{|m|\leq\sfrak}\mathfrak{f}_{+\sfrak,m}\frac{2^{2\sfrak+3}}{(2\sfrak+1)}\MQ_{m,\sfrak}Y_{m,\sfrak}^{+\sfrak}(\cos\theta)e^{im\pb}
\tau^{-3-2\sfrak-j}\bigg)\bigg|\notag\\
&\lesssim_{j,\delta,\alpha'} \tau^{-2\sfrak-3-j-\veps}((\IE{\reg,\delta}{\tau_0})^{\half}+D_0),
\end{align}
where
\begin{align}
\label{exp:f+sfrakm}
\mathfrak{f}_{+\sfrak,m}=&\mu^\sfrak+\frac{1}{(2\sfrak)!}
\sum_{n=1}^{2\sfrak}(C_{2\sfrak}^n-\delta_{-\sfrak,n})u_{-\sfrak, n}(r)\notag\\
&\qquad\qquad\qquad \times
(r-M)^{2\sfrak-n}(\R)^{n-\sfrak}\Big(\mu\prb+
\frac{2iam}{r^2+a^2}\Big)^{n-1}\Big(\frac{2iam}{r^2+a^2}\Big)
\end{align}
with $u_{-\sfrak, n}(r)$ and $\delta_{-\sfrak,n}$ as defined in Lemma \ref{lemma:expand:TST:righthand:V:short} and $f_{+\sfrak, m}=\mu^{\sfrak}+amO(r^{-1})$.

Further, if $\psipluss$ $(\sfrak\neq 0)$ is supported on an azimuthal $m$-mode, then on $\Horizon$,
\begin{align}
\label{eq:asymp:psiplus:horizon:pm1:d}
&\bigg|\Lxi^j\bigg(\psipluss\big|_{\Horizon}
-\frac{2^{2\sfrak+3}}{(2\sfrak+1)(2\sfrak)!}\MQ_{m,\sfrak}Y_{m,\sfrak}^{+\sfrak}(\cos\theta)e^{im\pb}
\sum_{n=1}^{2\sfrak}(C_{2\sfrak}^n-\delta_{-\sfrak,n})
 (r_+-M)^{2\sfrak-n}({2iam})^n\times \tb^{-2\sfrak-3}\bigg)\bigg|\notag\\
&\lesssim_{j,\delta,\alpha'} \tau^{-2\sfrak-3-j-\veps}((\IE{\reg,\delta}{\tau_0})^{\half}+D_0),
\end{align}
and for $am= 0$, the decay is faster by $\tb^{-1}$:
\begin{align}
\label{eq:asymp:psiplus:horizon:pm1:d:v2}
\big|\Lxi^j\big({\psipluss}\big|_{\Horizon}- D\MQ_{m,\sfrak}\tb^{-2\sfrak-4}Y_{m,\sfrak}^{+\sfrak}(\cos\theta)e^{im\pb}\big)
\big|
&\lesssim_{j,\delta,\alpha'} \tau^{-2\sfrak-4-j-\veps}((\IE{\reg,\delta}{\tau_0})^{\half}+D_0)
\end{align}
with the constants $D$ being explicitly calculated as in the proof.

Meanwhile, all the statements in this theorem are valid for $\abs{a}/M<1$ in the cases $\sfrak=1,2$ under the BEAM estimates assumption \ref{ass:BEAM:inhomogeneous}.
\end{subequations}
\end{thm}

\begin{proof}
Consider first the spin $-\sfrak$ component $\ellmode{\psiminuss}{m,\sfrak}$. We have achieved in Proposition \ref{prop:AP:pms:012} that
\begin{align}
\label{eq:pfofthm:PL:anymodes:near:pm1:12}
\abs{\Lxi^j\prb\ellmode{\psiminuss}{m,\sfrak}}
\lesssim_{j,\delta}{}&v^{-1}\tb^{-2\sfrak-3-j+\delta}(\IE{\reg(j),\delta}{\tau_0})^{\half},
\end{align}
for $\delta$ sufficiently small.
For any point $(\tb,\rb')\in\{r \leq v^{\alpha'}\}$, we integrate  $\Lxi^j\prb\ellmode{\psiminuss}{m,\sfrak}$ from point $(\tb,\rb')$ along constant $\tb$ up to the intersection point with the curve $\gamma_{\alpha'}$, thus, it holds
\begin{align}
\label{eq:pfofthm:PL:anymodes:near:pm1:1}
\Lxi^j\ellmode{\psiminuss}{m,\sfrak}(\tb,\rb')-\Lxi^j\ellmode{\psiminuss}{m,\sfrak}\big\vert_{\gamma_{\alpha'}}
&\leq {}\bigg|\int_{\rb'}^{\rb_{\gamma_{\alpha}}(\tb)}\Lxi^j\prb\ellmode{\psiminuss}{m,\sfrak}\di\rb\bigg|\notag\\
&\lesssim_{j,\delta}{}(\IE{\reg(j),\delta}{\tau_0})^{\half}\tb^{-2\sfrak-4-j+\delta}
\int_{\rb'}^{\rb_{\gamma_{\alpha}}(\tb)}\di\rb\notag\\
&\lesssim_{j,\delta,} \tb^{-2\sfrak-4-j+\delta+\alpha'}(\IE{\reg(j),\delta}{\tau_0})^{\half}\notag\\
&\lesssim_{j,\delta,\alpha'} \tb^{-2\sfrak-3-j-\veps}(\IE{\reg(j),\delta}{\tau_0})^{\half},
\end{align}
where we have used
\eqref{eq:pfofthm:PL:anymodes:near:pm1:12}  in the second step and chosen
$\delta$ small enough in the last step.
By the sharp decay estimate \eqref{eq:asymp:psiminus:away:pm1}, one has \begin{align}
\Lxi^j\ellmode{\psiminuss}{m,\sfrak}\big\vert_{\gamma_{\alpha'}}
={}&\frac{2^{2\sfrak+3}}{(2\sfrak)!}\MQ_{m,+\sfrak}
\Lxi^j\pu^{2\sfrak}\Big(\frac{(v-u)^{2\sfrak}
(v+(2\sfrak+1)u)}{(2\sfrak+2)(2\sfrak+1)v^{2\sfrak+2}u^2}\Big)
\bigg\vert_{\gamma_{\alpha'}}\notag\\
={}&\frac{2^{2\sfrak+3}\MQ_{m,+\sfrak}}{(2\sfrak+2)(2\sfrak+1)}
\Lxi^j\bigg(\frac{v+(2\sfrak+1)u}{v^{2\sfrak+2}u^{2}}\bigg)\bigg\vert_{\gamma_{\alpha'}}
+O(\tb^{-2\sfrak-4+\alpha'})\notag\\
={}&\frac{(-1)^j2^{4\sfrak+3}(2\sfrak+2+j)!\MQ_{m,+\sfrak}}{(2\sfrak+1)(2\sfrak+2)!}\tb^{-2\sfrak-3-j}
+O(\tb^{-2\sfrak-4+\alpha'})
\end{align}
where we have used $v=\tau+O(\tau^{\alpha'})$ and $u=\tau +O(\tau^{\alpha'})$ on $\gamma_\alpha'$.
Substituting this back into \eqref{eq:pfofthm:PL:anymodes:near:pm1:1}, and in view of the faster decay estimates \eqref{eq:AP:minuss:12:428} for $\geq \sfrak+1$ modes, we hence prove \eqref{eq:asymp:psiminus:near:int:pm1:d} in $\{r\leq v^{\alpha'}\}$.

Consider next the spin $+\sfrak$ component. We can obtain its asymptotics by utilizing the TSI \eqref{eq:TSI:simpleform} and \eqref{eq:TSIspin2:V-2} and the above estimates for the spin $-\sfrak$ component. Recall the TSI for $\sfrak=1,2$:
\begin{align}\
(\edthR'-ia\sin\theta\Lxi)^{2\sfrak} \psipluss -12 M(\sfrak-1)\overline{\Lxi\psipluss}={}&\Delta^{\sfrak} \VR^{2\sfrak}(\Delta^{\sfrak}\psiminuss),
\end{align}
which can again be expanded and written as
\begin{align}\label{TSI:spin:V:1:2}
(\edthR')^{2\sfrak}\psipluss={}&\Delta^{\sfrak} \VR^{2\sfrak}(\Delta^{\sfrak}\psiminuss)+\sum_{i_1\geq 1, i_1+i_2\leq 2\sfrak}O(1)(\edthR')^{i_2}\Lxi^{i_1}\psipluss +12M(\sfrak-1)\overline{\Lxi\psipluss}.
\end{align}

By \eqref{eq:TSIRHS:expand:short}, we can expand out the $(m,\sfrak)$ mode of the first term on the RHS as follows:
\begin{align}\label{expand:TSI:V:1:2:righthand}
\Proj{m,\sfrak}(\Delta^{\sfrak} \VR^{2\sfrak}(\Delta^{\sfrak}\psiminuss))
=&\sum_{j=0}^{2\sfrak} \tilde{u}_{-\sfrak, j}(r)\Proj{m,\sfrak}(V^j \psiminuss)\notag\\
=&\sum_{j=0}^{2\sfrak} \tilde{u}_{-\sfrak, j}(r)(\mu\prb+\mu\Hhyp\Lxi+
\frac{2iam}{r^2+a^2})^j\ellmode{\psiminuss}{m,\sfrak}\notag\\
=&\Big((\R)^{\sfrak}u_{-\sfrak,0}(r)+
\sum_{j=1}^{2\sfrak} \tilde{u}_{-\sfrak, j}(r)(\mu\prb+
\frac{2iam}{r^2+a^2})^{j-1}\frac{2iam}{r^2+a^2}\Big)\cdot\ellmode{\psiminuss}{m,\sfrak}\notag\\
&+\sum_{j+k\leq 2\sfrak-1}b_{j,k}\Lxi^j(\mu r\prb)^k\prb\ellmode{\psiminuss}{m,\sfrak}
+\sum_{j=1}^{2\sfrak}c_j\Lxi^j\ellmode{\psiminuss}{m,\sfrak}.
\end{align}
The last line of \eqref{expand:TSI:V:1:2:righthand} has faster $\tau^{-1}$ decay by the decay estimates in Proposition \ref{prop:AP:pms:012};
moreover,  the last two terms with $\Lxi$-derivatives on the RHS of  \eqref{TSI:spin:V:1:2} also have  faster decay in $\tb$. Thus, a projection onto $(m,\sfrak)$ mode for the TSI \eqref{TSI:spin:V:1:2} yields
\begin{align}
\label{eq:TSI:precise:ap:generals}
&\bigg|\Lxi^j((\R)^{-\sfrak}\ellmode{\psipluss}{m,\sfrak})
-\frac{1}{(2\sfrak)!} \Big(u_{-\sfrak,0}(r)+
(\R)^{-\sfrak}\sum_{i=1}^{2\sfrak} \tilde{u}_{-\sfrak, i}(r)(\mu\prb+
\frac{2aim}{r^2+a^2})^{i-1}\frac{2aim}{r^2+a^2}\Big)\Lxi^j\ellmode{\psiminuss}{m,\sfrak}\bigg|\notag\\
&\lesssim_{j,\delta,\alpha'} \tau^{-2\sfrak-3-j-\veps}((\IE{\reg(j),\delta}{\tau_0})^{\half}+D_0).
\end{align}
Substituting the asymptotic estimate \eqref{eq:asymp:psiminus:near:int:pm1:d} of $\ellmode{\psiminuss}{m,\sfrak}$ and the definition of $\tilde{u}_{-\sfrak, i}(r)$ in Lemma \ref{lemma:expand:TST:righthand:V:short}  into the above inequality and by the decay estimates \eqref{eq:AP:pluss:12:428} for the $\geq \sfrak+1$ modes of $\psipluss$, we obtain \eqref{eq:asymp:psiplus:near:int:pm1:d}. Meanwhile, it is manifest by the expression \eqref{exp:f+sfrakm} of $\mathfrak{f}_{+\sfrak, m}$ that $\mathfrak{f}_{+\sfrak, m}=\mu^{\sfrak}+amO(r^{-1})$.

Last, we discuss the sharp decay on the event horizon for the spin $+\sfrak$ components for $\sfrak=1,2$. Restricting \eqref{TSI:spin:V:1:2} on $\Horizon$, see for example \eqref{ellmode:TSI:onHorizon}, then we get
\begin{align}
\hspace{2ex}&\hspace{-2ex}
\bigg|\Lxi^j\ellmode{\psipluss}{m,\sfrak}-
\frac{1}{(2\sfrak)!}\sum_{n=1}^{2\sfrak}(C_{2\sfrak}^n-\delta_{-\sfrak,n})(r_+-M)^{2\sfrak-n}
 (2iam)^n\ellmode{\psiminuss}{m,\sfrak}\bigg|_{\Horizon}\notag\\
&\lesssim_{j,\delta,\alpha'}{} \tau^{-2\sfrak-3-j-\veps}((\IE{\reg(j),\delta}{\tau_0})^{\half}+D_0).
\end{align}
Substituting \eqref{eq:asymp:psiminus:near:int:pm1:d} into the above inequality and by the decay estimates \eqref{eq:AP:pluss:12:428} for the $(m,\geq \sfrak+1)$ modes of $\psipluss$, we achieve \eqref{eq:asymp:psiplus:horizon:pm1:d}.  Further, it can be easily check that the coefficient $\sum_{n=1}^{2\sfrak}(C_{2\sfrak}^n-\delta_{-\sfrak,n})(r_+-M)^{2\sfrak-n}
 (2iam)^n$ in fact vanishes if and only if $am=0$. Therefore, we have in the case $am=0$ that
 \begin{align}
\label{eq:asymp:psiplus:horizon:pm1:d:v5}
\big|\Lxi^j\psipluss\big|_{\Horizon}
\big|
&\lesssim_{j,\delta,\alpha'} \tau^{-2\sfrak-3-j-\veps}((\IE{\reg(j),\delta}{\tau_0})^{\half}+D_0).
\end{align}
We can use again the TSI \eqref{ellmode:TSI:onHorizon} and substitute in the above estimate, and this then yields
\begin{align}
\label{eq:asymp:psiplus:horizon:pm1:d:v8}
\big|\Lxi^j\big(\mellmode{\psipluss}{m}{\sfrak}\big|_{\Horizon}- D\MQ_{m,\sfrak}\tb^{-2\sfrak-4}Y_{m,\sfrak}^{+\sfrak}(\cos\theta)e^{im\pb}\big)
\big|
&\lesssim_{j,\delta,\alpha'} \tau^{-2\sfrak-4-j-\veps}((\IE{\reg(j),\delta}{\tau_0})^{\half}+D_0)
\end{align}
where $D$ can be calculated explicitly from the TSI \eqref{ellmode:TSI:onHorizon}.
By projecting the TSI \eqref{TSI:spin:V:1:2} on $(m,\geq \sfrak+1)$ modes and restricting on $\Horizon$, one finds the last two terms on the RHS are $O(\tb^{-2\sfrak-\frac{9}{2}+C\delta})$ by the above estimate \eqref{eq:asymp:psiplus:horizon:pm1:d:v8} and the decay estimates \eqref{eq:AP:pluss:12:428} for the $(m,\geq \sfrak+1)$ modes of $\psipluss$, and the first term on the RHS is  bounded as well by $O(\tb^{-2\sfrak-\frac{9}{2}+C\delta})$   using \eqref{eq:TSI:RHS:Horizon:exp}  and the decay estimates \eqref{eq:AP:minuss:12:428}. Hence, by taking $\delta$ sufficiently small, we arrive at
\begin{align}
\label{eq:asymp:psiplus:horizon:pm1:d:v9}
\big|\Lxi^j\mellmode{\psipluss}{m}{\geq \sfrak+1}\big|_{\Horizon}\big|
&\lesssim_{j,\delta,\alpha'} \tau^{-2\sfrak-4-j-\veps}((\IE{\reg(j),\delta}{\tau_0})^{\half}+D_0).
\end{align}
Combining this estimate with \eqref{eq:asymp:psiplus:horizon:pm1:d:v8} then proves
the estimate \eqref{eq:asymp:psiplus:horizon:pm1:d:v2}.
\end{proof}

\section*{Acknowledgement}

The first author S. M. acknowledges the support by the ERC grant ERC-2016 CoG 725589 EPGR. The second author L. Z. acknowledges the support by the National Natural Science Foundation of China (Grant No. 12201083). The authors are grateful to the anonymous referees for many valuable comments and suggestions.

\appendix

\section{Scalars constructed from the spin $\pm \sfrak$ components}
\label{app:scalars}

For the sake of convenience, we provide in the following table a list of scalars that are constructed or defined from the spin $\pm \sfrak$ components $\NPR_s$ in this work. 

\begin{table}[htbp]
\begin{center}
\begin{tabular}{l|ll}
&$s=+\sfrak$   & $s=-\sfrak$ \\
\hline
$\psi_s$ & $\Sigma^{\sfrak}\NPRpluss$ as in \eqref{eq:ssc} & $\Sigma^{-\sfrak}(r-ia\cos\theta)^{2\sfrak}\NPRminuss$ as in \eqref{eq:ssc}\\
$\Psi_s$ & $\sqrt{r^2+a^2}\psipluss$ as in \eqref{definition:radiationfield:kerr:spin} &  $\sqrt{r^2+a^2}\psiminuss$ as in \eqref{definition:radiationfield:kerr:spin}\\
$\PhisHigh{0}$ & $\mu^{-\sfrak}\Psi_{+\sfrak}$ as in \eqref{def:PhisHigh} &  $\mu^{\sfrak}\Psi_{-\sfrak}$ as in \eqref{def:PhisHigh}\\
$\PhisHigh{i}$ & $\curlVR^{i}\Phi_{+\sfrak}^{(0)}$ as in \eqref{def:PhisHigh} & $\curlVR^{i}\Phi_{-\sfrak}^{(0)}$ as in \eqref{def:PhisHigh} \\
$\Xi^{(0)}_{s}$ & $(\R)^{-\sfrak}\Psipluss$ as in \eqref{eq:DefOfphi012PosiSpinS2}& \textbackslash\\
$\Xi^{(i)}_{s}$ & $(-(\R) Y)^i\Xi^{(0)}_{+\sfrak}$ as in \eqref{eq:DefOfphi012PosiSpinS2}& \textbackslash\\
$\dot{\Phi}_s^{(2\sfrak)}$ & \textbackslash & as in Definition \ref{def:dotPhisHighi}\\
$\hat{\Phi}_{s}^{(i)}$ & $\hat{\Phi}_{+\sfrak}^{(i)}$ as in \eqref{ansatz:hatPhisHigh}&
$\hat{\Phi}_{-\sfrak}^{(2\sfrak+i)}$ as in \eqref{ansatz:hatPhisHigh}\\
$\tildePhisHighell{s}{\ell}$ & as in Proposition \ref{prop:wavesys:tildePhisHighi:ellmode} &
as in Proposition \ref{prop:wavesys:tildePhisHighi:ellmode}\\
$\tildePhisHighmell{s}{m}{\ell}$ & as in Proposition \ref{prop:wavesys:tildePhisHighi:ellmode} 
&as in Proposition \ref{prop:wavesys:tildePhisHighi:ellmode}\\
$\ellmode{\tilde{\Phi}_{s}^{(j)}}{m,\ell}$ & as in \eqref{ansatz:tildePhisHigh:ellmode:i} & \textbackslash
\end{tabular}
\end{center}
\caption{Scalars constructed from the spin $\pm \sfrak$ components.}
\label{table:scalars}
\end{table}

Let us in the end remark that by Definition \ref{def:projectiondefinitions}, $\ellmode{\varphi_s}{\ell}$, $\ellmode{\varphi_s}{\geq \ell}$ and $\ellmode{\varphi_s}{m,\ell}$ are the $\ell$ mode, the $\geq \ell$ modes, the $(m,\ell)$ mode of an arbitrary spin $s$ scalar $\varphi_s$, respectively. This definition works for the scalars in Table \ref{table:scalars}.


\newcommand{\arxivref}[1]{\href{http://www.arxiv.org/abs/#1}{{arXiv.org:#1}}}
\newcommand{\mnras}{Monthly Notices of the Royal Astronomical Society}
\newcommand{\prd}{Phys. Rev. D}
\newcommand{\apj}{Astrophysical J.}

\bibliographystyle{amsplain}

\providecommand{\MR}{\relax\ifhmode\unskip\space\fi MR }
\providecommand{\MRhref}[2]{%
  \href{http://www.ams.org/mathscinet-getitem?mr=#1}{#2}
}
\providecommand{\href}[2]{#2}

\end{document}